%% file: main.tex
\documentclass[11pt]{article}

\usepackage{etoolbox}
\makeatletter
\patchcmd{\@part}
  {\thepart\hspace{1em}}
  {\partname~\thepart\hspace{1em}}
  {}{\PackageError{main}{Could not customize part entries in the table of contents}{}}
\makeatother

\usepackage{tcs}
\usepackage{tabularx}
\begin{document}
\hypersetup{pageanchor=false}
\pagenumbering{gobble}

\title{Algorithmic Aspects of the Fermi--Hubbard Model}
\date{}
\author{
\begin{tabularx}{0.8\textwidth}{@{}*{3}{>{\centering\arraybackslash}X}@{}}
Ainesh Bakshi & Ankur Moitra & Xinyu Tan \\
\texttt{ainesh@nyu.edu} & \texttt{moitra@mit.edu} & \texttt{norahtan@mit.edu} \\
NYU & MIT & MIT
\end{tabularx}
}

\maketitle

\begin{abstract}
Preparing Gibbs states on a quantum computer offers a way to computationally probe how microscopic interactions give rise to macroscopic phenomena at finite temperature. However, an uncomfortable tension has set in: our strongest general guarantees for Gibbs state preparation hold in regimes where the resulting states can also be simulated classically. We propose investigating physical models with competing energy scales. The Fermi--Hubbard model is a natural target: the competition between electron hopping, on-site repulsion and chemical potential encapsulates many mysteries of condensed matter physics, including Mott insulators, striping and superconductivity.

Our main result identifies a simple guiding principle: the temperature scale for efficient preparation is set by the couplings between sites, while the on-site energy scale is unrestricted. For the Fermi--Hubbard model on an $n$-site nearest-neighbor lattice of fixed dimension, we construct a quantum Gibbs sampler that mixes in $\bigO{\log(n/\epsilon)}$ time whenever the temperature is at least a constant factor larger than the hopping strength, but is allowed to be independent of the repulsion and chemical potential. We obtain analogous guarantees for local qubit systems with bounded intersite interactions and arbitrary external fields. As a consequence, our algorithms can efficiently prepare Fermi--Hubbard Gibbs states that are provably outside the convex hull of fermionic Gaussian states and qubit Gibbs states for locally interacting Hamiltonians outside the set of product states.
\end{abstract}

\newpage
\tableofcontents

\clearpage
\pagenumbering{arabic}
\hypersetup{pageanchor=true}
\input{intro}

\part{Spin systems}
\input{all_qubit}

\part{Fermionic systems}
\input{all_fermion}

\section*{Acknowledgments}
AB thanks Soonwon Choi and XT thanks Aram Harrow for helpful discussions. 
XT is supported by the U.S. Department of Energy, Office of Science, National Quantum Information Science Research Centers, Co-design Center for Quantum Advantage (C2QA) under contract number DE-SC0012704, and Phyllis Ruby Block Fellowship. The authors note that this work subsumes~\cite{BT26}.

\paragraph{AI disclosure.} 
We used ChatGPT Thinking 5.4 to assist with proving two technical lemmas on bounding the moments of an explicit kernel function $b_{1,j}$. The tool provided assistance with the tedious computations in the proofs of \Cref{lem:b_1_zero_moment,lem:b_1_higher_moments}.
We used GPT 6 Astra for constructing an appropriate example in \Cref{ex:non-gaussian-gibbs}. 
GPT 5.6 Sol and GPT 6 Astra were used to assist with the technical writing. 
The authors verified the correctness and originality of all the content and take full responsibility for the contents of this work.

\begingroup
\setlength{\emergencystretch}{1em}
\printbibliography
\endgroup

\appendix

\input{appendix_lieb_robinson}
\input{appendix_qubit}

\end{document}

%% file: intro.tex
\section{Introduction}

Quantum Gibbs states offer a computational route towards understanding quantum matter. For a system with a Hamiltonian $H$ at inverse temperature $\beta > 0$, the Gibbs state, $\sigma=e^{-\beta H}/\tr(e^{-\beta H})$,
describes the system at thermal equilibrium. There has been intensive recent study of quantum algorithms for preparing Gibbs states~\cite{rfa24, BC25,TZ25,SMBB25,SMBB26rapid,rfa24a,blmt2025dobrushin,BT26,Ber26}. The hope is that, by preparing these states and measuring observables and correlations, we can computationally probe how microscopic interactions give rise to macroscopic phenomena in condensed matter physics and materials science.

While there has been rapid progress, an uncomfortable tension has set in: the regimes in which we have the strongest general guarantees are precisely those in which the resulting Gibbs state can also be described classically. A striking example is the case of qubit interacting local Hamiltonians. Rouze, Franca and Alhambra~\cite{rfa24,rfa24a} showed that there is a dissipative evolution, that can be efficiently implemented on a quantum computer, that rapidly thermalizes to the Gibbs state at sufficiently high temperature. However Bakshi, Liu, Moitra and Tang~\cite{blmt24b}, by way of a classical algorithm, showed that such states are actually completely unentangled.

Similarly, for fermionic systems, Ramkumar, Cai, Tong and Jiang~\cite{RCTJ26} showed that at sufficiently high temperature Gibbs states are mixtures of fermionic Gaussian states and that there are classical algorithms for estimating their observables. For any fixed temperature, and thus at low temperature too, Tong and Zhan~\cite{TZ25} showed that weakly interacting fermionic Gibbs states can be efficiently prepared. However they require the interactions to be sufficiently weak as a function of the inverse temperature. Later Han, Zhao and Liu~\cite{HZL26} showed that sufficiently weak interactions also imply a representation as a mixture of fermionic Gaussian states and efficient classical algorithms. In this work, we instead investigate physical models with competing energy scales, which will allow us to decouple the relationship between temperature and on-site parameters in the model and as a byproduct sidestep the classical obstructions to quantum advantage.

The Fermi--Hubbard model is a natural target, encapsulating many mysteries in condensed matter physics. It has been studied for more than sixty years. Already thirty years ago, Lieb~\cite{lieb1995hubbard} opined that
\begin{quote}
\begin{center}
\emph{After a third of a century of research, we are still not sure about many of its basic properties.}
\end{center}
\end{quote}
\noindent It has been described as the fruit fly of correlated electron physics~\cite{qin2022hubbard}, and even in the wake of large-scale computational investigations, its exotic behaviors have proven difficult to tame. Its appeal comes from the fact that just three simple ingredients give rise to a wide range of collective behaviors. Imagine electrons moving on a lattice, with one orbital at each site. Each electron has one of two spin states, and the Pauli exclusion principle prevents two electrons with the same spin from occupying the same orbital. Thus a site can contain at most two electrons, one of each spin. Its Hamiltonian decomposes as
\begin{equation*}
    H
    =\underbrace{t\,H_{\mathrm{hop}}}_{\text{electron hopping}}
    +\underbrace{U\,H_{\mathrm{rep}}}_{\text{double occupancy}}
    +\underbrace{\mu\,\widehat N}_{\text{electron number}}.
\end{equation*}
The hopping strength $t$ allows electrons to move from site to site, spreading across the lattice. The on-site repulsion $U>0$ penalizes double occupancy. And the chemical potential $\mu$ controls the average electron density. The competition is already apparent at \emph{half filling}, where there is one electron per site on average. When \(U\) is large compared to both the hopping strength and the temperature, moving an electron onto an occupied neighboring site creates a doubly occupied state and an empty site, at an energy cost of $U$. The movement of electrons is therefore suppressed, producing a Mott insulator even though the available orbitals are only half occupied~\cite{jordens2008mott}.

Suppressing the movement of electrons does not eliminate the effects of hopping. Consider two neighboring sites, each occupied by one electron. The hopping term mixes this configuration with a higher-energy configuration in which both electrons occupy the same site. A second hop restores having one electron per site but can exchange their spins. This two-step process is called \emph{virtual hopping}. Crucially, the Pauli exclusion principle blocks virtual hopping when neighboring electrons have the same spin, so the resulting energy gain favors neighboring spins pointing in opposite directions. As a consequence, the model behaves like an antiferromagnet at low temperature relative to an energy scale of $t^2/U$~\cite{mazurenko2017antiferromagnet}.

Changing the chemical potential can move the system away from half filling. Removing electrons introduces empty sites, or \emph{holes}, through which charge can move without creating double occupancy. In parts of the phase diagram, holes concentrate along stripes separating regions of antiferromagnetic order. The resulting state breaks spatial symmetry even though the underlying Hamiltonian can be translation invariant~\cite{zheng2017stripe}. The same magnetic correlations may also favor pairs of electrons on different sites, where they avoid the penalty for double occupancy~\cite{maier2006pairing}. If these pairs develop a common quantum phase across the system, superconductivity can emerge~\cite{dong2022superconductivity}. This competition between the three aforementioned energy scales, hopping, repulsion and chemical potential, leads to the central question of our work:
\begin{quote}
\begin{center}
\emph{Which energy scales govern the temperature at which \\efficient Gibbs state preparation is possible?}
\end{center}
\end{quote}

Our main result identifies a simple principle: \emph{the temperature scale for efficient Gibbs state preparation is set by interactions between sites, while the on-site energy scale is unrestricted.} For fixed local dimension, locality, and degree, we construct quantum Gibbs samplers that mix in $\bigO{\log(n/\epsilon)}$ time at sufficiently high temperature relative to the intersite couplings. Crucially both the temperature threshold and the mixing bound are independent of the on-site energies. For the Fermi--Hubbard model, this means that the temperature need only be high relative to the hopping strength $t$, while the repulsion $U$ and chemical potential $\mu$ can be arbitrarily large. We obtain the same guarantee for qubit systems with arbitrary external fields.
We further show that suitable choices of large on-site repulsions and chemical potentials can produce Fermi--Hubbard Gibbs states that are provably not convex mixtures of fermionic Gaussian states. Similarly, Kuwahara and Hatano~\cite{kh11} showed that suitably large external fields can produce entangled qubit Gibbs states even at high temperature, from which we conclude that our algorithms prepare qubit Gibbs states for locally interacting Hamiltonians that are provably outside the set of product states.

\subsection{Background}
\label{sec:background}
We consider a set $\Lambda$ of $n$ sites, each carrying $m$ fermionic modes. Each mode can be empty or occupied by one fermion. The resulting Hilbert space is the \emph{Fock space} $\mathcal H_\Lambda$, of dimension $2^{mn}$. After fixing an ordering of the modes, its occupation-number basis is indexed by bitstrings $\mathbf s\in\{0,1\}^{\Lambda\times[m]}$, where $s_{i,a}$ records whether mode $a$ at site $i$ is occupied. The creation operator $c_{i,a}^\dagger$ adds a fermion to this mode, and the annihilation operator $c_{i,a}$ removes one, with a sign determined by the fixed ordering. The result is zero if the mode is already occupied or empty, respectively. These operators satisfy the canonical anticommutation relations (CAR):
\begin{equation*}
    \{c_{i,a},c_{j,b}\}=0,
    \qquad
    \{c_{i,a}^\dagger,c_{j,b}^\dagger\}=0,
    \qquad
    \{c_{i,a},c_{j,b}^\dagger\}=\delta_{ij}\delta_{ab}I,
\end{equation*}
where $\{A,B\}=AB+BA$. The number operator $n_{i,a}=c_{i,a}^\dagger c_{i,a}$ is the projector onto states in which this mode is occupied. 

\paragraph{The Fermi--Hubbard model.}
The spin-$\frac12$ Fermi--Hubbard model has two modes per site, labeled $\uparrow$ and $\downarrow$. Thus, each site has four possible occupations: empty, spin-up, spin-down, or doubly occupied. On a graph $(\Lambda,E)$, the Hamiltonian is
\begin{align}\label{eq:fermi-hubbard-intro}
    H_{\mathsf{FH}}
    &=-\sum_{\{i,j\}\in E}\sum_{a\in\{\uparrow,\downarrow\}}
    t_{ij}\left(c_{i,a}^\dagger c_{j,a}+c_{j,a}^\dagger c_{i,a}\right)+\sum_{i\in\Lambda}\left(U_i n_{i,\uparrow}n_{i,\downarrow}
    +\mu_i(n_{i,\uparrow}+n_{i,\downarrow})\right).
\end{align}
Here the real hopping amplitudes $t_{ij}$ govern motion between neighboring sites, the electronic repulsion $U_i$ penalizes double occupancy when $U_i>0$, while the chemical potential $\mu_i$ shifts the energy of each occupied mode. We write $H_{\mathsf{FH}}=V+W$, with $V=\sum_i v_i$ containing the on-site interaction and chemical-potential terms, and $W$ containing the hopping terms. This decomposition is grouped by spatial support: the Hubbard interaction belongs to $V$, even though it couples the two modes at a site.

\paragraph{Fermionic parity and physical states.}
The total particle-number operator and total fermionic parity are $\widehat N=\sum_{i\in\Lambda}\sum_{a=1}^m n_{i,a}$, $\Pi=(-1)^{\widehat N}.$
The operator $\Pi$ acts as $+1$ on states with an even number of fermions and as $-1$ on states with an odd number. An operator $A$ is \emph{even} if $\Pi A\Pi=A$ and \emph{odd} if $\Pi A\Pi=-A$. We say that $A$ has definite parity if one of these identities holds. Creation and annihilation operators are odd, whereas number operators and all terms in $H_{\mathsf{FH}}$ are even.
We work with physical density operators $\rho$, which are positive semidefinite, have trace one, and commute with $\Pi$.  Physical observables are likewise even Hermitian operators, and odd operators have zero expectation in every physical state. For an even Hamiltonian $H$, the Gibbs state at inverse temperature $\beta$ is $\sigma=e^{-\beta H}/\tr(e^{-\beta H})$ and is itself even. For fermionic operators, scalar traces and trace norms are taken on the full Fock space $\mathcal H_\Lambda$.

\paragraph{Fermionic locality.}
An operator is supported in a set of sites $S$ if it can be written as a linear combination of products of creation and annihilation operators from $S$ and the identity. We denote the algebra of all such operators by $\mathfrak A_S$. For example, the number operator, $n_{i,a}$, is supported at site $i$, while the hopping operator $c_{i,a}^\dagger c_{j,a}$ is supported on $\{i,j\}$. The CAR do not require operators on disjoint sets to commute, in particular they commute when at least one is even, while two odd operators on disjoint sets anticommute. 
To construct our sampler, we need local operators corresponding to the matrices $\ketbra{\mu}{\nu}$, where $\ket{\mu}$ and $\ket{\nu}$ are basis states on the $2^m$-dimensional Fock space of a single site. We obtain these operators by expressing the matrices in terms of creation and annihilation operators. For example, when site $i$ contains one mode, write $c_i$ for its annihilation operator and $n_i=c_i^\dagger c_i$. The four matrix units are
\begin{equation*}
    P^i_{00}=I-n_i,
    \qquad
    P^i_{11}=n_i,
    \qquad
    P^i_{10}=c_i^\dagger,
    \qquad
    P^i_{01}=c_i.
\end{equation*}
These expressions define operators on the full system, all belonging to $\mathfrak A_{\{i\}}$. They obey the same multiplication rules as the corresponding two-dimensional matrices, i.e. $P^i_{01}P^i_{10}=c_ic_i^\dagger=I-n_i=P^i_{00}$. For $m$ modes, we proceed in the same way: express each one-site matrix $\ketbra{\mu}{\nu}$ using the one-site creation and annihilation operators, and substitute the full-system operators $c_{i,a}^\dagger,c_{i,a}$ into that expression. The resulting operators $P^i_{\mu\nu}\in\mathfrak A_{\{i\}}$ satisfy the usual matrix-unit identities (\Cref{def:local_basis_matrix_unit}).  Then $P^i_{\mu\nu}$ is even when $\ket{\mu}$ and $\ket{\nu}$ have the same parity and odd otherwise. These are the local operators from which we construct the fermionic Gibbs sampler.

\paragraph{Interaction parameters.}
For both qubit and fermionic systems, we write $H=V+W$, where $V=\sum_i v_i$ contains the on-site terms and $W=\sum_a h_a$ contains terms supported on at least two sites. We combine terms with the same support. The locality $L$ bounds the number of sites in each $h_a$, and the degree $D$ bounds the number of interaction terms containing any given site. We call $W$ a \emph{bounded $L$-local degree-$D$ Hamiltonian} if additionally $\norm{h_a}\leq1$ for every $a$; see \Cref{def:interaction-parameters} for the precise conventions. We write $h=\max_i\norm{v_i}$ for the on-site energy scale, on which we impose no bound unless stated otherwise.

\paragraph{Quantum Markov Chains.}
A Lindbladian $\calL$ describes continuous-time evolution of a quantum system interacting with an environment:
\begin{equation*}
   \calL(\rho_t)
    =\underbrace{-\ii[C,\rho_t]}_{\text{coherent term}}
    +\underbrace{\sum_a\left(A_a\rho_t A_a^\dagger
    -\frac12\{A_a^\dagger A_a,\rho_t\}\right)}_{\text{dissipative term}}\,,
\end{equation*}
where $C$ is a Hermitian operator, the $A_a$ are jump operators, and $\rho_t=e^{t\calL}(\rho_0)$. The coherent term generates reversible unitary evolution under $C$. The dissipative term models jumps induced by the environment and can erase information about the initial state, while the anticommutator correction keeps the trace fixed. For more background, we refer the reader to the notes of Lin~\cite{Lin24} and the survey of Lidar~\cite{lidar19}.

\subsection{Our Results}
\label{sec:our-results}

We begin with our rapid-mixing guarantee for the Fermi-Hubbard model.

\begin{theorem}[Rapid mixing for the Fermi--Hubbard model, \Cref{cor:fermionic-fermi-hubbard} (informal)]
\label{thm:informal-fermi-hubbard}
Consider the Fermi-Hubbard model on an $n$-site nearest-neighbor lattice of fixed dimension $d$, with maximum hopping amplitude $t_{\max}>0$. For any inverse temperature $0<\beta\leq c/(d^3t_{\max})$, where $c>0$ is a universal constant, there is a quantum Gibbs sampler that mixes to trace-norm error $\epsilon$ in $\bigO{\log(n/\epsilon)}$ time, starting from any physical fermionic state.
\end{theorem}

\Cref{thm:fermionic-rapid-mixing} gives the same guarantee for general local fermionic Hamiltonians with $m$ modes per site, bounded $L$-local degree-$D$ interactions, and arbitrary even on-site terms, provided $0<\beta\leq c\,4^{-m}(DL)^{-3}$. For fixed $m,L,D$, this is a constant temperature threshold, independent of system size and the on-site energies.

\begin{remark}[Comparison with weakly interacting fermions]
Earlier work established fast mixing near the Hubbard atomic limit, and subsequently logarithmic mixing for sufficiently weak hopping and for perturbations of nonhopping quadratic Hamiltonians~\cite{SMBB25,SMBB26rapid}. In our bounded-locality setting, these small-perturbation guarantees follow by placing the on-site reference Hamiltonian in $V$ and taking the intersite couplings $W$ to be sufficiently small. Our temperature threshold and mixing bound are uniform in the on-site energies.
The general weakly interacting regime of~\cite{TZ25,SMBB25} allows a quadratic reference Hamiltonian with hopping, at any fixed temperature, provided the nonquadratic interaction is sufficiently weak. Our result recovers its on-site-reference case, but does not cover the regime where the temperature is low relative to hopping. The upshot is that we allow arbitrarily strong on-site repulsion while requiring $\beta t_{\max}$ to be small. See \Cref{sec:fermionic-prior-work} for a detailed comparison.
\end{remark}

In a similar manner, we also obtain rapid mixing for qubit Hamiltonians with arbitrarily large external fields:

\begin{theorem}[Rapid mixing for qubits, \Cref{thm:rapid-mixing-with-external-field} (informal)]
\label{thm:informal-qubit-mixing}
Let $H=V+W$ act on $n$ qubits, with bounded $L$-local degree-$D$ interactions $W$ and arbitrary on-site fields $V$. For any inverse temperature $0<\beta\leq c(DL)^{-3}$, where $c>0$ is a universal constant, there is a quantum Gibbs sampler that mixes to trace-norm error $\epsilon$ in $\bigO{\log(n/\epsilon)}$ time from any initial state. Both the temperature threshold and mixing bound are independent of the field strength $h$.
\end{theorem}

\begin{remark}[Comparison with prior qubit preparation results]
Rouz{\'e}, Fran{\c{c}}a and Alhambra~\cite{rfa24a} and Bakshi, Liu, Moitra and Tang~\cite{blmt2025dobrushin} established rapid mixing for local qubit Hamiltonians at high temperature. Their temperature conditions control every local term, including the external field, and in particular require $\beta h$ to be small. We believe the techniques in these work can be extended to handle fields up to $\log(1/\beta)/\beta$. However, as we show next, for fixed locality and degree and sufficiently small $\beta$, the field regime covered by these guarantees remains separable (\Cref{thm:informal-separability}). Our result retains $\bigO{\log(n/\epsilon)}$ mixing for arbitrary $h$, including fields large enough to produce entanglement.
\end{remark}

Next, we characterize how the on-site energies affect the structure of the Gibbs state. The point is that these energies can take the state outside the convex hull of product or Gaussian states while remaining within the regime of our rapid-mixing guarantees.

\begin{theorem}[Fermi--Hubbard states outside Gaussian mixtures, \Cref{ex:non-gaussian-gibbs} (informal)]
\label{thm:informal-non-gaussian}
For every inverse temperature $\beta>0$, there is a two-site Fermi--Hubbard Hamiltonian with unit hopping and suitably chosen on-site repulsions and chemical potentials whose Gibbs state is not a convex mixture of fermionic Gaussian states.
\end{theorem}

With unit hopping, any fixed $0<\beta\leq c$ places this two-site example within the parameter settings of \Cref{thm:informal-fermi-hubbard}. Thus, at a constant temperature covered by our algorithmic guarantee, suitably chosen on-site energies produce Gibbs states outside the convex hull of Gaussian states. Next consider the qubit case. We show:

\begin{theorem}[Separability with external fields, \Cref{thm:separability-with-external-field} (informal)]
\label{thm:informal-separability}
For a qubit Hamiltonian as in \Cref{thm:informal-qubit-mixing}, there is a constant inverse-temperature threshold $\beta_0>0$, with $\beta_0^{-1}=D\exp(\bigO{L})$, such that the Gibbs state is a convex combination of product states whenever
\begin{equation*}
    0<\beta\leq\beta_0
    \qquad\text{and}\qquad
    h\leq\frac{c}{\beta L}\log\Paren{\frac{1}{DL\beta}},
\end{equation*}
where $c>0$ is a universal constant. For fixed $L,D$, the field threshold is of order $\log(1/\beta)/\beta$.
\end{theorem}

Combined wiht the results of Kuwahara and Hatano~\cite{kh11}, who constructed two-qubit examples in which suitably chosen fields of this order produce entanglement, this implies that the bound on the strength of the field in the above theorem are optimal up to constant factors. These examples exist at arbitrarily high finite temperatures, including the constant-temperature regime of \Cref{thm:informal-qubit-mixing}. Our samplers can therefore also efficiently prepare entangled qubit Gibbs states that are, of course, provably not product states.

\section{Technical Overview}
\label{sec:technical-overview}

We consider local Hamiltonians of the form $H=V+W$ on $n$ sites, where $V=\sum_i v_i$ contains the on-site terms and $W$ contains the interactions between sites. Throughout this overview, we take the local dimension, locality, and degree to be bounded constants, and normalize the interaction terms in $W$ to have norm at most one. Our goal is to construct a Lindbladian $\calL_{\mathsf{FR}}=\sum_i\calL_{\mathsf{FR},i}$ that mixes to the Gibbs state $\sigma=e^{-\beta H}/\tr(e^{-\beta H})$ in $\bigO{\log(n/\epsilon)}$ time. The temperature must be sufficiently high compared to the interactions, but we impose no bound on the on-site energies. For the Fermi--Hubbard model, the electronic repulsions and chemical potentials are on-site terms, while the hopping amplitudes set the temperature scale.

We follow the quantum Dobrushin framework of Bakshi, Liu, Moitra and Tang~\cite{blmt2025dobrushin}. The basic principle is that an update should erase more disagreement at its site than it spreads to other sites. To quantify this comparison, we decompose the difference between two states into single-site disagreements and track their costs. We first recall the framework for qubits, then explain how to extend it to fermions and arbitrary on-site terms.

\paragraph{Local disagreement and quantum Wasserstein distance.}
Two qubit states $\rho$ and $\sigma$ agree outside site $j$ if $\tr_j(\rho-\sigma)=0$, meaning that every observable on the remaining sites has the same expectation in both states. A \emph{transport plan} extends this notion to a general traceless Hermitian operator $X$: we write $X=\sum_j X_j$, where each $X_j$ is Hermitian and locally traceless at site $j$, i.e.\ $\tr_j(X_j)=0$. The associated cost vector $x$ has entries $x_j=\frac12\norm{X_j}_1$. The quantum Wasserstein norm of order $1$, introduced by De Palma, Marvian, Trevisan, and Lloyd~\cite{dmtl21}, is the minimum total cost over these plans:
\begin{equation*}
    \norm{X}_{W_1}
    =\inf\left\{\frac12\sum_{j=1}^n\norm{X_j}_1
    \,\middle|\,
    X=\sum_{j=1}^n X_j,\quad X_j=X_j^\dagger,\quad \tr_j(X_j)=0\right\}.
\end{equation*}
The condition $\tr_j(X_j)=0$ does not require $X_j$ to be supported at $j$. It expresses agreement outside $j$, and therefore also applies to correlated states.

\paragraph{The Dobrushin condition and path coupling.}
Consider an averaged update $\Phi=\frac1n\sum_i\Phi_i$, where each $\Phi_i$ is a trace-preserving linear map associated with an update at site $i$. In our application, $\Phi_i=\calI+\delta\calL_{\mathsf{FR},i}$ is an infinitesimal Lindbladian update. The quantum Dobrushin influence matrix measures the effect of each update on a disagreement at site $j$:
\begin{equation*}
    D^{(\Phi)}_{i,j}
    =\frac1n\sup\left\{\norm{\Phi_i(X)}_{W_1}
    \,\middle|\,
    X=X^\dagger,\quad \tr_j(X)=0,\quad \norm{X}_{W_1}=1\right\}.
\end{equation*}
The $j$-th column sum bounds the cost after updating a uniformly random site, starting from a disagreement at $j$ of unit Wasserstein norm. Apply this bound to each $X_j$ in a transport plan, add the resulting plans, and minimize over the original plan. If every column sums to at most $1-\gamma$, we obtain
\begin{equation*}
    \norm{\Phi(X)}_{W_1}
    \leq \norm{D^{(\Phi)}}_{1\to1}\norm{X}_{W_1}
    \leq (1-\gamma)\norm{X}_{W_1}.
\end{equation*}
This is the quantum path-coupling argument: contraction for locally traceless inputs suffices to prove contraction between arbitrary states.

\paragraph{Ingredients for rapid mixing.}
To apply this criterion, we need a Lindbladian that fixes the Gibbs state and whose on-site contraction dominates the influence of the interactions. Detailed balance ensures stationarity, and Chen, Kastoryano and Gily\'en~\cite{ckg23} show how to construct Lindbladians with this property. We separate each site generator into its on-site contribution and terms supported on growing balls, with norms that decay geometrically with the radius. Lieb--Robinson bounds control real-time spreading, and estimates on the filters turn this into quasi-locality of the generator. Transport plans then convert these estimates into bounds on the influence matrix. At sufficiently high temperature, the total contribution of the off-site terms is smaller than the on-site contraction.
The rapid-mixing argument of Rouz{\'e}, Fran{\c{c}}a and Alhambra~\cite{rfa24a} admits a similar interpretation in observable space: they prove contraction of the adjoint evolution in the oscillator norm. This norm is bounded above and below by multiples of the dual quantum Wasserstein norm (see~\cite[Section~5.4, Remark~5.9]{blmt2025dobrushin} for a detailed discussion).

\subsection{Obstacles to rapid mixing}
\label{sec:overview-obstacles}

\paragraph{Fermionic locality.}
The first obstacle is to define local agreement for fermions. For qubits, $\tr_j(\rho-\sigma)=0$ expresses agreement on the tensor factors outside $j$. For fermions, locality is specified by the creation and annihilation operators appearing in an observable, and operators on disjoint sets of modes need not commute. A Jordan--Wigner transformation maps the algebra to qubits, but its parity strings can turn local fermionic operators into operators supported on many qubits. Thus, the qubit partial trace in a fixed representation does not directly describe agreement on the complementary fermionic algebra. We need a fermionic partial trace and transport plans that respect this locality. To evolve these plans, we also need identities that apply to odd jumps, even though the states themselves are even.

\paragraph{Arbitrary on-site energies.}
The second obstacle is already present for qubits: large on-site energies can spoil either locality or contraction in the existing analyses. The first starting point is the sampler of Bakshi, Liu, Moitra and Tang~\cite{blmt2025dobrushin}, whose jump operators are imaginary-time evolutions of single-site Paulis, $A^P=e^{-\beta H/4}Pe^{\beta H/4}$. Unlike real-time evolution, imaginary-time evolution need not preserve operator norms. In the nested-commutator expansions used to bound its locality, arbitrarily many on-site terms can appear between successive interaction terms. These insertions amplify the coefficients without enlarging support. For example, write $h=\max_i\norm{v_i}$ and let $W^{(i)}$ contain the interaction terms incident to $i$. The imaginary time evolution operator counts connected clusters of $k$ interacting terms and contributes a term of norm $\frac{1}{k!}\Paren{\frac{e^{4\beta hL}-1}{hL}}^k$ 
for each cluster. Thus, even when the interaction strengths are bounded, large on-site energies can overwhelm the decay needed for quasi-locality of the jump operators. The second starting point is the CKG Lindbladian~\cite{ckg23}, whose jump operators use filtered real-time evolutions $e^{\ii Ht}Pe^{-\ii Ht}$. Here, a modified Lieb--Robinson bound controls propagation independently of the on-site energies. However, the canonical filters used in~\cite{ckg23,rfa24a,rfa24,TZ25} are fixed at the thermal scale $1/\beta$. Even for an isolated qubit, a large energy gap can make the on-site contraction exponentially small. With interactions present, this contraction need not dominate the off-site influence.

\subsection{Fermionic partial trace and transport plans}
\label{sec:overview-fermionic-transport}

\paragraph{Local agreement through observables.}
Reduced fermionic states have a standard formulation through tracial conditional expectations onto local CAR algebras, developed by Araki and Moriya~\cite{AM03}. 
We define local agreement through the fermionic algebras introduced in \Cref{sec:background}. Two physical states $\rho$ and $\sigma$ agree outside $S$ if every even observable in $\mathfrak A_{S^c}$ has the same expectation in both states. Since odd operators have zero expectation in even states, this is equivalent to agreement on the entire complementary algebra.
To construct transport plans, we need a linear map that sends a state difference to zero exactly when the states agree outside $S$. For qubits, tracing out $S$ does this by retaining all expectation values on the remaining sites. We define the fermionic analogue using the complementary algebra: the normalized fermionic partial trace $\mathbb E_S(X)$ is the unique element of $\mathfrak A_{S^c}$ that preserves these expectation values:
\begin{equation*}
    \tr\Paren{Z\, \mathbb E_S(X)}=\tr\Paren{ZX}
    \qquad\text{for every }Z\in\mathfrak{A}_{S^c}.
\end{equation*}
Thus, agreement outside $S$ is equivalent to $\mathbb E_S(\rho-\sigma)=0$. If each site carries $m$ fermionic modes, the unnormalized partial trace is $\operatorname{tr}_S=2^{m|S|}\mathbb E_S$. We view $\mathfrak A_{S^c}$ as a subalgebra of operators on the full Fock space and take scalar traces and trace norms on this full space. With this convention, the normalized map $\mathbb E_S$ preserves trace and contracts trace norm. We note that our map $\mathbb E_S$ is precisely the conditional expectation considered by Araki and Moriya~\cite{AM03}, but projected onto $\mathfrak A_{S^c}$, with $S$ denoting the sites removed. We also note that Szalay, Zimbor\'as, M\'at\'e, Barcza, Schilling, and Legeza~\cite{Szalay21} considered a similar partial trace with its output on the smaller Fock space.

\paragraph{Partial trace and parity.}
To evolve transport plans, we need partial-trace identities for the both even and odd operators. Even local operators commute with every operator on the complement, so they can be cycled through the partial trace, just as for qubits. An odd local operator can instead introduce a minus sign when moved past an odd operator on the complement. Nevertheless, we show that if $X$ is even and $A,B\in\mathfrak A_S$ are both even or both odd, then
\begin{equation*}
    \operatorname{tr}_S(AXB)=\operatorname{tr}_S(BAX).
\end{equation*}
The point is that the parity assumptions make both $AXB$ and $BAX$ even. Thus, odd operators on the complement have zero trace against either one, and it suffices to compare their traces against even $Z\in\mathfrak A_{S^c}$. To see why $Z$ commutes with $B$, expand $Z$ into monomials with an even number $2k$ of creation and annihilation operators on $S^c$. Moving a monomial of $B$ with $r$ generators on $S$ past such a term produces the sign $(-1)^{2kr}=1$. Since the supports are disjoint, the CAR introduce no identity terms. By linearity, $ZB=BZ$, even when $B$ is odd. Ordinary trace cyclicity then gives $\tr(ZAXB)=\tr(BZAX)=\tr(ZBAX)$, which suffices to obtain the identity.

\paragraph{Fermionic transport plans.}
With this notion of partial trace, we can define transport plans as De Palma, Marvian, Trevisan, and Lloyd did for qubits~\cite{dmtl21}. For an even traceless Hermitian operator $X$, a fermionic transport plan is a decomposition $X=\sum_iX_i$, where each $X_i$ is even, Hermitian, and satisfies $\operatorname{tr}_i(X_i)=0$. The cost at site $i$ is $x_i=\frac12\norm{X_i}_1$, and minimizing the total cost defines the fermionic Wasserstein norm. If $X$ is locally traceless at one site, this norm is exactly half its trace norm.
The useful consequence in our context is that a disagreement vanishing under $\operatorname{tr}_S$ can be charged entirely to sites in $S$. To see this, enumerate $S=\{i_1,\ldots,i_k\}$ and let $S_\ell=\{i_1,\ldots,i_\ell\}$, with $S_0=\varnothing$. Since $\mathbb E_S(X)=0$, we have the telescoping decomposition
\begin{equation*}
    X=\sum_{\ell=1}^k
    \left(\mathbb E_{S_{\ell-1}}(X)-\mathbb E_{S_\ell}(X)\right).
\end{equation*}
The maps $\mathbb E_S$ preserve parity and compose over disjoint regions. Applying $\mathbb E_{\{i_\ell\}}$ to the $\ell$-th summand therefore sends both terms to $\mathbb E_{S_\ell}(X)$, so the summand is locally traceless at $i_\ell$. We also prove that the trace-norm contraction bounds its cost by $\norm{X}_1$, giving total cost at most $|S|\norm{X}_1$ (\Cref{prop:fermionic-cost-vector-region-partial-trace-zero}). In particular, for analyzing quasi-local jump operators, this enables us to charge such a jump to a ball whenever its partial trace over that ball vanishes. As we see subsequently, this converts the quasi-local expansions into bounds on transport costs.

We note that fermionic optimal transport has been well-studied in the literature, but the constructions differ in what objects they optimize. Carlen and Maas~\cite{CM14} define a Wasserstein-$2$ distance through paths of states satisfying a fermionic continuity equation, with a density-dependent quadratic cost. The plans of Duvenhage, van Zyl, and Zurlo~\cite{DvZZ25} are joint states on graded tensor products, and Bertini, De Sole, Posta, and Presilla~\cite{BDSP25} also use the conditional expectation geometry for fermionic lattices with one mode per site, studying convergence of the distance per unit volume. Our norm agrees with their finite-volume norm on even states in this setting: parity averaging preserves the partial-trace constraint and cannot increase the cost. We generalize this notion and group $m$ modes at each site into one subsystem and track the cost at each site under quasi-local updates.

\subsection{Lieb-Robinson bounds with arbitrary on-site terms}
\label{sec:overview-onsite-lr}

On-site evolution can change a local operator arbitrarily quickly, but cannot move information between sites. We formalize this intuition with a Lieb-Robinson bound whose velocity is independent of the on-site energies. Following Nachtergaele, Raz, Schlein, and Sims~\cite{NRSS09}, we convert the onsite terms into a time-dependent Hamiltonian. Concretely, recall $H=V+W$, where $W=\sum_Z w_Z$, and define
\begin{equation*}
    U(t)=e^{\ii tV}e^{-\ii tH},
    \qquad
    W(t)=e^{\ii tV}We^{-\ii tV}
    =\sum_Z \underbrace{e^{\ii tV}w_Ze^{-\ii tV}}_{w_Z(t)}.
\end{equation*}
Then $\ii\partial_tU(t)=W(t)U(t)$, so $U(t)$ is generated by the time-dependent interaction $W(t)$. The point of this transformation is that conjugation by the on-site evolution preserves the support and norm of every interaction term. Thus, $W(t)$ has the same interaction geometry and strength as $W$: arbitrarily large on-site energies can make $w_Z(t)$ vary rapidly in time, but do not increase its norm or enlarge its support.
We note that for fermions, this reduction uses the evenness of the on-site terms. Since the $v_i$ are even and supported on distinct sites, they commute, and hence $e^{\ii tV}=\prod_i e^{\ii t v_i}$. Each factor is also even and commutes with every operator supported away from site $i$, including the unphysical odd operators. In particular, when conjugating an operator supported on $Z$, the factors outside $Z$ cancel, while those inside $Z$ preserve its support and parity. Consequently, each $w_Z(t)$ remains even, is supported on $Z$, and satisfies $\norm{w_Z(t)}=\norm{w_Z}$. Since the transformed interactions are even, we can conclude that the resulting Lieb-Robinson truncation estimates apply to both even and odd local operators.

We can now apply a time-dependent extension of the Lieb--Robinson bound of Haah, Hastings, Kothari, and Low~\cite{hhkl21}. The bound depends on the supports and norms of the terms in $W(t)$, with no dependence on how quickly they vary. The relevant interaction strength is therefore
\begin{equation*}
    \zeta=\max_i\sum_{Z\ni i}|Z|\,\norm{w_Z}\leq DL,
\end{equation*}
independently of the on-site energies. Note, in the Fermi--Hubbard model, this scale is set by the hopping amplitudes, with no dependence on the on-site repulsions or chemical potentials. Next, we convert this Lieb-Robinson bound into a shell-decomposition, where the $r$-th term is supported on a ball of radius $r$. In particular, for an operator $A$ supported at site $i$, let $B_r$ be the ball of radius $r$ around $i$ in the interaction graph. Let $H_r=V+\sum_{Z\subseteq B_r}w_Z$, be the interaction-truncated Hamiltonian, which continues to retain all on-site terms. Performing evolutions undersuccessive truncations yields the telescoping behavior required for a shell decomposition, with norms controlled by the Lieb-Robinson bound as follow:
\begin{align*}
    e^{\ii Ht}Ae^{-\ii Ht}
    &=e^{\ii t v_i}Ae^{-\ii t v_i}
    +\sum_{r\geq1}\left(e^{\ii H_rt}Ae^{-\ii H_rt}
    -e^{\ii H_{r-1}t}Ae^{-\ii H_{r-1}t}\right),\\
    \textrm{where } \quad & \norm{e^{\ii H_rt}Ae^{-\ii H_rt}
    -e^{\ii H_{r-1}t}Ae^{-\ii H_{r-1}t}}
    \leq\norm{A}\,\frac{(2\zeta|t|)^r}{r!}
    \qquad(r\geq1).
\end{align*}
We note that for fermions, each $H_r$ is even, so the telescoping continue to retain the parity of $A$. We refer the reader to \Cref{lem:LR_radius_r,lem:fermionic-shell} for a formal proof.

\subsection{Field-resonant Lindbladians}
\label{sec:overview-field-resonant}

To recover on-site contraction, we allow the parameters in the CKG construction to depend on the local energies. For qubits, the parameters are site-dependent, whereas for fermions, we have to make them jump-dependent. Unlike a qubit, a fermionic site can have several distinct energy gaps, and thus we tune the filter for each jump to the pair of levels it connects.

\paragraph{Adapting the CKG filters.}
For an operator $B$ supported at site $i$, the filtered jump and its transition weight are
\begin{equation*}
    A_B(\omega)=\frac1{\sqrt{2\pi}}\int_{\mathbb R}e^{\ii Ht}Be^{-\ii Ht}e^{-\ii\omega t}f_B(t)\,\diff t,
    \quad f_B(t)=\sqrt{\sigma_B}(2/\pi)^{1/4}e^{-\sigma_B^2t^2},
    \quad \gamma_B(\omega)=e^{-\frac{(\omega+\Delta_B)^2}{2\eta_B^2}}.
\end{equation*}
The jump $A_B(\omega)$ is a smoothed frequency component of the evolution: it favors matrix elements of $B$ between energy eigenstates whose energy difference is near $\omega$. The parameter $\sigma_B^{-1}$ sets the time scale of the filter, while $\gamma_B$ weights these jumps around frequency $-\Delta_B$. Then, the single site generator has the following form:, the site generator has the form
\begin{equation*}
    \calL_{\mathsf{FR},i}(X)=-\ii[C_i,X]+\sum_{B\in\calB_i}\int_{\mathbb R}\gamma_B(\omega)D_{A_B(\omega)}(X)\,\diff\omega\,,
\end{equation*}
where $D_A(X)=AXA^\dagger-\frac12\{A^\dagger A,X\}$.
The coherent term is $C_i=\sum_{B\in\calB_i}C_B$, where
\begin{equation*}
    C_B=\int_{\mathbb R}b_1(t,\sigma_B)e^{-\ii Ht}
    \left(\int_{\mathbb R}b_2(t',\eta_B,\Delta_B)
    e^{\ii Ht'}B^\dagger e^{-2\ii Ht'}Be^{\ii Ht'}\,\diff t'\right)
    e^{\ii Ht}\,\diff t.
\end{equation*}
Here $b_1$ and $b_2$ are the scalar kernels defined in \Cref{def:fermionic-general-lindblad}, with the same filter parameters $(\sigma_B,\eta_B,\Delta_B)$. To retain KMS detailed balance, we require $\beta(\eta_B^2+\sigma_B^2)=2\Delta_B$, take $\calB_i$ to be closed under adjoints, and assign identical parameters to $B$ and $B^\dagger$. For fermions, we also require every $B$ to have either even or odd parity. Since $H$ is even, this ensures that the dynamics preserves even states, including when the jumps themselves are odd. These conditions ensure the stationarity of the Gibbs state (see~\Cref{prop:fermionic-detailed-balance} for a formal proof).

\paragraph{Fermionic on-site contraction.}
At site $i$, since $v_i$ is even, we can diagonalize it in the basis of parity eigenstates. We then shift its minimum eigenvalue to zero, and write $v_i=\sum_\mu\varepsilon_\mu P^i_{\mu\mu}$, where $P^i_{\mu\nu}$ are the matrix basis elements in the on-site fermionic algebra. We take all these matrix elements as jumps, and observe that each has definite parity, and $(P^i_{\mu\nu})^\dagger=P^i_{\nu\mu}$. We choose the filters so that, when the interactions are switched off, the jump from level $\nu$ to level $\mu$ has rate $\kappa_\star e^{-\beta\varepsilon_\mu}$. This depends only on the destination level and makes the on-site dynamics contract toward the local Gibbs state. Specifically, set
\begin{equation*}
    a_\mu=\sqrt{\varepsilon_\mu+ 1/(4\beta)},
    \qquad
    \Delta_{\mu\nu}=(a_\mu+a_\nu)^2,
    \qquad
    \sigma_{\mu\nu}=\eta_{\mu\nu}=\sqrt{\Delta_{\mu\nu}/\beta}.
\end{equation*}
These parameters satisfy the detailed-balance conditions above. To see why they yield the desired rates, consider the thought experiment of switching off the interactions. Then $P^i_{\mu\nu}$ evolves at frequency $\Omega_{\mu\nu}=\varepsilon_\mu-\varepsilon_\nu$, and the Gaussian overlap gives the rate
\begin{equation*}
    \kappa_{\mu\nu}
    =\frac1{\sqrt2}\exp\!\left(-\frac{\beta(\Delta_{\mu\nu}+\Omega_{\mu\nu})^2}{4\Delta_{\mu\nu}}\right)
    =\kappa_\star e^{-\beta\varepsilon_\mu},
    \qquad \kappa_\star=\frac{e^{-1/4}}{\sqrt2}.
\end{equation*}
The key observation here is that $\kappa_{\mu\nu}$ depends only on the destination $\mu$. Let $\calL_{\mathsf{FR},i,0}$ be the on-site dissipative generator obtained by switching off the interactions. Summing over the matrix units gives, for every even $X$,
\begin{equation*}
    \calL_{\mathsf{FR},i,0}(X)=\kappa_\star\left(e^{-\beta v_i}\operatorname{tr}_i(X)-z_iX\right),
    \qquad z_i=\sum_\mu e^{-\beta\varepsilon_\mu}\geq1.
\end{equation*}
If $\operatorname{tr}_i(X)=0$, this reduces to $\calL_{\mathsf{FR},i,0}(X)=-\kappa_\star z_iX$. Since $z_i\geq1$, the on-site evolution contracts the transport cost at rate at least $\kappa_\star$, independently of the on-site energies. The condition allows $X$ to be correlated with other sites (\Cref{lem:fermionic-onsite-contraction}).

\paragraph{On-site contraction for qubits.}
For qubits, a single filter choice per site suffices. Let $g_i$ be the gap between the two eigenvalues of $v_i$. We use the three single-site Paulis as jumps, with
\begin{equation*}
    \Delta_i=\max\{g_i,1/\beta\},
    \qquad \sigma_i=\eta_i=\sqrt{\Delta_i/\beta}.
\end{equation*}
When $g_i\geq1/\beta$, the transition weight is centered at the downward transition frequency $-g_i$, and its overlap with that transition is $1/\sqrt2$. For smaller gaps, the overlap is at least $\kappa_\star$. This yields uniform trace-norm contraction for Hermitian $X$ satisfying $\tr_i(X)=0$, and hence contraction of the transport cost at site $i$ (\Cref{lem:diss_contraction}).
We call this choice of jumps and filters \emph{field-resonant Lindbladians}. Both choices satisfy $\Delta_B\geq1/\beta$ and $\sigma_B^{-1}\leq\beta$, which enables us to recover on-site contraction while keeping the time scale of each filter at most $\beta$. 

\subsection{Quasi-locality of the field-resonant Lindbladian}
\label{sec:overview-generator-locality}

It remains to prove that these filter choices preserve quasi-locality. We combine the Lieb-Robinson shell bounds with the carefully chosen filters above to obtain geometric decay for the off-site terms.

\paragraph{Jump operators and dissipative evolution.}
For a norm-one jump $B$ at site $i$, integrate its real-time shell decomposition against $f_B$. Gaussian moment bounds give
\begin{equation*}
    A_B(\omega)=\sum_{r\geq0}G_{i,r}^{B,\omega},
    \qquad
    \norm{G_{i,r}^{B,\omega}}
    \leq(2\pi)^{-1/4}\sigma_B^{-1/2}(2\zeta/\sigma_B)^r.
\end{equation*}
Each term is supported on the ball $B_r$ around $i$, and $G_{i,0}^{B,\omega}$ is precisely the filtered on-site evolution used in the contraction argument. Since $\sigma_B^{-1}\leq\beta$, the spatial decay is controlled by $\beta\zeta$, independently of the on-site energies (\Cref{prop:fermionic-jump-operators-quasi-local,prop:jump-operators-quasi-local}).

The dissipative evolution also involves integration against $\gamma_B$, whose width grows with the on-site energies. Expanding its quadratic terms gives pairs $G_{i,r}^{B,\omega}$ and $G_{i,s}^{B,\omega}$, both supported on the ball of radius $\max\{r,s\}$. Their norms satisfy
\begin{equation*}
    \int_{\mathbb R}\gamma_B(\omega)
    \norm{G_{i,r}^{B,\omega}}\norm{G_{i,s}^{B,\omega}}\,\diff\omega
    \leq\frac{\eta_B}{\sigma_B}(2\zeta/\sigma_B)^{r+s}
    \leq(2\zeta\beta)^{r+s}.
\end{equation*}
Here $\int\gamma_B=\sqrt{2\pi}\eta_B$, and $\eta_B=\sigma_B$ cancels the normalization from the two factors. Thus, the dissipative bounds remain independent of the on-site energies. The term $r=s=0$ gives the on-site evolution; the remaining terms decay geometrically. For fermions, each $G_{i,r}^{B,\omega}$ has the same parity as $B$, so the resulting maps preserve evenness.

\paragraph{Cancellation in the coherent kernel.}
The coherent terms require more care. Let $C_i=\sum_{B\in\calB_i}C_B$, where each $C_B$ integrates a product of two real-time evolutions against scalar kernels $b_1(t)b_2(t')$. Suppressing their dependence on $B$, these kernels are
\begin{equation*}
    b_1(t)=\frac{\sigma_B}{\pi\beta}e^{\beta^2\sigma_B^2/8}
    \left(\sech\left(\frac{2\pi t}{\beta}\right)
    *_t\left(\sin(-\beta\sigma_B^2t)e^{-2\sigma_B^2t^2}\right)\right),
    \quad b_2(t)=2\eta_B e^{-4\Delta_Bt^2/\beta-2\ii\Delta_Bt}\,,
\end{equation*}
where $*_t$ denotes convolution in $t$. Integrating the Lieb-Robinson shell bounds therefore requires control of the absolute moments of these kernels as the on-site energies grow. Due to the prefactor of $e^{\beta^2\sigma_B^2/8}=e^{\beta\Delta_B/8}$ in $b_1$, naively bounding the convolution in that expression by absolute values would retain this exponential dependence on the local energy scale. We show that is apparent amplification is merely cosmetic, an evaluating the Fourier transform instead leads to a seemingly magical cancelattion against the shifted Gaussian kernels, yielding the following equivalent expression
\begin{equation*}
    b_1(t)=\frac{\sigma_B}{\pi\beta}\int_0^\infty
    \frac{e^{-2\sigma_B^2(t+u)^2}-e^{-2\sigma_B^2(t-u)^2}}
    {\sinh(2\pi u/\beta)}\,\diff u.
\end{equation*}
Although $1/\sinh(2\pi u/\beta)$ diverges as $1/u$ near the origin, the two Gaussians agree at $u=0$, and their difference vanishes at least linearly in $u$. Thus, the integrand remains bounded near the origin. This expression gives $k$-th absolute moments at most $\bigO{k!\beta^k}$ for $k\geq1$, where the implicit constant is universal. The Gaussian kernel $b_2$ has moments at most $\bigO{k!\sigma_B^{-k}}$ for every $k\geq0$. The zeroth absolute moment of $b_1$ requires separate treatment: its bound still grows as $\bigO{\log(\sigma_B\beta)}$ (see~\Cref{lem:b_1_zero_moment,lem:b_1_higher_moments,lem:b_2_moments} for a formal proof). Combining the Lieb-Robinson bound with these moment estimates gives geometric decay at sufficiently high temperature, independently of the on-site energies. Thus,
\begin{equation*}
    C_B=K_{i,0}^B+\sum_{r\geq1}K_{i,r}^B,
    \qquad \norm{K_{i,r}^B}\leq\bigO{\zeta\beta}^r
    \quad(r\geq1),
\end{equation*}
where each $K_{i,r}^B$ is Hermitian, supported on $B_r$, and even in the fermionic case (\Cref{prop:fermionic-coherent-quasi-local,prop:coherent-quasi-local}).
For the fermionic matrix-unit jumps, $K_{i,0}^B=0$ since the on-site operator product is proportional to $P^i_{\nu\nu}$, and the remaining integral vanishes by oddness of $b_1$. Any on-site coherent term retained in the qubit argument generates a local unitary and preserves transport cost. Thus, only the terms with $r\geq1$ contribute to spreading, and our corresponding bounds are independent of the on-site energies. Combining these locality estimates with the on-site contraction, we obtain rapid mixing at sufficiently high temperature by the same Dobrushin path-coupling argument as Bakshi, Liu, Moitra and Tang~\cite{blmt2025dobrushin}.

\section{Prior work and comparison}
\label{sec:prior-work}

Chen, Kastoryano and Gily\'en~\cite{ckg23} constructed efficiently implementable Lindbladians satisfying exact KMS detailed balance for general noncommuting Hamiltonians. Ding, Li and Lin~\cite{ding2024kms} developed related constructions with finitely many jump operators. These constructions ensure that the Gibbs state is stationary, but efficient preparation requires a separate bound on the mixing time. Throughout this section, we use \emph{fast mixing} for polynomial mixing times and \emph{rapid mixing} for polylogarithmic mixing times in the system size.

\subsection{Fermionic systems}
\label{sec:fermionic-prior-work}

Free fermions, whose Hamiltonians are quadratic in creation and annihilation operators, provide a natural starting point. Early work established rapid mixing and entropy-decay bounds for free-fermion thermalization at infinite and finite temperatures~\cite{CL93,TPK14,CM17}. More recent work on Gibbs samplers proved rapid mixing from Gaussian initial states~\cite{SMBB25}, and subsequently from arbitrary physical initial states~\cite{SMBB26rapid}.

These guarantees extend to sufficiently weak perturbations of free fermions. At any fixed temperature, fast mixing is known on fixed-dimensional lattices when the nonquadratic interactions are sufficiently weak~\cite{TZ25,SMBB25}. The proofs use stability of a parent Hamiltonian gap~\cite{Has19}. The allowed perturbation is independent of system size, but its threshold depends on temperature and the local energy scales of the quadratic Hamiltonian, including hopping and chemical potential. The temperature dependence of this threshold is implicit.

A different starting point is the \emph{atomic limit} of the Fermi--Hubbard model, where hopping is zero but the on-site repulsion need not be small. For any fixed temperature and on-site interaction strength $U$, fast mixing was established whenever the hopping is below a system-size-independent threshold~\cite{SMBB25}. Subsequent work proved rapid mixing with an explicit sufficient bound on the hopping that contains an exponentially small factor in $\beta|U|$~\cite{SMBB26rapid}.

Our result removes this dependence on the on-site energies. On a fixed-dimensional lattice, we obtain logarithmic mixing whenever $\beta t_{\max}$ is sufficiently small, independently of the repulsions $U_i$ and chemical potentials $\mu_i$. Increasing these parameters requires neither weaker hopping nor a longer mixing time, even when they vary across sites and grow with system size. The distinction from weakly interacting fermions is that we allow arbitrarily strong on-site interactions, while those results also cover temperatures low relative to hopping. Our general theorem further allows arbitrary even on-site Hamiltonians and bounded-degree interaction hypergraphs, without requiring a fixed-dimensional lattice (\Cref{thm:fermionic-rapid-mixing}).

There is a related structural result due to Ueltschi~\cite{Uel99}, who proved analyticity of the free energy and exponential decay of correlations for the translation-invariant Fermi--Hubbard model when $\beta t$ is sufficiently small. His threshold is independent of the on-site interaction strength and chemical potential. Our result establishes rapid mixing with the same independence from these energy scales. Polynomial-time preparation is also known through direct coupling to a bath for free and weakly interacting fermions~\cite{DZPL25,Slezak26,WD25}.

On the classical side, Ramkumar, Cai, Tong and Jiang~\cite{RCTJ26} showed that fermionic Gibbs states are mixtures of Gaussian states when the temperature is sufficiently high relative to all local terms. They also gave a polynomial-time classical algorithm for sampling this mixture. Han, Zhao and Liu~\cite{HZL26} obtained a similar representation and classical sampling algorithm for sufficiently weak perturbations of finite-range free-fermion Hamiltonians at any fixed finite temperature. Both results constrain the nonquadratic on-site interactions. Our guarantee allows these interactions to be arbitrarily strong and includes Fermi--Hubbard Gibbs states outside the convex hull of Gaussian states (\Cref{thm:informal-non-gaussian}).

\subsection{Spin systems}
\label{sec:spin-prior-work}

Rouz\'e, Fran\c{c}a and Alhambra established fast mixing for noncommuting local spin Hamiltonians at sufficiently high temperature~\cite{rfa24}, and subsequently strengthened this to rapid mixing~\cite{rfa24a}. Bakshi, Liu, Moitra and Tang~\cite{blmt2025dobrushin} obtained rapid mixing through a quantum Dobrushin condition. In these results, high temperature is measured against the strengths of all local terms, including the external field. Our result keeps both the temperature threshold and mixing bound independent of the field strength.

Fast mixing is also known for finite-range spin chains at every fixed temperature~\cite{BC25}, and for all-to-all systems at temperatures controlled by the aggregate interaction strength~\cite{Ber26}. These results address different geometries and temperature regimes.\footnote{Related results for Davies generators include spectral-gap bounds under a nondegeneracy assumption~\cite{temme2013lower}, entropy-decay bounds for finite-range, translation-invariant commuting spin chains~\cite{bardet2024entropy}, and comparisons with the embedded classical chain under spectral assumptions~\cite{BGSS+25}.}
Very recently, Kim~\cite{Kim26} and Rouz{\'e} and Fran{\c{c}}a~\cite{RF26} established rapid mixing for finite-range one-dimensional spin systems at every fixed finite temperature.

For quantum Ising models, Crosson and Slezak~\cite{crosson2025classical} obtained a classical path-integral Monte Carlo algorithm whose temperature threshold and number of worldline updates are independent of the longitudinal and transverse fields. Their algorithm samples measurement outcomes in the Ising basis and approximates the partition function. Our mixing theorem applies to general bounded-degree qubit interactions, without requiring the quantum Ising form or stoquasticity.

Efficient quantum preparation can also coexist with classical sampling hardness. Bergamaschi, Chen and Liu~\cite{bergamaschi2024quantum} constructed families of commuting Hamiltonians whose Gibbs states at constant temperature admit efficient quantum preparation, while sampling their computational-basis measurement outcomes is classically hard under standard complexity assumptions.

Another route to classical algorithms is through complex zeros of the partition function. Using Barvinok's interpolation method, Harrow, Mehraban and Soleimanifar~\cite{HMS20} obtained quasi-polynomial-time approximation algorithms under suitable zero-free assumptions. Yao, Yin and Zhang~\cite{YYZ22} improved these guarantees to polynomial time for bounded-degree local Hamiltonians. Zlokapa and Kiani~\cite{ZK26} obtained related quasi-polynomial-time guarantees for local thermal expectations.

For interpolation in inverse temperature, however, the required zero-free neighborhoods need not remain uniform as the on-site fields grow. This is already visible for a product Hamiltonian.\footnote{For $H_n=h_n\sum_{i=1}^n\sigma_Z^{(i)}$ with $h_n>0$, the partition function is $Z_n(z)=[2\cosh(zh_n)]^n$, with zeros at $z=\pm\ii\pi/(2h_n)$. These approach the origin as $h_n$ grows, ruling out a field-independent zero-free neighborhood of $[0,\beta]$. The model itself is easy to sample. The example only concerns interpolation in inverse temperature and does not exclude other classical algorithms, including interpolation in the intersite couplings while holding the on-site terms fixed.}
Thus, these guarantees do not directly yield an algorithm uniform in the external field.

\section{Preliminaries and notation}
\label{sec:preliminaries}

\paragraph{Notation.}
We write $[n] = \{1, 2,\ldots, n\}$ and $\ii = \sqrt{-1}$ as the imaginary unit. 
For a finite site set $\Lambda$ with $|\Lambda|=n$, we write $\R^\Lambda$ for real vectors indexed by sites in $\Lambda$; in the spin setting, we take $\Lambda=[n]$.
For $i\in\Lambda$, we use $e_i\in\R^\Lambda$ to denote the standard basis vector whose $i$-th entry is $1$ and whose other entries are zero.
For $S\subseteq\Lambda$, we write $e_S = \sum_{i\in S}e_i$ and $E_S = e_S e_S^T$.
For a vector $x$, we write $\norm{x}$ for the Euclidean norm and $\norm{x}_1$ for the $\ell_1$ norm. 
For vectors with real entries, we use $x \leq_v y$ to refer to entrywise comparison, i.e.\ $x_i \leq y_i$ for all entries $i$. 
For a matrix $A$, we write $\norm{A}$ for the operator norm and $\norm{A}_1$ for the Schatten-$1$ norm, and the $1\to 1$ norm $\norm{A}_{1\to 1} = \max_{x: \norm{x}_1=1} \norm{Ax}_1 = \max_{j} \sum_{i} \abs{A_{ij}}$. 

We refer to a function as a linear map, or a
map for short, if it is linear and takes Hermitian matrices to Hermitian matrices. 
We use $\calI$ denote the identity map (as opposed to $I$, which denotes the identity matrix).

The following conventions describe the interaction parameters of both spin and fermionic systems.
Operator support is defined separately in the spin and fermionic parts.

\begin{definition}[Interaction strength and bounded couplings]
\label{def:interaction-parameters}
Fix a decomposition
\[
H=V+W,\qquad V=\sum_{i\in\Lambda}v_i,\qquad
W=\sum_{a\in\mathcal J}h_a,
\]
where $\mathcal J$ is a finite index set, each $v_i$ is Hermitian and supported on site $i$, and each $h_a$ is nonzero and Hermitian with support $Z_a\subseteq\Lambda$ of size at least two.
Terms with the same support are combined into a single term, so the sets $Z_a$ are pairwise distinct.
The interaction hypergraph has vertex set $\Lambda$ and hyperedge family $\mathcal E_H:=\{Z_a:a\in\mathcal J\}$.
The \emph{local interaction strength}, also called the \emph{local multi-site coupling strength}, is
\[
\zeta:=\max_{i\in\Lambda}\sum_{a\in\mathcal J:\,i\in Z_a}|Z_a|\cdot\norm{h_a}.
\]
For integers $L\geq2$ and $D\geq1$, we say $W$ is a \emph{bounded $L$-local degree-$D$ Hamiltonian} or $W$  has \emph{bounded $L$-site degree-$D$ couplings}, if
\[
\norm{h_a}\leq1,\qquad |Z_a|\leq L\quad\text{for every }a\in\mathcal J,
\qquad
\max_{i\in\Lambda}\big|\{Z\in\mathcal E_H:i\in Z\}\big|\leq D.
\]
\end{definition}

The interaction hypergraph determines the following notion of distance between sites.

\begin{definition}[Interaction hypergraph and induced distance]
\label{def:support-hypergraph}
Let $\graf=(\Lambda,\mathcal E)$ be a hypergraph on a finite set of sites $\Lambda$.
For distinct $i,j\in\Lambda$, define $\dist(i,j)$ to be the minimum integer $k\geq 1$ for which there exist
$Z_1,\dots,Z_k\in\mathcal E$ satisfying
\[
i\in Z_1,\qquad j\in Z_k,\qquad
Z_p\cap Z_{p+1}\neq\varnothing\quad\text{for all }1\leq p<k.
\]
Set $\dist(i,i)=0$, and set $\dist(i,j)=\infty$ if no such chain exists.
For $X,Y\subseteq\Lambda$, define
\[
\dist(X,Y):=\min_{i\in X,\,j\in Y}\dist(i,j),
\]
with $\dist(X,Y)=\infty$ if either set is empty.
We abbreviate $\dist(\{i\},X)$ as $\dist(i,X)$.
For $r\geq 0$, define
\[
\ball(X,r):=\{i\in\Lambda:\dist(i,X)\leq r\}.
\]
We also write $\ball(i,r):=\ball(\{i\},r)$.
\end{definition}

\begin{fact}[Interaction strength and ball growth]
\label{fact:interaction-growth}
Under the bounded-coupling assumptions of \Cref{def:interaction-parameters},
\[
\zeta\leq DL,\qquad
|\ball(i,r)|\leq(DL)^r
\quad\text{for every }i\in\Lambda\text{ and integer }r\geq0.
\]
\end{fact}
\begin{proof}
The bound on $\zeta$ follows directly from its definition.
Each site has at most $D(L-1)$ other sites at distance one, so $|\ball(i,1)|\leq1+D(L-1)\leq DL$.
Since $\ball(i,r+1)\subseteq\bigcup_{j\in\ball(i,r)}\ball(j,1)$ and $\ball(i,0)=\{i\}$, induction gives the ball-growth bound.
\end{proof}

\begin{definition}[Interaction truncation and restriction to a region]
\label{def:interaction-truncated-Ham}
For a decomposition as in \Cref{def:interaction-parameters} and a region $\Omega\subseteq\Lambda$, define the \emph{interaction-truncated Hamiltonian}
\begin{equation}\label{eq:coupling-truncated-hamiltonian}
H_\Omega^{\mathrm{int}}:=V+\sum_{a\in\mathcal J:\,Z_a\subseteq\Omega}h_a.
\end{equation}
This truncation retains all onsite terms in $V$.
The \emph{restricted Hamiltonian} retains only terms inside $\Omega$:
\[
H_\Omega:=\sum_{i\in\Omega}v_i+\sum_{a\in\mathcal J:\,Z_a\subseteq\Omega}h_a.
\]
For $i\in\Lambda$ and an integer $r\geq0$, write
\[
H_{i,r}:=H_{\ball(i,r)}^{\mathrm{int}}.
\]
When the center $i$ is understood, we abbreviate $H_{i,r}$ as $H_r$.
In particular, $H_{i,0}=V$.
\end{definition}

%% file: all_qubit.tex
\section{Preliminaries}

\begin{definition}[Support of an operator]
    For an $n$-qubit operator $P$, its \emph{support}, $\supp(P) \subseteq [\qubits]$ is the subset of qubits that $P$ acts non-trivially on.
    That is, $\supp(P)$ is the minimal set of qubits such that $P$ can be written as $P = O_{\supp(P)} \otimes \id_{[n] \setminus \supp(P)}$ for some operator $O$.
\end{definition}

Next, recall the definition of the partial trace:

\begin{definition}[Partial trace]
\label{def:partial-trace}
Let $\mathcal H=(\mathbb C^2)^{\otimes n}$ and let $S\subseteq[n]$ be a subset of qubits, with complement $S^c=[n]\setminus S$.
We write $\Tr_{S^c}$ for the partial trace over the qubits in $S^c$, i.e.\ the unique linear map
$\Tr_{S^c}:\mathcal L(\mathcal H)\to \mathcal L\!\big((\mathbb C^2)^{\otimes |S|}\big)$
satisfying
\[
\Tr\!\big(M\,\Tr_{S^c}(X)\big) \;=\; \Tr\!\big((M\otimes I_{S^c})\,X\big)
\quad\text{for all }X\in\mathcal L(\mathcal H),\; M\in \mathcal L\!\big((\mathbb C^2)^{\otimes |S|}\big),
\]
where $I_{S^c}$ denotes the identity on the qubits in $S^c$.
\end{definition}

Next, recall the Pauli matrices, which form a basis for the real vector space of $2\times 2$ Hermitian matrices.
\begin{definition}[Pauli matrices] \label{def:paulis}
    The Pauli matrices are the following $2 \times 2$ Hermitian matrices.
    \begin{equation*}
    \sigma_X = \begin{pmatrix}
        0 & 1 \\
        1 & 0
    \end{pmatrix}, \qquad \sigma_Y = \begin{pmatrix}
        0 & -\ii \\
        \ii & 0
    \end{pmatrix}, \qquad \sigma_Z = \begin{pmatrix}
        1 & 0\\
        0& -1
    \end{pmatrix}.
    \end{equation*}
    For $j\in [n]$ and $P\in \{X, Y, Z\}$, we denote by $\sigma_P^{(j)}$ the $n$-qubit operator that acts as $\sigma_P$ on qubit $j$ and as the identity on all other qubits. 
\end{definition}

\begin{definition}[Local Hamiltonians with arbitrary external field]
\label{def:local-hamiltonian}
Let $H = V + W$ be a Hamiltonian where $V= \sum_{i\in[n]} v_i$ are on-site potentials and $W = \sum_{a \in [m]} H_a$ are the bounded interaction terms. 
WLOG we assume that $|\supp(H_a)|\geq 2$. 
The local interaction strength $\zeta$ and the bounded $L$-local degree-$D$ convention are as in \Cref{def:interaction-parameters}.
The associated interaction hypergraph, distance $\dist$, and balls $\ball$ are as in \Cref{def:interaction-parameters,def:support-hypergraph}.
\end{definition}

\section{Lindbladian evolution under external field}

In this section, we revisit the Lindbladian construction of Chen, Kastoryano and Gilyén~\cite{ckg23} for preparing the Gibbs state $\sigma \propto e^{-\beta H}$. Although~\cite{ckg23} states all of its formal results for the special case of $\sigma = \eta = \Delta = 1/\beta$ (also used in~\cite{rfa24a, TZ25}), the underlying construction in fact admits a family of parameters $(\sigma, \eta, \Delta)$ satisfying $\beta(\eta^2 + \sigma^2) = 2\Delta$. We state this general family explicitly, and then further generalize it by allowing the parameters to be \emph{site-dependent}. This extra flexibility lets us tailor the construction to Hamiltonians that include an arbitrary on-site external field, yielding what we call the \emph{field-resonant Lindbladian}. We show that the field-resonant Lindbladian continues to satisfy detailed balance and that its building blocks remain quasi-local even when the on-site field is exponentially large.

\subsection{Field-resonant Lindbladians}
\label{subsec:field-resonant-lindbladian}
We begin with the general construction, and then specialize to the field-resonant choice of parameters.

\begin{definition}[General Lindbladians from filter and transition-weight functions; adapted from~\cite{ckg23}]
\label{def:general_lindblad}
    Let $H$ be an $n$-qubit Hamiltonian and $\beta > 0$ be an inverse temperature. 
    For each site $j\in [n]$, given $\sigma_j > 0$, 
    define the \emph{filter function} at site $j$ as
    \begin{equation*}
        f_j(t) \coloneqq (2/\pi)^{1/4}\cdot \sigma_j^{1/2} \cdot \exp(-\sigma_j^2 t^2).
    \end{equation*}
    Let $P$ be a Pauli operator at site $j$. The associated \emph{jump operator} is given by
    \begin{equation*}
        A^P(\omega) \coloneqq \frac{1}{\sqrt{2\pi}} \int_{-\infty}^{\infty} e^{\ii Ht} {P} e^{-\ii H t} e^{-\ii \omega t}  f_j(t) \diff t. 
    \end{equation*}
    Given $\Delta_j, \eta_j>0$ such that $\beta(\eta_j^2 + \sigma_j^2) = 2\Delta_j$, 
    define the \emph{transition-weight function} at site $j$ as
    \begin{equation*}
        \gamma_j(\omega) \coloneqq \exp\left( - \frac{(\omega + \Delta_j)^2}{ 2\eta_j^2} \right). 
    \end{equation*}
    The associated \emph{dissipative evolution} is given by
    \begin{equation*}
        \calL_{\mathrm{diss}}^P(\rho) \coloneqq \int_{-\infty}^{\infty} \gamma_j(\omega) \cdot \Paren{  A^{P}(\omega) \rho A^{P}(\omega)^\dagger - \frac12\braces{A^{P}(\omega)^\dagger A^{P}(\omega), \rho} } \diff \omega . 
    \end{equation*}
    The \emph{coherent term} associated to $P$ is given by a Hermitian matrix
    \begin{align*}
        C^{P} \coloneqq \int_{-\infty}^\infty b_{1, j}(t)\cdot e^{-\ii H t} \left(\int_{-\infty}^\infty b_{2, j}(t') \cdot e^{\ii H t'} {P}^\dagger e^{-2\ii H t'}  {P} e^{\ii H t'}  \diff t' \right) e^{\ii H t} \diff t, 
    \end{align*}
    where
    \begin{align*}
        b_{1,j}(t) &= \frac{\sigma_j}{\pi\beta} e^{\beta^2\sigma_j^2 / 8} \Paren{\sech\left(\frac{2\pi t}{\beta}\right) *_t \sin(-\beta \sigma_j^2 t) e^{-2\sigma_j^2 t^2}},\\
        b_{2,j}(t) &= 2 \eta_j \cdot  \exp\Paren{-\frac{4 \Delta_j t^2}{\beta} - 2\ii \Delta_j t}. 
    \end{align*}
    Here $*_t$ denotes the convolution of two functions over variable $t$. 
    The Lindbladian at site $j$ is given by
    \begin{align*}
        \calL_j(\rho) = \sum_{P \in \jumps{j}} \calL^P(\rho), \quad \text{where } \calL^P(\rho) \coloneqq - \ii [C^P, \rho] + \calL_{\mathrm{diss}}^P(\rho) . 
    \end{align*}
    The full Lindbladian determined by $\{(\sigma_j, \eta_j, \Delta_j)\}_{j\in [n]}$ is given by
    $\calL(\rho) = \sum_{j=1}^n \calL_j(\rho)$. 
\end{definition}

The family in \Cref{def:general_lindblad} differs from the Lindbladian of~\cite{ckg23} in two ways. First, we allow the parameters $(\sigma_j, \eta_j, \Delta_j)$ to vary across sites, whereas~\cite{ckg23} uses a single global choice. Second, while~\cite{ckg23} introduces the filter and transition-weight functions in this general parameterized form, all of its formal statements and proofs are carried out for the special case of $\sigma = \eta = \Delta = 1/\beta$, and the coherent term for the general case is only written down implicitly inside the proof of the special case. 
So for completeness we give a self-contained proof for the detailed balance of the general site-dependent family in \Cref{sec:kms-detailed-balance}; the argument closely follows the one in~\cite{ckg23}. We record the resulting property formally:

\begin{restatable}[Detailed balance]{proposition}{kms}\label{prop:kms}
    For any Hamiltonian $H$ and any inverse temperature $\beta >0$, the general Lindbladian with parameters $\{(\sigma_j, \eta_j, \Delta_j)\}_{j\in [n]}$ where $\beta(\eta_j^2 + \sigma_j^2) = 2\Delta_j$, as defined in \Cref{def:general_lindblad}, satisfies KMS detailed balance, and the Gibbs state $e^{-\beta H}/\tr(e^{-\beta H})$ is a stationary state.
\end{restatable}

With the general construction in hand, we now introduce the particular choice of parameters that we use to handle on-site potentials.

\begin{definition}[Field-resonant Lindbladians]\label{def:field_lindblad}
    Let $\beta > 0$ be an inverse temperature and $H = \sum_{j=1}^n v_j + W$ be a local Hamiltonian with unbounded external field as defined in \Cref{def:local-hamiltonian}. 
    For each $j\in [n]$, 
    let $\Delta_j$ be the maximum between $1/\beta$ and $\lambda_{\max}(v_j) - \lambda_{\min}(v_j)$, i.e. 
    the spectral gap of $v_j$. 
    Set $\eta_j = \sigma_j = \sqrt{\Delta_j / \beta}$. 
    The Lindbladian with respect to such $\{\Delta_j, \eta_j, \sigma_j\}_{j\in [n]}$ as defined in \Cref{def:general_lindblad} is called a \emph{field-resonant Lindbladian} of $H$ at inverse temperature $\beta$. We denote it by $\calL_{\mathsf{FR}}=\sum_j\calL_{\mathsf{FR},j}$, where $\calL_{\mathsf{FR},j}$ is the site contribution $\calL_j$ evaluated at these parameters. 
\end{definition}

\subsection{Field-independent Lieb-Robinson bound}

In this subsection, we prove a Lieb-Robinson bound for local Hamiltonians with an unbounded external field. Lieb-Robinson bounds are a basic tool for formalizing the following physics principle: under time evolution by a local Hamiltonian, any initial perturbation cannot instantly propagate across the system. Concretely, if we start with a local operator $O_X$ supported on $X$, and time evolve it with respect to a Hamiltonian $H$, then for short enough time, it can be well-approximated by the time evolution under the Hamiltonian supported near $X$. Further, its influence on a distant region $Y$ decays exponentially in the distance between $X$ and $Y$. The corresponding upper bound on the propagation speed is known as the Lieb-Robinson velocity.

Existing off-the-shelf Lieb-Robinson bounds typically assume that every term in the Hamiltonian is uniformly bounded, so the resulting Lieb-Robinson velocity is governed by the largest local energy scale.
In our setting, a direct application would give an unnecessarily field-dependent Lieb-Robinson velocity.
A standard way around this, going back to Lieb-Robinson bounds for Hamiltonians with unbounded on-site terms, is to factor out the on-site dynamics and work with the interaction picture \cite{NRSS09}. We follow this strategy here, combined with a time-dependent Lieb-Robinson bound generalized from \cite{hhkl21}.
Formally, we prove the following:

\begin{theorem}[Field-independent Lieb-Robinson bound]
\label{thm:lr-potential}
Let $H = V + W$ be an $n$-qubit local Hamiltonian with external field as defined in \Cref{def:local-hamiltonian} with local interaction strength $\zeta$.
Use the interaction-truncated Hamiltonians $H_\Omega^{\mathrm{int}}$ from \Cref{def:interaction-truncated-Ham}.
Let $X, \Omega, \Lambda$ be three regions where $X\subseteq \Omega \subseteq \Lambda \subseteq [n]$.
Suppose $\ell = \dist(X, \Lambda \setminus \Omega) \geq 2$.
Let $O_X$ be any operator supported on $X$.
Then for any evolution time $T\in \R$,
\begin{align*}
    \norm*{ e^{\ii H_\Lambda^{\mathrm{int}} T}\cdot O_X \cdot e^{-\ii H_\Lambda^{\mathrm{int}} T} - e^{\ii H_\Omega^{\mathrm{int}} T} \cdot O_X \cdot e^{-\ii H_\Omega^{\mathrm{int}} T}} \leq \norm{O_X} \cdot \abs{X} \cdot \frac{(2 \zeta \abs{T})^\ell}{\ell !}.
\end{align*}
\end{theorem}

We note that we only require the following corollary that enables us to bound the quasi-locality of the jump operators and coherent term of our field-resonant Lindbladian:
\begin{corollary}[Shell decomposition for Heisenberg evolution with unbounded potentials]\label{lem:LR_radius_r}
    Let $H = V + W$ be a local Hamiltonian with external field as defined in \Cref{def:local-hamiltonian} with local interaction strength $\zeta$. Suppose $A$ is an operator supported at site $i \in [n]$. For any integer $r\geq 1$, let $H_r$ be the interaction-truncated Hamiltonian at site $i$ with radius $r$ as defined in \Cref{def:interaction-truncated-Ham}.
    Write
    \begin{equation*}
        F_r^{A, t} \coloneqq e^{\ii H_r t} A e^{-\ii H_r t} - e^{\ii H_{r-1} t} A e^{-\ii H_{r-1} t}.
    \end{equation*}
    Then $F_r^{A,t}$ is supported on $\ball_r$ and $\norm{F_r^{A,t}} \leq \norm{A} \cdot \frac{(2 \zeta \abs{t})^r}{r!}$.
\end{corollary}

The proofs of \Cref{thm:lr-potential,lem:LR_radius_r} are given in \Cref{sec:LR_time_dependent_proof}.

\subsection{Quasi-locality of the field-resonant Lindbladian}

\paragraph{Obstacles to quasi-locality.} A priori, the Heisenberg evolution may spread information
across the entire system, since it involves an integral over all time $t$. To account for this, the particular choice of the \emph{time kernels}, $f(t), b_1(t)$ and $b_2(t)$ exponentially suppress large $t$ and without on-site potentials can be controlled via a standard Lieb-Robinson bound. When the Hamiltonian has a large external field, i.e. $H=V+W$, where $V$ is on-site and potentially unbounded, one must additionally ensure
that the locality estimates are \emph{uniform in the field strength}. This is precisely where we use our shell decomposition for Heisenberg evolution from the previous subsection (see \cref{lem:LR_radius_r}). We begin by showing that the jump operators are quasi-local: at a high level, we write $A^P(\omega)$, the jump operator at site $i$, as a sum of operators, where each one is supported on a radius $r$ ball around site $i$ and show that the operator norm of each one decays exponentially in the radius, as long as $\zeta \cdot \beta < 1/4$.

\begin{proposition}[Jump operators are quasi-local]
\label{prop:jump-operators-quasi-local}
Let $0< \beta \leq 1/(4\zeta)$ and $P$ be a Pauli operator at site $i$. 
Recall that the jump operator in the field-resonant Lindbladian defined in \Cref{def:field_lindblad}
\begin{equation*}
    A^P(\omega) = \frac{1}{\sqrt{2\pi}} \int_{-\infty}^{\infty} e^{\ii Ht} {P} e^{-\ii H t} e^{-\ii \omega t}  f_i(t) \diff t , \quad \text{where }f_i(t) = \sqrt{\sigma_i \sqrt{2/\pi}} \cdot e^{-\sigma_i^2 t^2}. 
\end{equation*}
Then we can write $ A^P(\omega) = \sum_{r\geq 0} G^{P, \omega}_r$,
where $G_r^{P, \omega}$ is supported on $\ball(\{i\}, r)$ and 
\begin{equation*}
    \norm{G_r^{P, \omega}} \leq (2\pi)^{-1/4} \cdot \sigma_i^{-1/2} \cdot (2\zeta \sigma_i^{-1})^r,
\end{equation*}
for each integer $r\geq 0$. 
In particular,
\begin{equation*}
    G_0^{P, \omega} = \frac{1}{\sqrt{2\pi}} \int_{-\infty}^{\infty}  e^{\ii t v_i} P e^{-\ii tv_i} \cdot e^{-\ii \omega t} f_i(t) \diff t.
\end{equation*}
\end{proposition}
\begin{proof}
    For any integer $r\geq 0$, let $H_r$ be the interaction-truncated Hamiltonian of $H$ at site $i$ with radius $r$ as defined in \Cref{def:interaction-truncated-Ham}. Recall, this truncation only affects the interaction terms and keeps all the on-site terms around. 
    Then we can write
    \begin{equation*}
        e^{\ii Ht}\cdot P\cdot e^{-\ii H t} = \Big(\sum_{r\geq 1} \underbrace{e^{\ii H_r t}\cdot P\cdot e^{-\ii H_r t} - e^{\ii H_{r-1} t}\cdot P\cdot e^{-\ii H_{r-1} t}}_{  F_r^{P, t}} \Big) + \underbrace{e^{\ii H_0 t}\cdot P\cdot e^{-\ii H_0 t}}_{ F_0^{P, t}}. 
    \end{equation*}
    Here, $F_r^{P, t}$ is roughly the "gain in information" due to the Heisenberg evolution, when you allow interactions to distance $r$ instead of $r-1$.  
    Since $P$ is a Pauli at site $i$, then
    \begin{equation*}
        F_0^{P,t} = e^{\ii t V} P e^{-\ii tV} = e^{\ii t v_i} P e^{-\ii tv_i}
    \end{equation*}
    is supported only on $\ball(\{i\}, 0) = \{i\}$ and has norm exactly $1$. 
    For $r\geq 1$, it follows from our shell decomposition in \Cref{lem:LR_radius_r} that $F_r^{P,t}$ is supported on $\ball(\{i\}, r)$ and $\norm{F_r^{P,t}} \leq \frac{(2\zeta \abs{t})^r}{r!}$. 

    Then the jump operator can be written as
    \begin{align*}
        A^P(\omega) &= \frac{1}{\sqrt{2\pi}} \int_{-\infty}^{\infty} e^{\ii Ht} P e^{-\ii H t} e^{-\ii \omega t} f_i(t) \diff t  \\
        &= \sum_{r\geq 0} \underbrace{\frac{1}{\sqrt{2\pi}} \int_{-\infty}^{\infty}  F_r^{P,t} e^{-\ii \omega t} f_i(t) \diff t}_{  G_r^{P, \omega}} \ ,
    \end{align*}
    where $G_r^{P,\omega}$ is supported on $\ball(\{i\}, r)$ and 
    \begin{align*}
        \norm{G_r^{P, \omega}} &\leq \frac{1}{\sqrt{2\pi}} \int_{-\infty}^{\infty} \norm{F_r^{P,t}}\cdot f_i(t) \diff t \\
        &\leq \frac{1}{\sqrt{2\pi}} \int_{-\infty}^{\infty} \frac{(2\zeta \abs{t})^r}{r!} \cdot \sqrt{\sigma_i \sqrt{2/\pi}} \cdot e^{-\sigma_i^2 t^2} \diff t \\
        &= \sqrt{\frac{\sigma_i \sqrt{2/\pi}}{2\pi}} \cdot \frac{(2\zeta)^r}{r!} \int_{-\infty}^{\infty} |t|^r \cdot e^{-\sigma_i^2 t^2} \diff t \\
        &= \sqrt{\frac{\sigma_i \sqrt{2/\pi}}{2\pi}} \cdot \frac{(2\zeta)^r}{r!} \cdot \sigma_i^{-(r+1)} \cdot \Gamma\left(\frac{r+1}{2}\right) \tag{Moments of Gaussian distributions} \\
        &\leq (2\pi)^{-1/4} \cdot \sigma_i^{-1/2} \cdot (2\zeta \sigma_i^{-1})^r,
    \end{align*}
    where the last inequality is because the Gamma function $\Gamma\left(\frac{r+1}{2}\right) \leq \sqrt{\pi} (r!)$ for any integer $r\geq 0$. 
    Since $\sigma_i = \sqrt{\Delta_i/\beta}$ and $\Delta_i\geq 1/\beta$, we have $\sigma_i^{-1}\leq \beta$. 
    Since $0<\beta \leq 1/(4\zeta)$, we have the desired bound on $\norm{G}_{r}^{P, \omega}$. 
\end{proof}

Next, we use a similar argument to prove that the coherent term is quasi-local, with the new caveat that we are integrating a product of Heisenberg evolution on two Paulis. At a high level, we first replace the full dynamics by interaction-truncated dynamics inside growing metric balls around the support of $P$, and telescope the resulting expression into ``shell increments'' indexed by the radius $r$.
Each shell increment is supported on $\ball(\{i\},r)$ and, by our shell decomposition bounds from
\Cref{lem:LR_radius_r}, its norm is suppressed by a factorial factor $(2\zeta|t|)^r/r!$ (uniformly in the external
field strength).
Second, the moment estimates for the kernels $b_1$ and $b_2$, proved in \Cref{subsec:kernels_moments},
grow at most factorially in the moment order, canceling the factorial in the Lieb-Robinson estimate.
Together, these estimates yield the following quasi-local decomposition of the coherent term.
\begin{proposition}[Coherent terms are quasi-local]
\label{prop:coherent-quasi-local}
    Let $0<\beta \leq 1/(12 \zeta)$ and $P$ be a Pauli operator at site $i$. Then we can write the coherent term in the field-resonant Lindbladian as
    \begin{equation*}
        C^P = \sum_{r\geq 0} K_r^P, \quad \textrm{ where } \; \norm{K_0^P}\leq 3\log(2\sqrt{\Delta_i \beta })  \; \textrm{ and } \; \norm{K_r^P} \leq 6(6\zeta \beta)^r \;\; \textrm{ for all } r\geq 1\,.
    \end{equation*}
    Further, each $K_r^P$ is Hermitian and is supported on $\ball(\{i\}, r)$. 
\end{proposition}
\begin{proof}
    For any integer $r\geq 0$, let $H_r$ be the interaction-truncated Hamiltonian of $H$ at site $i$ with radius $r$ as defined in \Cref{def:interaction-truncated-Ham}. 
    Then we can write
    \begin{align*}
        P(t) &\coloneqq e^{\ii Ht}P e^{-\ii Ht}, \\
        P_r(t) &\coloneqq e^{\ii H_r t}P e^{-\ii H_r t}, \quad \text{for }r\geq 0. 
    \end{align*}
    Then the corresponding coherent term
    \begin{align*}
        C^P = \iint_{-\infty}^\infty b_1(t)b_2(t') \underbrace{e^{-\ii H t} e^{\ii H t'} P e^{-\ii H t'}  e^{\ii  H t} }_{P(t'-t)} \cdot \underbrace{e^{-\ii H t} e^{-\ii H t'} P e^{\ii H t'} e^{\ii H t}}_{P(-(t'+t))} \diff t' \diff t. 
    \end{align*}
    Note that we drop the index $i$ in $b_{1,i}(t)$ and $b_{2,i}(t')$ for simplicity. 
    Let us split $P(t'-t) \cdot P(-(t'+t))$ into pieces with increasing support:
    \begin{align*}
        P(t'-t) \cdot P(-(t'+t)) &= \Big(\sum_{r\geq 1} \underbrace{P_r(t'-t) \cdot P_r(-(t'+t)) - P_{r-1}(t'-t) \cdot P_{r-1}(-(t'+t))}_{  R_{r}^{P, t, t'}}\Big) \\
        & \qquad + \underbrace{P_0(t'-t) \cdot P_0(-(t'+t))}_{  R_0^{P, t, t'}}. 
    \end{align*}
    Because both $P_0(t'-t)$ and $P_0(-(t'+t))$ are supported on site $i$, $R_0^{P, t, t'}$ is also supported on site $i$ and 
    \begin{equation*}
        \norm{R_0^{P, t, t'}} \leq \norm{P_0(t'-t)}\cdot \norm{P_0(-(t'+t))} = 1. 
    \end{equation*}
    Similarly, it follows from \Cref{lem:LR_radius_r} that for each $r\geq 1$, $R_r^{P, t, t'}$ is supported on $\ball(\{i\}, r)$. 
    We now bound the operator norm of $R_r^{P, t, t'}$. Note that any matrices $X_1, X_2, Y_1, Y_2$, we can bound the operator norm
    \begin{align*}
        \norm{X_1 X_2 - Y_1 Y_2} &= \norm{X_1 X_2 - X_1 Y_2 + X_1 Y_2 - Y_1 Y_2} \\
        &\leq \norm{X_1 X_2 - X_1 Y_2} + \norm{X_1 Y_2 - Y_1 Y_2} \\
        &\leq \norm{X_1} \cdot \norm{X_2 - Y_2} + \norm{X_1 - Y_1} \cdot \norm{Y_2}. 
    \end{align*}
    Using this inequality and $\norm{P_r(t'-t)} = \norm{P_{r-1}(-(t'+t))} = 1$, we have that
    \begin{align*}
        \norm{R_r^{P, t, t'}} &\leq \norm{P_r(t'-t) - P_{r-1}(t'-t)} + \norm{P_r(-(t'+t)) - P_{r-1}(-(t'+t))} \\
        &\leq \frac{(2\zeta )^r }{r!}\cdot (\abs{t'-t}^r + \abs{t'+t}^r) \tag{\Cref{lem:LR_radius_r}} \,.
    \end{align*}
    So overall,
    \begin{align*}
        C^P = \sum_{r\geq 0} \underbrace{\iint_{-\infty}^\infty b_1(t)b_2(t') R_r^{P, t, t'} \diff t' \diff t}_{\coloneqq \widetilde{K}_r^P}, 
    \end{align*}
    where for each $r\geq 0$, $\widetilde{K}_r^P$ is supported on $\ball(\{i\}, r)$. 
    Set $K_r^P \coloneqq \frac12(\widetilde{K}_r^P + (\widetilde{K}_r^P)^\dagger)$ which is Hermitian and is also supported on $\ball(\{i\}, r)$. 
    Since the coherent term $C^P$ is Hermitian, we have $C^P = \sum_{r\geq 0} K_r^P$. 
    Furthermore,
    \begin{align*}
        \norm{K_r^P} \leq \norm{\widetilde{K}_r^P} &\leq \iint_{-\infty}^\infty \abs{b_1(t)b_2(t')}\cdot \norm{R_r^{P, t, t'}} \diff t' \diff t \\
        &\leq \frac{(2\zeta )^r }{r!} \iint_{-\infty}^{\infty}\abs{b_1(t)b_2(t')} \cdot (\abs{t'-t}^r + \abs{t'+t}^r) \diff t' \diff t \,.
    \end{align*}
    We now use the moment bounds for $b_1$ and $b_2$ proved in \Cref{lem:b_1_zero_moment,lem:b_1_higher_moments,lem:b_2_moments}, which we record here for the field-resonant parameters.
    For any integer $k\geq 0$, let us write
    \begin{equation*}
        M_{1,k} \coloneqq \int_{-\infty}^{\infty} \abs{b_1(t)}\cdot \abs{t}^k \diff t \,, \quad \text{and}\quad M_{2,k} \coloneqq \int_{-\infty}^{\infty} \abs{b_2(t)}\cdot \abs{t}^k \diff t \,.
    \end{equation*}
    By the definition of the field-resonant Lindbladian, we have $\eta = \sigma = \sqrt{\Delta / \beta}$ and $\Delta \geq 1/\beta$. (Note that here we drop the index $i$ in $\sigma_i$, $\eta_i$, and $\Delta_i$ for simplicity.)
    Therefore, 
    \begin{itemize}
        \item (\Cref{lem:b_1_zero_moment}) $M_{1,0}\leq \frac{1}{\sqrt{2} \pi^{3/2}} \Paren{\frac{2}{\sqrt{\pi}} + \log \Paren{1+ \frac{\sqrt{2}}{\pi}\cdot \sqrt{\Delta \beta}}}$
        \item (\Cref{lem:b_1_higher_moments}) $M_{1,k}\leq \frac{2^{(k+5)/2} \cdot k!}{ \pi^{3/2}} \Paren{\sqrt{\frac{\beta}{\Delta}} + \frac{\beta}{\pi}}^{k}$ for $k\geq 1$
        \item (\Cref{lem:b_2_moments}) $M_{2,k} = 2^{-k}\cdot \Paren{\frac{\beta}{\Delta}}^{k/2} \cdot \Gamma\Paren{\frac{k+1}{2}}$ for $k\geq 0$
    \end{itemize}
    
    When $r=0$, 
    \begin{align*}
        \norm{K_0^P} &\leq 2 \iint_{-\infty}^{\infty}\abs{b_1(t)b_2(t')} \diff t' \diff t \\
        &= 2 M_{1,0} \cdot M_{2,0} \\
        &\leq 2\cdot \frac{1}{\sqrt{2} \pi^{3/2}} \Paren{\frac{2}{\sqrt{\pi}} + \log \Paren{1+\frac{\sqrt{2}}{\pi}\cdot\sqrt{\Delta \beta}}} \cdot \sqrt{\pi} \tag{$\Gamma(\frac{k+1}{2})\leq \sqrt{\pi}\cdot  k!$}\\
        &\leq 1 + \log (1 + \sqrt{\Delta \beta}) \\
        &\leq 3\log(2\sqrt{\Delta \beta }) \tag{$1\leq \sqrt{\Delta \beta}$ and $1\leq 2\log(2\sqrt{\Delta \beta})$}\,.
    \end{align*}
    When $r\geq 1$,
    \begin{align*}
        \norm{K_r^P} &\leq \frac{(2\zeta )^r }{r!} \iint_{-\infty}^{\infty}\abs{b_1(t)b_2(t')} \cdot 2(\abs{t'} + \abs{t})^r \diff t' \diff t \tag{triangle inequality}\\
        &= \frac{2(2\zeta )^r }{r!} \sum_{k=0}^{r} \binom{r}{k} \cdot M_{1,k} \cdot M_{2,r-k} \\
        &= 2(2\zeta )^r  \sum_{k=0}^{r}  \frac{M_{1,k}}{k!} \cdot \frac{M_{2,r-k}}{(r-k)!} \,.
    \end{align*}
    For any $m\geq 0$, 
    \begin{align*}
        \frac{M_{2,m}}{m!} = 2^{-m}\cdot \Paren{\frac{\beta}{\Delta}}^{m/2} \cdot \Gamma\Paren{\frac{m+1}{2}} \frac{1}{m!} \leq \sqrt{\pi} \Paren{\frac{\beta}{4\Delta}}^{m/2}. 
    \end{align*}
    Then
    \begin{align*}
        \sum_{k=1}^{r}  \frac{M_{1,k}}{k!} \cdot \frac{M_{2,r-k}}{(r-k)!} &\leq \sum_{k=1}^r \frac{2^{(k+5)/2}}{ \pi^{3/2}} \Paren{\sqrt{\frac{\beta}{\Delta}} + \frac{\beta}{\pi}}^{k} \cdot \sqrt{\pi} \Paren{\frac{\beta}{4\Delta}}^{(r-k)/2} \\
        &\leq \frac{2^{5/2}}{\pi} \sum_{k=1}^r 2^{k/2} \Paren{\beta + \frac{\beta}{\pi}}^{k} \cdot \Paren{\frac{\beta}{2}}^{r-k}
        \tag{$1/\Delta \leq \beta$} \\
        &\leq \frac{2^{5/2}}{\pi} \Paren{ \sqrt{2} \Paren{\beta + \frac{\beta}{\pi}} + \frac{\beta}{2}}^{r} \tag{$\sum_{k=1}^r A^k B^{r-k} \leq (A+B)^r$ for $A,B\geq 0$} \\
        &\leq 2\cdot 3^r \beta^r
    \end{align*}
    and
    \begin{align*}
        M_{1,0} \cdot \frac{M_{2,r}}{r!} &\leq \frac{1}{\sqrt{2} \pi^{3/2}} \Paren{\frac{2}{\sqrt{\pi}} + \log \Paren{1+ \frac{\sqrt{2}}{\pi}\cdot \sqrt{\Delta \beta}}} \sqrt{\pi} \Paren{\frac{\beta}{4\Delta}}^{r/2} \\
        &\leq \frac{1}{\sqrt{2} \pi} \Paren{\frac{2}{\sqrt{\pi}} + \frac{\sqrt{2}}{\pi}\cdot \sqrt{\Delta \beta}}  \Paren{\frac{\beta}{4\Delta}}^{r/2} \tag{$\log(1+x)\leq x$ for $x\geq 0$} \\
        &= \frac{1}{\sqrt{2} \pi} \Paren{\frac{2}{\sqrt{\pi}} \sqrt{\frac{\beta}{4\Delta}} + \frac{\sqrt{2}}{\pi}\cdot \sqrt{\Delta \beta} \sqrt{\frac{\beta}{4\Delta}}}  \Paren{\frac{\beta}{4\Delta}}^{(r-1)/2} \\
        &\leq \frac{1}{\sqrt{2} \pi} \Paren{\frac{\beta}{\sqrt{\pi}} + \frac{\beta}{\sqrt{2} \pi}}  \Paren{\frac{\beta}{2}}^{r-1} \tag{$1/\Delta \leq \beta$} \\
        &\leq 2^{-r}  \beta^r \,.
    \end{align*}
    Overall, when $r\geq 1$, 
    \begin{equation*}
        \norm{K_r^P}\leq 2(2\zeta)^r \Paren{2\cdot 3^r \beta^r + 2^{-r}  \beta^r} \leq 6(6\zeta \beta)^r \,.
    \end{equation*}
    Since $0< \beta \leq 1/(12\zeta)$, we have that $\sum_{r\geq 0} \norm{K_{r}^P}$ converges. 
\end{proof}

\subsection{Efficient simulation of the field-resonant Lindbladian}\label{sec:simulation_efficiency}

In this section, we include two remarks explaining how one expects the field-resonant Lindbladian evolution $e^{\calL_{\mathsf{FR}} t}$ on a constant-dimensional lattice can be implemented with two-qubit gate count of $\bigOt{tn}$ on a quantum computer. 

First, as a result of the quasi-locality of the jump operators and the coherent terms, we show in \Cref{lem:truncation} that the resulting Lindbladian can be truncated to radius $R=\bigO{\log(nt/\epsilon)}$ with diamond-norm error at most $\epsilon$; thus, up to this error, the dynamics is generated by strictly local terms on radius-$R$ patches. 
Then, combining this truncation with the block-encoding approach underlying \cite[Theorem~I.2]{ckg23} and the large-field Hamiltonian simulation result of Low and Wiebe \cite[Theorem~10]{LW18} suggests that, on a constant-dimensional lattice, the field-resonant Lindbladian can be compiled using $\bigOt{n}$ gates when $\beta=\bigO{1}$ and $t=\bigO{\log(n/\epsilon)}$.

\begin{lemma}[Lindbladian trunctaion error]\label{lem:truncation}
    Suppose $\beta \leq 1/(12\zeta)$. Recall the expansions of the jump operators $A_P(\omega)$ and the coherent terms $C^P$ in \Cref{prop:jump-operators-quasi-local,prop:coherent-quasi-local}. 
    For any integer $R\geq 1$, define the truncated jump operators and coherent terms at radius $R$ respectively by 
    \begin{equation*}
        A_{R}^P(\omega) \coloneqq \sum_{r= 0}^R G_r^{P, \omega}, \quad \text{and} \quad C^P_{R} \coloneqq \sum_{r=0}^R {K_r^P} \,,
    \end{equation*}
    and define the corresponding truncated field-resonant Lindbladian $\calL_{\mathsf{FR},R}$. Then
    \begin{equation*}
        \norm{e^{t \calL_{\mathsf{FR}}} - e^{t \calL_{\mathsf{FR}, R}}}_\diamond \leq 60 nt\cdot 2^{-R} \,.
    \end{equation*}
    As a result, when $R = \bigO{\log(nt/\epsilon)}$, we have $\norm{e^{t \calL_{\mathsf{FR}}} - e^{t \calL_{\mathsf{FR},R}}}_\diamond \leq \epsilon$. 
\end{lemma}
Recall from the proofs of \Cref{prop:jump-operators-quasi-local,prop:coherent-quasi-local}, these truncations of the jump operators and the coherent terms are consequences of truncating the Hamiltonian $H$ to radius $r$ centered at site $i$. 
\begin{proof}[Proof of \Cref{lem:truncation}]
We first bound the truncation error for a jump operator in operator norm:
\begin{align*}
    \norm{A^P(\omega) - A_{R}^P(\omega)} &\leq (2\pi)^{-1/4} \sigma_i^{-1/2} \sum_{r \geq R+1} (2\zeta \sigma_i^{-1})^r \tag{\Cref{prop:jump-operators-quasi-local}}\\
    &\leq (2\pi)^{-1/4} \sigma_i^{-1/2} \sum_{r \geq R+1} (2\zeta \beta )^r \tag{$\sigma_i^{-1}\leq \beta$} \\
    &\leq (2\pi)^{-1/4} \sigma_i^{-1/2} (2\zeta \beta )^{R+1} 2 \tag{$\beta \leq 1/(4\zeta)$} \,.
\end{align*}
Similarly, with \Cref{prop:coherent-quasi-local}, the truncation error for a coherent term is bounded by
\begin{align*}
    \norm{C^P - C^P_R} \leq 6\sum_{r\geq R+1}(6\zeta \beta)^r \leq 12 (6\zeta \beta)^{R+1} \tag{$\beta \leq 1/(12\zeta)$} \,.
\end{align*}
We now convert these operator truncation into Lindbladian truncation. 
Write $D_A(X) \coloneqq AXA^\dagger - \frac12 \{A^\dagger A, X\}$ and the map $X\mapsto LXR$ as $\calM_{L, R}$. Then
\begin{align}
    \norm{D_A - D_B}_\diamond &= \norm*{AXA^\dagger - AXB^\dagger + AXB^\dagger - BXB^\dagger - \frac12 \Paren{CX + XC} }_\diamond \tag{Set $C = A^\dagger A - B^\dagger B$} \nonumber\\
    &= \norm{AX(A - B)^\dagger}_\diamond + \norm{(A-B) X B^\dagger}_\diamond + \frac12\norm{CX}_\diamond + \frac12\norm{XC}_\diamond \tag{triangle inequality} \nonumber \\
    &\leq \norm{A} \cdot \norm{A-B} + \norm{A-B} \cdot \norm{B} + \norm{C} \nonumber\\
    &\leq 2 (\norm{A} + \norm{B}) \cdot \norm{A-B} \,,
    \label{eq:diamond_to_operator_norm_product}
\end{align}
where the first inequality is because the diamond of the map $X\mapsto LXR$ is upper bounded by $\norm{L} \cdot \norm{R}$ (see e.g.\ \cite[Theorem 3.62]{watrous18}). 
Therefore, the approximation error in the dissipative evoluation is given by
\begin{align*}
    &\norm*{\int_{-\infty}^{\infty} \gamma_i(\omega) \Paren{D_{A^P(\omega)} - D_{A_R^P(\omega)}} \diff \omega}_\diamond \\
    &\leq \int_{-\infty}^{\infty} \gamma_i(\omega) \norm*{\Paren{D_{A^P(\omega)} - D_{A_R^P(\omega)}}}_\diamond \diff \omega \tag{triangle inequality}\\
    &\leq 2\int_{-\infty}^{\infty} \gamma_i(\omega) \Paren{\norm{A^P(\omega)} + \norm{A_R^P(\omega)}}\cdot \norm{A^P(\omega) - A_{R}^P(\omega)} \diff \omega \\
    &\leq 2\int_{-\infty}^{\infty} \gamma_i(\omega)\cdot \Paren{4 (2\pi)^{-1/4} \sigma_i^{-1/2}} \cdot \Paren{2 (2\pi)^{-1/4} \sigma_i^{-1/2} (2\zeta \beta)^{R+1}} \diff \omega \\
    &= 16 (2\pi)^{-1/2} \sigma_i^{-1} (2\zeta \beta)^{R+1} \sqrt{2\pi} \eta_j  \tag{$\int_{-\infty}^{\infty} \gamma_i(\omega) \diff \omega = \sqrt{2\pi} \eta_j$} \\
    &= 16 (2\zeta \beta)^{R+1} \tag{$\sigma_i = \eta_i$} \\
    &\leq 8 \cdot 2^{-R} \,. \tag{$2\zeta \beta \leq 1/2$}
\end{align*}
For the approximation error in the coherent term, under the assumption $6\zeta \beta \leq 1/2$, we have
\begin{equation*}
    \norm*{-\ii [C^P, \cdot] + \ii [C^P_R, \cdot]}_\diamond \leq 2 \norm{C^P - C^P_R} \leq 12 \cdot 2^{-R} \,,
\end{equation*}
where the first inequality is analogous to \Cref{eq:diamond_to_operator_norm_product}. 
Therefore, the truncated Lindbladin $\calL^P_R$ with respect to a Pauli $P$ at site $i$ satisfies
\begin{equation*}
    \norm{\calL^P - \calL^P_R}_\diamond \leq 20 \cdot 2^{-R} \,.
\end{equation*}
Summing over the $3n$ site-Pauli terms gives
\begin{equation*}
    \norm{\calL_{\mathsf{FR}} - \calL_{\mathsf{FR},R}}_\diamond \leq 60 n\cdot 2^{-R} \,.
\end{equation*}
Therefore, for any $t>0$, 
\begin{equation*}
    \norm{e^{t \calL_{\mathsf{FR}}} - e^{t \calL_{\mathsf{FR},R}}}_\diamond \leq t \norm{\calL_{\mathsf{FR}} - \calL_{\mathsf{FR},R}}_\diamond \leq 60 nt\cdot 2^{-R} \,,
\end{equation*}
where the first inequality is a special case of Duhamel's fomula (also see a proof in \cite[Lemma 6(i)]{vE24}). 
\end{proof}

\begin{remark}[Expected black-box implementation cost]
We do not prove a full compilation theorem for the field-resonant Lindbladian in this work.
However, adapting the block-encoding constructions used in the proof of \cite[Theorem~I.2]{ckg23}
suggests the following black-box cost.

After truncating the time integrals to $[-T, T]$ and discretizing both the time and frequency integrals, let $\Lambda_f$, $\Lambda_1$, and $\Lambda_2$ denote the corresponding $\ell_1$-normalization factors of the discretized kernels $f_j$, $b_{1,j}$, and $b_{2,j}$, maximized over $j\in[n]$. 
Then one expects that the Lindbladian evolution $e^{\mathcal L^\ast t}$ for any $t\ge 1$
can be implemented to error $\epsilon$ in diamond norm using
\begin{align*}
    &\bigOt*{t T \cdot \Paren{1+ \Lambda_f^2+ \Lambda_1 \Lambda_2}}  \quad && \text{total Hamiltonian simulation time, and} \\
    &\bigOt*{t \cdot \Paren{1+ \Lambda_f^2+ \Lambda_1 \Lambda_2}} \quad && \text{uses of a block-encoding for $\sum_{j\in [n]}\sum_{P\in \{X, Y, Z\}} \ket{j, P} \otimes \sigma_P^{(j)}$} \,,
\end{align*}
up to additional polylogarithmic factors from discretization, Hamiltonian simulation, and the
target precision.

For the field-resonant choice
$\Delta_j=\max\{1/\beta,\lambda_{\max}(v_j)-\lambda_{\min}(v_j)\}$ and
$\sigma_j=\eta_j=\sqrt{\Delta_j/\beta}$,
the kernel bounds from \Cref{subsec:kernels_moments} give
\begin{equation*}
    \Lambda_f=\bigO{\sqrt{\beta}},\qquad \Lambda_2=\bigO{1},\qquad
\Lambda_1= \bigO*{1+\log(1+\sqrt{\beta\Delta_{\max}})},
\end{equation*}
where $\Delta_{\max}:=\max_j \Delta_j$.\footnote{We remark that a straightforward adaptation of the proof in \cite{ckg23} would result in this logarithmic dependence on $\Delta_{\max}$. But it might be avoidable if one uses a different and more careful argument as the Fourier transform of $b_{1,j}$ exhibits a nice form. Since the algorithmic efficieny is not the main focus of this paper and a direct adaption can already handle exponentially large external field, we will leave a more careful analysis for future work. } 
Moreover, the explicit Gaussian forms of $f_j$ and $b_{2,j}$ indicate a time cutoff
$T=\bigOt{\beta}$. 
For $b_{1,j}$, \Cref{lem:b_1_fourier} shows that its Fourier transform $\widehat{b}_{1,j}(\omega)$ is analytic in the strip $\abs{\mathrm{Im}(\omega)} < 2\pi /\beta$. 
A standard contour-shift estimate for inverse Fourier transforms therefore yields $\abs{b_{1,j}(t)} \leq C \sigma_j e^{-\pi \abs{t} /\beta}$. Hence the $b_{1,j}$-integral can also be truncated to $[-T, T]$ with $T = \bigOt{\beta}$. 
\end{remark}

\newcommand{\Vol}{\mathrm{Vol}}

\begin{remark}[Expected gate complexity on a $d$-dimensional lattice]
Assume that the interaction Hamiltonian $W$ is local on a $d$-dimensional lattice where $d = \bigO{1}$.
Let $\mathcal L_R^\ast$ be the radius-$R$ truncation of the field-resonant Lindbladian from
\Cref{lem:truncation}. Then
\begin{equation*}
\|e^{t\mathcal L^\ast}-e^{t\mathcal L_R^\ast}\|_\diamond \le \epsilon
\qquad\text{for}\qquad
R=\bigO{\log(nt/\epsilon)}.
\end{equation*}
Thus it suffices to compile $e^{t\mathcal L_R^\ast}$.

Each truncated Lindbladian term centered at site $i$ is supported on the ball
$S=\ball(i,R)$. Since the interaction geometry is a $d$-dimensional lattice, one has $|S| = \bigO{R^d}$. 
For its local simulation, use the restricted Hamiltonian from \Cref{def:interaction-truncated-Ham}, written as
\begin{equation*}
H_S = V_S + W_S,
\qquad
V_S := \sum_{j\in S} v_j,
\qquad
W_S := \sum_{a:\,\supp(H_a)\subseteq S} H_a.
\end{equation*}
Assume access to the standard oracle for the patch Hamiltonian $W_S$. 
This oracle can be implemented using $\bigOt{|S|}=\bigOt{R^d}$ gates, since
$W_S$ contains only $\bigO{|S|}$ local interaction terms on the patch $S$, and each such term acts on
$\bigO{1}$ qubits and can be indexed and implemented with only polylogarithmic overhead.
Then \cite[Theorem 10]{LW18} by Low and Wiebe implies that $e^{-i\tau H_S}$ can be
implemented with query complexity $\bigOt{|S|\cdot \tau}$, independently of $\|V\|$ up to polylogarithmic factors.

Since the truncated Lindbladian consists of $\bigO{n}$ site-centered terms, combining this with the
expected black-box implementation cost from the previous remark yields
\begin{equation*}
\bigOt*{
n\,R^{2d}\,t\,T\,(1+\Lambda_f^2+\Lambda_1\Lambda_2)
}
\end{equation*}
gates for implementing $e^{t\mathcal L_R^\ast}$, and hence the same bound for
$e^{t\mathcal L^\ast}$ up to the truncation error above.
In particular, when $d=\bigO{1}$, $\beta=\bigO{1}$, and $t=\bigO{\log(n/\epsilon)}$, this gives
$\bigOt{n}$ gates.
\end{remark}

\section{Evolution of transport plans}

In this section, we develop tools that allow us to control how ``local mass'' propagates under one step of the Lindbladian evolution. We borrow the notion of transport plans from~\cite{blmt2025dobrushin}, which decompose a traceless Hermitian operator $X$ into pieces $(X_i)_{i\in[n]}$ that are assigned to individual sites and are locally traceless in the sense that $\tr_i(X_i)=0$.  The associated cost vector $\Paren{ \tfrac12\|X_i\|_1}_{i \in [n]}$ quantifies how much of $X$ is ``charged'' to each site. 

\begin{definition}[Transport plan and cost vector]\label{def:transport_plan_cost_vec}
    For a traceless Hermitian matrix $X$, we say that it has a \emph{transport plan} of $(X_i)_{i\in [n]}$ if $X = \sum_{i=1}^n X_i$, and $\tr_i(X_i) = 0$ and $X_i$ is Hermitian and traceless for every $i\in [n]$. 
    The \emph{cost vector} of this transport plan is a length-$n$ vector $x$ where $x_i = \frac12 \norm{X_i}_1$. 
\end{definition}

This transport plan formalism reduces the evolution of a traceless Hermitian perturbation to the evolution of a non-negative cost vector. Applying a linear, trace-preserving map $\Phi$ may redistribute this mass across the entire system. Next, we use the notion of an \emph{update matrix}, which is a convenient way to capture the worst-case one-step redistribution rule: it upper bounds, entry-wise, the cost vector of a suitable transport plan for $\Phi(X)$ in terms of the original cost vector $x$.

\begin{definition}[Update matrix]
    Let $\Phi$ be a trace-preserving linear map. $Q$ is an \emph{update matrix} of $\Phi$ if for any traceless Hermitian matrix $X$ and any transport plan $(X_i)_i$ of $X$, the matrix $\Phi(X)$ has a transport plan with cost vector $y \leq_v Qx$, where $x$ is the cost vector of $(X_i)_i$.
\end{definition}

An \emph{update matrix} $Q$ formalizes the statement that after applying the Lindbladian evolution, one can choose a new transport plan whose costs are entry-wise bounded by $Qx$. 
In particular, the quasi-locality of the Lindbladian implies that $Q$ has most of its weight near the updated site, with contributions from radius-$r$ balls decaying rapidly in $r$. The main theorem we prove obtains such an update matrix for the Lindbladian evolution, with a negative diagonal contribution corresponding to local contraction and a controlled quasi-local tail capturing leakage to larger balls around $i$. 

\begin{theorem}[Update matrix of the Lindbladian evolution at a site]
\label{thm:wass_growth_per_site}
    Assume that the inverse temperature $\beta$ satisfies $0<\beta \leq 1/(12\zeta)$. For each $i\in [n]$ and any $\delta\in [0, 1/(6\sqrt{2})]$, let $\Phi_i \coloneqq \calI + \delta \calL_{\mathsf{FR},i}$,
    where the Lindbladian evolution $\calL_{\mathsf{FR},i}$ is defined in \Cref{def:field_lindblad}. Then $\Phi_i$ has an update matrix of the form 
    \begin{equation*}
        (1+c_i\delta^2)\cdot I + \delta \cdot Q^{(i)}, \quad \text{where} \quad Q^{(i)} = - \frac{e^{-1/4}}{\sqrt{2}} E_{\{i\}} + 120 \sum_{r\geq 1} (6\beta \zeta)^r \cdot E_{\ball(\{i\}, r)} \,,
    \end{equation*}
    and $c_i = 144+ 324 (\log(2\sqrt{\Delta_i \beta }) + 4)^2$. 
\end{theorem}

\subsection{Properties of transport plans}

Before diving into the proofs, we recall some basic but useful properties of transport plans. First, transport plans are linear and the cost vectors satisfy triangle inequality. 

\begin{proposition}[Linearity of transport plans]\label{fact:linear_transport_plan}
    If two traceless Hermitian matrices $X$ and $Y$ have transport plans $(X_i)_{i\in [n]}$ and $(Y_i)_{i\in [n]}$ with cost vectors $x$ and $y$, then for all $\alpha, \beta \in \R$, the matrix $\alpha X + \beta Y$ has a transport plan $(\alpha X_i + \beta Y_i)_{i\in [n]}$ with a cost vector $z$ which satisfies $z \leq_v \abs{\alpha}\cdot x+ \abs{\beta}\cdot y$.
\end{proposition}

Next, if an operator vanishes after tracing out a subsystem $S$, then it can be charged entirely to $S$. In particular, the condition $\tr_S(X)=0$ implies that the operator $X$ is ``invisible'' when we forget all the sites in $S$, i.e $X$ has no influence on sites outside $S$. 

\begin{proposition}[{\cite[Proposition 2]{dmtl21}}]\label{lem:cost_vec_for_partial_trace_zero}
    Any traceless Hermitian matrix $X$ with $\tr_i(X) = 0$ for some $i\in [n]$ has a transport plan with cost vector $\frac12 \norm{X}_1 \cdot e_i$. 
\end{proposition}

\begin{proposition}[{\cite[Proposition 5]{dmtl21}}]
\label{lem:wass_norm_commutator}
    Any traceless Hermitian matrix $X$ with $\tr_S(X) = 0$ for some $S\subseteq [n]$ has a transport plan with cost vector entry-wise upper bounded by $\norm{X}_1 \cdot e_S$.
\end{proposition}

Next, we note that a Lindbladian-like map $\Psi$ that acts on operators $K$ and $L$  has a transport plan whose cost vector is supported on $\supp(K)\cup\supp(L)$, with magnitude controlled by the operator norms of $K$ and $L$. \cite{blmt2025dobrushin} formalize this in 
\Cref{lem:wass_norm_symmetric_map} stated below, and it is proved using \Cref{lem:wass_norm_commutator}.

\begin{lemma}[Cost vector after a few-site map, {\cite[Lemma 4.16]{blmt2025dobrushin}}]\label{lem:wass_norm_symmetric_map}
    For any matrices $K$ and $L$, 
    consider a map $\Psi(\rho) = K \rho L^\dagger + L \rho K^\dagger - \frac12 \{L^\dagger K + K^\dagger L, \rho\}$. 
    If $X$ is a traceless Hermitian matrix, then $\Psi(X)$ has a transport plan whose cost vector is entry-wise upper bounded by 
    \begin{equation*}
        4 \norm{K}\cdot \norm{L}\cdot \norm{X}_1 \cdot e_{\supp(K) \cup \supp(L)}. 
    \end{equation*}
\end{lemma}

We will also frequently use the following two properties of partial trace, which we abstract out below.
\begin{proposition}[Partial trace identities]\label{prop:partial_trace_identities}
    Given $S\subseteq [n]$. Suppose $Y_S$ is an operator on $S$ and $Z_{S^c}$ is an operator on $S^c = [n]\setminus S$. Then for any operator $X$ on $[n]$, 
    \begin{align*}
        \tr_S( (Z_{S^c} \otimes I_S) \cdot X) = Z_{S^c} \cdot \tr_S(X) & \tag{Module property of partial trace}
    \end{align*}
    and
    \begin{align*}
        \tr_S((I_{S^c} \otimes Y_S) \cdot X) = \tr_S(X \cdot (I_{S^c} \otimes Y_S)) &\tag{Cyclicity of partial trace}. 
    \end{align*}
\end{proposition}
The module property captures the intuition that after we trace out the subsystem $S$, anything acting on the complementary subsystem remains unchanged. The cyclicity property allows us to cyclically permute operators that only act on the $S$ subsystem.\footnote{We remind the readers to be careful when using the properties of partial trace. For example, it might be tempting to believe that if $\tr_S(X) = 0$, then we will always have $\tr_S((I_{S^c} \otimes Y_S) \cdot X) = 0$ because $Y_S$ is supported on $S$. This is not true. A simple counterexample is to take $Y_S = \sigma_Z$ and $X = I_{S^c} \otimes \sigma_Z$. }

We also require the following bound on the trace norm of the map $\Phi$, where $\Phi$ maps $X$ to a one-step Lindbladian-type update:
\begin{lemma}[{\cite[Lemma 3.6 and Corollary 3.8]{blmt2025dobrushin}}]\label{lem:one_norm_lindbladian}
    Consider the map $\Phi(X) = X - \ii [G, X] + KX K^\dagger - \frac12 \{K^\dagger K, X\}$. Then 
    \begin{align*}
        \norm{\Phi(X)}_1\leq \left(1+ 2\norm{G}^2 + \frac12 \norm{K}^4 + 2 \norm{G} \cdot \norm{K}^2\right)\cdot \norm{X}_1 .
    \end{align*}
\end{lemma}

\subsection{Proof idea}

As the Lindbladian evolution $\calL_{\mathsf{FR},i}$ at site $i$ can be split into a coherent component and dissipative component, the map $\Phi_i$ admits a corresponding decomposition using $\Phi_{i, \mathrm{cohe}}$ and $\Phi_{i, \mathrm{diss}}$. 
We then prove the update matrices of $\Phi_{i, \mathrm{cohe}}$ and $\Phi_{i, \mathrm{diss}}$ admit a weighted decomposition, where the weights decay as a function of the radius $r$ around site $i$  (\Cref{prop:wass_cohe} and \Cref{prop:wass_diss}). 
Then the update matrix of $\Phi_i$ follows from the linearity of transport plans (\Cref{fact:linear_transport_plan})
the fact that both $\Phi_{i, \mathrm{cohe}}$ and $\Phi_{i, \mathrm{diss}}$ are linear maps. 

\begin{proof}[Proof of \Cref{thm:wass_growth_per_site}]
    For any traceless Hermitian matrix $X$, let us split 
    \begin{align*}
        \Phi_i(X) &= X + \delta \sum_{P\in \{\sigma_X^{(i)}, \sigma_Y^{(i)}, \sigma_Z^{(i)}\}} \calL^P(X) \\
        &= \frac12 \left(X - 2\delta \ii \sum_{P\in \{\sigma_X^{(i)}, \sigma_Y^{(i)}, \sigma_Z^{(i)}\}} [C^P, X] \right) \\
        & \quad + \frac12 \left(X + 2\delta \sum_{P\in \{\sigma_X^{(i)}, \sigma_Y^{(i)}, \sigma_Z^{(i)}\}} \int_{-\infty}^{\infty} \gamma(\omega) \left( A^P(\omega) X A^P(\omega)^\dagger - \frac12 \left\{A^P(\omega)^\dagger A^P, X\right\} \right) \diff \omega\right). 
    \end{align*}
    Using the linearity of transport plans (\Cref{fact:linear_transport_plan}), we conclude the proof by combining \Cref{prop:wass_cohe} and \Cref{prop:wass_diss}. 
\end{proof}

The first step in proving both \Cref{prop:wass_cohe} and \Cref{prop:wass_diss} is to reduce to the case of ``normalized transport plans.''
Suppose that $Q$ is a matrix that satisfies the following. 
For any $j\in [n]$ and any traceless Hermitian matrix $X$ where $\tr_j(X) = 0$ and $\norm{X}_1 = 2$, i.e. $X$ has a transport plan with cost vector $e_j$ by \Cref{lem:cost_vec_for_partial_trace_zero}, the matrix $\Phi(X)$ has a transport plan with cost vector entry-wise upper bounded by $Qe_j$. 
Then by the linearity of transport plans in \Cref{fact:linear_transport_plan} and the fact that $\Phi$ is a linear map, $Q$ must be an update matrix of $\Phi$. 

So now we can fix some $j\in [n]$. 
To bound the cost vector of some transport plan of $\Phi(X)$, we frequently adopt the following strategy. 
\begin{itemize}
    \item Split $\Phi(X)$ into a sum of few terms and use
    the partial trace identities in \Cref{prop:partial_trace_identities} to show that for each term, taking the partial trace over an appropriate subsystem $S\subseteq [n]$ yields $0$. 
    \item Then by either \Cref{lem:cost_vec_for_partial_trace_zero} or \Cref{lem:wass_norm_commutator}, each term has a transport plan with cost vector entry-wise upper bounded by $\norm{ \cdot }_1 \cdot e_{S}$. 
    \item Bound the corresponding $1$-norms and conclude the cost vector of some transport plan of $\Phi(X)$ using the linearity of transport plans (\Cref{fact:linear_transport_plan}). 
\end{itemize}

The proof of the dissipative evolution in \Cref{prop:diss_growth} is more complicated than the proof of the coherent evolution in \Cref{prop:wass_cohe}. This is because for the dissipative evolution at site $i$, we have to further consider the cases where $j=i$ and $j\neq i$. Crucially, the rapid mixing property comes from the         ``contraction'' which happens in the case of $j=i$ during the dissipative evolution.  

\subsection{Update matrix of the coherent evolution}

We begin by analyzing the update matrix for the coherent evolution. 

\begin{lemma}[Update matrix of the coherent evolution]\label{prop:wass_cohe}
Assume $0<\beta \leq 1/(12\zeta)$. 
For each $i\in [n]$, define
\begin{equation*}\label{eq:cohe_map}
    \Phi_{i, \mathrm{cohe}}(\rho) \coloneqq \rho - \ii \delta \sum_{P\in \{\sigma_X^{(i)}, \sigma_Y^{(i)}, \sigma_Z^{(i)}\}} [C^P, \rho],
\end{equation*}
where $C^P$ is the coherent term as defined in \Cref{def:field_lindblad}. 
Then $\Phi_{i, \mathrm{cohe}}$ has an update matrix of the form
\begin{equation*}
    \left(1+ 162 (\log(2\sqrt{\Delta_{i} \beta }) + 4)^2 \cdot \delta^2 \right)I + \delta Q_{\mathrm{cohe}}^{(i)}, \quad \text{where} \quad Q_{\mathrm{cohe}}^{(i)} = 72 \sum_{r\geq 1} (6\beta \zeta)^r \cdot E_{\ball(\{i\}, r)} \ .
\end{equation*} 
\end{lemma}
\begin{proof}
    Since $\Phi_{i, \mathrm{cohe}}$ is a linear map, by the linearity of transport plans (\Cref{fact:linear_transport_plan}), 
    it suffices to show that for any $j\in [n]$ and any traceless Hermitian matrix $X$ with $\tr_j(X) = 0$ and $\norm{X}_1 = 2$,
    $\Phi_{i, \mathrm{cohe}}(X)$ has a transport plan with cost vector entry-wise upper bounded by 
    \begin{align*}
        \Paren{\left(1+ 162 (\log(2\sqrt{\Delta_i \beta }) + 4)^2 \cdot \delta^2 \right)I + \delta Q_{\mathrm{cohe}}^{(i)}} \cdot e_j \,.
    \end{align*}
    Let us fix $j\in [n]$. It follows from the expansion in \Cref{prop:coherent-quasi-local} that
    \begin{align*}
        [C^P, X] &= \sum_{r\geq 0} [K_r^P, X] \\
        &=  \underbrace{\sum_{r\geq 1} \Iver{j\notin\ball(\{i\}, r)} \cdot [K_r^P, X]}_{  S_1^P} +  \underbrace{[K_0^P, X]}_{  S_2^P} + \underbrace{\sum_{r\geq 1} \Iver{j \in\ball(\{i\}, r)} \cdot [K_r^P, X]}_{  S_3^P}\,, 
    \end{align*}
    where we split the summation into three cases, one where $r=0$, and when $r > 0$, whether site $j$ lies in a ball of radius $r$ or not. 
    We first note that $\tr_j(S_1^P) = 0$. This follows from the module property of partial trace in \Cref{prop:partial_trace_identities}: since $j\notin \supp(K_r^P)$ and $\tr_j(X) = 0$, then 
    \begin{equation*}
        \tr_j(K_r^P X) = 0 = \tr_j(X K_r^P). 
    \end{equation*}
    We also have $\tr_j(S_2^P) = 0$. This is because $K_0^{P}$ is entirely supported on site $i$. 
    \begin{itemize}
        \item If $i\neq j$, meaning that $j$ is not in the support of $K_0^P$, then similarly using the module property of partial trace in \Cref{prop:partial_trace_identities}, we have $\tr_j(K_0^{P} X) = 0 = \tr_j(X K_0^{P} )$. 
        \item If $i=j$, by the cyclicity of partial trace in \Cref{prop:partial_trace_identities}, we have $\tr_j(K_0^P X) = \tr_j (X K_0^P)$ and thus $\tr_j([K_0^{P}, X]) = 0$.\footnote{Note that it is not necessarily true that $\tr_j(K_0^P X) = \tr_j (X K_0^P) = 0$ despite $\tr_j(X) = 0$. }
    \end{itemize}

    Hence in
    \begin{equation}\label{eq:wass_cohe_split}
        \Phi_{i, \mathrm{cohe}}(X) = \underbrace{X - \ii \delta \sum_{P\in \{\sigma_X^{(i)}, \sigma_Y^{(i)}, \sigma_Z^{(i)}\}} \Paren{ S_1^P + S_2^P} }_{  \eqref{eq:wass_cohe_split}.(1)}  \; \; \; \underbrace{- \ii \delta \sum_{P\in \{\sigma_X^{(i)}, \sigma_Y^{(i)}, \sigma_Z^{(i)}\}} S_3^P}_{  \eqref{eq:wass_cohe_split}.(2)}, 
    \end{equation}
    we have $\tr_j(\eqref{eq:wass_cohe_split}.(1)) = 0$. 
    Let us focus on \eqref{eq:wass_cohe_split}.(1) first. 
    By \Cref{lem:cost_vec_for_partial_trace_zero}, 
    the transport plan of \eqref{eq:wass_cohe_split}.(1) has cost vector entry-wise upper bounded by $\frac12 \norm{\eqref{eq:wass_cohe_split}.(1)}_1 \cdot e_j$. By the linearity of the commutator, 
    \begin{align*}
        \eqref{eq:wass_cohe_split}.(1) = X - \ii \delta \sum_{P\in \{\sigma_X^{(i)}, \sigma_Y^{(i)}, \sigma_Z^{(i)}\}} \left[\left(\sum_{r\geq 1} \Iver{j\notin\ball(\{i\}, r)} \cdot K_r^P\right) + K_0^P, X\right].
    \end{align*}
    Then using \Cref{lem:one_norm_lindbladian} and the fact that $\norm{X}_1 = 2$, we have
    \begin{align*}
        \frac12\norm{\eqref{eq:wass_cohe_split}.(1)}_1 &\leq 1+ 2\delta^2 \norm*{\sum_{P\in \{\sigma_X^{(i)}, \sigma_Y^{(i)}, \sigma_Z^{(i)}\}} \Paren{ \left(\sum_{r\geq 1} \Iver{j\notin\ball(\{i\}, r)} \cdot K_r^P\right) + K_0^P }}^2 \\
        &\leq 1+ 2\delta^2\left( \sum_{P\in \{\sigma_X^{(i)}, \sigma_Y^{(i)}, \sigma_Z^{(i)}\}}\sum_{r\geq 0} \norm{K_r^P} \right)^2 \tag{triangle inequality} \\
        &\leq 1+ 2\delta^2\cdot 9\left( 3\log(2\sqrt{\Delta_i \beta }) + \sum_{r\geq 1} 6(6\zeta\beta)^r \right)^2 \tag{\Cref{prop:coherent-quasi-local}} \\
        &\leq 1+ 162 \delta^2 \cdot \Paren{\log(2\sqrt{\Delta_i \beta }) + 4}^2 \tag{$6\zeta\beta\leq 1/2$} \,.
    \end{align*}
    So the transport plan of \eqref{eq:wass_cohe_split}.(1) has cost vector at most $(1+ 162 (\log(2\sqrt{\Delta_i \beta }) + 4)^2 \cdot \delta^2) e_j$. 
    
    We now analyze \eqref{eq:wass_cohe_split}.(2):
    \begin{align*}
        \eqref{eq:wass_cohe_split}.(2) = \sum_{P\in \{\sigma_X^{(i)}, \sigma_Y^{(i)}, \sigma_Z^{(i)}\}} \sum_{r\geq 1} \ii \left[-\delta  \Iver{j \in\ball(\{i\}, r)} \cdot K_r^P, X\right]. 
    \end{align*}
    Since each $K_r^P$ is supported on $\ball(\{i\}, r)$, by the cyclicity of partial trace in \Cref{prop:partial_trace_identities}, we have
    \begin{align*}
        \tr_{\ball(\{i\}, r)}\left(K_r^P X\right) = \tr_{\ball(\{i\}, r)}\left(X K_r^P\right),
    \end{align*}
    which implies that
    \begin{align*}
        \tr_{\ball(\{i\}, r)}\left( \ii \left[-\delta  \Iver{j \in\ball(\{i\}, r)} \cdot K_r^P, X\right] \right) = 0. 
    \end{align*}
    Then by \Cref{lem:wass_norm_commutator}, the traceless Hermitian matrix $\ii \left[-\delta  \Iver{j \in\ball(\{i\}, r)} \cdot K_r^P, X\right]$ 
    has a transport plan with cost vector at most
    \begin{align*}
        \norm*{\ii \left[-\delta  \Iver{j \in\ball(\{i\}, r)} \cdot K_r^P, X\right]}_1 \cdot e_{\ball(\{i\}, r)} \leq_v 2\delta \norm{X}_1\cdot \norm{K_r^P} \cdot E_{\ball(\{i\}, r)}\cdot e_j,
    \end{align*}
    where we use $\Iver{j \in\ball(\{i\}, r)}  \cdot e_{\ball(\{i\}, r)} = E_{\ball(\{i\}, r)}\cdot e_j$ and
    \begin{equation*}
        \norm{[A,X]}_1 \leq \norm{AX}_1 + \norm{XA}_1 \leq 2\norm{X}_1 \cdot \norm{A} .\tag{triangle and Hölder's inequalities}
    \end{equation*}
    Therefore, by the linearity of transport plans (\Cref{fact:linear_transport_plan}), 
    \eqref{eq:wass_cohe_split}.(2) has a transport plan with cost vector entry-wise upper bounded by 
    \begin{align*}
        \sum_{P\in \{\sigma_X^{(i)}, \sigma_Y^{(i)}, \sigma_Z^{(i)}\}} \sum_{r\geq 1} 2\delta \norm{X}_1\cdot \norm{K_r^P} \cdot E_{\ball(\{i\}, r)}\cdot e_j &\leq_v 24\delta \sum_{P\in \{\sigma_X^{(i)}, \sigma_Y^{(i)}, \sigma_Z^{(i)}\}} \sum_{r\geq 1} (6\zeta\beta)^r \cdot E_{\ball(\{i\}, r)} \cdot e_j \\
        &= 72\delta \sum_{r\geq 1} (6\beta \zeta)^r \cdot E_{\ball(\{i\}, r)} \cdot e_j.
    \end{align*}
    This completes the proof. 
\end{proof}

\subsection{Update matrix of the dissipative evolution}
Next, we provide an expression for the update matrix of the dissipative evolution:
\begin{lemma}[Update matrix of the dissipative evolution]\label{prop:wass_diss}
    \label{prop:diss_growth}
    Assume $0<\beta \leq 1/(4\zeta)$. 
    For each $i\in [n]$ and $\delta \in [0, 1/(3\sqrt{2})]$, define
    \begin{equation*}\label{eq:diss_map}
        \Phi_{i, \mathrm{diss}}(\rho) \coloneqq \rho + \delta \sum_{P\in \{\sigma_X^{(i)}, \sigma_Y^{(i)}, \sigma_Z^{(i)}\}} \int_{-\infty}^{\infty} \gamma(\omega) \left( A^P(\omega) \rho A^P(\omega)^\dagger - \frac12 \left\{A^P(\omega)^\dagger A^P, \rho\right\} \right) \diff \omega,
    \end{equation*}
    where $\gamma(\omega)$ and the jump operators $A^P(\omega)$ are defined in \Cref{def:field_lindblad}. 
    Then $\Phi_{i, \mathrm{diss}}$ has an update matrix of form
    \begin{equation*}
        (1+72\delta^2)I + \delta Q_{\mathrm{diss}}^{(i)} , \quad \text{where} \quad
        Q_{\mathrm{diss}}^{(i)} = - \frac{e^{-1/4}}{\sqrt{2}} E_{\{i\}} + 48 \sum_{r\geq 1} (2\beta \zeta)^r \cdot E_{\ball(\{i\}, r)} \ .
    \end{equation*}
\end{lemma}
\begin{proof}
    Since $\Phi_{i, \mathrm{diss}}$ is a linear map, by the linearity of transport plans (\Cref{fact:linear_transport_plan}), 
    it suffices to show that for any $j\in [n]$ and any traceless Hermitian matrix $X$ with $\tr_j(X) = 0$ and $\norm{X}_1 = 2$,
    $\Phi_{i, \mathrm{diss}}(X)$ has a transport plan with cost vector entry-wise upper bounded by $((1+72 \delta^2)I + \delta Q_{\mathrm{diss}}^{(i)})e_j$. 

    Let us fix $j\in [n]$. 
    Using the expansion in \Cref{prop:jump-operators-quasi-local}, we have
    \begin{align}\label{eq:diss_i_equal_j}
        \Phi_{i, \mathrm{diss}}(X) = X + \delta \sum_{r_1, r_2\geq 0}\sum_{P\in \{\sigma_X^{(i)}, \sigma_Y^{(i)}, \sigma_Z^{(i)}\}} \int_{-\infty}^{\infty} \gamma(\omega) \left( G_{r_1}^{P,\omega} X (G_{r_2}^{P,\omega})^\dagger - \frac12 \left\{(G_{r_2}^{P,\omega})^\dagger G_{r_1}^{P,\omega}, X\right\} \right) \diff \omega. 
    \end{align}

    \paragraph{Suppose $j=i$.} 
    Let us split $\Phi_{i, \mathrm{diss}}(X) = \eqref{eq:diss_i_equal_j}.(1) + \eqref{eq:diss_i_equal_j}.(2)$ where
    \begin{align*}
        \eqref{eq:diss_i_equal_j}.(1) &= X + \delta \sum_{P\in \{\sigma_X^{(i)}, \sigma_Y^{(i)}, \sigma_Z^{(i)}\}} \int_{-\infty}^{\infty} \gamma(\omega) \left( G_{0}^{P,\omega} X (G_{0}^{P,\omega})^\dagger - \frac12 \left\{(G_{0}^{P,\omega})^\dagger G_{0}^{P,\omega}, X\right\} \right) \diff \omega, \\
        \eqref{eq:diss_i_equal_j}.(2) &= \delta \sum_{r_1+ r_2\geq 1}\sum_{P\in \{\sigma_X^{(i)}, \sigma_Y^{(i)}, \sigma_Z^{(i)}\}} \int_{-\infty}^{\infty} \gamma(\omega) \left( G_{r_1}^{P,\omega} X (G_{r_2}^{P,\omega})^\dagger - \frac12 \left\{(G_{r_2}^{P,\omega})^\dagger G_{r_1}^{P,\omega}, X \right\} \right) \diff \omega. 
    \end{align*}
    We first analyze \eqref{eq:diss_i_equal_j}.(1). 
    Since $G_0^{P,\omega}$ is entirely supported on site $i$, using the cyclicity of partial trace in \Cref{prop:partial_trace_identities}, we have that
    \begin{align*}
        \tr_i\left( G_{0}^{P,\omega} X (G_{0}^{P,\omega})^\dagger - \frac12 \left\{(G_{0}^{P,\omega})^\dagger G_{0}^{P,\omega}, X\right\} \right) = 0. 
    \end{align*}
    Since $\tr_i(X) = 0$, we have that $\tr_i(\eqref{eq:diss_i_equal_j}.(1))=0$. 
    By \Cref{lem:cost_vec_for_partial_trace_zero}, we know that $\eqref{eq:diss_i_equal_j}.(1)$ has a transport plan with cost vector entry-wise bounded by $\frac12 \norm{\eqref{eq:diss_i_equal_j}.(1)}_1 e_i$.
    We show in \Cref{lem:diss_contraction} that $\norm{\eqref{eq:diss_i_equal_j}.(1)}_1 \leq 2(1 - \frac{e^{-1/4}}{\sqrt{2}}\delta)$. 

    We now analyze the cost vector for the transport plan of \eqref{eq:diss_i_equal_j}.(2). 
    We first symmetrize each element in the sum and write
    \begin{align}\label{eq:wass_diss_symmetrization}
        S^{P, \omega}_{r_1, r_2} &\coloneqq G_{r_1}^{P,\omega} X (G_{r_2}^{P,\omega})^\dagger + G_{r_2}^{P,\omega} X (G_{r_1}^{P,\omega})^\dagger - \frac12 \left\{(G_{r_2}^{P,\omega})^\dagger G_{r_1}^{P,\omega} + (G_{r_1}^{P,\omega})^\dagger G_{r_2}^{P,\omega}, X \right\},
    \end{align}
    and then 
    \begin{equation*}
        \eqref{eq:diss_i_equal_j}.(2) = \frac{\delta}{2} \sum_{r_1+ r_2\geq 1}\sum_{P\in \{\sigma_X^{(i)}, \sigma_Y^{(i)}, \sigma_Z^{(i)}\}} \int_{-\infty}^{\infty} \gamma(\omega) S_{r_1,r_2}^{P, \omega}  \diff \omega. 
    \end{equation*}
    By \Cref{lem:wass_norm_symmetric_map}, each $S^{P, \omega}_{r_1, r_2}$ has a transport plan with cost vector entry-wise upper bounded by 
    \begin{align}\label{eq:wass_diss_symmetrization_cost}
        4 \norm{G_{r_1}^{P,\omega}} \cdot \norm{G_{r_2}^{P,\omega}} \cdot \norm{X}_1 \cdot e_{\ball(\{i\}, \max\{r_1, r_2\})}. 
    \end{align}
    Hence, \eqref{eq:diss_i_equal_j}.(2) has a transport plan with cost vector at most
    \begin{equation*}
        4\delta \sum_{P\in \{\sigma_X^{(i)}, \sigma_Y^{(i)}, \sigma_Z^{(i)}\}} \int_{-\infty}^{\infty} \gamma(\omega) \sum_{r_1 + r_2\geq 1} \norm{G_{r_1}^{P, \omega}}\cdot \norm{G_{r_2}^{P, \omega}} \cdot \diff \omega \cdot E_{\ball(\{i\}, \max\{r_1, r_2\})} \cdot e_i
    \end{equation*}

    \paragraph{Suppose $j\neq i$.}
    Let us define
    \begin{align*}
        G^{P, \omega}_{\text{in}} \coloneqq \sum_{r\in r_{\mathrm{in}}} G_{r}^{P, \omega}, \quad & \text{where} \quad r_{\mathrm{in}} = \{r\geq 0: j\in \supp(G_{r}^{P, \omega})\}, \quad \text{and}\\
        G^{P, \omega}_{\text{out}} \coloneqq \sum_{r\in r_{\mathrm{out}}} G_{r}^{P, \omega}, \quad & \text{where} \quad r_{\mathrm{out}} = \{r\geq 0: j\notin \supp(G_{r}^{P, \omega})\}. 
    \end{align*}
    Then $j\notin \supp(G^{P, \omega}_{\text{out}})$. 
    We now split $\Phi_{i, \mathrm{diss}}(X) = \eqref{eq:diss_i_equal_j}.(3) + \eqref{eq:diss_i_equal_j}.(4)$ as
    \begin{align*}
        \eqref{eq:diss_i_equal_j}.(3) &= X + \delta \sum_{P\in \{\sigma_X^{(i)}, \sigma_Y^{(i)}, \sigma_Z^{(i)}\}} \int_{-\infty}^{\infty} \gamma(\omega) \left( G_{\text{out}}^{P,\omega} X (G_{\text{out}}^{P,\omega})^\dagger - \frac12 \left\{(G_{\text{out}}^{P,\omega})^\dagger G_{\text{out}}^{P,\omega}, X\right\} \right) \diff \omega, \\
        \eqref{eq:diss_i_equal_j}.(4) &= \delta \sum_{P\in \{\sigma_X^{(i)}, \sigma_Y^{(i)}, \sigma_Z^{(i)}\}} \int_{-\infty}^{\infty} \gamma(\omega) \Big(G_{\text{in}}^{P,\omega} X (G_{\text{in}}^{P,\omega})^\dagger + G_{\text{in}}^{P,\omega} X (G_{\text{out}}^{P,\omega})^\dagger + G_{\text{out}}^{P,\omega} X (G_{\text{in}}^{P,\omega})^\dagger \\
        & \hspace{1.5in} - \frac12 \left\{(G_{\text{in}}^{P,\omega})^\dagger G_{\text{in}}^{P,\omega} +  (G_{\text{out}}^{P,\omega})^\dagger G_{\text{in}}^{P,\omega} +   (G_{\text{in}}^{P,\omega})^\dagger G_{\text{out}}^{P,\omega}, X \right\} \Big) \diff \omega. 
    \end{align*}
    Let us first analyze \eqref{eq:diss_i_equal_j}.(4) as it is similar to how we handle \eqref{eq:diss_i_equal_j}.(2) in the case of $j=i$ previously. 
    Recall the symmetrization in \Cref{eq:wass_diss_symmetrization} and the cost vector of its transport plan is bounded in \Cref{eq:wass_diss_symmetrization_cost}. Then
    \begin{align*}
        \eqref{eq:diss_i_equal_j}.(4) &= \frac{\delta}{2} \sum_{P\in \{\sigma_X^{(i)}, \sigma_Y^{(i)}, \sigma_Z^{(i)}\}} \int_{-\infty}^{\infty} \gamma(\omega) \left(\sum_{r_1\in r_{\mathrm{in}} \text{ or } r_2\in r_{\mathrm{in}}} S_{r_1, r_2}^{P, \omega} \right) \diff \omega,
    \end{align*}
    and its transport plan has a cost vector at most 
    \begin{align}
        &4\delta \sum_{P\in \{\sigma_X^{(i)}, \sigma_Y^{(i)}, \sigma_Z^{(i)}\}} \int_{-\infty}^{\infty} \gamma(\omega) \sum_{r_1\in r_{\mathrm{in}} \text{ or } r_2\in r_{\mathrm{in}}} \norm{G_{r_1}^{P, \omega}}\cdot \norm{G_{r_2}^{P, \omega}} \cdot \diff \omega \cdot E_{\ball(\{i\}, \max\{r_1, r_2\})} \cdot e_j \nonumber\\
        &\leq_v 4\delta \sum_{P\in \{\sigma_X^{(i)}, \sigma_Y^{(i)}, \sigma_Z^{(i)}\}} \int_{-\infty}^{\infty} \gamma(\omega) \sum_{r_1 + r_2 \geq 1 } \norm{G_{r_1}^{P, \omega}}\cdot \norm{G_{r_2}^{P, \omega}} \cdot \diff \omega \cdot E_{\ball(\{i\}, \max\{r_1, r_2\})} \cdot e_j\,.\label{eq:wass_diss_term_4}
    \end{align}
    We now analyze \eqref{eq:diss_i_equal_j}.(3). 
    By the module property of partial trace in \Cref{prop:partial_trace_identities}, 
    since $j \notin \supp(G_{\text{out}}^{P,\omega})$ and $\tr_j(X) = 0$, we have
    \begin{align*}
        \tr_j\left( G_{\text{out}}^{P,\omega} X (G_{\text{out}}^{P,\omega})^\dagger\right) = 0 \,, \quad \tr_j\Paren{(G_{\text{out}}^{P,\omega})^\dagger G_{\text{out}}^{P,\omega} X} = 0 \,, \quad \text{and} \quad 
        \tr_j\Paren{X (G_{\text{out}}^{P,\omega})^\dagger G_{\text{out}}^{P,\omega}}=0 \,.
    \end{align*}
    Hence $\tr_j(\eqref{eq:diss_i_equal_j}.(3)) = 0$. 
    Again using \Cref{lem:cost_vec_for_partial_trace_zero}, we know that \eqref{eq:diss_i_equal_j}.(3) has a transport plan with cost vector at most $\frac12 \norm{\eqref{eq:diss_i_equal_j}.(3)}_1 e_j$. 
    With $\norm{\gamma}_1 = \int_{-\infty}^{\infty} \gamma(\omega) d\omega = \sqrt{2\pi} \cdot \eta_i$, we can write 
    \begin{align*}
        \eqref{eq:diss_i_equal_j}.(3) = \frac{1}{3\norm{\gamma}_1}\sum_{P\in \{\sigma_X^{(i)}, \sigma_Y^{(i)}, \sigma_Z^{(i)}\}} \int_{-\infty}^{\infty} \gamma(\omega) \left( X + \widehat{G}_{\text{out}}^{P,\omega} X (\widehat{G}_{\text{out}}^{P,\omega})^\dagger - \frac12 \left\{(\widehat{G}_{\text{out}}^{P,\omega})^\dagger \widehat{G}_{\text{out}}^{P,\omega}, X\right\}\right) \diff \omega
    \end{align*}
    where $\widehat{G}_{\text{out}}^{P,\omega} = \sqrt{3\norm{\gamma}_1 \delta } \cdot G_{\text{out}}^{P,\omega}$. By the triangle inequality,
    \begin{align*}
        \norm{\eqref{eq:diss_i_equal_j}.(3)}_1
        &\leq \frac{1}{3\norm{\gamma}_1} \sum_{P\in \{\sigma_X^{(i)}, \sigma_Y^{(i)}, \sigma_Z^{(i)}\}} \int_{-\infty}^{\infty} \gamma(\omega) \norm*{ X + \widehat{G}_{\text{out}}^{P,\omega} X (\widehat{G}_{\text{out}}^{P,\omega})^\dagger - \frac12 \left\{(\widehat{G}_{\text{out}}^{P,\omega})^\dagger \widehat{G}_{\text{out}}^{P,\omega}, X\right\} }_1 \diff \omega \\
        &\leq \frac{1}{3\norm{\gamma}_1}\sum_{P\in \{\sigma_X^{(i)}, \sigma_Y^{(i)}, \sigma_Z^{(i)}\}} \int_{-\infty}^{\infty} \gamma(\omega) \left( 1 + \frac12 \norm{\widehat{G}_{\text{out}}^{P,\omega}}^4 \right)\norm{X}_1 \diff \omega \tag{\Cref{lem:one_norm_lindbladian}} \\ 
        &= 2 + \frac{1}{3\norm{\gamma}_1}\sum_{P\in \{\sigma_X^{(i)}, \sigma_Y^{(i)}, \sigma_Z^{(i)}\}} \int_{-\infty}^{\infty} \gamma(\omega) \cdot \left(\delta \cdot 3\norm{\gamma}_1\right)^2 \norm*{G_{\text{out}}^{P,\omega}}^4 \cdot \diff \omega. 
    \end{align*}
    Using the triangle inequality and the expansion in \Cref{prop:jump-operators-quasi-local}, we can bound
    \begin{align*}
        \norm{G_{\text{out}}^{P,\omega}} \leq \sum_{r\geq 0} \norm{G_{r}^{P,\omega}} \leq (2\pi)^{-1/4} \sigma_i^{-1/2} \sum_{r\geq 0} (2\zeta \sigma_i^{-1})^r \leq (2\pi)^{-1/4} \sigma_i^{-1/2} \cdot 2,
    \end{align*}
    where the last bound is due to the assumption $2\zeta \sigma_i^{-1} = 2\zeta \sqrt{\beta/\Delta} \leq 2\zeta\beta \leq 1/2$. 
    Continuing,
    \begin{align*}
        \frac12\norm{\eqref{eq:diss_i_equal_j}.(3)}_1 &\leq 1 + \frac{1}{2\norm{\gamma}_1} \int_{-\infty}^{\infty} \gamma(\omega) \cdot \left(\delta \cdot3\norm{\gamma}_1\right)^2 \frac{16}{2\pi \sigma_i^2} \cdot \diff \omega \\
        &= 1 + \frac{36 \norm{\gamma}_1^2}{\pi \sigma_i^2}\delta^2 \tag{$\eta_i = \sigma_i$ and $\norm{\gamma}_1 = \sqrt{2\pi} \cdot \eta_i$} \\
        &= 1 + 72 \delta^2 \,.
    \end{align*}
    Therefore \eqref{eq:diss_i_equal_j}.(3) has a transport plan with cost vector at most $(1 + 72\delta^2) e_j$. 

    \paragraph{Put two cases together.}
    Now combining the cost vectors of the transport plans of \eqref{eq:diss_i_equal_j}.(1--4), we have that the cost vector of the transport plan of \eqref{eq:diss_i_equal_j} is at most
    \begin{align*}
        \bigg(-\frac{e^{-1/4}}{\sqrt{2}} \delta E_{\{i\}} &+ (1 + 72\delta^2)I \\
        &+ 4\delta \sum_{P\in \{\sigma_X^{(i)}, \sigma_Y^{(i)}, \sigma_Z^{(i)}\}} \int_{-\infty}^{\infty} \gamma(\omega) \sum_{r_1 + r_2\geq 1} \norm{G_{r_1}^{P, \omega}}\cdot \norm{G_{r_2}^{P, \omega}} \cdot \diff \omega \cdot E_{\ball(\{i\}, \max\{r_1, r_2\})}\bigg)e_j
    \end{align*}
    The last term can be bounded as 
    \begin{align*}
        &4\delta \sum_{P\in \{\sigma_X^{(i)}, \sigma_Y^{(i)}, \sigma_Z^{(i)}\}} \int_{-\infty}^{\infty} \gamma(\omega) \sum_{r_1 + r_2\geq 1} \norm{G_{r_1}^{P, \omega}}\cdot \norm{G_{r_2}^{P, \omega}} \cdot \diff \omega \cdot E_{\ball(\{i\}, \max\{r_1, r_2\})} \cdot e_j \\
        &\leq_v 12\delta \int_{-\infty}^{\infty} \gamma(\omega) \sum_{r_1 + r_2\geq 1} (2\pi)^{-1/2} \sigma^{-1} (2\zeta \sigma^{-1})^{r_1+r_2} \cdot \diff \omega \cdot E_{\ball(\{i\}, \max\{r_1, r_2\})} \cdot e_j\\
        &= 12\delta \sum_{r_1+r_2\geq 1} (2\zeta \beta)^{r_1+r_2} \cdot E_{\ball(\{i\}, \max\{r_1, r_2\})} \cdot e_j \tag{$\int_{-\infty}^{\infty} \gamma(\omega) d\omega = \sqrt{2\pi} \cdot \eta_i $ and $\sigma_i = \eta_i\geq 1/\beta$}\\
        &\leq_v 24\delta \sum_{r_1\geq 1} (2\zeta \beta)^{r_1}\sum_{r_2=0}^{r_1} (2\zeta \beta)^{r_2} \cdot E_{\ball(\{i\}, r_1)} \cdot e_j \\
        &\leq_v 48\delta \sum_{r\geq 1} (2\zeta \beta)^{r} \cdot E_{\ball(\{i\}, r)} \cdot e_j. \tag{$2\zeta \beta \leq 1/2$}
    \end{align*}
    This completes the proof. 
\end{proof}

\subsection{Contraction in the dissipative evolution}

\begin{lemma}[Contraction in the dissipative evolution]
\label{lem:diss_contraction}
    Given $i\in [n]$, let $X$ be any Hermitian traceless matrix such that $\tr_i(X) = 0$.
    Let $\eta_i, \Delta_i, \sigma_i>0$ be the parameters for a field-resonant Lindbladian.
    Let
    \begin{align*}
        G_0^{P}(\omega) = \frac{1}{\sqrt{2\pi}} \int_{-\infty}^{\infty}  e^{\ii t v_i} P e^{-\ii tv_i} \cdot e^{-\ii \omega t} f_i(t) \diff t, \quad \text{where} \quad
        f_i(t) = \sqrt{\sigma_i} (2/\pi)^{1/4} \cdot e^{-\sigma_i^2 t^2}.
    \end{align*}
    Let $\gamma_i(\omega)= \exp\left(-\frac{(\omega + \Delta_i)^2}{2\eta_i^2}\right)$.
    For any $\delta \in [0, \frac{1}{3\sqrt{2}}]$, define
    \begin{equation*}
        S_i(X) = X + \delta \sum_{P\in \{\sigma_X^{(i)}, \sigma_Y^{(i)}, \sigma_Z^{(i)}\}} \int_{-\infty}^{\infty} \gamma_i(\omega) \left( G_{0}^{P}(\omega) X G_{0}^{P}(\omega)^\dagger - \frac12 \left\{G_{0}^{P}(\omega)^\dagger G_{0}^{P}(\omega), X\right\} \right) \diff \omega.
    \end{equation*}
    Then $\norm{S_i(X)}_1 \leq \left(1 -  \frac{e^{-1/4}}{\sqrt{2} }\delta\right) \cdot \norm{X}_1$.
\end{lemma}
\begin{proof}
    Without loss of generality, we can assume $i=1$.
    For simplicity, we will sometimes drop the index $i=1$ when the context is clear.
    For example, we will write $\Delta_i$ as $\Delta$.
    Since each $G_0^{P}(\omega)$ where $P\in \{\sigma_X^{(i)}, \sigma_Y^{(i)}, \sigma_Z^{(i)}\}$ is effectively a single-qubit operator, we can just focus on its restriction to the first qubit, for which we denote by $G^{\sigma_X}(\omega), G^{\sigma_Y}(\omega), G^{\sigma_Z}(\omega)$.

    Without loss of generality, let us assume that $v_1 =  \frac{h}{2} \sigma_Z \otimes I_{2^{n-1}}$ where $h\geq 0$ (since any shift in the spectrum can be considered as a global phase after exponentiation, which will be canceled due to the unitary conjugation).
    This is because at the start of the full Gibbs state prepartion algorithm, we can always just apply single-qubit unitary to each site $i$ to rotate $v_i$ to the Pauli $Z$-basis.

    Recall by the definition of a field-resonant Lindbladian, we have $\eta_i = \sigma_i = \sqrt{\Delta_i/ \beta}$ and $\Delta_i = \max\{h, 1/\beta\}$.

    We first analyze the single-qubit operator $G^{P}(\omega)$ using the Fourier transform of $f$:
    \begin{equation*}
        \widehat{f}(\nu) = \frac{1}{\sqrt{2\pi}}\int_{-\infty}^{\infty} f(t) e^{-\ii \nu t} \diff t = \frac{(2\pi)^{-1/4}}{\sqrt{\sigma}} e^{-\frac{\nu^2}{4\sigma^2}} \,.
    \end{equation*}
    Then
    \begin{equation*}
        G^{\sigma_Z}(\omega) = \frac{1}{\sqrt{2\pi}} \int_{-\infty}^{\infty}  e^{\ii \frac{t h}{2} \sigma_Z} \cdot \sigma_Z\cdot e^{-\ii \frac{t h}{2} \sigma_Z} \cdot e^{-\ii \omega t} f(t) \diff t = \widehat{f}(\omega) \sigma_Z \,.
    \end{equation*}
    Write $\sigma_{+} = \frac12 (\sigma_X + \ii \sigma_Y) = \ketbra{0}{1}$ and $\sigma_{-} = \frac12 (\sigma_X - \ii \sigma_Y) = \ketbra{1}{0}$. Then
    \begin{align*}
        e^{\ii \frac{t h}{2} \sigma_Z}\cdot \sigma_{+}\cdot e^{-\ii \frac{t h}{2} \sigma_Z} &= \left(e^{\ii \frac{t h}{2}} \ketbra{0}{0} + e^{-\ii \frac{t h}{2}} \ketbra{1}{1}\right) \ketbra{0}{1} \left(e^{-\ii \frac{t h}{2}} \ketbra{0}{0} + e^{\ii \frac{t h}{2}} \ketbra{1}{1}\right) = e^{\ii t h} \sigma_+ \,, \\
        e^{\ii \frac{t h}{2} \sigma_Z}\cdot \sigma_{-}\cdot e^{-\ii \frac{t h}{2} \sigma_Z} &= \left(e^{\ii \frac{t h}{2}} \ketbra{0}{0} + e^{-\ii \frac{t h}{2}} \ketbra{1}{1}\right) \ketbra{1}{0} \left(e^{-\ii \frac{t h}{2}} \ketbra{0}{0} + e^{\ii \frac{t h}{2}} \ketbra{1}{1}\right) = e^{-\ii t h} \sigma_- \,.
    \end{align*}
    Therefore,
    \begin{align*}
        G^{\sigma_X}(\omega) &= \frac{1}{\sqrt{2\pi}} \int_{-\infty}^{\infty}  e^{\ii \frac{t h}{2} \sigma_Z} \cdot (\sigma_+ + \sigma_-) \cdot e^{-\ii \frac{t h}{2} \sigma_Z}  e^{-\ii \omega t} f(t) \diff t = \widehat{f}(\omega - h) \sigma_+ + \widehat{f}(\omega + h) \sigma_- \,, \\
        G^{\sigma_Y}(\omega) &= \frac{1}{\sqrt{2\pi}} \int_{-\infty}^{\infty}  e^{\ii \frac{t h}{2} \sigma_Z} \cdot \ii(-\sigma_+ + \sigma_-)\cdot e^{-\ii \frac{t h}{2} \sigma_Z} e^{-\ii \omega t} f(t) \diff t = -\ii \widehat{f}(\omega - h) \sigma_+ + \ii \widehat{f}(\omega + h) \sigma_- \,.
    \end{align*}
    We now analyze $\sum_{P} G^P(\omega)^\dagger G^P(\omega)$.
    Write $a = \widehat{f}(\omega - h)$, $b = \widehat{f}(\omega + h)$ and $c = \widehat{f}(\omega)$ which are all real numbers.
    Then
    \begin{align*}
        G^{\sigma_X}(\omega)^\dagger G^{\sigma_X}(\omega) &= (a \sigma_- + b\sigma_+)(a \sigma_+ + b\sigma_-) = a^2 \sigma_-\sigma_+ + b^2 \sigma_+\sigma_- \,, \\
        G^{\sigma_Y}(\omega)^\dagger G^{\sigma_Y}(\omega) &= \ii (a \sigma_- - b\sigma_+)(-\ii)(a \sigma_+ - b\sigma_-) = a^2 \sigma_-\sigma_+ + b^2  \sigma_+\sigma_- \,, \\
        G^{\sigma_Z}(\omega)^\dagger G^{\sigma_Z}(\omega) &= c^2 \sigma_Z^2 \,.
    \end{align*}
    Therefore,
    \begin{align*}
        \sum_{P\in \{\sigma_X, \sigma_Y, \sigma_Z\}} G^P(\omega)^\dagger  G^P(\omega) = 2a^2 \sigma_-\sigma_+ + 2b^2  \sigma_+\sigma_- + c^2 \sigma_Z^2 \,.
    \end{align*}
    For any $n$-qubit operator $X$, for simplicity, we will write, for example, $\sigma_{+} X \sigma_{-}$ to mean
    $(\sigma_{+} \otimes I_{2^{n-1}})\cdot  X \cdot (\sigma_{-} \otimes I_{2^{n-1}})$.
    We now analyze $\sum_{P} G^P(\omega) X G^P(\omega)^\dagger$. The easy case is
    \begin{equation*}
        G^{\sigma_Z}(\omega) X G^{\sigma_Z}(\omega)^\dagger = c^2 \sigma_Z X \sigma_Z \,.
    \end{equation*}
     The two other cases are
    \begin{align*}
        G^{\sigma_X}(\omega) X G^{\sigma_X}(\omega)^\dagger &= (a \sigma_+ + b\sigma_-)X(a \sigma_- + b\sigma_+) \\
        &= a^2 \sigma_+ X \sigma_- + b^2 \sigma_- X \sigma_++ ab (\sigma_+ X \sigma_+ + \sigma_- X \sigma_- )
    \end{align*}
    and
    \begin{align*}
        G^{\sigma_Y}(\omega) X G^{\sigma_Y}(\omega)^\dagger &= -\ii (a \sigma_+ - b\sigma_-)X \cdot \ii (a \sigma_- - b\sigma_+) \\
        &= a^2 \sigma_+ X \sigma_- + b^2 \sigma_- X \sigma_+ - ab (\sigma_+ X \sigma_+ + \sigma_- X \sigma_- ) \,.
    \end{align*}
    Therefore,
    \begin{align*}
        \sum_{P\in \{\sigma_X, \sigma_Y, \sigma_Z\}} G^P(\omega) X G^P(\omega)^\dagger = 2a^2 \sigma_+ X \sigma_- + 2b^2 \sigma_- X \sigma_+ + c^2 \sigma_Z X \sigma_Z \,.
    \end{align*}
    Let us write $D_{A}(X) \coloneqq A X A^\dagger - \frac12 \{A^\dagger A, X\}$. Then
    \begin{align}\label{eq:contraction_simplify}
        \sum_{P\in \{\sigma_X, \sigma_Y, \sigma_Z\}} G^P(\omega) X G^P(\omega)^\dagger - \frac12 \left\{ G^P(\omega)^\dagger  G^P(\omega), X \right\} = 2a^2 D_{\sigma_+}(X) + 2b^2 D_{\sigma_-}(X) + c^2 D_{\sigma_Z}(X) \,.
    \end{align}
    The condition of $\tr_1(X) = 0$ implies that $X$, when writing as an expansion in Pauli operators, does not have an identity component in the first qubit, i.e.\ we can write
    \begin{equation*}
        X = \sigma_X \otimes B_X + \sigma_Y \otimes B_Y + \sigma_Z \otimes B_Z \,.
    \end{equation*}
    Then the easy case is
    \begin{align*}
        D_{\sigma_Z}(X) = \sigma_Z X \sigma_Z - X &= - \sigma_X \otimes B_X - \sigma_Y \otimes B_Y + \sigma_Z \otimes B_Z - X \\
        &= - 2\sigma_X \otimes B_X - 2\sigma_Y \otimes B_Y \,.
    \end{align*}
    Since
    \begin{align*}
        D_{\sigma_+}(\sigma_X) &= \ketbra{0}{1} \sigma_X \ketbra{1}{0} -\frac12 (\ketbra{1}{1}\sigma_X + \sigma_X \ketbra{1}{1}) = -\frac12 \sigma_X, \\
        D_{\sigma_+}(\sigma_Y) &= \ketbra{0}{1} \sigma_Y \ketbra{1}{0} -\frac12 (\ketbra{1}{1}\sigma_Y + \sigma_Y \ketbra{1}{1}) = -\frac12 \sigma_Y,\\
        D_{\sigma_+}(\sigma_Z) &= \ketbra{0}{1} \sigma_Z \ketbra{1}{0} -\frac12 (\ketbra{1}{1}\sigma_Z + \sigma_Z \ketbra{1}{1}) = - \sigma_Z \,,
    \end{align*}
    then
    \begin{align*}
        D_{\sigma_+}(X) = -\frac12 \sigma_X \otimes B_X -\frac12 \sigma_Y \otimes B_Y - \sigma_Z \otimes B_Z \,.
    \end{align*}
    Similarly,
    \begin{align*}
        D_{\sigma_-}(X) = -\frac12 \sigma_X \otimes B_X -\frac12 \sigma_Y \otimes B_Y - \sigma_Z \otimes B_Z \,.
    \end{align*}
    Therefore we can further simplify
    \begin{align*}
        \eqref{eq:contraction_simplify} &= -(a^2 + b^2) (\sigma_X \otimes B_X + \sigma_Y \otimes B_Y + 2\sigma_Z \otimes B_Z) - 2c^2 (\sigma_X \otimes B_X + \sigma_Y \otimes B_Y) \\
        &= -(a^2 + b^2 + 2c^2) (\sigma_X \otimes B_X + \sigma_Y \otimes B_Y) - 2(a^2 + b^2) \sigma_Z \otimes B_Z \,.
    \end{align*}
    Integrating over $\omega$ weighted by $\gamma(\omega)$ gives
    \begin{align*}
        S(X) &= X + \delta\int_{-\infty}^{\infty} \gamma(\omega) \cdot \eqref{eq:contraction_simplify} \cdot \diff \omega \\
        &= \left(1 - \delta\int_{-\infty}^\infty \gamma(\omega)(a^2+b^2+2c^2) \diff \omega\right)(\sigma_X \otimes B_X + \sigma_Y \otimes B_Y) \\
        &\qquad + \left(1 - \delta\int_{-\infty}^{\infty} \gamma(\omega)(2a^2+2b^2) \diff \omega\right) \sigma_Z \otimes B_Z \,.
    \end{align*}
    Here is the key step for the contraction.
    \begin{align*}
        \int_{-\infty}^{\infty} \gamma(\omega) \cdot b^2 \cdot \diff \omega &= \int_{-\infty}^{\infty} \gamma(\omega) \cdot \left(\widehat{f}(\omega + h)\right)^2 \cdot \diff \omega \\
        &=\int_{-\infty}^{\infty} \exp\left(-\frac{(\omega + \Delta)^2}{2\eta^2}\right) \cdot \left(\frac{(2\pi)^{-1/4}}{\sqrt{\sigma}} e^{-\frac{(\omega + h)^2}{4\sigma^2}}\right)^2 \cdot \diff \omega \\
        &= \frac{1}{\sigma \sqrt{2\pi}} \int_{-\infty}^{\infty} \exp\left(-\frac{(\omega + \Delta)^2}{2\eta^2} - \frac{(\omega + h)^2}{2\sigma^2} \right) \diff \omega \\
        &= \frac{1}{\sigma \sqrt{2\pi}} \sqrt{2\pi} \frac{\eta \sigma}{\sqrt{\eta^2 + \sigma^2}}\exp\Paren{-\frac{(\Delta - h)^2}{2(\eta^2 + \sigma^2)}} \\
        &= \frac{\eta}{\sqrt{\eta^2 + \sigma^2}}\exp\Paren{-\frac{(\Delta - h)^2}{2(\eta^2 + \sigma^2)}} \,,
    \end{align*}
    where the last second equality is due to the formula for the integral of the product of two Gaussian functions, i.e.
    \begin{equation}\label{eq:int_prod_gaussians}
        \int_{-\infty}^{\infty} e^{-a(x-\mu_1)^2}e^{-b(x-\mu_2)^2} \diff x = \sqrt{\frac{\pi}{a+b}} \exp \Paren{-\frac{ab}{a+b}(\mu_1 - \mu_2)^2} \,.
    \end{equation}
    If $h>1/\beta$, then $\Delta = h$ and thus $\int_{-\infty}^{\infty} \gamma(\omega) \cdot b^2 \cdot \diff \omega = \frac{1}{\sqrt{2}}$.
    If $h\leq 1/\beta$, then $\Delta = 1/\beta$ and $\eta = \sigma = \sqrt{\Delta/\beta} = 1/\beta$. Since $h\geq 0$, then
    \begin{align*}
        \int_{-\infty}^{\infty} \gamma(\omega) \cdot b^2 \cdot \diff \omega = \frac{1}{\sqrt{2}}\exp\Paren{-\frac{(1 - \beta h)^2}{4}} \geq \frac{e^{-1/4}}{\sqrt{2}} \,.
    \end{align*}
    So combining this two cases, since $e^{-1/4}\leq 1$, we always have
    \begin{align*}
        \min \left\{ \int_{-\infty}^\infty \gamma(\omega)(a^2+b^2+2c^2) \diff \omega,\ \int_{-\infty}^\infty \gamma(\omega)(2a^2+2b^2) \diff \omega\right\} \geq \int_{-\infty}^\infty \gamma(\omega)b^2 \diff \omega \geq \frac{e^{-1/4}}{\sqrt{2}} \,.
    \end{align*}
    On the other hand, to apply \Cref{lem:pauli_channel_norm}, we need the conditions that $1 - \delta \int_{-\infty}^\infty \gamma(\omega)(a^2+b^2+2c^2) \diff \omega \geq 0$ and $1 - \delta \int_{-\infty}^\infty \gamma(\omega)(2a^2+2b^2) \diff \omega\geq 0$. Since
    \begin{align*}
        &\max \left\{ \int_{-\infty}^\infty \gamma(\omega)(a^2+b^2+2c^2) \diff \omega,\ \int_{-\infty}^\infty \gamma(\omega)(2a^2+2b^2) \diff \omega\right\} \\
        &\leq 2\int_{-\infty}^\infty \gamma(\omega)(a^2+b^2+c^2) \diff \omega \\
        &= \frac{2\eta}{\sqrt{\eta^2 + \sigma^2}} \Paren{\exp\Paren{-\frac{(\Delta - h)^2}{2(\eta^2 + \sigma^2)}} + \exp\Paren{-\frac{(\Delta + h)^2}{2(\eta^2 + \sigma^2)}} + \exp\Paren{-\frac{\Delta^2}{2(\eta^2 + \sigma^2)}}} \\
        &\leq \frac{6\eta}{\sqrt{\eta^2 + \sigma^2}} = 3\sqrt{2} \,,
    \end{align*}
    we can apply \Cref{lem:pauli_channel_norm} to any $\delta$ such that $0\leq \delta \leq \frac{1}{3\sqrt{2}}$ and conclude that
    \begin{align*}
        \norm{S(X)}_1 \leq \left(1 -  \frac{e^{-1/4}}{\sqrt{2}}\delta\right) \cdot \norm{X}_1 \,. & \qedhere
    \end{align*}
\end{proof}

\begin{lemma}\label{lem:pauli_channel_norm}
    Let $B_X, B_Y, B_Z$ be Hermitian operators on $n-1$ qubits and set $B = \sigma_X \otimes B_X + \sigma_Y \otimes B_Y + \sigma_Z \otimes B_Z$. Then for any $c_x, c_y, c_z\geq 0$, we have that
    \begin{equation*}
        \norm{c_x \sigma_X \otimes B_X + c_y \sigma_Y \otimes B_Y + c_z\sigma_Z \otimes B_Z}_1 \leq \max\{c_x, c_y, c_z\}\cdot \norm{B}_1 .
    \end{equation*}
\end{lemma}
\begin{proof}
    For each $P\in \{\sigma_X, \sigma_Y, \sigma_Z\}$, define the unitary conjugate $B^{(P)} = (P\otimes I_{2^{n-1}}) B (P\otimes I_{2^{n-1}})$:
    \begin{align*}
        B^{(\sigma_X)} &= \sigma_X \otimes B_X - \sigma_Y \otimes B_Y - \sigma_Z \otimes B_Z , \\
        B^{(\sigma_Y)} &= -\sigma_X \otimes B_X + \sigma_Y \otimes B_Y - \sigma_Z \otimes B_Z , \\
        B^{(\sigma_Z)} &= -\sigma_X \otimes B_X - \sigma_Y \otimes B_Y + \sigma_Z \otimes B_Z .
    \end{align*}
    Due to the unitary invariance, $\norm{B^{(P)}}_1 = \norm{B}_1$.
    For any real numbers $p_0, p_1, p_2, p_3$ such that
    \begin{align*}
        c_x &= p_0 + p_1 - p_2 - p_3 ,\\
        c_y &= p_0 - p_1 + p_2 - p_3 ,\\
        c_z &= p_0 - p_1 - p_2 + p_3 ,
    \end{align*}
    we have
    \begin{align*}
        \norm{c_x \sigma_X \otimes B_X + c_y \sigma_Y \otimes B_Y + c_z\sigma_Z \otimes B_Z}_1 &=
        \norm{p_0 B + p_1 B^{(\sigma_X)} + p_2 B^{(\sigma_Y)}+ p_3 B^{(\sigma_Z)}} \\
        &\leq (\abs{p_0} + \abs{p_1} + \abs{p_2} + \abs{p_3})\cdot \norm{B}_1 \tag{triangle inequality}
    \end{align*}
    Without loss of generality, let us assume that $c_x \geq c_y\geq c_z\geq 0$. Then it suffices to show that there exist $p_0,p_1,p_2,p_3$ such that $\abs{p_0} + \abs{p_1} + \abs{p_2} + \abs{p_3} = c_x$.

    One can verify that a valid choice is to
    set $p_0 = \frac12 (c_x + c_z) \geq 0$, $p_1 = \frac12(c_x - c_y) \geq 0$, $p_2 = 0$, $p_3 = \frac12(c_z - c_y) \leq 0$. This completes the proof.
\end{proof}

\section{Rapid mixing from quantum Dobrushin condition}

Recall the definitions of a transport plan and its associated cost vector in \Cref{def:transport_plan_cost_vec}. 
\begin{definition}[Quantum Wasserstein norm]
    For a traceless Hermitian matrix $X$, its \emph{quantum Wasserstein norm} of order $1$ is given by
    \begin{equation*}
        \norm{X}_{W_1} = \min \left\{ \norm{x}_1: x \text{ is the cost vector of a transport plan } (X_i)_{i\in [n]} \text{ of }X\right\}. 
    \end{equation*}
\end{definition}

\begin{proposition}\label{prop:trace_wasserstein_relation}
    For any traceless Hermitian matrix $X$, we have $\frac12 \norm{X}_1 \leq \norm{X}_{W_1}\leq \frac{n}{2} \norm{X}_1$. 
\end{proposition}

\begin{definition}[Quantum Dobrushin influence matrix]\label{def:dobrushin_influence}
    Let $\Phi$ be a trace-preserving linear map which admits a decomposition into trace-preserving linear maps $\Phi = \frac1n(\Phi_1 + \cdots + \Phi_n)$. The associated \emph{Dobrushin influence matrix} $D^{(\Phi)}$ is defined to have entries
    \begin{align*}
        D_{i,j}^{(\Phi)} = \frac1n \max_{\substack{X: \tr_j(X) = 0 \\ \norm{X}_{W_1}=1}} \norm{\Phi_i(X)}_{W_1}.
    \end{align*}
\end{definition}

\begin{lemma}[Contraction from Dobrushin influence, {\cite[Lemma 5.4]{blmt2025dobrushin}}]\label{lem:contraction_dobrushin}
    Let $\Phi$ be a trace-preserving linear map $\Phi$ which admits a decomposition into trace-preserving linear maps $\Phi = \frac{1}{n} (\Phi_1 + \cdots + \Phi_n)$. Denote by $D^{(\Phi)}$ the associated Dobrushin influence matrix.
    Then for any traceless Hermitian matrix $X$,
    \begin{equation*}
        \norm{\Phi(X)}_{W_1} \leq \norm{D^{(\Phi)}}_{1\to 1} \cdot \norm{X}_{W_1}. 
    \end{equation*}
\end{lemma}
\begin{proof}
    Let $X_1, \ldots, X_n$ be the optimal transport plan of $X$ with cost vector $x$, i.e.\ $X = \sum_{j=1}^n X_j$, $\tr_j(X_j) = 0$, $x_j = \frac12 \norm{X_j}_1$ and $\norm{X}_{W_1} = \norm{x}_1$. Then
    \begin{align*}
        \norm{\Phi(X)}_{W_1} = \norm*{\frac{1}{n}\sum_{i=1}^n \sum_{j=1}^n \Phi_i(X_j)}_{W_1} &\leq \frac{1}{n}\sum_{i=1}^n \sum_{j=1}^n \norm*{ \Phi_i(X_j)}_{W_1} \tag{triangle inequality} \\
        &\leq \sum_{j=1}^n \sum_{i=1}^n  D_{i,j}^{(\Phi)} \cdot \norm{X_j}_{W_1} \tag{\Cref{def:dobrushin_influence}}\\
        &\leq \sum_{j=1}^n \sum_{i=1}^n  D_{i,j}^{(\Phi)} \cdot x_j \tag{$\norm{X_j}_{W_1} \leq \frac12 \norm{X_j}_{1} = x_j$} \\
        &= \norm{D^{(\Phi)}x}_1 \tag{$D_{i,j}^{(\Phi)}, x_j\geq 0$} \\
        &\leq \norm{D^{(\Phi)}}_{1\to 1} \cdot \norm{X}_{W_1} \tag{$\norm{X}_{W_1} = \norm{x}_1$}
    \end{align*}
    This completes the proof. 
\end{proof}

\begin{lemma}[Bounding the Dobrushin influence matrix]
\label{lem:dobrushin_influence_bound}
    Let $H = V + W$, where $V = \sum_{i\in [n]} v_i$ is an external field and $W$ is a bounded $L$-local degree-$D$ Hamiltonian. 
    For any inverse temperature $\beta \in (0, (DL)^{-3}/28800 ]$ and any $\delta \in [0, 1/(6\sqrt{2})]$, define $\Phi = \calI + \frac{\delta}{n} \calL_{\mathsf{FR}}$ where $\calL_{\mathsf{FR}}$ is the associated field-resonant Lindbladian defined in \Cref{def:field_lindblad}.
    Denote by $D^{(\Phi)}$ the Dobrushin influence matrix of $\Phi$. Then
    \begin{equation*}
        \norm{D^{(\Phi)}}_{1\to 1}\leq 1 - \frac{\delta}{2n} + c\delta^2 \,.
    \end{equation*}
    where $c=144+ 324 (\log(2\sqrt{\Delta_{\max} \beta }) + 4)^2$ and $\Delta_{\max} \coloneqq \max_{i\in [n]}\Delta_i$.  
\end{lemma}

\begin{proof}
    We first split our Lindbladian as follows:
    \begin{equation*}
        \Phi = \calI + \frac{\delta}{n} \calL_{\mathsf{FR}} = \frac{1}{n} \sum_{i=1}^n (\underbrace{\calI + \delta \calL_{\mathsf{FR},i}}_{= \Phi_i}). 
    \end{equation*}
    We will analyze the Dobrushin influence matrix $D^{(\Phi)}$ of $\Phi$. 
    Since each entry of $D^{(\Phi)}$ is nonnegative, then $\norm{D^{(\Phi)}}_{1\to 1} = \max_j \sum_{i=1}^n D_{i,j}^{(\Phi)}$, i.e.\ the maximum column sum. 

    We first upper bound each entry of $D^{(\Phi)}$. 
    Let $X$ be a Hermitian traceless matrix with $\tr_j(X) = 0$ and $\norm{X}_{W_1} = 1$, which means that $X$ has a transport plan with cost vector at most $e_j$. Then by \Cref{thm:wass_growth_per_site}, $\Phi_i(X)$ has a transport plan with cost vector entry-wise upper bounded by
    \begin{equation*}
        \Paren{(1+c_i\delta^2)\cdot I + \delta \cdot Q^{(i)}}e_j, \quad \text{where} \quad Q^{(i)} = - \frac{e^{-1/4}}{\sqrt{2}} E_{\{i\}} + 120 \sum_{r\geq 1} (6\beta \zeta)^r \cdot E_{\ball(\{i\}, r)} \,,
    \end{equation*}
    and $c_i = 144+ 324 (\log(2\sqrt{\Delta_i \beta }) + 4)^2 \leq c$. 
    Hence,
    \begin{align*}
        n D_{i,j}^{(\Phi)} &= \max_{\substack{X: \tr_j(X) = 0 \\ \norm{X}_{W_1}=1}} \norm{\Phi_i(X)}_{W_1} \\
        &\leq \norm*{\left((1+c\delta^2)\cdot I + \delta \cdot \left(-\frac{e^{-1/4}}{\sqrt{2}} E_{\{i\}} + 120 \sum_{r\geq 1} (6\beta \zeta)^r \cdot E_{\ball(\{i\}, r)}\right)\right) e_j}_1 \\
        &\leq \Paren{1 - \Iver{i=j} \cdot \frac{e^{-1/4}}{\sqrt{2}}\delta } + 120\delta \cdot \norm*{\sum_{r\geq 1} (6\beta \zeta)^r \cdot E_{\ball(\{i\}, r)}\cdot e_j}_1 + c\delta^2. 
    \end{align*}
    Now, we sum the entries over a column. Using the triangle inequality again, we have
    \begin{align*}
        \sum_{i=1}^n D_{i,j}^{(\Phi)} &\leq 1 - \frac{e^{-1/4}}{\sqrt{2}} \cdot \frac{\delta}{n} + c\delta^2 + \frac{120\delta}{n} \sum_{i=1}^n \norm*{\sum_{r\geq 1} (6\beta \zeta)^r \cdot E_{\ball(\{i\}, r)}\cdot e_j}_1,
    \end{align*}
    where
    \begin{align*}
        \sum_{i=1}^n \norm*{\sum_{r\geq 1} (6\beta \zeta)^r \cdot E_{\ball(\{i\}, r)}\cdot e_j}_1 &\leq \sum_{i=1}^n \sum_{r\geq 1} (6\beta \zeta)^r \cdot \norm*{E_{\ball(\{i\}, r)}\cdot e_j}_1 \\
        &= \sum_{r\geq 1} (6\beta \zeta)^r \sum_{i\in [n]: \dist(i,j)\leq r} \abs{\ball(\{i\}, r)} \,.
    \end{align*}
    By \Cref{fact:interaction-growth}, we have $\zeta \leq DL$, $\abs{\ball(\{j\}, r)}\leq (DL)^r$, and $\abs{\ball(\{i\}, r)}\leq (DL)^r$. So
    \begin{equation*}
        \sum_{i\in [n]: \dist(i,j)\leq r} \abs{\ball(\{i\}, r)} \leq \abs{\ball(\{j\}, r)} \cdot \max_{i\in [n]} \abs{\ball(\{i\}, r)} \leq (DL)^{2r} \,.
    \end{equation*}
    Continuing, with the assumption that $6 \beta (DL)^3\leq 1/2$, we have
    \begin{equation*}
        \sum_{i=1}^n \norm*{\sum_{r\geq 1} (6\beta \zeta)^r \cdot E_{\ball(\{i\}, r)}\cdot e_j}_1 \leq \sum_{r\geq 1} (6\beta DL)^r (DL)^{2r} \leq 12 \beta (DL)^3 \,.
    \end{equation*}
    Since $e^{-1/4}/\sqrt{2}>0.55$ and $\beta \leq 0.05/(12\cdot 120 \cdot (DL)^3)$, we have
     \begin{align*}
        \sum_{i=1}^n D_{i,j}^{(\Phi)} &\leq 1 - \frac{e^{-1/4}}{\sqrt{2}}\cdot\frac{\delta}{n} + c\delta^2 + \frac{120\delta}{n} \cdot 12\beta (DL)^3 \leq 1 - \frac{\delta}{2n} + c\delta^2, 
    \end{align*}
    which completes the proof. 
\end{proof}

We are now ready to prove our main theorem.
\begin{theorem}[Rapid mixing with arbitrary external field]
\label{thm:rapid-mixing-with-external-field}
    Let $H = V + W$, where $W$ is a bounded $L$-local degree-$D$ Hamiltonian and $V = \sum_{i\in [n]} v_i$ is an external field with $v_i$ supported on site $i$ only.
    Let $\sigma = e^{-\beta H}/\tr(e^{-\beta H})$ be the Gibbs state.
    For any inverse temperature $\beta \in (0, (DL)^{-3}/28800 ]$ and any initial state $\rho$, the field-resonant Lindbladian $\calL_{\mathsf{FR}}$ defined in \Cref{def:field_lindblad} outputs a state $\widehat{\rho}$ such that $\norm{\widehat{\rho} - \sigma}_1 \leq \epsilon$ after running for $\bigO{\log(2n/\epsilon)}$ time.
\end{theorem}
\begin{proof}
    Since the Gibbs state $\sigma$ is a fixed point of $\calL_{\mathsf{FR}}$ (by \Cref{prop:kms}), we have $\Phi(\sigma) = \sigma$. 
    Let $X = \rho - \sigma$ which is traceless and Hermitian. 
    Then by \Cref{lem:contraction_dobrushin}, 
    \begin{align*}
        \norm{\Phi(\rho - \sigma)}_{W_1} \leq \norm{D^{(\Phi)}}_{1\to 1} \cdot \norm{\rho - \sigma}_{W_1}. 
    \end{align*}
    This means that for any interger $\tau \geq 1$, we have
    \begin{align*}
        \norm{\Phi^{\tau}(\rho - \sigma)}_{W_1} \leq \norm{D^{(\Phi)}}^{\tau}_{1\to 1} \cdot \norm{\rho - \sigma}_{W_1}. 
    \end{align*}    
    Using \Cref{prop:trace_wasserstein_relation}, we know $\norm{\Phi^\tau(\rho - \sigma)}_{W_1} \geq \frac12 \norm{\Phi^\tau(\rho - \sigma)}_1$ and $\norm{\rho - \sigma}_{W_1} \leq \frac{n}{2} \norm{\rho - \sigma}_1$. Therefore
    \begin{align*}
        \norm{\Phi^\tau(\rho) - \sigma}_{1} \leq n \cdot \norm{D^{(\Phi)}}^\tau_{1\to 1} \cdot \norm{\rho - \sigma}_{1}. 
    \end{align*}
    Plugging in \Cref{lem:dobrushin_influence_bound}, we have
    \begin{align*}
        \norm*{\left(\calI + \frac{\delta}{n} \calL_{\mathsf{FR}}\right)^\tau(\rho) - \sigma}_{1} \leq n \cdot \left(1 - \frac{\delta}{2n} + c\delta^2\right)^\tau \cdot \norm{\rho - \sigma}_{1}. 
    \end{align*}
    Due to the technicality that $\tau$ needs to be an integer, we consider a particular sequence $\{\delta_k\}_{k=1}^{\infty}$ where each $\delta_k = \frac{1}{k}$. 
    Let $t_0 = \lceil 2 \log(2n/\epsilon) \rceil$ which is an integer. 
    Then $\tau_k = t_0 n/\delta_k = t_0 nk$ is also an integer. 
    Now for any $t\geq t_0$, we have
    \begin{align*}
        \norm*{e^{t\calL_{\mathsf{FR}}}(\rho) - \sigma}_1 \leq \norm*{e^{t_0\calL_{\mathsf{FR}}}(\rho) - \sigma}_1 &= \lim_{k \to \infty} \ \norm*{\left(\calI + \frac{\delta_k}{n}\calL_{\mathsf{FR}}\right)^{t_0 \frac{n}{\delta_k}}(\rho) -\sigma}_1  \\
        &= \lim_{k \to \infty} \ \norm*{\left(\calI + \frac{1}{nk}\calL_{\mathsf{FR}}\right)^{t_0 n k}(\rho) -\sigma}_1  \\
        &\leq n \cdot\norm{\rho - \sigma}_1 \cdot \lim_{k \to \infty} \ \left(1 - \frac{1}{2kn} + \frac{c}{k^2}\right)^{t_0 nk}  \\
        &= n \cdot\norm{\rho - \sigma}_1 \cdot e^{-t_0/2}\leq \epsilon,
    \end{align*}
    which completes the proof. In particular, this implies that $\sigma$ is the unique stationary state of the semigroup $(e^{t\calL_{\mathsf{FR}}})_{t\geq 0}$. 
\end{proof}

\section{Separability with external field}
\label{sec:separability}
In this section, we show that high-temperature Gibbs states of local Hamiltonians remain separable as long as the magnitude of the external field is bounded by $\Theta(\log(1/\beta)/\beta)$, assuming the degree and locality are a fixed constant. 

The main theorem we prove in this section is as follows:

\begin{theorem}[Separability]
\label{thm:separability-with-external-field}
    Let $H = V + W$, where $W$ is a bounded $L$-local degree-$D$ Hamiltonian and $V = \sum_{i\in [n]} V_i$ is an external field with $V_i$ supported on site $i$ only. Let $h = \max_i \norm{V_i}$.
    If
    \begin{equation*}
        0<\beta \leq \frac{1}{8DL (56)^{2L}} \quad \textrm{ and } \quad  h \leq \frac{1}{8\beta L}\log\Paren{\frac{1}{4DL\beta}} \,,
    \end{equation*}
    then the Gibbs state $e^{-\beta H}/\tr(e^{-\beta H})$ is a convex combination of product states.
\end{theorem}

The following definition of clusters will be used frequently in this section. 
\begin{definition}[Clusters]
\label{def:cluster}
Given an integer $k\geq 1$ and a bounded $L$-local degree-$D$ Hamiltonian $W = \sum_{a\in [m]} W_a$, we say that $\ba \in [m]^k$ is a cluster of length $k$ from site $i \in [n]$ if $i\in \supp(W_{a_1})$ and
\begin{equation*}
    \supp(W_{a_\ell}) \cap \bigg(\bigcup_{j\in [\ell - 1]} \supp(W_{a_j})\bigg) \neq \varnothing, \quad \text{ for each } \ell = 2, \ldots, k \,.
\end{equation*}
\end{definition}

The backbone of the proof of \Cref{thm:separability-with-external-field} is the following statement on the behavior of the Araki expansional in the presence of an external field, whose proof is deferred to \Cref{sec:quasi-locality-expansion}: 

\begin{corollary}[Araki expansionals are quasi-local]
\label{cor:araki-expansionals-quasi-local}
    Let $W^{(1)} = \sum_{a\in [m]: \ 1 \in \supp(W_a)} W_a$. 
    For any inverse temperature $0<\beta \leq \frac{1}{4DL 56^L}$ and $0\leq h \leq \frac{1}{8\beta L}\log\Paren{\frac{1}{4DL\beta}}$, we have
    \begin{equation*}
        e^{-\beta H} e^{\beta (H - W^{(1)})} = I + \sum_{k=1}^{\infty} \sum_{\substack{\ba \in [m]^k: \\ \text{cluster from }1}} \; R_{\ba} \,,
    \end{equation*}
    where for all $\ba$,
    $\abs{\supp(R_{\ba})}\leq kL$ and $\sum_{\ba} \Norm{ R_{\ba}} \leq (1/56)^{kL}$.
\end{corollary}

We also require the following lemma from Bakshi, Choi, and Pilatowski-Cameo~\cite{BCP26,pilatowsky2026one}:

\begin{lemma}[Local perturbations, {\cite[Lemma 6.2]{BCP26}}]
\label{lem:local-perturbations}
Let $X =I + Y$ such that $\abs{\supp(Y)}\leq \ell$ and $\norm{Y}\leq 4^{-\ell}$. Then, $X$ can be decomposed into a convex combination of stabilizer product states.
\end{lemma}

We are now ready to prove the separability theorem for external field.

\begin{proof}[Proof of \Cref{thm:separability-with-external-field}]
For any subset $S\subseteq[n]$, use the restricted Hamiltonian $H_S$ from \Cref{def:interaction-truncated-Ham}, with $W=\sum_{a\in[m]}W_a$.
Further, for any $i \in S$, let $W^{(i)}_{S}$ be the interaction terms in $H_S$ that touch qubit $i$, i.e.\ $W^{(i)}_{S} = \sum_{a\;:\; i \in \supp(W_a)\subseteq S} W_a$. Then, let 
\begin{equation*}
     H_S^{(i)} = H_S - W_S^{(i)} = H_{S \setminus \set{i}} + v_i\,.
\end{equation*}
Observe, by construction, $H_S^{(i)}$ does not have any interaction term acting at site $i$ and therefore
\begin{equation*}
    e^{-\beta H_S^{(i)}} =  e^{-\beta v_i}  \; \otimes  \;e^{-\beta H_{S \setminus \set{i}} }  \,,
\end{equation*}
where we abuse the notation $e^{-\beta v_i}$ to denote the single-qubit operator restricted to site $i$ only and $e^{-\beta H_{S \setminus \set{i}} }$ to denote the operator restricted to action on qubits $[n]\setminus \{i\}$. 
Next, observe 
\begin{equation}
\label{eqn:nice-conjugation}
    e^{-\beta H_S } = e^{-\beta H_S^{(i)}/2} \underbrace{\Paren{ e^{\beta H_S^{(i)}/2}  \cdot e^{-\beta H_S/2} }}_{X_{S,i}^\dagger} \underbrace{ \Paren{ e^{-\beta H_S/2} \cdot e^{\beta H_S^{(i)}/2}    }}_{X_{S,i}}   e^{-\beta H_S^{(i)}/2} 
\end{equation}
Applying \Cref{cor:araki-expansionals-quasi-local} with $\beta \gets \beta/2$ to $X_{S,i}$, we have $X_{S,i} = I + \sum_{k=1}^{\infty} \sum_{\lambda \in \Lambda_{S,i,k}} \; R_{\lambda} $, where each $\lambda$ is a cluster of $k$ interacting terms in $W_S$ starting at qubit $i$ as defined in \Cref{def:cluster}. 
Let us write the support of the cluster $\lambda$ as $U_\lambda = \cup_{a \in \lambda} \supp(W_a)$. Then, we know that $i \in U_\lambda$, $\abs{U_{\lambda}}\leq kL$, $\supp(R_{\lambda}) \subseteq U_{\lambda}$, and    $\sum_{\lambda \in \Lambda_{S,i,k}} \Norm{ R_{\lambda}} \leq (1/56)^{kL}$. 

\paragraph{Single-site pining.} Given a Pauli string $A$ supported on some $U\subseteq S$ such that $i\in U$, $A \neq I$, and some $c$ such that $\abs{c}\leq 7^{-\abs{U}}$, we claim that we can write
\begin{equation*}
    X_{S,i}^{\dagger} \Paren{I + cA} X_{S,i} = \sum_{j } w_j \ketbra{s_j}{s_j}_i \otimes \Paren{I + c_j' B_j}\,,
\end{equation*}
where $w_j\geq 0$, $\ket{s_j}$ is a stabilizer state at qubit $i$, $\abs{c_j'} \leq 7^{-\abs{U'}}$ and $U_j' = \supp(B)$. 

To see why, first write $R_\lambda$ in the Pauli basis as follows: $R_\lambda = \sum_P \alpha_{\lambda,P} \cdot P$. Since $\abs{U_\lambda}\leq kL$, there are at most $4^{kL}$ such Pauli terms and thus  
\begin{equation*}
    \calZ_k = \sum_{\lambda \in \Lambda_{S, i,k}} \sum_{P} \;\; \abs{ \alpha_{\lambda,P}} \leq 4^{kL} \sum_{\lambda \in \Lambda_{S, i,k}} \Norm{R_{\lambda}} \leq 14^{-kL}\,. 
\end{equation*}
Next, consider a random operator $Z$ sampled based on the following process: 
\begin{itemize}
    \item Sample an index $k$, where $k \geq 1$ is sampled with probabiltiy $2^{-kL}$ and $k=0$ is sampled with probabiltiy $1 - \sum_{k\geq 1} 2^{-kL}$. 
    \item If $k =0$, set $Z = I$.
    \item If $k\geq 1$, sample $(\lambda, P)$ with probability $\abs{\alpha_{\lambda,P}}/ \calZ_k$ and then set $Z = I + zP$ where $z = 2^{kL} \cdot \calZ_k \cdot \alpha_{\lambda, P}/\abs{\alpha_{\lambda, P}}$. 
\end{itemize}
Then, $\expecf{}{Z} = X_{S,i}$ and the support of the Pauli string $P$ is contained within $U_{\lambda}$, $i \in U_{\lambda}$, and $\abs{z}\leq 7^{-\abs{U_{\lambda}}}$. We refer to $U_{\lambda}$ as the "associated support" for the sampled Pauli string $P$.  

Let $Z= I + zP$ and $Z'=I+z'P'$ be two independent copies drawn as described above, and let $T$, $T'$ be the associated support.  Then, $\supp(P)\subseteq T, \supp(P')\subseteq T'$, $\abs{z}\leq 7^{-\abs{T}}$ and $\abs{z'} \leq 7^{-\abs{T'}}$. 
\begin{equation*}
     \frac{1}{2} \expecf{}{ Z^\dagger \Paren{I + cA} Z' + (Z')^\dagger \Paren{I + cA} Z  } =  X_{S,i}^{\dagger} \Paren{I + cA} X_{S,i} \,.
\end{equation*}
Using the definition of $Z,Z'$ and multiplying things out, we have 
\begin{equation*}
    \begin{split}
          Z^\dagger \Paren{I + cA} Z' = I + cA + z'P' + \bar{z}P + cz'AP' + c\bar{z} PA + z' \bar{z} P P' + cz' \bar{z} PAP'\,,
    \end{split}
\end{equation*}
and similarly
\begin{equation*}
    \begin{split}
          (Z')^\dagger \Paren{I + cA} Z = I + cA + zP + \bar{z}'P' + czAP + c\bar{z}' P'A + z \bar{z}' P'P + cz \bar{z}' P'AP\,.
    \end{split}
\end{equation*}
Averaging them together, we have
\begin{equation*}
    \begin{split}
        \frac{1}{2} \Paren{ Z^\dagger\Paren{I + cA} Z' + (Z')^\dagger \Paren{I + cA} Z  } & = I + cA + \mathfrak{R}(z) P + \mathfrak{R}(z') P' \\
        & \hspace{0.2in} +  \frac{c}{2}\Paren{zAP + \bar{z}PA  }   + \frac{c}{2}\Paren{ z'AP' + \bar{z}' P'A}  \\
        & \hspace{0.2in} +  \frac{1}{2}\Paren{z \bar{z}' P'P + z'\bar{z} PP'} + \frac{c}{2}\Paren{ z\bar{z}'P'AP + z' \bar{z} PAP' } \,.
    \end{split}
\end{equation*}
In particular, we can write the $7$ terms in the RHS above as $\frac{1}{7} \sum_{\nu=1}^7 \Paren{I + 7c_\nu Q_\nu}$. Now, consider the following associated support for each term: $U'_1 = U$, $U'_2 = T$, $U'_3 = T'$, $U'_4 = U \cup T$, $U'_5 = U \cup T'$, $U'_6 = T \cup T'$ and $U'_7 = U \cup T \cup T'$. Then, for each $\nu \in [7]$, $\supp(Q_\nu) \subseteq U'_{\nu}$. The coefficients can be bounded as follows: $7 \abs{c_1} \leq 7^{-\abs{U} +1}$, $7 \abs{c_2} \leq 7^{-\abs{T} +1}$, $7 \abs{c_3} \leq 7^{-\abs{T'} +1}$, $7 \abs{c_4} \leq 7^{-(\abs{U} + \abs{T}) +1}$, $7 \abs{c_5} \leq 7^{-(\abs{U} + \abs{T'}) +1}$, $7 \abs{c_6} \leq 7^{-(\abs{T} +\abs{T'}) +1}$ and $7 \abs{c_7} \leq 7^{-(\abs{U}+\abs{T}+\abs{T'}) +1}$. 
Since $i \in U, T, T'$, after pinning site $i$, the residual support is 
$U_{\nu}'' \coloneqq U_\nu' \setminus \{1\}$ for each $\nu \in [7]$. 
Then $\supp(B_\nu) \subseteq U_{\nu}''$ and $7\abs{c_\nu} \leq 7^{-\abs{U''_{\nu}}}$, maintaining the desired invariant. 
We can then invoke \cref{lem:local-perturbations} to conclude that each $I+7c_\nu Q_\nu$ is a convex combination of stabilizer product states, yielding the overall claim.

\paragraph{Induction.}Recall \cref{eqn:nice-conjugation} with $S=[n]$:
\begin{equation*}
\begin{split}
    e^{-\beta H_{[n]} } & = e^{-\beta H_{[n]}^{(1)}/2}  \Paren{ e^{\beta H_{[n]}^{(1)}/2}  \cdot e^{-\beta H_{[n]}/2} }  \Paren{ e^{-\beta H_{[n]}/2} \cdot e^{\beta H_{[n]}^{(1)}/2}    }   e^{-\beta H_{[n]}^{(1)}/2} \\
    & = e^{-\beta H_{[n]}^{(1)}/2} \Paren{\sum_{j } w_j \ketbra{s_j}{s_j}_1 \otimes \Paren{I + c'_j B_j} } e^{-\beta H_{[n]}^{(1)}/2}  \\
    & = \Paren{ \sum_j w_j e^{-\beta v_1/2} \ketbra{s}{s}_1 e^{-\beta v_1/2}}  \otimes \Paren{ e^{-\beta H_{[n]\setminus \{1\}}/2} \Paren{I + c_1' B_1} e^{-\beta H_{[n]\setminus \{1\}}/2} } \,,
\end{split}
\end{equation*}
where we applied the single site pinning with $cA =0$. We now recurse and iteratively pin each site, completing the proof.
\end{proof}

\subsection{Quality-local perturbation of the identity}
\label{sec:quasi-locality-expansion}

The proof of \Cref{cor:araki-expansionals-quasi-local} follows from \Cref{lem:separability_combo} and counting the number of clusters.

\begin{proof}[Proof of \Cref{cor:araki-expansionals-quasi-local}]
    For any integer $k\geq 1$, let us count the number of clusters $\ba\in [m]^k$ from site $1$. Since each site is acted upon by at most $D$ terms, there are $D$ ways to pick $a_1$. 
    For $a_2$, since $W_{a_2}$ must overlap with $W_{a_1}$ and $\abs{\supp(W_{a_1})}\leq L$, then there are at most $DL$ ways to pick $a_2$. 
    Continuing, suppose we have picked $a_1, \ldots, a_i$, then there are at most $iLD$ ways to pick $a_{i+1}$. Overall, taking the product, we conclude that the number of clusters $\ba\in [m]^k$ from site $1$ is at most $D^k L^{k-1} (k-1)!$. 
    So $\abs{\supp(R_{\ba})} \leq kL$ and
    \begin{equation*}
        \sum_{\substack{\ba \in [m]^k: \\ \text{cluster from }1}} \norm{R_{\ba}} \leq D^k L^{k-1} (k-1)! \cdot \frac{1}{k!} \Paren{\frac{e^{4\beta h L}-1}{hL}}^k \leq \Paren{\frac{D(e^{4\beta h L}-1)}{h}}^k \,.
    \end{equation*}
    It remains to bound the right-hand side by $(1/56)^{kL}$, which follows from the assumptions $0<\beta \leq \frac{1}{4DL 56^{2L}}$ and $0\leq h \leq \frac{1}{8\beta L}\log\Paren{\frac{1}{4DL\beta}}$. 
\end{proof}

The rest of the section proves the following lemma. 
\begin{lemma}[Quasi-local perturbation of the identity]\label{lem:separability_combo}
    Let $W^{(1)} = \sum_{a\in [m]: \ 1 \in \supp(W_a)} W_a$. 
    For any inverse temperature $\beta >0$, 
    \begin{equation*}
        e^{-\beta H} e^{\beta (H - W^{(1)})} = I + \sum_{k=1}^{\infty} \sum_{\substack{\ba \in [m]^k: \\ \text{cluster from }1}}  R_{\ba} \,,
    \end{equation*}
    where $R_{\ba}$ is supported on $S_{\ba} = \supp(W_{a_1}) \cup \cdots \cup \supp(W_{a_k})$ and $\norm{R_{\ba}} \leq \frac{1}{k!} \Paren{\frac{e^{4\beta h L}-1}{hL}}^k$.  
\end{lemma}
\begin{proof}
    We begin with
    \begin{equation*}
        e^{-\beta H} e^{\beta (H - W^{(1)})} = \sum_{t=0}^{\infty} \frac{\beta^t}{t!} f_t(H, W^{(1)}) \,,
    \end{equation*}
    where $f_0(H, W^{(1)}) = I$ and for any $t\geq 1$, 
    \begin{equation*}
        f_t(H, W^{(1)}) = [f_{t-1}(H, W^{(1)}), H] - f_{t-1}(H, W^{(1)}) W^{(1)} \,.
    \end{equation*}
    Each recursion step is therefore a sum of two contributions, so $f_t(H, W^{(1)})$ can be written as a sum of $2^t$ terms.
    Concretely, an index string $\bb\in \{0,1\}^t$ specifies one particular sequence of choices in this expansion. Starting with $g_0(\bb) = I$, define $g_i(\bb)$ iteratively for $i=1, \ldots, t$ by
    \begin{align*}
        g_i(\bb) = \begin{cases}
            [g_{i-1}(\bb), H]\,, & \text{if } b_i=1\,,\\
            -g_{i-1}(\bb)W^{(1)}\,, & \text{if } b_i=0\,.
        \end{cases}
    \end{align*}
    Summing over all sequences gives $f_t(H, W^{(1)}) = \sum_{\bb\in \{0,1\}^t} g_t(\bb)$. 

    Because both $H$ and $W^{(1)}$ are sums of local terms, we can further define a more detailed sequence of operations that describes which particular local terms are taken into the commutator or being multiplied. 
    Concretely, let us write the restricted index sets from $W^{(1)}$ as $M_1 \coloneqq \{a\in [m]: 1\in \supp(W_a)\}$. 
    For each step $i$, define the allowed index set
    \begin{align*}
        \calI(b_i, c_i) \coloneqq \begin{cases}
            [m], & \text{if }(b_i, c_i) = (1,1) \ (\text{take the commutator with some }W_a) \\
            [n], & \text{if }(b_i, c_i) = (1,0) \ (\text{take the commutator with some }v_j) \\
            M_1, & \text{if }(b_i, c_i) = (0,1) \ (\text{multiply by some }W_a)
        \end{cases}
    \end{align*}
    So $c_i=1$ means the step uses a $W$-term and $c_i=0$ means the step uses a $V$-term.  
    We define a set of $\bb$-admissible $(\bc, \ba)$ as
    \begin{equation*}
        \calA_{\bb} \coloneqq \{(\bc, \ba): \text{for each }i\in [t], \ b_i=0 \Rightarrow c_i=1,  \text{ and } a_i \in \calI(b_i, c_i) \} \,.
    \end{equation*}
    Then for any $\bb\in \{0,1\}^t$ and $(\bc, \ba)\in \calA_{\bb}$, starting with $F_0(\bb, \bc, \ba) = I$, define $F_i(\bb, \bc, \ba)$ iteratively for $i=1, \ldots, t$ by
    \begin{align*}
        F_i(\bb, \bc, \ba) = \begin{cases}
            [F_{i-1}(\bb, \bc, \ba), W_{a_i}], & \text{if }(b_i, c_i)=(1,1) \\
            [F_{i-1}(\bb, \bc, \ba), v_{a_i}], & \text{if }(b_i, c_i)=(1,0) \\
            -F_{i-1}(\bb, \bc, \ba) W_{a_i}, & \text{if }(b_i, c_i)=(0,1)
        \end{cases}
    \end{align*}
    Summing over such $\bb$ and $(\bc, \ba)\in \calA_{\bb}$ gives
    \begin{equation}\label{eq:separable_f_expand}
        f_t(H, W^{(1)}) = \sum_{\bb\in \{0,1\}^t} \sum_{(\bc, \ba)\in \calA_{\bb}} F_t(\bb, \bc, \ba) \,.
    \end{equation}
    We now define a projection map that extracts the $W$-indices. Let $I_W(\bc) = \{i\in [t]: c_i=1\}$. 
    Suppose $\abs{I_W} = k$ with elements $i_1 < i_2 < \cdots < i_k$. Define
    \begin{equation*}
        \pi(\bc, \ba) = \widetilde{\ba} = (a_{i_1}, \ldots, a_{i_k}) \in [m]^k \,.
    \end{equation*}
    Then $\calA_{\bb}$ is the disjoint union of 
    \begin{equation*}
        \calA_{\bb} = \bigsqcup_{k=0}^t \bigsqcup_{\widetilde{\ba}\in [m]^k} \calA_{\bb}(\widetilde{\ba}), \quad \text{where }\calA_{\bb}(\widetilde{\ba}) \coloneqq \{(\bc, \ba) \in \calA_{\bb}: \pi(\bc, \ba) = \widetilde{\ba}\} \,.
    \end{equation*}
    Plugging this back to \Cref{eq:separable_f_expand} yields
    \begin{equation*}
        f_t(H, W^{(1)}) = \sum_{k=0}^t \sum_{\widetilde{\ba}\in [m]^k} \sum_{\bb\in \{0,1\}^t} \sum_{(\bc, \ba) \in \calA_{\bb}(\widetilde{\ba})} F_t(\bb, \bc, \ba)  \,,
    \end{equation*}
    We first isolate the contribution of $\widetilde{\ba} = \varnothing$. 
    In this case, $c_i=0$ for all $i$. By admissibility this forces $b_i=1$ for all $i$. 
    So every step is of the form 
    $F_i = [F_{i-1}, v_{a_i}]$.
    Starting from $F_0 = I$, the first step gives $F_1 = [I, v_{a_i}] = 0$. So all $t\geq 1$ contributions vanish. 
    Then
    \begin{equation*}
        f_t(H, W^{(1)}) = I \cdot \Iver{t=0} + \sum_{k=1}^t \sum_{\widetilde{\ba}\in [m]^k} \sum_{\bb\in \{0,1\}^t} \sum_{(\bc, \ba) \in \calA_{\bb}(\widetilde{\ba})} F_t(\bb, \bc, \ba)  \,,
    \end{equation*}
    Here are the two key observations:
    \begin{itemize}
        \item If $\widetilde{\ba}$ is not a cluster from $1$, then $F_t(\bb, \bc, \ba) = 0$.
        \item If there exists $i\in [t]$ such that $c_i = 0$ and $a_i \notin \bigcup_{j<i: c_j=1} \supp(W_{a_j})$, then $F_t(\bb, \bc, \ba) = 0$. 
    \end{itemize}
    Observation $2$ motivates us to define
    \begin{equation*}
        \calA(\widetilde{\ba}) \coloneqq \Big\{(\bc, \ba): \pi(\bc, \ba) = \widetilde{\ba} \text{ and for each }i\in [t] \text{ with }c_i=0, \ a_i \in \bigcup_{j<i: \ c_j=1} \supp(W_{a_j})\Big\} \,.
    \end{equation*}
    Then 
    \begin{equation*}
        f_t(H, W^{(1)}) = I \cdot \Iver{t=0} + \sum_{k=1}^t \sum_{\substack{\widetilde{\ba} \in [m]^k: \\ \text{cluster from }1}} \underbrace{\sum_{\bb\in \{0,1\}^t} \sum_{(\bc, \ba) \in \calA_{\bb}(\widetilde{\ba}) \cap \calA(\widetilde{\ba})} F_t(\bb, \bc, \ba)}_{\coloneqq G_{t, \widetilde{\ba}}} \,.
    \end{equation*} 
    where $\supp(G_{t, \widetilde{\ba}}) \subseteq S_{\widetilde{\ba}}$. 
    We now bound $\norm{G_{t, \widetilde{\ba}}}$. Since $\norm{[X, Y]}\leq 2\norm{X}\cdot \norm{Y}$ and $\norm{v_i}\leq h$, one can prove by induction that 
    \begin{equation*}
        \norm{F_t(\bb, \bc, \ba)} \leq 2^{b_1+\cdots + b_t}\cdot  h^{t-(c_1 + \cdots + c_t)} \leq 2^t \cdot  h^{t-(c_1 + \cdots + c_t)} \,.
    \end{equation*}
    So given $\widetilde{\ba} \in [m]^k$, we have
    \begin{equation*}
        \norm{G_{t, \widetilde{\ba}}} \leq \sum_{\bb\in \{0,1\}^t} \sum_{(\bc, \ba) \in \calA_{\bb}(\widetilde{\ba})\cap \calA(\widetilde{\ba})} \norm{F_t(\bb, \bc, \ba)} \leq 2^t \cdot h^{t-k} \cdot 2^t \cdot \abs{\calA(\widetilde{\ba})} \,.
    \end{equation*}
    We now give a clean bound for $\abs{\calA(\widetilde{\ba})}$. Recall that $\abs{\supp(W_a)}\leq L$. For a fixed $\widetilde{\ba} \in [m]^k$:
    \begin{itemize}
        \item Choose $\bc\in \{0,1\}^t$ with exactly $k$ ones (these are the positions where $W$ occurs).
        \item Once $\bc$ is fixed, the $W$-indices $a_i$ for which $c_i=1$ are forced by $\widetilde{\ba}$.
        \item For each $V$-position (where $c_i = 0$), nonzero contributions require the site index $a_i$ to be in the support of the union of prior $W$-terms, i.e.\
        \begin{align*}
            a_i \in \bigcup_{j<i: \ c_j=1} \supp(W_{a_j}) \,.
        \end{align*}
        Note that $\abs{\bigcup_{j<i: \ c_j=1} \supp(W_{a_j})} \leq (c_1 + \cdots + c_{j-1})L $ . 
    \end{itemize}
    This implies that
    \begin{align*}
        \abs{\calA(\widetilde{\ba})} \leq A_{k, t} \coloneqq \sum_{\substack{\bc\in \{0,1\}^t:\\ c_1 + \cdots + c_t=k}} \prod_{j\in [t]} \Paren{ \Iver{c_j=1} + \Iver{c_j=0} \cdot \Paren{c_1 + \cdots + c_{j-1} }L}. 
    \end{align*}
    Overall, we have
    \begin{align*}
        e^{-\beta H} e^{\beta (H - W^{(1)})}
        &= \sum_{t=0}^{\infty} \frac{\beta^t}{t!} I \cdot \Iver{t=0}
         + \sum_{t=0}^{\infty} \frac{\beta^t}{t!}
            \sum_{k=1}^t
            \sum_{\substack{\widetilde{\ba} \in [m]^k:\\ \text{cluster from }1}}
            G_{t,\widetilde{\ba}} \\
        &= I
         + \sum_{k=1}^{\infty}
            \sum_{\substack{\widetilde{\ba} \in [m]^k:\\ \text{cluster from }1}}
            \underbrace{\sum_{t\ge k} \frac{\beta^t}{t!} G_{t,\widetilde{\ba}}}_{\coloneqq R_{\widetilde{\ba}}}\,,
    \end{align*}
    where $R_{\widetilde{\ba}}$ is supported on $S_{\widetilde{\ba}}$ and
    \begin{align*}
        \norm{R_{\widetilde{\ba}}} \leq \sum_{t\geq k} \frac{\beta^t}{t!} \norm{G_{t, \widetilde{\ba}}} \leq \sum_{t\geq k} \frac{\beta^t}{t!} 4^t \cdot h^{t-k} \cdot A_{k,t} = \frac{1}{k!} \Paren{\frac{e^{4\beta h L}-1}{hL}}^k \,,
    \end{align*}
    where the last equality follows from \Cref{lem:combo_identity} with $\alpha = 4\beta h$, which is purely combinatorial. 
    We conclude the proof by replacing the notation $\widetilde{\ba}$ with $\ba$. 
\end{proof}

\begin{lemma}\label{lem:combo_identity}
    For any $\alpha, L>0$ and any integers $k,t$ such that $0\leq k\leq t$, define
    \begin{align*}
        A_{k,t} = \sum_{\substack{\bc\in \{0,1\}^t:\\ c_1 + \cdots + c_t=k}} \prod_{j\in [t]} \Paren{ \Iver{c_j=1} + \Iver{c_j=0} \cdot \Paren{c_1 + \cdots + c_{j-1} }L} \,.
    \end{align*}
    Then for any integer $k\geq 0$, 
    \begin{align*}
        \sum_{t\geq k}\frac{\alpha^t}{t!} \cdot A_{k,t} = \frac{1}{k!}\cdot \Paren{\frac{e^{\alpha L} - 1}{L}}^k \,.
    \end{align*}
\end{lemma}
\begin{proof}
    We first derive a recurrence relation for $A_{k,t}$. For the bitstrings $\bc\in \{0,1\}^t$ with $c_1 + \cdots + c_t = k$, if $c_t=1$, then their contribution to $A_{k,t}$ is $A_{k-1, t-1}$. If $c_t = 0$, then their contribution is $Lk\cdot A_{k, t-1}$. 
    So for any $1\leq k\leq t$, we have
    \begin{equation*}
        A_{k,t} = Lk\cdot A_{k,t-1} + A_{k-1,t-1} \,.
    \end{equation*}
    To make this recurrence well-defined, we set $A_{a,b} = 0$ for any $a>b$ (so $A_{k,t-1}$ is always well-defined). 

    Next we identify the case $k=0$. When $t=0$, the only bitstring is the empty string, so $A_{0,0}=1$. 
    When $t\geq 1$, the only bitstring with $c_1+\cdots+c_t=0$ is the all-zero string, and its product contains the factor $(c_1+\cdots+c_0)L = 0$ at $j=1$. Hence $A_{0,t} = 0$ for all $t\geq 1$. 
    
    For any integer $k\geq 0$, define $G_k(\alpha) \coloneqq \sum_{t\geq k}\frac{\alpha^t}{t!} \cdot A_{k,t}$ and $F_k(\alpha) \coloneqq \frac{1}{k!} (\frac{e^{\alpha L} - 1}{L})^k$. 
    The goal is to show that $G_k(\alpha) = F_k(\alpha)$. 

    When $k=0$, we have $G_0(\alpha) = F_0(\alpha)=1$ because $A_{0,0} = 1$ and $A_{0,t} = 0$ for $t\geq 1$. 

    Now fix $k\geq 1$. Taking the derivative of $G_k(\alpha)$ with respect to $\alpha$ gives
    \begin{align*}
        G_k'(\alpha)
        &= \sum_{t\geq k}\frac{t\alpha^{t-1}}{t!}\,A_{k,t} \\
        &= \sum_{s\geq k-1}\frac{\alpha^s}{s!}\,A_{k,s+1}
         \tag{\text{Set } $s=t-1$} \\
        &= \sum_{s\geq k-1}\frac{\alpha^s}{s!}\,\Paren{Lk\,A_{k,s}+A_{k-1,s}} \\
        &= Lk \sum_{s\geq k-1}\frac{\alpha^s}{s!}\,A_{k,s}
           + \sum_{s\geq k-1}\frac{\alpha^s}{s!}\,A_{k-1,s}.
    \end{align*}
    Since $A_{k,k-1}=0$, the first sum is just $G_k(\alpha)$, and the second is $G_{k-1}(\alpha)$. Thus
    \begin{equation*}
        G_k'(\alpha)=Lk\,G_k(\alpha)+G_{k-1}(\alpha),
        \qquad
        G_k(0)=0 \qquad (k\geq 1).
    \end{equation*}
    On the other hand, for $k\geq 1$ we have $F_k(0)=0$, and
    \begin{align*}
        F_k'(\alpha)
        &= \frac{1}{k!}\cdot
           k\Paren{\frac{e^{\alpha L}-1}{L}}^{k-1}\cdot \frac{Le^{\alpha L}}{L} \\
        &= \frac{e^{\alpha L}}{(k-1)!}\Paren{\frac{e^{\alpha L}-1}{L}}^{k-1} \\
        &= \frac{1}{(k-1)!}\Paren{\frac{e^{\alpha L}-1}{L}}^{k-1}
           + \frac{Lk}{k!}\Paren{\frac{e^{\alpha L}-1}{L}}^k \\
        &= F_{k-1}(\alpha)+Lk\,F_k(\alpha).
    \end{align*}
    So for each $k\geq 1$, under the induction hypothesis $G_{k-1} = F_{k-1}$, 
    both $G_k$ and $F_k$ satisfy the same linear ODE and initial condition at $\alpha=0$.
    Since the induction hypothesis is true for the case of $k=0$, we complete the induction and conclude that $G_k(\alpha)=F_k(\alpha)$. 
\end{proof}

\section{Classical hardness at high temperature}
\label{sec:classical-hardness}

We now state our classical-hardness result.

\begin{restatable}[Classical hardness with external field]{theorem}{hardness}
\label{thm:classical-hardness}
Given any integer $n\geq 1$ and any $\beta \in (0, 1)$,
there exists a family of Hamiltonians $H = W_t + V$ parameterized by some integer $t\geq 1$ on $n$ data qubits (system $\calA$) and $M = \bigO{nt}$ ancilla qubits (system $\calR$) such that
\begin{itemize}
    \item $V$ is an on-site external field with potential $0$ on sites in $\calA$ and potential $h$ on sites in $\calR$.
    \item $W_t$ is a $6$-local degree-$5t$ Hamiltonian whose local terms are orthogonal projectors.
\end{itemize}
Moreover, in either of the following two parameter regimes
\begin{itemize}
    \item $t = \lceil c / \beta \rceil + 1$ and any $h\geq \log(4c/\beta)/\beta$, or
    \item $t = \lceil 2c / \beta \rceil$ and any $h\geq 1/2$,
\end{itemize}
where $c = 1.87$,
no classical randomized polynomial time algorithm can output samples from a distribution $Q_n$ such that
\begin{equation*}
    \Norm{ Q_n - P_{H} }_1 \leq 2^{-7 n}\, \quad \textrm{ where } \quad P_H(x) = \frac{\bra{x} e^{-\beta H} \ket{x} }{\tr\Paren{e^{-\beta H}}} \textrm{ for } x\in\{0,1\}^{n+M}\,,
\end{equation*}
unless the polynomial hierarchy collapses to the third level.
\end{restatable}

We use the following base hardness statement from Rajakumar and Watson~\cite{rajakumar2026gibbs}:

\begin{lemma}[Classical hardness for constant temperature, Theorem 6~\cite{rajakumar2026gibbs}]
\label{thm:RW-classical-hardness}
There exists a family of commuting projector Hamiltonians $H_C$ over $n$ qubits with locality $5$ and degree $5$, such that for any $\beta\geq  1.87$, there is no randomized classical polynomial time algorithm that outputs samples from any distribution $Q_n$ that is $2^{-7 n}$-close to $$P_{H_C}(x) = \frac{\bra{x} e^{-\beta H_C} \ket{x} }{\tr\Paren{e^{-\beta H_C}}} \, \quad x \in \{0,1 \}^n\,, $$ 
in total variation distance, unless the polynomial hierarchy collapses to the third level. 
\end{lemma}

The key lemma we prove is that a low, but constant temperature Gibbs state of any local Hamiltonian with projective terms can be encoded into a high-temperature Gibbs state of a local Hamiltonian with slightly larger degree and locality.

\begin{lemma}[Field refrigeration reduction]
\label{lem:field-refrigeration}
Let $H_C = \sum_{a =1}^m P_a$ be an $L$-local, degree-$D$ Hamiltonian on $n$ qubits (system $\mathcal{A}$), where $P_a$'s are orthogonal projectors and pairwise commuting.
For any $\beta>0$, $h\geq 0$, and integer $t \geq 1$, there exists a Hamiltonian $H = V + W$ on system $\mathcal{A}$ and $tm$ ancilla qubits (system $\mathcal{R}$) such that
\begin{itemize}
    \item $V = \sum_{i \in [tm]} v_i$ is the external field, where each $v_i$ acts on the $i$-th qubit in $\mathcal{R}$ with $\Norm{v_i}= h$;
    \item $W = \sum_{a \in [m], \ell \in [t]} W_{a,\ell}$ is a $(L+1)$-local, degree-$(tD)$ interaction term which
    couples $\mathcal{A}$ and $\mathcal{R}$ with $\Norm{W_{a,\ell}}= 1$.
\end{itemize}
Moreover, the Gibbs state $\rho_H = e^{-\beta H}/\tr\Paren{e^{-\beta H}}$ satisfies
\begin{equation*}
    \tr_{\mathcal{R}}\Paren{\rho_H} = \frac{e^{-\beta_{\eff}\; H_C }}{ \tr\Paren{e^{-\beta_{\eff}\; H_C }}} \,, \quad \textrm{where} \quad \beta_{\eff} = t \log\Paren{\frac{1+e^{-\beta h}}{e^{-\beta}+e^{-\beta h}}} \,.
\end{equation*}
That is, tracing out the ancilla system $\mathcal{R}$ recovers the Gibbs state of $H_C$ at effective inverse temperature $\beta_{\eff}$.
\end{lemma}
\begin{proof}
Denote by $\{r_{a,\ell}: a\in [m], \ell \in [t]\}$ the $mt$ ancilla qubits in $\mathcal{R}$.  
For each $a \in [m]$ and $\ell \in [t]$, define the following operator on $\mathcal{R}$ which projects onto the state $\ket{1}$ on qubit $r_{a,\ell}$: 
\begin{equation*}
    n_{a,\ell} = \frac{I - \sigma_Z^{(r_{a,\ell})}}{2} = \ketbra{1}{1}_{r_{a, \ell}} \,,
\end{equation*}
where $\ketbra{1}{1}_{r_{a, \ell}}$ is the shorthand notation for $\ketbra{1}{1}_{r_{a, \ell}} \otimes I_{\mathcal{R} \setminus \{r_{a,\ell}\}}$. 
Then the overall Hamiltonian given by
\begin{equation*}
    H = \sum_{a \in [m]} \sum_{\ell \in [t]} H_{a,\ell} \,, \quad \textrm{ where } \quad H_{a,\ell} = \underbrace{h \cdot I_{\mathcal{A}} \otimes n_{a,\ell}}_{\coloneqq v_{a,\ell}} + \underbrace{P_a \otimes \Paren{I_{\mathcal{R}} - n_{a,\ell}}}_{\coloneqq W_{a,\ell}}\,, 
\end{equation*}
admits the form of $H = V+W$ as described in the statement with an external field term $V = \sum_{a, \ell} v_{a,\ell}$ and a local interaction term $W = \sum_{a, \ell} W_{a,\ell}$. 
Observe that the interaction term $W$ has locality at most $L+1$ and degree at most $tD$. 
Since
\begin{equation*}
    H_{a,\ell}\Paren{ I_{\calA} \otimes \ket{0}_{r_{a,\ell}} } =  P_a \otimes \ket{0}_{r_{a,\ell}}   \, \quad \textrm{ and }\quad H_{a,\ell}\Paren{ I_{\calA} \otimes \ket{1}_{r_{a,\ell}} } = h I_{\calA} \otimes \ket{1}_{r_{a,\ell}} \,,
\end{equation*}
then
\begin{align*}
    \tr_{r_{a,\ell}}\Paren{ e^{-\beta H_{a,\ell}} } &= (I_{\calA} \otimes \bra{0}_{r_{a, \ell}}) \cdot e^{-\beta H_{a,\ell}} \cdot (I_{\calA} \otimes \ket{0}_{r_{a,\ell}}) + (I_{\calA} \otimes \bra{1}_{r_{a, \ell}}) \cdot e^{-\beta H_{a,\ell}} \cdot (I_{\calA} \otimes \ket{1}_{r_{a,\ell}})\\
    &= \Paren{e^{-\beta P_a} + e^{-\beta h} I_{\calA}} \otimes I_{\calR\setminus \{r_{a,\ell}\}}
\end{align*}
Since the $P_a$'s mutually commute and each $H_{a,\ell}$ acts on a distinct ancilla qubit $r_{a,\ell}$, we have $[H_{a,\ell}, H_{a',\ell'}] = 0$ for all $a,a', \ell,\ell'$. 
Therefore the matrix exponential factorizes:
\begin{equation*}
    e^{-\beta H} = \prod_{a \in [m], \ell \in [t]} e^{-\beta H_{a,\ell}} \,.
\end{equation*}
Now, each factor $e^{-\beta H_{a,\ell}}$ acts nontrivially only on $\mathcal{A}$ and the single ancilla qubit $r_{a,\ell}$, and acts as identity on all other ancillae.
Therefore, let us take the partial trace sequentially over all $a\in [m]$ and $\ell\in [t]$. Then by the module property of the partial trace in \Cref{prop:partial_trace_identities}, we have
\begin{align*}
    \tr_{\calR}\Paren{e^{-\beta H}} = \prod_{a \in [m], \ell \in [t]} \Paren{e^{-\beta P_a} + e^{-\beta h} I_{\calA}} \,.
\end{align*}
Since $P_a$ is an orthogonal projector, we have $e^{-\beta P_a} = I-P_a + e^{-\beta} \cdot P_a$ and thus
\begin{align*}
    e^{-\beta P_a} + e^{-\beta h} I &= (1 + e^{-\beta h}) (I - P_a) + (e^{-\beta} + e^{-\beta h}) P_a \\
    &= (1 + e^{-\beta h}) \Paren{I - P_a + \frac{e^{-\beta} + e^{-\beta h}}{1 + e^{-\beta h}} P_a}\,.
\end{align*}
Therefore, 
\begin{align*}
    \tr_{\calR}\Paren{e^{-\beta H}} &= (1 + e^{-\beta h})^{mt} \prod_{a \in [m]}  \Paren{I - P_a + \Paren{\frac{e^{-\beta} + e^{-\beta h}}{1 + e^{-\beta h}}}^t P_a} \\
    &=  (1 + e^{-\beta h})^{mt} \prod_{a \in [m]}  e^{-\beta_{\eff} P_a} \tag{Set $\beta_{\eff} =t \log\Paren{\frac{1+e^{-\beta h}}{e^{-\beta}+e^{-\beta h}}}$} \\
    &= (1 + e^{-\beta h})^{mt}\cdot e^{-\beta_{\eff}H_C}\,.
\end{align*}
Taking the full trace yields $\tr\Paren{e^{-\beta H}} = (1 + e^{-\beta h})^{mt}\cdot \tr\Paren{e^{-\beta_{\eff}H_C}}$
which implies that
\begin{align*}
    \tr_{\calR}\Paren{ \frac{e^{-\beta H } }{\tr\Paren{e^{-\beta H }}} } = \frac{e^{-\beta_{\eff}\; H_C }}{ \tr\Paren{e^{-\beta_{\eff}\; H_C }}}\, . & \qedhere 
\end{align*}
\end{proof}

We can now complete the proof of \cref{thm:classical-hardness} by applying \cref{lem:field-refrigeration} to the hard instance from \cref{thm:RW-classical-hardness}.

\begin{proof}[Proof of~\cref{thm:classical-hardness}]
Let $H_C$ be as defined in \Cref{thm:RW-classical-hardness}. Then with \cref{lem:field-refrigeration}, 
we have a Hamiltonian $H = V + W$ where $V$ is an external field, $W$ is an interaction term with locality $6$ and degree $5t$, and tracing out the ancillary system $\calR$ yields the state $e^{-\beta_{\eff} H_C}/ \tr\Paren{e^{-\beta_{\eff} H_C}}$. We inherit the computational hardness whenever $\beta_{\eff} \geq 1.87$.
This is because if there were a polynomial-time classical algorithm that sampled a distribution $Q$ with $\norm{Q-P_H}_1\leq 2^{-7n}$, then by discarding the ancilla bits we would obtain a polynomial-time sampler for the marginal $Q_{\mathcal A}$ satisfying $\norm{Q_{\mathcal A}-(P_H)_{\mathcal A}}_1\leq 2^{-7n}$ where
\begin{align*}
    (P_H)_{\mathcal A}(x)=\sum_{y\in\{0,1\}^{tm}} P_H(x,y)=\bra{x}\tr_{\mathcal R}(\rho_H)\ket{x}=\frac{\bra{x}e^{-\beta_{\eff}H_C}\ket{x}}{\tr(e^{-\beta_{\eff}H_C})} \,.
\end{align*}
This would contradict \Cref{thm:RW-classical-hardness}.

First notice that $\beta_{\eff}(h) = t \log\Paren{\frac{1+e^{-\beta h}}{e^{-\beta}+e^{-\beta h}}}$ is increasing in $h$ because
\begin{equation*}
    \frac{\diff}{\diff h} \beta_{\eff}(h) = t \beta e^{-\beta h} \Paren{ \frac{1}{e^{-\beta} + e^{-\beta h}} - \frac{1}{1 + e^{-\beta h}}} > 0 \,.
\end{equation*}
Write $c = 1.87$. 

\paragraph{Case 1. }
When $h=1/2$, we have
\begin{equation*}
    \beta_{\eff}\Paren{\frac12} = t \log\Paren{\frac{1+e^{-\beta /2}}{e^{-\beta}+e^{-\beta /2}}} = t \log\Paren{e^{\beta /2}} = t\beta / 2 \,.
\end{equation*}
Setting $t = \lceil \frac{2 c}{\beta} \rceil$ implies that $\beta_{\eff}\Paren{\frac12} \geq c$. Since $\beta_{\eff}(h)$ is increasing in $h$, we can conclude that with $t = \lceil \frac{2c}{\beta} \rceil$, if $h\geq 1/2$, then $\beta_{\eff} \geq c$. 

\paragraph{Case 2. }
When $h=\frac1\beta \log\!\left(\frac{4c}{\beta}\right)$, let us set $t=\lceil \frac{c}{\beta}\rceil+1$, $x=\frac{\beta_{\eff}}{t}$, and $F(u)=\frac{e^u-1}{1-e^{u-\beta}}$. 
One can check directly that
\begin{equation*}
    F(x)=e^{\beta h}=\frac{4c}{\beta}.
\end{equation*}
Suppose for contradiction that \(\beta_{\eff}<c\). Since \(t\ge c/\beta+1\), we have
\begin{equation*}
    x=\frac{\beta_{\eff}}{t}<\frac{c}{c/\beta+1}=\frac{c\beta}{c+\beta}=:s \,.
\end{equation*}
Since $F$ is increasing on $(0, \beta)$ and $0<x<s<\beta < 1$, then
\begin{equation*}
    F(x)<F(s) = \frac{e^s-1}{1-e^{s-\beta}} \leq \frac{2s}{(\beta - s)/2} = \frac{4c}{\beta} \,,
\end{equation*}
where the inequality is because $e^s - 1 < 2s$ for any $s\in (0,1)$ and $1-e^{-\nu} \geq \nu/2$ for any $\nu \in (0, 1)$. 
This contradicts $F(x) = 4c/\beta$. Therefore $\beta_{\eff}\geq c=1.87$.

In either case, the locality of the Hamiltonian is $6$ and the degree is $5t$. 
\end{proof}

%% file: all_fermion.tex
\section{Preliminaries: local fermionic Hamiltonians}

\subsection{Fock space and fermionic algebra}

Let $\Lambda$ be a set of $n$ sites, each carrying $m$ complex fermionic modes.
Fix an ordering of the modes in $\Lambda\times[m]$.
Let $\calF_{\Lambda} \coloneqq \{0,1\}^{\Lambda \times[m]}$ denote the set of bitstrings written in this order. 
The underlying many-body Hilbert space is the fermionic \emph{Fock space} $\mathcal H_\Lambda$, with orthonormal \emph{occupation-number basis} $\{\ket{\mathbf s}:\mathbf s\in\mathcal F_\Lambda\}$, where $s_{i,a}=1$ means that mode $a$ at site $i$ is occupied, and $s_{i,a}=0$ means that it is empty. 

The corresponding annihilation and creation operators $c_{i,a}$ and $c_{i,a}^\dagger$ act on $\mathcal H_\Lambda$ and satisfy the canonical anticommutation relations (CAR):
\begin{equation*}
\{c_{i,a},c_{j,b}\}=0,
\qquad
\{c_{i,a}^\dagger,c_{j,b}^\dagger\}=0,
\qquad
\{c_{i,a},c_{j,b}^\dagger\}=\delta_{ij}\delta_{ab}\id .
\end{equation*}

\begin{definition}[Local fermionic algebra]
For $S\subseteq\Lambda$, the local fermionic algebra is the unital $C^*$-algebra
\[
\mathfrak A_S
:=
C^*\bigl(c_{i,a},c_{i,a}^\dagger:\ i\in S,\ a\in[m]\bigr),
\]
i.e., consisting of finite linear combinations of products of these generators and the identity. In particular, $\mathfrak A_\Lambda$ is the full operator algebra on $\mathcal H_\Lambda$, $\mathfrak A_S\subseteq\mathfrak A_\Lambda$, and $\mathfrak A_\varnothing=\mathbb C\id$.
An operator $A\in\mathfrak A_\Lambda$ is \emph{supported in} $S$ if
$A\in\mathfrak A_S$, and its support is
\[
\supp(A):=\bigcap\{S\subseteq\Lambda:\ A\in\mathfrak A_S\}.
\]
\end{definition}

We note that locality only specifies which fermionic generators are used. Unlike the qubit systems, operators with disjoint supports need not commute (see \Cref{prop:supercommute}).
We refer the reader to Zhao's thesis for a detailed treatment~\cite{zhao2024learning}.

\begin{definition}[Local occupation matrix units]\label{def:local_basis_matrix_unit}
For $S\subseteq\Lambda$, let $\mathcal H_S$ be the Fock space of the modes in $S\times[m]$, with occupation-number basis
\[
\{\ket{\mu}_S:\mu\in\mathcal F_S\},
\qquad \mathcal F_S:=\{0,1\}^{S\times[m]},
\]
using the fixed ordering restricted to these modes. Write $c_{i,a}^S$ and $(c_{i,a}^S)^\dagger$ for the annihilation and creation operators on $\mathcal H_S$. They generate the full operator algebra on $\mathcal H_S$ denoted by $\calB(\calH_S)$.
Let $\iota_S: \calB(\calH_S) \to \mathfrak{A}_S$ be the unital $*$-isomorphism specified on generators by $c_{i,a}^S\mapsto c_{i,a}$. 
For $\mu,\nu\in\mathcal F_S$, define
\[
P^S_{\mu\nu}
:=\iota_S\!\left(\ket{\mu}_S\bra{\nu}_S\right).
\]
The family \(\{P^S_{\mu\nu}\}_{\mu,\nu\in\mathcal F_S}\) forms a matrix-unit basis of $\mathfrak A_S$, satisfying
\[
P^S_{\mu\nu}P^S_{\mu'\nu'}
=\delta_{\nu\mu'}P^S_{\mu\nu'},
\qquad
(P^S_{\mu\nu})^\dagger=P^S_{\nu\mu},
\qquad
\sum_{\mu\in\mathcal F_S}P^S_{\mu\mu}=\id.
\]
\end{definition}

We note that $\ket{\mu}_S\bra{\nu}_S$ has rank one on $\mathcal H_S$, whereas $P^S_{\mu\nu}$ acts on $\mathcal H_\Lambda$ and has rank $2^{m(n-|S|)}$.

\paragraph{Trace convention.}
For \(X\in\mathfrak A_\Lambda\), write \(\tr(X)\) for the ordinary unnormalized matrix trace on the full Fock space \(\mathcal H_\Lambda\). If an element of a local algebra appears inside $\tr$, it is traced as an operator on $\mathcal H_\Lambda$. For example,
\[
\tr(P^S_{\mu\nu})=2^{m(n-|S|)}\delta_{\mu\nu}.
\]

\subsection{Parity and Supercommutation}

The parity automorphism $\Theta:\mathfrak A_\Lambda\to \mathfrak A_\Lambda$ is the unique $*$-automorphism determined by
\[
\Theta(c_{i,a})=-c_{i,a},
\qquad
\Theta(c_{i,a}^\dagger)=-c_{i,a}^\dagger,
\qquad
\text{for every } i\in \Lambda,\ 1\le a\le m.
\]
Thus $\Theta$ is linear and preserves products and adjoints:
\[
\Theta(XY)=\Theta(X)\Theta(Y),
\qquad
\Theta(X^\dagger)=\Theta(X)^\dagger,
\qquad
\text{for all } X,Y\in \mathfrak A_\Lambda.
\]
Equivalently, $\Theta(A)=\Pi A\Pi$ for all $A\in\mathfrak A_\Lambda$, where $\Pi$ is the parity operator given by
\[
\Pi:=(-1)^{\widehat N},
\qquad
\widehat N:=\sum_{i\in\Lambda}\sum_{a=1}^m c_{i,a}^\dagger c_{i,a}.
\]

\begin{definition}[Even and odd subspaces]
For $S\subseteq \Lambda$, define
\[
\mathfrak A_S^{\mathrm{ev}} := \{A\in \mathfrak A_S : \Theta(A)=A\},
\qquad
\mathfrak A_S^{\mathrm{odd}} := \{A\in \mathfrak A_S : \Theta(A)=-A\}.
\]
An operator is called \emph{homogeneous} if it belongs to either $\mathfrak A_\Lambda^{\mathrm{ev}}$ or $\mathfrak A_\Lambda^{\mathrm{odd}}$.
If a nonzero $A$ is homogeneous, its parity $p(A)\in \{0,1\}$ is defined by
\[
\Theta(A) = (-1)^{p(A)} A.
\]
\end{definition}
For example, the matrix-unit operator $P_{\mu\nu}^S$ is homogeneous with parity $p(P_{\mu\nu}^S) = (\abs{\mu} + \abs{\nu}) \pmod 2$. Here $\abs{\mu}$ is the Hamming weight of the bitstring $\mu$. 

The CAR imply the following standard supercommutation relation for homogeneous operators with disjoint supports.
\begin{proposition}[Supercommutation of disjoint local operators]
\label{prop:supercommute}
Let $A$ and $B$ be homogeneous operators in $\mathfrak A_\Lambda$ with disjoint supports.
Then
\[
AB = (-1)^{p(A)p(B)} BA.
\]
In particular, if one of $A$ or $B$ is even, then $[A,B]=0$.
\end{proposition}

\begin{definition}[Fermionic density operators]
A \emph{fermionic density operator} is a PSD operator
$\rho\in\mathfrak A_\Lambda^{\mathrm{ev}}$ with $\tr(\rho)=1$. 
\end{definition}

\begin{remark}[Fermion-parity superselection]
The evenness condition expresses fermion-parity superselection: physical density operators and observables must commute with the parity operator $\Pi$.
Let $\mathcal H_\Lambda^{\mathrm{ev}}$ and $\mathcal H_\Lambda^{\mathrm{odd}}$ be the $+1$- and $-1$-eigenspaces of $\Pi$. 
Then fermionic density operators must be block diagonal with respect to this decomposition of $\mathcal H_\Lambda = \mathcal H_\Lambda^{\mathrm{ev}} \oplus \mathcal H_\Lambda^{\mathrm{odd}}$.
Thus $\ketbra{\psi}{\psi}$ is a fermionic density operator if and only if
$\ket{\psi}$ has definite parity.
Mixtures of even and odd states are allowed, but coherent off-diagonal terms between the two parity sectors are not.
For example, each occupation-number projector $\ketbra{\mathbf s}{\mathbf s}$ is even,
whereas in a one-mode system the projector onto $\frac{1}{\sqrt 2}(\ket{0}+\ket{1})$ is not.
\end{remark}

\subsection{Local fermionic Hamiltonians}

\begin{definition}[Local fermionic Hamiltonians with onsite terms]
\label{def:local-fermionic-ham}
A \emph{local fermionic Hamiltonian with onsite terms} on $\Lambda$ is a Hermitian operator
$H$ together with a decomposition
\[
H = V + W,
\qquad
V = \sum_{i\in\Lambda} v_i,
\qquad
W = \sum_{Z\subseteq\Lambda, \, \abs{Z}\geq 2} w_Z,
\]
where
\begin{enumerate}
\item each $v_i\in\mathfrak A_{\{i\}}^{\mathrm{ev}}$ and $v_i=v_i^\dagger$;
\item each $w_Z\in\mathfrak A_Z^{\mathrm{ev}}$ and $w_Z=w_Z^\dagger$;
\item $\supp(w_Z)=Z$ whenever $w_Z\neq 0$.
\end{enumerate}
We call $V$ the \emph{onsite part}, $W$ the \emph{multi-site part}, and the
$w_Z$ \emph{multi-site coupling terms}.
The local multi-site coupling strength $\zeta$ and bounded $L$-site degree-$D$ couplings are as in \Cref{def:interaction-parameters}.
The support family of $H$ is $\mathcal E_H:=\{Z\subseteq\Lambda: \abs{Z}\geq 2, \, w_Z\neq 0\}$. 
\end{definition}

Throughout the fermionic-system part, $\dist$ and $\ball$ refer to the hypergraph $(\Lambda,\mathcal E_H)$,
as in \Cref{def:support-hypergraph}.

\begin{example}[Spin-$\frac12$ Fermi--Hubbard Hamiltonian]
\label{def:fermi-hubbard}
Take $m=2$ complex fermionic modes per site, labeled by $\uparrow$ and $\downarrow$, and write
\[
n_{i,\sigma}:=c_{i,\sigma}^\dagger c_{i,\sigma} \,, \quad \text{for each }\sigma \in \{\uparrow, \downarrow\}.
\]
Let $E$ be a nearest-neighbor edge set on $\Lambda$, with real hopping
amplitudes $t_{ij}$, onsite interactions $U_i$, and chemical potentials
$\mu_i$. The spin-$\frac12$ Fermi--Hubbard Hamiltonian is
$H_{\mathsf{FH}}=V+W$ as in \Cref{def:local-fermionic-ham}, with
\begin{align*}
v_i
&:=
U_i\,n_{i,\uparrow}n_{i,\downarrow}
+\mu_i(n_{i,\uparrow}+n_{i,\downarrow}),
\\
w_{\{i,j\}}
&:=
-\sum_{\sigma\in\{\uparrow,\downarrow\}}
t_{ij}\bigl(c_{i,\sigma}^\dagger c_{j,\sigma}+c_{j,\sigma}^\dagger c_{i,\sigma}\bigr),
\qquad \{i,j\}\in E,
\end{align*}
and $w_Z:=0$ for all other $Z\subseteq\Lambda$.
Each $v_i$ is Hermitian, parity-even, and supported on $\{i\}$, and each nonzero $w_{\{i,j\}}$
is Hermitian, parity-even, and supported on $\{i,j\}$.
The decomposition is by spatial support: the onsite Hubbard interaction
belongs to $V$, while hopping belongs to $W$.
\end{example}

\section{Fermionic Wasserstein norm and Dobrushin condition}

\subsection{Fermionic partial trace}

\begin{definition}[Fermionic partial trace over a region]
Let \(S\subseteq\Lambda\), and write \(S^c:=\Lambda\setminus S\). 
Use the notation \(\mathcal F_S\) and \(P^S_{\mu\nu}\) from \Cref{def:local_basis_matrix_unit}, and set \(d_S:=|\mathcal F_S|=2^{m|S|}\).
By the CAR, every product of creation and annihilation
generators can be reordered, with the corresponding signs, so that the factors
supported in \(S\) are written before the factors supported in \(S^c\). Hence
every \(X\in\mathfrak A_\Lambda\) has a unique expansion
\[
X=\sum_{\mu,\nu\in\mathcal F_S} P^S_{\mu\nu}X_{\mu\nu}^{(S)},
\qquad
X_{\mu\nu}^{(S)}\in\mathfrak A_{S^c}.
\]
The unnormalized fermionic partial trace over \(S\) is the linear map
\(\operatorname{tr}_S:\mathfrak A_\Lambda\to\mathfrak A_{S^c}\) defined by
\[
\operatorname{tr}_S(X):=\sum_{\mu\in\mathcal F_S}X_{\mu\mu}^{(S)}.
\]
Equivalently, it is characterized by
\begin{equation}\label{eq:alt_def_partial_trace}
\operatorname{tr}_S\!\big(P^S_{\mu\nu}Y\big)=\delta_{\mu\nu}Y
\qquad
\text{for all }\mu,\nu\in\mathcal F_S
\text{ and }Y\in\mathfrak A_{S^c}.
\end{equation}
Its normalized version is
\(\mathbb E_S:=d_S^{-1}\operatorname{tr}_S\).
\end{definition}

\begin{lemma}[Composition of partial traces]
\label{lem:fermionic-partial-trace-composition}
For disjoint $R,S\subseteq\Lambda$, it holds that $\operatorname{tr}_R\circ\operatorname{tr}_S=\operatorname{tr}_{R\cup S}$ and $\mathbb E_R\circ\mathbb E_S=\mathbb E_{R\cup S}$. 
\end{lemma}

\begin{remark}
For those who are more familiar with the qubit systems, fermionic partial trace can be confusing because their image is a subalgebra. Let us clarify some special cases: 
\begin{enumerate}
    \item For every $Y\in\mathfrak A_{S^c}$, $\operatorname{tr}_S(Y)=d_SY$ and thus $\E_S(Y) = Y$. In particular, $\operatorname{tr}_S(\id)=d_S\id$.
    \item The map \(\operatorname{tr}_\Lambda\) is still
    algebra-valued: \(
    \operatorname{tr}_\Lambda(X)\in\mathfrak A_\varnothing=\mathbb C\id\). It is related to the scalar-valued matrix trace $\tr$ by $\operatorname{tr}_\Lambda(X)=\tr(X)\id$.
    \item (Trace compatibility) For every \(X\in\mathfrak A_\Lambda\), 
    \begin{equation}\label{eq:fermionic-partial-trace-trace-compatibility}
    \tr\!\left(\operatorname{tr}_S(X)\right)=d_S\tr(X) \,.
    \end{equation}
    To see why, let $Y = \tr_S(X)\in \mathfrak{A}_{S^c}$. Then $\tr_S(Y) = d_S Y$ and thus, by \Cref{lem:fermionic-partial-trace-composition},
    \[
    \tr(Y) I = \tr_{\Lambda}(Y) = \tr_{S^c}(\tr_S(Y)) = d_S \tr_{S^c}(Y) = d_S  \tr_{\Lambda}(X) = d_S\tr(X) I. 
    \]
\end{enumerate}
\end{remark}

\subsubsection{Basic properties}
\label{sec:fermion_partial_trace}
\begin{lemma}[Local trace facts]
\label{lem:fermionic-local-trace-facts}
Let \(S\subseteq\Lambda\). The fermionic partial trace has the following local
properties.
\begin{enumerate}
\item If \(C\in\mathfrak A_S^{\mathrm{odd}}\), then $\operatorname{tr}_S(C)=0$. 
\item If \(C,D\in\mathfrak A_S\), then $\operatorname{tr}_S(CD)=\operatorname{tr}_S(DC)$. 
\end{enumerate}
\end{lemma}

\begin{proof}
For any $R\in \mathfrak{A}_S$, we can write $R=\sum_{\mu,\nu\in\mathcal F_S} r_{\mu\nu} P^S_{\mu\nu}$ for some scalars $r_{\mu\nu} \in \C$ and $\operatorname{tr}_S(R)=\sum_{\mu\in\mathcal F_S}r_{\mu\mu}\id$.
If \(C\in\mathfrak A_S^{\mathrm{odd}}\), then its expansion contains only odd \(P^S_{\mu\nu}\). Recall that \(P^S_{\mu\nu}\) has parity
\(|\mu|+|\nu|\pmod 2\). 
So the expansion contains no \(P^S_{\mu\mu}\) and thus \(\operatorname{tr}_S(C)=0\). The second
identity follows from cyclicity of the ordinary local trace.
\end{proof}

\begin{proposition}
\label{prop:fermionic-partial-trace-identities}
Let \(S\subseteq\Lambda\). The fermionic partial trace satisfies the following
identities.
\begin{enumerate}
\item \emph{(Module property over the complement.)}
If \(Z\in\mathfrak A_{S^c}\), then for every
\(X\in\mathfrak A_\Lambda\),
\[
\operatorname{tr}_S(ZX)=Z\operatorname{tr}_S(X),
\qquad
\operatorname{tr}_S(XZ)=\operatorname{tr}_S(X)Z .
\]
\item \emph{(Cyclicity for even local operators.)}
If \(Y\in\mathfrak A_S^{\mathrm{ev}}\), then for every
\(X\in\mathfrak A_\Lambda\),
\[
\operatorname{tr}_S(YX)=\operatorname{tr}_S(XY).
\]
\item \emph{(Same-parity cyclicity.)}
If \(A,B\in\mathfrak A_S\) are homogeneous with the same parity and
\(X\in\mathfrak A_\Lambda^{\mathrm{ev}}\), then
\[
\operatorname{tr}_S(AXB)=\operatorname{tr}_S(BAX).
\]
\end{enumerate}
\end{proposition}

\begin{proof}
Write the expansion from the definition of the fermionic partial trace as
\[
X=\sum_{\mu,\nu\in\mathcal F_S}P^S_{\mu\nu}X^{(S)}_{\mu\nu},
\qquad
X^{(S)}_{\mu\nu}\in\mathfrak A_{S^c}.
\]
The local matrix unit \(P^S_{\mu\nu}\) is homogeneous with parity \(p(\mu,\nu):=|\mu|+|\nu|\pmod 2\). 

For the first identity, by linearity, it suffices to take \(Z\) homogeneous.
Using supercommutation between \(\mathfrak A_S\) and \(\mathfrak A_{S^c}\),
\[
ZP^S_{\mu\nu}
=
(-1)^{p(Z)p(\mu,\nu)}P^S_{\mu\nu}Z .
\]
The module property then follows from \cref{eq:alt_def_partial_trace}: 
\[
\operatorname{tr}_S(ZX)
=
\sum_{\mu,\nu\in\mathcal F_S} (-1)^{p(Z)p(\mu,\nu)} \operatorname{tr}_S \paren*{P^S_{\mu\nu} Z X^{(S)}_{\mu\nu}}
=
\sum_{\mu\in\mathcal F_S}ZX^{(S)}_{\mu\mu}
=
Z\operatorname{tr}_S(X).
\]
Similarly one can show $\operatorname{tr}_S(XZ)=\operatorname{tr}_S(X)Z$. 

For the remaining identities, we first combine the module property with trace compatibility in \cref{eq:fermionic-partial-trace-trace-compatibility} to obtain a useful identity:
for every $T\in\mathfrak A_\Lambda$ and $Z\in\mathfrak A_{S^c}$,
\begin{equation}\label{eq:norm_partial_trace_id}
\tr\!\left(Z\operatorname{tr}_S(T)\right)
=\tr\!\left(\operatorname{tr}_S(ZT)\right)
=d_S \tr(ZT).
\end{equation}
Consequently, for arbitrary $A,B\in\mathfrak A_\Lambda$, to prove that $\tr_S(A) = \tr_S(B)$, it suffices to show that $\tr(ZA) = \tr(ZB)$ for every $Z\in\mathfrak A_{S^c}$.
Indeed, setting $D=\tr_S(A)-\tr_S(B)\in\mathfrak A_{S^c}$, 
\cref{eq:norm_partial_trace_id} gives
\[
\tr(ZD)=d_S\bigl(\tr(ZA)-\tr(ZB)\bigr)=0.
\]
Taking $Z=D^\dagger$ gives $\tr(D^\dagger D)=0$, hence $D=0$.

For the second identity, since $Y\in\mathfrak A_S^{\mathrm{ev}}$ commutes with every $Z\in\mathfrak A_{S^c}$, ordinary trace cyclicity gives $\tr(ZYX)=\tr(YZX)=\tr(ZXY)$.
The preceding argument then gives $\operatorname{tr}_S(YX)=\operatorname{tr}_S(XY)$.

For the third identity, $AXB$ and $BAX$ are even because $X$ is even and
$A,B$ have the same parity. By linearity, it suffices to test against
homogeneous $Z\in\mathfrak A_{S^c}$. If $Z$ is even, it commutes with $B$, so
\[
\tr(ZAXB)=\tr(BZAX)=\tr(ZBAX).
\]
If $Z$ is odd, then both $ZAXB$ and $ZBAX$ are odd and thus $\tr(ZAXB) = \tr(ZBAX) = 0$ by \Cref{lem:fermionic-local-trace-facts} with $S=\Lambda$.
Therefore we conclude $\operatorname{tr}_S(AXB)=\operatorname{tr}_S(BAX)$.
\end{proof}

The partial trace also preserves parity and Hermiticity.
\begin{lemma}[Parity and adjoints under fermionic partial trace]
\label{lem:fermionic-partial-trace-parity-adjoint}
For every \(S\subseteq\Lambda\) and \(X\in\mathfrak A_\Lambda\),
\[
\operatorname{tr}_S(\Theta(X))=\Theta(\operatorname{tr}_S(X)),
\qquad
\operatorname{tr}_S(X^\dagger)=\operatorname{tr}_S(X)^\dagger .
\]
\end{lemma}
\begin{proof}
Write
\[
X=\sum_{\mu,\nu\in\mathcal F_S}P^S_{\mu\nu}X^{(S)}_{\mu\nu},
\qquad
X^{(S)}_{\mu\nu}\in\mathfrak A_{S^c}.
\]
Since \(P^S_{\mu\nu}\) has parity
\(p(\mu,\nu):=|\mu|+|\nu|\pmod 2\), then
\[
\Theta(X)
=
\sum_{\mu,\nu}
(-1)^{p(\mu,\nu)}
P^S_{\mu\nu}\Theta\!\left(X^{(S)}_{\mu\nu}\right).
\]
By \cref{eq:alt_def_partial_trace}, only the terms $\mu=\nu$ survive under \(\operatorname{tr}_S\), and
\(p(\mu,\mu)=0\), so
\[
\operatorname{tr}_S(\Theta(X))
=
\sum_{\mu\in\mathcal F_S}
\Theta\!\left(X^{(S)}_{\mu\mu}\right)
=
\Theta(\operatorname{tr}_S(X)).
\]
For the adjoint identity, the module property and \cref{eq:alt_def_partial_trace} give
\begin{align*}
\operatorname{tr}_S(X^\dagger)
=\sum_{\mu,\nu}
\bigl(X^{(S)}_{\mu\nu}\bigr)^\dagger
\operatorname{tr}_S(P^S_{\nu\mu})=\sum_\mu\bigl(X^{(S)}_{\mu\mu}\bigr)^\dagger
=\operatorname{tr}_S(X)^\dagger. & \qedhere
\end{align*}
\end{proof}

\begin{lemma}[Contraction of normalized partial trace]
\label{lem:fermionic-conditional-expectation-contraction}
For every Hermitian \(X\in\mathfrak A_\Lambda\),
\[
\norm{\mathbb E_S(X)}_1\le \norm{X}_1 .
\]
\end{lemma}
\begin{proof}
By \Cref{lem:fermionic-partial-trace-parity-adjoint},
\(\mathbb E_S(X)\in\mathfrak A_{S^c}\) is Hermitian. Hence
\(Y:=\operatorname{sgn}(\mathbb E_S(X))\) lies in \(\mathfrak A_{S^c}\)
and satisfies \(\norm{Y}\le1\).
Thus \cref{eq:norm_partial_trace_id} and trace-norm duality give
\[
\norm{\mathbb E_S(X)}_1
=\tr\!\left(Y\mathbb E_S(X)\right)
=\tr(YX)
\le\norm{X}_1. \qedhere
\]
\end{proof}

\subsection{Transport plans}

We now adapt the transport-plan and update-matrix formalism to the
fermionic algebra. 

\begin{definition}[Fermionic transport plan and cost vector]
\label{def:fermionic-transport-plan-cost-vector}
Let \(X\in\mathfrak A_\Lambda^{\mathrm{ev}}\) be Hermitian and traceless.
A \emph{fermionic transport plan} for \(X\) is a family
\((X_i)_{i\in\Lambda}\) such that for every site \(i\), 
\[
X=\sum_{i\in\Lambda}X_i,
\qquad
X_i=X_i^\dagger,\qquad
X_i\in\mathfrak A_\Lambda^{\mathrm{ev}},
\qquad
\operatorname{tr}_i(X_i)=0 \,.
\]
The \emph{cost vector} of this transport plan is the nonnegative vector $x = (x_i)_{i\in \Lambda}$ where $x_i:=\frac12\norm{X_i}_1$. 
\end{definition}

The following proposition is immediate from the triangle inequality. 
\begin{proposition}[Linearity of fermionic transport plans]
\label{prop:fermionic-linear-transport-plan}
If \(X\) and \(Y\) have fermionic transport plans \((X_i)_i\) and \((Y_i)_i\)
with cost vectors \(x\) and \(y\), then for every \(\alpha,\beta\in\mathbb R\),
\(\alpha X+\beta Y\) has a fermionic transport plan
\((\alpha X_i+\beta Y_i)_i\) with cost vector \(z\) satisfying $z\le_v |\alpha|x+|\beta|y$. 
\end{proposition}

\begin{definition}[Fermionic update matrix]
\label{def:fermionic-update-matrix}
Let \(\Phi:\mathfrak A_\Lambda^{\mathrm{ev}}\to
\mathfrak A_\Lambda^{\mathrm{ev}}\) be a linear map that preserves Hermiticity
and trace. A matrix \(Q\), indexed by \(\Lambda\times\Lambda\), is a
\emph{fermionic update matrix} for \(\Phi\) if for every traceless Hermitian
even operator \(X\) and every fermionic transport plan \((X_i)_i\) of \(X\)
with cost vector \(x\), the operator \(\Phi(X)\) has a fermionic transport plan
with cost vector \(y\) satisfying
\[
y\le_v Qx .
\]
\end{definition}

With \Cref{prop:fermionic-linear-transport-plan}, if we want to show that $Q$ is an update matrix for $\Phi$, it suffices to only bound the cost vector of $\Phi(X)$ restricted to those $X$ with normalized cost vector: 
\begin{proposition}[Checking an update matrix on normalized charges]
\label{prop:fermionic-normalized-update-check}
Let \(\Phi\) be a map as in \Cref{def:fermionic-update-matrix}, and let
\(Q\) be a matrix indexed by \(\Lambda\times\Lambda\). Suppose that for
every site \(j\) and every even Hermitian \(X\) satisfying $\operatorname{tr}_j(X)=0$ and $\norm{X}_1=2$,
the operator \(\Phi(X)\) has a fermionic transport plan with cost vector
entry-wise bounded by \(Qe_j\). Then \(Q\) is a fermionic update matrix for
\(\Phi\).
\end{proposition}

\Cref{prop:fermionic-cost-vector-single-partial-trace-zero,prop:fermionic-cost-vector-region-partial-trace-zero} below are the central propositions being used repeatedly in \Cref{sec:evolution}. They are the fermionic analogues of \cite[Propositions 2 and 5]{dmtl21}. 
\begin{proposition}[Charging a single site]
\label{prop:fermionic-cost-vector-single-partial-trace-zero}
Let \(X\in\mathfrak A_\Lambda^{\mathrm{ev}}\) be Hermitian. If
\(\operatorname{tr}_i(X)=0\), then \(X\) has a fermionic transport plan with cost vector $\frac12\norm{X}_1\,e_i$. 
\end{proposition}

\begin{proof}
Set \(X_i=X\) and \(X_j=0\) for \(j\neq i\). This is a fermionic transport plan
by the hypothesis \(\operatorname{tr}_i(X)=0\), and its cost vector is exactly $\frac12\norm{X}_1\,e_i$.
\end{proof}

\begin{proposition}[Charging a region]
\label{prop:fermionic-cost-vector-region-partial-trace-zero}
Let \(X\in\mathfrak A_\Lambda^{\mathrm{ev}}\) be Hermitian. If
\(\operatorname{tr}_S(X)=0\) for some \(S\subseteq\Lambda\), then \(X\) has a
fermionic transport plan with cost vector entry-wise upper bounded by $\norm{X}_1\,e_S$. 
\end{proposition}
\begin{proof}
Enumerate \(S=\{i_1,\ldots,i_k\}\). Set $S_\ell:=\{i_1,\ldots,i_\ell\}$ and $S_0:=\varnothing$. 
Since \(\operatorname{tr}_S(X)=0\), we have \(\mathbb E_S(X)=0\). Therefore, we can write $X$ as a telescoping sum of normalized partial traces: 
\[
X
=
\sum_{\ell=1}^k
\left(\mathbb E_{S_{\ell-1}}(X)-\mathbb E_{S_\ell}(X)\right).
\]
Define
\[
X_{i_\ell}:=\mathbb E_{S_{\ell-1}}(X)-\mathbb E_{S_\ell}(X),
\qquad
X_j:=0\quad\text{for }j\notin S.
\]
We now show that \((X_i)_{i\in\Lambda}\) is a fermionic transport
plan for \(X\).
By \Cref{lem:fermionic-partial-trace-parity-adjoint}, each \(X_{i_\ell}\) is even and Hermitian. It remains to verify that \(\operatorname{tr}_{i_\ell}(X_{i_\ell})=0\).
By the composition of partial traces in \Cref{lem:fermionic-partial-trace-composition},
$\mathbb E_{\{i_\ell\}}\circ\mathbb E_{S_{\ell-1}}=\mathbb E_{S_\ell}$.
Also, $\mathbb E_{S_\ell}(X)\in\mathfrak A_{S_\ell^c}$ is fixed by $\mathbb E_{\{i_\ell\}}$ because $i_\ell\in S_\ell$. Hence
\[
\mathbb E_{\{i_\ell\}}(X_{i_\ell})
=\mathbb E_{S_\ell}(X)-\mathbb E_{S_\ell}(X)=0.
\]
It remains to bound
\(\frac12\norm{X_{i_\ell}}_1\) for each \(\ell\).
By the triangle inequality and the contraction of $\E_S$ in \Cref{lem:fermionic-conditional-expectation-contraction},
\[
\frac12\norm{X_{i_\ell}}_1
\le
\frac12\left(
\norm{\mathbb E_{S_{\ell-1}}(X)}_1+\norm{\mathbb E_{S_\ell}(X)}_1
\right)
\le
\norm{X}_1 .
\]
This proves the claimed cost-vector bound.
\end{proof}

\subsubsection{Helper lemmas for the cost vectors of some linear maps}

We now prove some lemmas that are helpful for bounding the transport plan evolution of a Euler step Lindbladian evolution in \Cref{sec:evolution}. In particular, \Cref{lem:fermionic-wass-norm-symmetric-map,lem:fermionic-one-norm-dissipative} are helpful for the dissipative evolution and \Cref{lem:fermionic-one-norm-coherent} is helpful for the coherent evolution. 

\begin{lemma}[Cost vector after a same-parity few-site map]
\label{lem:fermionic-wass-norm-symmetric-map}
Let \(S\subseteq\Lambda\) and \(K,L\in\mathfrak A_S\) be homogeneous
operators with the same parity. Define
\[
\Psi(X)
:=
KXL^\dagger+LXK^\dagger
-\frac12\big\{L^\dagger K+K^\dagger L,X\big\}.
\]
If \(X\in\mathfrak A_\Lambda^{\mathrm{ev}}\) is Hermitian, then
\(\Psi(X)\) has a fermionic transport plan whose cost vector is entry-wise upper bounded by
\[
4\norm{K}\cdot \norm{L}\cdot \norm{X}_1 \cdot e_S .
\]
\end{lemma}
\begin{proof}
Set \(M:=L^\dagger K+K^\dagger L\), which is even, Hermitian, and supported
on \(S\). Since \(K\) and \(L\) have the same parity and \(X\) is even and
Hermitian, \(\Psi(X)\) is also even and Hermitian.
By same-parity cyclicity and cyclicity for even local operators in
\Cref{prop:fermionic-partial-trace-identities},
\[
\operatorname{tr}_S(\Psi(X))
=\operatorname{tr}_S(MX)
-\frac12\operatorname{tr}_S(MX+XM)=0.
\]
By
\Cref{prop:fermionic-cost-vector-region-partial-trace-zero}, \(\Psi(X)\) has a
transport plan with cost vector at most \(\norm{\Psi(X)}_1e_S\). The final trace-norm
bound is given by
\begin{align*}
\norm{\Psi(X)}_1
&\le
\norm{KXL^\dagger}_1+\norm{LXK^\dagger}_1
+\frac12\norm{MX}_1+\frac12\norm{XM}_1 \\
&\le
2\norm{K}\cdot \norm{L}\cdot \norm{X}_1+\norm{M}\cdot \norm{X}_1 \\
&\le
4\norm{K}\cdot \norm{L}\cdot \norm{X}_1 . \qedhere
\end{align*}
\end{proof}

\begin{lemma}[Coherent one-step trace-norm estimate]
\label{lem:fermionic-one-norm-coherent}
For any Hermitian operators \(X\) and $G$,
\[
\norm{X-\ii[G,X]}_1
\le
\left(1+2\norm{G}^2\right)\norm{X}_1 .
\]
\end{lemma}

\begin{proof}
The identity
\[
X-\ii[G,X]=(I-\ii G)X(I+\ii G)-GXG
\]
and Hermiticity of \(G\) give
\begin{align*}
\norm{X-\ii[G,X]}_1 \le \left(\norm{I-\ii G}^2+\norm{G}^2\right)\norm{X}_1 =\left(1+2\norm{G}^2\right)\norm{X}_1 . & \qedhere
\end{align*}
\end{proof}

\begin{lemma}[Dissipative one-step trace-norm estimate]
\label{lem:fermionic-one-norm-dissipative}
Let \((K_a)_a\) be a finite family of operators and set
\(D:=\sum_a K_a^\dagger K_a\). For any Hermitian \(X\),
\[
\norm{X+\sum_a K_aXK_a^\dagger-\frac12\{D,X\}}_1
\le
\paren[\Big]{1+\frac12\norm{D}^2}\norm{X}_1 .
\]
The same estimate holds when the sums are replaced by integrals against a nonnegative measure $\mu$, provided that $\int \norm{K_a}^2 \diff\mu(a)<\infty$.
\end{lemma}

\begin{proof}
Set \(A:=I-\frac12D\) and
\(\Gamma(X):=AXA^\dagger+\sum_a K_aXK_a^\dagger\). Then
\[
X+\sum_a K_aXK_a^\dagger-\frac12\{D,X\}
=
\Gamma(X)-\frac14DXD.
\]
Let \(X_+\) and \(X_-\) be the positive and negative parts of \(X\), so
\(X=X_+-X_-\) and \(\abs{X}=X_++X_-\).
Since \(\Gamma\) is positive and \(A^\dagger A+D=I+\frac14D^2\),
ordinary trace cyclicity gives
\begin{equation*}
\norm{\Gamma(X)}_1
\le \tr\!\left(\Gamma(X_++X_-)\right) =\tr\!\left(\paren[\Big]{I+\frac14D^2}\abs{X}\right) \le \paren[\Big]{1+\frac14\norm{D}^2}\norm{X}_1.
\end{equation*}
Moreover,
\[
\frac14\norm{DXD}_1\le \frac14\norm{D}^2\norm{X}_1.
\]
Combining the two bounds proves the claim. The same proof applies to the integral version.
\end{proof}

\subsection{Fermionic Wasserstein norm}

We now describe the fermionic analogue of the quantum Wasserstein norm of order $1$ introduced by~\cite{dmtl21}.
\begin{definition}[Fermionic Wasserstein norm]
\label{def:fermionic-wasserstein-norm}
For a traceless Hermitian $X\in \mathfrak A_\Lambda^{\mathrm{ev}}$, its \emph{fermionic Wasserstein norm} of order $1$ is given by
\begin{equation*}
\norm{X}_{W_1} = \inf \left\{ \norm{x}_1 : x \text{ is the cost vector of a fermionic transport plan } (X_i)_{i\in \Lambda} \text{ of } X \right\} .
\end{equation*}
\end{definition}

We remark that $\norm{\cdot}_{W_1}$ is a well-defined norm on the real vector space of traceless even Hermitian operators: \Cref{prop:fermionic-linear-transport-plan} gives the triangle inequality and absolute homogeneity, and the lower bound in \Cref{prop:fermionic-trace-wasserstein-relation} below gives positive definiteness.

\begin{proposition}[Trace norm versus Wasserstein norm]
\label{prop:fermionic-trace-wasserstein-relation}
For any traceless Hermitian $X\in \mathfrak A_\Lambda^{\mathrm{ev}}$, we have\footnote{For qubit systems, \cite{dmtl21} prove the stronger upper bound $\frac{n}{2}\norm{X}_1$. The weaker constant $n$ suffices for us.
}
\[
\frac12 \norm{X}_1 \le \norm{X}_{W_1} \le n \norm{X}_1 .
\]
If $\operatorname{tr}_j(X)=0$ for some $j\in\Lambda$, then
$\norm{X}_{W_1}=\frac12\norm{X}_1$.
\end{proposition}
\begin{proof}
Since $\operatorname{tr}_\Lambda(X)=\tr(X)\id=0$,
\Cref{prop:fermionic-cost-vector-region-partial-trace-zero} with $S=\Lambda$
provides a transport plan with cost vector bounded by $\norm{X}_1e_\Lambda$, giving
\[
\norm{X}_{W_1} \le \norm{\norm{X}_1\, e_\Lambda}_1 = n\norm{X}_1 .
\]
For the lower bound, let $(X_i)_{i\in \Lambda}$ be any fermionic transport plan of $X$ with cost vector $x$. By the triangle inequality,
\[
\norm{X}_1 = \norm*{\sum_{i\in \Lambda} X_i}_1 \le \sum_{i\in \Lambda} \norm{X_i}_1 = 2\sum_{i\in \Lambda} x_i = 2\norm{x}_1 .
\]
Taking the infimum over all fermionic transport plans of $X$ gives $\norm{X}_1 \le 2\norm{X}_{W_1}$.

If $\operatorname{tr}_j(X)=0$, then
\Cref{prop:fermionic-cost-vector-single-partial-trace-zero} gives
$\norm{X}_{W_1}\le\frac12\norm{X}_1$, so $\norm{X}_{W_1}=\frac12\norm{X}_1$.
\end{proof}

\subsection{Dobrushin influence and contraction}

\begin{definition}[Fermionic Dobrushin influence matrix]
\label{def:fermionic-dobrushin-influence}
Let $\Phi,\Phi_1, \ldots, \Phi_n: \mathfrak A_\Lambda^{\mathrm{ev}} \to \mathfrak A_\Lambda^{\mathrm{ev}}$ be trace-preserving, Hermiticity-preserving linear maps such that $\Phi = \frac1n(\Phi_1 + \cdots + \Phi_n)$. The associated \emph{Dobrushin influence matrix} $D^{(\Phi)}$, indexed by $\Lambda \times \Lambda$, is defined to have entries
\begin{align*}
D^{(\Phi)}_{i,j} = \frac1n \max_{\substack{X \in \mathfrak A_\Lambda^{\mathrm{ev}},\ X = X^\dagger \\ \operatorname{tr}_j(X) = 0,\ \norm{X}_{W_1} = 1}} \norm{\Phi_i(X)}_{W_1} .
\end{align*}
\end{definition}

Note that $\operatorname{tr}_j(X) = 0$ implies $\tr(X)=0$ by the trace compatibility in \cref{eq:fermionic-partial-trace-trace-compatibility}, and each $\Phi_i(X)$ is again even, Hermitian, and traceless, so all the Wasserstein norms above are well-defined.

\begin{lemma}[Contraction from Dobrushin influence; fermionic analogue of {\cite[Lemma 5.4]{blmt2025dobrushin}}]
\label{lem:fermionic-contraction-dobrushin}
Let $\Phi$ be a map as in \Cref{def:fermionic-dobrushin-influence} with Dobrushin influence matrix $D^{(\Phi)}$. Then for any traceless Hermitian $X\in \mathfrak A_\Lambda^{\mathrm{ev}}$,
\begin{equation*}
\norm{\Phi(X)}_{W_1} \le \norm{D^{(\Phi)}}_{1\to 1} \cdot \norm{X}_{W_1} .
\end{equation*}
\end{lemma}
\begin{proof}
Let $(X_j)_{j\in \Lambda}$ be any fermionic transport plan of $X$ with cost vector $x$.
By \Cref{prop:fermionic-trace-wasserstein-relation}, $\norm{X_j}_{W_1}=x_j$ for every $j$.
By \Cref{def:fermionic-dobrushin-influence},
$\norm{\Phi_i(X_j)}_{W_1}\le nD^{(\Phi)}_{i,j}x_j$.
Then
\begin{align*}
\norm{\Phi(X)}_{W_1} &= \norm*{\frac1n \sum_{i\in \Lambda}\sum_{j\in \Lambda} \Phi_i(X_j)}_{W_1} \\
&\le \frac1n \sum_{i\in \Lambda}\sum_{j\in \Lambda} \norm{\Phi_i(X_j)}_{W_1} \tag{triangle inequality} \\
&\le \sum_{j\in \Lambda}\sum_{i\in \Lambda} D^{(\Phi)}_{i,j} \cdot x_j \\
&= \norm{D^{(\Phi)}x}_1 \tag{$D^{(\Phi)}_{i,j}, x_j \ge 0$} \\
&\le \norm{D^{(\Phi)}}_{1\to 1}\cdot \norm{x}_1 .
\end{align*}
Taking the infimum over all fermionic transport plans of $X$ proves the claim.
\end{proof}
Therefore $\norm{D^{(\Phi)}}_{1\to 1} < 1$ is a sufficient condition for strict Wasserstein norm contraction.

\section{Field-resonant Lindbladians for fermionic systems}

We first give a generalization of the CKG Lindbladian. The original CKG Lindbladian uses fixed parameters across all sites and \cite{BT26} generalizes the parameters to be \emph{site-dependent}. We take one step further: we make the parameters \emph{jump-dependent} and adapt it to the fermionic systems. 

\begin{definition}[Generalization of CKG Lindbladian]
\label{def:fermionic-general-lindblad}
Let $\beta>0$ and $H$ be a fermionic Hamiltonian on $\Lambda$.
For each site $i\in\Lambda$, let $\calB_i \subseteq \mathfrak{A}_{\{i\}}$ be a finite set of homogeneous operators that is closed under taking adjoints. 
For each $B\in \calB_i$, choose parameters $\Delta_B, \sigma_B, \eta_B>0$ such that 
\begin{equation}\label{eq:DB_cond}
    \beta(\eta_B^2+\sigma_B^2)=2\Delta_B \,, \qquad \Delta_B = \Delta_{B^\dagger} \,, \qquad \sigma_B = \sigma_{B^\dagger} \,, \qquad \eta_B = \eta_{B^\dagger} \,.
\end{equation}
Define 
\begin{equation}\label{eq:jump}
    A_B(\omega) :=
\frac{1}{\sqrt{2\pi}}
\int_{-\infty}^{\infty}
e^{\ii Ht} B e^{-\ii Ht}\, e^{-\ii \omega t} f_B(t)\,\mathrm dt \,,
\end{equation} 
where $f_B$ is the filter function given by
\begin{equation*}
    f_B(t):=\sqrt{\sigma_B}(2/\pi)^{1/4}e^{-\sigma_B^2 t^2} \,.
\end{equation*}
Then the dissipative evolution is given by
\[
\mathcal L_{\mathrm{diss},B}(\rho)
:=
\int_{-\infty}^{\infty}
\gamma_B(\omega)
\Paren{
A_B(\omega)\rho A_B(\omega)^\dagger
-\frac12\braces{A_B(\omega)^\dagger A_B(\omega),\rho}
}\,\mathrm d\omega \,,
\]
where $\gamma_B$ is the transition-weight function given by
\begin{equation*}
    \gamma_B(\omega):=
\exp\!\left(-\frac{(\omega+\Delta_B)^2}{2\eta_B^2}\right) \,. 
\end{equation*}
Define two kernel functions
\begin{align*}
    b_{1}(t, \sigma)
&:=
\frac{\sigma}{\pi\beta} e^{\beta^2\sigma^2/8}
\Paren{
\sech\Paren{\frac{2\pi t}{\beta}}
*_t
\sin(-\beta\sigma^2 t)e^{-2\sigma^2 t^2}
} \,,
\\
b_{2}(t, \eta, \Delta)
&:=
2\eta
\exp\Paren{-\frac{4\Delta t^2}{\beta}-2\ii \Delta t} \,.
\end{align*}
Then the coherent term is given by
\begin{equation}\label{eq:cohe}
C_B
:=
\int_{-\infty}^{\infty}
b_{1}(t,\sigma_B)\, e^{-\ii Ht}
\Paren{
\int_{-\infty}^{\infty}
b_{2}(t',\eta_B, \Delta_B)\, e^{\ii Ht'} B^\dagger e^{-2\ii Ht'} B e^{\ii Ht'}\,\mathrm dt'
}
e^{\ii Ht}\,\mathrm dt.
\end{equation}
The Lindbladian generated by $B$ at site $i$ is
\[
\mathcal L_B(\rho)
:=
-\ii[C_B,\rho]+\mathcal L_{\mathrm{diss},B}(\rho).
\]
The site contribution is $\mathcal L_i = \sum_{B\in \mathcal B_i}\mathcal L_B$ 
and the full fermionic Lindbladian is $\mathcal L = \sum_{i\in\Lambda}\mathcal L_i$.
\end{definition}

We note that because the Hamiltonian $H$ is even, the Lindbladian defined above preserves fermion parity: each $\mathcal L_B$ maps $\mathfrak A_\Lambda^{\mathrm{ev}}$ into itself, even when $B$ is odd.

The conditions in \cref{eq:DB_cond} ensure that the Lindbladian satisfies KMS detailed balance and the proof is almost identical to the proof in \cite[Appendix B]{BT26}. 

\begin{proposition}\label{prop:fermionic-detailed-balance}
    The Lindbladian defined in \Cref{def:fermionic-general-lindblad} satisfies KMS detailed balance, and the Gibbs state $e^{-\beta H}/\tr(e^{-\beta H})$ is a stationary state, i.e.\ $\calL\Paren{e^{-\beta H}/\tr(e^{-\beta H})} = 0$.
\end{proposition}

Throughout, we will work with a harmless assumption, justified below, that each onsite term $v_i$ is diagonal in its local occupation-number basis with minimum eigenvalue $0$.  
\begin{remark}[Choice of onsite basis]
Since each onsite term $v_i$ is even and Hermitian, it preserves the
local even and odd parity sectors and can be diagonalized separately
within them. Hence there exists an even unitary $U_i$ supported on site
$i$ such that $U_i v_i U_i^\dagger$ is diagonal in the fixed local
occupation-number basis.
The unitaries $U_i$ commute for distinct sites, so $U:=\prod_{i\in\Lambda}U_i$
is well defined. Set $\widetilde H:=UHU^\dagger$ and conjugate every
jump operator according to $\widetilde B:=UBU^\dagger$. The
Lindbladians constructed from these data satisfy
\[
\widetilde{\mathcal L}(U\rho U^\dagger)
=
U\mathcal L(\rho)U^\dagger.
\]
This conjugation preserves parity, supports, operator norms, Gibbs
states, and trace distance. We may therefore work with
$\widetilde H$, drop the tildes, and assume that every $v_i$ is
diagonal in the local occupation-number basis.
Subtracting the minimum eigenvalue from each onsite term changes $H$ only
by a scalar multiple of the identity, leaving its Gibbs state and
Heisenberg evolution unchanged.
\end{remark}

\begin{definition}[Field-resonant Lindbladian]
    \label{def:field-resonant-lindblad}
    Fix a site $i\in\Lambda$, and write $\mathcal F_i:=\mathcal F_{\{i\}}$.
    Using the matrix units from \Cref{def:local_basis_matrix_unit}, set
    \[
    P^i_{\mu\nu}:=P^{\{i\}}_{\mu\nu},
    \qquad
    \calB_i:=\{P^i_{\mu\nu}:\mu,\nu\in\mathcal F_i\}.
    \]
    $\calB_i$ is closed under taking adjoints, and every $P^i_{\mu\nu}$ is supported on site $i$
    only and $\|P^i_{\mu\nu}\|=1$. 

    Assume that each onsite term $v_i$ is diagonal in the local occupation-number basis:
    \[
    v_i=\sum_{\mu\in\mathcal F_i}\varepsilon_\mu P^i_{\mu\mu}\, \qquad \text{and} \qquad \min_{\mu\in\mathcal F_i}\varepsilon_\mu=0.
    \]
    For each $P^i_{\mu\nu}\in \calB_i$, set the associated parameter $\Delta_{\mu\nu}=\Delta_{\nu\mu} \coloneqq (a_\mu + a_\nu)^2$ where 
    \begin{equation*}
        a_\mu := \sqrt{\varepsilon_\mu + 1/(4\beta)},
        \qquad
        a_\nu := \sqrt{\varepsilon_\nu + 1/(4\beta)}.
    \end{equation*}
    Then set
    \[
    \sigma_{\mu\nu}=\eta_{\mu\nu}:=\sqrt{\Delta_{\mu\nu}/\beta} \,.
    \]
    The jumps $\calB_i$ together with the associated parameters $\Delta, \eta,\sigma$ give the field-resonant Lindbladian $\calL_{\mathsf{FR}}=\sum_{i\in\Lambda}\calL_{\mathsf{FR},i}$, where $\calL_{\mathsf{FR},i}$ is the site contribution from \Cref{def:fermionic-general-lindblad} evaluated at these jumps and parameters. 
\end{definition}

For every jump $B=P^i_{\mu\nu}\in\mathcal B_i$, these parameters satisfy
\begin{equation}\label{eq:field-resonant-parameter-bounds}
\begin{aligned}
\beta\Delta_B
&=\left(\sqrt{\beta\varepsilon_\mu+1/4}
+\sqrt{\beta\varepsilon_\nu+1/4}\right)^2\ge1,\\
\sigma_B^{-1}&=\sqrt{\beta/\Delta_B}\le\beta.
\end{aligned}
\end{equation}

\subsection{A fermionic Lieb-Robinson bound with onsite terms}

For a local fermionic Hamiltonian $H=V+W$ as in \Cref{def:local-fermionic-ham}, we use the interaction-truncated Hamiltonian $H_\Omega^{\mathrm{int}}$ from \Cref{def:interaction-truncated-Ham}, which retains all onsite terms in $V$.
Following the spin-system terminology of~\cite{BT26}, we call the bound below \emph{field-independent} because its constants are uniform over the onsite terms $v_i$'s. 

\begin{restatable}[Field-independent fermionic Lieb-Robinson bound]{theorem}{fieldindependentfermionic}
\label{thm:field-independent-fermionic}
Let $H=V+W$ be a local fermionic Hamiltonian on $\Lambda$ with onsite terms and $\zeta>0$ be the associated local multi-site coupling strength as defined in \Cref{def:interaction-parameters}.
Let $X\subseteq \Omega \subseteq \Lambda$ and set $\ell := \dist(X,\Lambda\setminus \Omega)$. 
If $2\leq \ell < \infty$, then for every $T\in \R$ and every $A\in \mathfrak A_X$,
\begin{equation}
\norm{e^{\ii T H} A e^{-\ii T H}-e^{\ii T H_\Omega^{\mathrm{int}}} A e^{-\ii T H_\Omega^{\mathrm{int}}}}
\le
\norm{A}\cdot\abs{X}\cdot\frac{(2\zeta \abs{T})^\ell}{\ell!}.
\label{eq:field-independent-fermionic}
\end{equation}
\end{restatable}

The proof of \Cref{thm:field-independent-fermionic} follows directly from \cite{BT26} as justified in the remark below. Nevertheless for completeness, we include the proofs in \Cref{sec:technical-proofs-fermionic-lr}. 

\begin{remark}[Why the \cite{BT26} proof adapts]
The only structural difference from the qubit case is that two fermionic operators with disjoint support do not necessarily commute.
But because every term in a fermionic Hamiltonian is even, the supercommutation in \Cref{prop:supercommute} supplies exactly the relation our proof needs: two homogeneous operators with disjoint support commute in the ordinary sense as long as one of them is even.
An inhomogeneous observable is handled by splitting $A=A^{\mathrm{ev}}+A^{\mathrm{odd}}$ and treating each parity component separately.

This observation enters in two places. First, the onsite terms in $V$ still factor site-by-site: because each $v_i$ is even and supported on $\{i\}$, the family $\{v_i\}_{i\in \Lambda}$ pairwise commutes and
\begin{equation}\label{eq:onsite_factors}
    e^{\ii tV}=\prod_{i\in \Lambda} e^{\ii t v_i}.
\end{equation}
Hence conjugation by $e^{\ii tV}$ preserves the support and norm of every multi-site term $w_Z$ and every local observable $A$, exactly as in \cite{BT26}.
Second, the truncation recursion rests on a differential inequality whose proof requires each Hamiltonian term to commute with every observable supported away from it. Conjugation by $e^{\ii tV}$ preserves parity, so every term stays even in the rotating frame, and \Cref{prop:supercommute} again supplies exactly this commutation. The recursion therefore closes without any further change.
Therefore, once that ordinary-commutator differential inequality (\Cref{lem:fermionic-differential}) is in place, the rotating-frame reduction and the subsequent truncation bound follow almost line by line from the proof in \cite{BT26}.
\end{remark}

\subsection{Quasi-locality of the field-resonant Lindbladian}
\label{subsec:quasi-locality}

Throughout this subsection, we use the field-resonant Lindbladian from
\Cref{def:field-resonant-lindblad}.

For each fixed jump $B\in\mathcal B_i$, the scalar filter and kernel
functions in our construction are the same as those in
\cite[Section 6.2]{BT26}, evaluated at the jump-dependent parameters
$(\sigma_B,\eta_B,\Delta_B)$. Hence the moment estimates of
\cite[Lemmas 6.4, 6.7, and 6.8]{BT26} apply separately to every $B$.
Combining these estimates with the fermionic shell decomposition in
\Cref{lem:fermionic-shell}, the arguments of
\cite[Propositions 6.9 and 6.10]{BT26} give the following two
quasi-locality statements.

\begin{proposition}[Jump operators are quasi-local]
\label{prop:fermionic-jump-operators-quasi-local}
Let $0<\beta\le 1/(4\zeta)$, fix a site $i\in \Lambda$, and let $B\in \mathcal B_i$.
Then the jump operator $A_B(\omega)$ defined in \cref{eq:jump} admits a decomposition
\[
A_B(\omega)=\sum_{r\ge 0} G_{i,r}^{B,\omega}, \qquad \norm{G_{i,r}^{B,\omega}}
\le
(2\pi)^{-1/4}\sigma_B^{-1/2}(2\zeta \sigma_B^{-1})^r, 
\]
where $G_{i,r}^{B,\omega}$ is homogeneous with the same parity as $B$
and is supported on $\ball(\{i\},r)$. 
In particular, 
\begin{equation}\label{eq:fermionic-onsite-jump}
G_{i,0}^{B,\omega}
=
\frac{1}{\sqrt{2\pi}}
\int_{-\infty}^{\infty}
e^{\ii t v_i} B e^{-\ii t v_i}\, e^{-\ii \omega t} f_B(t)\,\mathrm dt.
\end{equation}
\end{proposition}

\begin{proposition}[Coherent terms are quasi-local]
\label{prop:fermionic-coherent-quasi-local}
Let $0<\beta\le 1/(12\zeta)$, fix a site $i\in \Lambda$, and let $B\in \mathcal B_i$.
Then the coherent term defined in \cref{eq:cohe} admits a decomposition
\[
C_B=\sum_{r\ge 1} K_{i,r}^B, \qquad \norm{K_{i,r}^B}\le 6(6\zeta\beta)^r,
\]
where each $K_{i,r}^B$ is even, Hermitian, and supported on $\ball(\{i\},r)$.
\end{proposition}

\begin{proof}
The support and norm bounds follow from \cite[Proposition 6.10]{BT26}. We only need to show that there is no $K_{i,0}^B$ in the decomposition of $C_B$: 
\[
\begin{aligned}
K_{i,0}^B
&=\iint_{\mathbb R^2} b_1(t,\sigma_B)b_2(t',\eta_B,\Delta_B)\cdot e^{-\ii t v_i}
\left(e^{\ii t'v_i}B^\dagger e^{-2\ii t'v_i}B e^{\ii t'v_i}\right)
e^{\ii t v_i}\,\mathrm dt'\,\mathrm dt.
\end{aligned}
\]
For $B=P^i_{\mu\nu}$, we have
\[
e^{\ii t'v_i}B^\dagger e^{-2\ii t'v_i}B e^{\ii t'v_i}
=e^{-2\ii t'(\varepsilon_\mu-\varepsilon_\nu)}P^i_{\nu\nu}.
\]
Since $P^i_{\nu\nu}$ commutes with $v_i$ and the function $b_1(t,\sigma_B)$ is odd, we have $K_{i,0}^B=0$.
\end{proof}

\section{Evolution of fermionic transport plans}\label{sec:evolution}

Throughout this section we work with the field-resonant Lindbladian from
\Cref{def:field-resonant-lindblad}. Fix a site \(i\).
Recall the parameter bounds in \cref{eq:field-resonant-parameter-bounds}.
Note that $|\mathcal B_i| = 4^m$ and $\int_{-\infty}^{\infty}\gamma_B(\omega)\diff\omega=\sqrt{2\pi}\,\eta_B$. 
We define a magic constant $\kappa_\star:=e^{-1/4}/\sqrt2$. 
The main theorem of this section is the following. 

\begin{theorem}[Update matrix of the fermionic Lindbladian evolution at a site]
\label{thm:fermionic-wass-growth-per-site}
Assume \(0<\beta\le 1/(12\zeta)\).
For every \(0\le\delta\le 1/(\kappa_\star 2^{m+1})\), the map
\(\Phi_i^{(\delta)}:=\calI+\delta\calL_{\mathsf{FR},i}\) has a fermionic update matrix
\[
(1+160\cdot16^m\delta^2)I+\delta Q^{(i)}, \qquad \text{where} \qquad Q^{(i)}
:=
-\kappa_\star E_{\{i\}}
+40\cdot4^m\sum_{r\ge 1}(6\beta\zeta)^rE_{\ball(\{i\},r)} .
\]
\end{theorem}

\begin{proof}
For every even traceless Hermitian \(X\),
\[
\Phi_i^{(\delta)}(X)
=
\frac12\Phi_{i,\mathrm{cohe}}^{(2\delta)}(X)
+\frac12\Phi_{i,\mathrm{diss}}^{(2\delta)}(X).
\]
The restriction on \(\delta\) ensures that \(2\delta\le(\kappa_\star2^m)^{-1}\),
so \Cref{lem:fermionic-wass-diss} applies at step size \(2\delta\). Applying
\Cref{lem:fermionic-wass-cohe,lem:fermionic-wass-diss} at step size \(2\delta\), using the linearity of transport plans and \((2\beta\zeta)^r\le(6\beta\zeta)^r\), gives the claimed update matrix. 
\end{proof}

\subsection{Update matrix of the coherent evolution}

\begin{lemma}[Update matrix of the coherent evolution]
\label{lem:fermionic-wass-cohe}
Assume \(0<\beta\le 1/(12\zeta)\). For every \(\delta\ge 0\), the map
\[
\Phi_{i,\mathrm{cohe}}^{(\delta)}(X)
:=
X-\ii\delta\sum_{B\in\mathcal B_i}[C_B,X]
\]
has a fermionic update matrix
\[
(1+72\cdot16^m\delta^2)I+\delta Q_{\mathrm{cohe}}^{(i)}, \qquad \text{where} \qquad Q_{\mathrm{cohe}}^{(i)}
:=
24\cdot4^m\sum_{r\ge 1}(6\beta\zeta)^r E_{\ball(\{i\},r)} .
\]
\end{lemma}

\begin{proof}
By \Cref{prop:fermionic-normalized-update-check}, it suffices to fix a site
\(j\) and an even Hermitian \(X\) such that
\(\operatorname{tr}_j(X)=0\) and \(\norm{X}_1=2\).

Using the expansion $C_B=\sum_{r\ge 1}K_{i,r}^B$ from
\Cref{prop:fermionic-coherent-quasi-local}, define
\[
G_j
:=
\sum_{B\in\mathcal B_i}
\sum_{\substack{r\ge 1\\j\notin\ball(\{i\},r)}}K_{i,r}^B,
\qquad
S_j
:=
\sum_{B\in\mathcal B_i}
\sum_{\substack{r\ge 1\\j\in\ball(\{i\},r)}}K_{i,r}^B.
\]
Then \(\sum_{B\in\mathcal B_i}C_B=G_j+S_j\), so we can split
\[
\Phi_{i,\mathrm{cohe}}^{(\delta)}(X)
=
\underbrace{X-\ii\delta[G_j,X]}_{=:T_1}
\underbrace{{}-\ii\delta[S_j,X]}_{=:T_2}.
\]
Since \(G_j\in \mathfrak A_{\{j\}^c}\) is even and Hermitian, by the module property in
\Cref{prop:fermionic-partial-trace-identities}, we have
\[
\operatorname{tr}_j(T_1)
=\operatorname{tr}_j(X)-\ii\delta[G_j,\operatorname{tr}_j(X)]=0.
\]
Therefore
\Cref{prop:fermionic-cost-vector-single-partial-trace-zero} charges \(T_1\) to
site \(j\) with cost \(\frac12\norm{T_1}_1e_j\).
Using \Cref{lem:fermionic-one-norm-coherent}, we have $\frac12\norm{T_1}_1
\le
1+2\delta^2\norm{G_j}^2$. 
By \Cref{prop:fermionic-coherent-quasi-local} and $6\zeta\beta\le1/2$,
\[
\norm{G_j}
\le 6\abs{\mathcal B_i}\sum_{r\ge1}(6\zeta\beta)^r
\le 6\abs{\mathcal B_i}.
\]
Hence \(T_1\) has a transport plan with cost vector at most $(1+72\abs{\mathcal B_i}^2\delta^2)e_j$. 

For \(T_2\), we estimate the commutators in the expansion of \(S_j\)
separately. Fix \(B\) and \(r\ge 1\). Since \(K_{i,r}^B\) is even and
supported on \(\ball(\{i\},r)\), cyclicity for even local operators in \Cref{prop:fermionic-partial-trace-identities} gives
\[
\operatorname{tr}_{\ball(\{i\},r)}
\!\left([K_{i,r}^B,X]\right)=0.
\]
The operator
\[
-\ii\delta\Iver{j\in\ball(\{i\},r)}\cdot [K_{i,r}^B,X]
\]
is even and Hermitian. By
\Cref{prop:fermionic-cost-vector-region-partial-trace-zero},
it has a fermionic transport plan with cost vector at most
\[
2\delta\norm{X}_1\cdot \norm{K_{i,r}^B}
\cdot E_{\ball(\{i\},r)} \cdot e_j .
\]
Summing over \(B\) and \(r\) by linearity and absolute convergence, and using
\(\norm{X}_1=2\), gives a cost vector for \(T_2\) bounded by
\[
24\abs{\calB_i}\delta\sum_{r\ge 1}(6\beta\zeta)^r
E_{\ball(\{i\},r)}e_j .
\]
Combining the plans for \(T_1\) and \(T_2\) and using
\(\abs{\calB_i}=4^m\) proves the claim.
\end{proof}

\subsection{Update matrix of the dissipative evolution}

\begin{lemma}[Update matrix of the dissipative evolution]
\label{lem:fermionic-wass-diss}
Assume \(0<\beta\le 1/(4\zeta)\).
For any \(0\le\delta\le(\kappa_\star2^m)^{-1}\), define
\[
\Phi_{i,\mathrm{diss}}^{(\delta)}(X)
:=
X+\delta\sum_{B\in\mathcal B_i}
\int_{-\infty}^{\infty}
\gamma_B(\omega)
\left(
A_B(\omega)XA_B(\omega)^\dagger
-\frac12\{A_B(\omega)^\dagger A_B(\omega),X\}
\right)\diff\omega .
\]
Then \(\Phi_{i,\mathrm{diss}}^{(\delta)}\) has a fermionic update matrix
\[
(1+8\cdot16^m\delta^2)I+\delta Q_{\mathrm{diss}}^{(i)}, \qquad \text{where} \qquad Q_{\mathrm{diss}}^{(i)}
:=
-\kappa_\star E_{\{i\}}
+16\cdot4^m\sum_{r\ge 1}(2\beta\zeta)^rE_{\ball(\{i\},r)} .
\]
\end{lemma}

\begin{proof}
Again by \Cref{prop:fermionic-normalized-update-check}, it suffices to fix a site
\(j\) and an even Hermitian \(X\) such that
\(\operatorname{tr}_j(X)=0\) and \(\norm{X}_1=2\).

Expanding \(A_B(\omega)=\sum_{r\ge 0}G_{i,r}^{B,\omega}\) as in
\Cref{prop:fermionic-jump-operators-quasi-local}, we have
\begin{align}\label{eq:euler_diss_expansion}
\Phi_{i,\mathrm{diss}}^{(\delta)}(X)
=
X
&+\delta\sum_{B\in\mathcal B_i}
\sum_{r_1,r_2\ge 0}
\int_{-\infty}^{\infty}\gamma_B(\omega)
\Bigl(
G_{i,r_1}^{B,\omega}X(G_{i,r_2}^{B,\omega})^\dagger
-\frac12\{(G_{i,r_2}^{B,\omega})^\dagger G_{i,r_1}^{B,\omega},X\}
\Bigr)\diff\omega .
\end{align}
All double sums below are absolutely convergent after integration against
\(\gamma_B(\omega)\), using the quasi-local bounds and
\(\sigma_B^{-1}\le\beta\).

For \(r_1,r_2\ge 0\) with \(r_1+r_2\ge 1\), define the symmetrized summand
\begin{align*}
S_{r_1,r_2}^{B,\omega}
&:=
G_{i,r_1}^{B,\omega}X(G_{i,r_2}^{B,\omega})^\dagger
+G_{i,r_2}^{B,\omega}X(G_{i,r_1}^{B,\omega})^\dagger
-\frac12\left\{
(G_{i,r_2}^{B,\omega})^\dagger G_{i,r_1}^{B,\omega}
+(G_{i,r_1}^{B,\omega})^\dagger G_{i,r_2}^{B,\omega},
X
\right\}.
\end{align*}
Since $G_{i,r_1}^{B,\omega}$ and $G_{i,r_2}^{B,\omega}$ have the same parity, then
\Cref{lem:fermionic-wass-norm-symmetric-map} gives a transport plan for
\(S_{r_1,r_2}^{B,\omega}\) with cost vector bounded by
\[
4\norm{G_{i,r_1}^{B,\omega}}\cdot \norm{G_{i,r_2}^{B,\omega}}\cdot \norm{X}_1
\cdot e_{\ball(\{i\},\max\{r_1,r_2\})}.
\]

\paragraph{Suppose \(j=i\).}
Split \cref{eq:euler_diss_expansion} into \(T_0+T_{\mathrm{tail}}\),
where \(T_0\) contains the term \(X\) and all summands with \(r_1=r_2=0\).
By the onsite contraction lemma in \Cref{lem:fermionic-onsite-contraction},
\[
T_0=\Phi_{i,0}^{(\delta)}(X)=(1-\kappa_\star z_i\delta)X,
\qquad \text{for some} \quad
1\le z_i\le 2^m.
\]
With \(0\le\delta\le(\kappa_\star2^m)^{-1}\), we have $\norm{T_0}_1\le 2(1-\kappa_\star\delta)$ and $\tr_i(T_0)=\tr_i(X)=0$.
Therefore by \Cref{prop:fermionic-cost-vector-single-partial-trace-zero}, \(T_0\) has a transport plan with cost vector at most \((1-\kappa_\star\delta)e_i\).

On the other hand, $T_{\mathrm{tail}}$ satisfies
\[
T_{\mathrm{tail}}
=
\frac{\delta}{2}\sum_{B\in\mathcal B_i}
\int_{-\infty}^{\infty}\gamma_B(\omega)
\sum_{r_1+r_2\ge 1}S_{r_1,r_2}^{B,\omega}\,\diff\omega .
\]
Hence, by linearity and absolute convergence,
\(T_{\mathrm{tail}}\) has a transport plan with cost vector at most
\begin{equation}
\label{eq:fermionic-diss-tail-cost}
4\delta\sum_{B\in\mathcal B_i}
\int_{-\infty}^{\infty}\gamma_B(\omega)
\sum_{r_1+r_2\ge 1}
\norm{G_{i,r_1}^{B,\omega}} \cdot \norm{G_{i,r_2}^{B,\omega}} \cdot 
E_{\ball(\{i\},\max\{r_1,r_2\})}\cdot \diff\omega \cdot e_i\, .
\end{equation}

\paragraph{Suppose \(j\neq i\).} 
For each \(B\) and \(\omega\), set
\[
G_{\mathrm{out}}^{B,\omega}
:=\sum_{\substack{r\ge 0\\j\notin\ball(\{i\},r)}}G_{i,r}^{B,\omega},
\]
and
\[
T_{\mathrm{out}}
:=
X+\delta\sum_{B\in\mathcal B_i}\int_{-\infty}^{\infty}\gamma_B(\omega)
\left(
G_{\mathrm{out}}^{B,\omega}X(G_{\mathrm{out}}^{B,\omega})^\dagger
-\frac12\{(G_{\mathrm{out}}^{B,\omega})^\dagger
G_{\mathrm{out}}^{B,\omega},X\}
\right)\diff\omega .
\]
Since \(G_{\mathrm{out}}^{B,\omega}\in\mathfrak A_{\{j\}^c}\), the module
property in \Cref{prop:fermionic-partial-trace-identities} and
\(\operatorname{tr}_j(X)=0\) imply $\operatorname{tr}_j(T_{\mathrm{out}})=0$. 
Thus \(T_{\mathrm{out}}\) can be charged to \(j\) with cost
\(\frac12\norm{T_{\mathrm{out}}}_1e_j\).

Set
\[
D_{\mathrm{out}}
:=\delta\sum_{B\in\mathcal B_i}\int_{-\infty}^{\infty}\gamma_B(\omega)
(G_{\mathrm{out}}^{B,\omega})^\dagger G_{\mathrm{out}}^{B,\omega}\,\diff\omega.
\]
Then \Cref{prop:fermionic-jump-operators-quasi-local} and \(2\zeta\beta\le 1/2\) imply that
\[
\norm{G_{\mathrm{out}}^{B,\omega}}
\le
\sum_{r\ge 0}(2\pi)^{-1/4}\sigma_B^{-1/2}
(2\zeta\sigma_B^{-1})^r
\le
2(2\pi)^{-1/4}\sigma_B^{-1/2}.
\]
Using \(\int_{-\infty}^{\infty}\gamma_B(\omega)\,\diff\omega
=\sqrt{2\pi}\eta_B\) and \(\eta_B=\sigma_B\), we obtain
\begin{align*}
\norm{D_{\mathrm{out}}}
&\le \delta\sum_{B\in\mathcal B_i}\int_{-\infty}^{\infty}\gamma_B(\omega)
\norm{G_{\mathrm{out}}^{B,\omega}}^2\,\diff\omega \le 4\delta\sum_{B\in\mathcal B_i}\frac{\eta_B}{\sigma_B}
=4\abs{\calB_i}\delta.
\end{align*}
Thus the integral version of \Cref{lem:fermionic-one-norm-dissipative}
and \(\norm{X}_1=2\) give
\[
\frac12\norm{T_{\mathrm{out}}}_1
\le 1+\frac12\norm{D_{\mathrm{out}}}^2
\le 1+8\abs{\calB_i}^2\delta^2.
\]

The remaining part is
\[
T_{\mathrm{in}}
:=
\frac{\delta}{2}\sum_{B\in\mathcal B_i}\int_{-\infty}^{\infty}\gamma_B(\omega)
\sum_{\substack{r_1,r_2\ge 0\\j\in\ball(\{i\},\max\{r_1,r_2\})}}
S_{r_1,r_2}^{B,\omega}\,\diff\omega .
\]
Since \(j\neq i\), every pair in this sum satisfies \(r_1+r_2\ge 1\).
For each such pair,
\[
e_{\ball(\{i\},\max\{r_1,r_2\})}
=
E_{\ball(\{i\},\max\{r_1,r_2\})}e_j .
\]
By linearity and absolute convergence, the cost bounds for the symmetrized
summands give a cost bound for \(T_{\mathrm{in}}\). These bounds are
nonnegative, so extending their sum to all pairs with \(r_1+r_2\ge 1\)
bounds the cost vector by the right-hand side of
\cref{eq:fermionic-diss-tail-cost} with \(e_i\) replaced by \(e_j\).

It remains to bound
\begin{align*}
&\int_{-\infty}^{\infty}\gamma_B(\omega)
\norm{G_{i,r_1}^{B,\omega}}\cdot \norm{G_{i,r_2}^{B,\omega}}\,\diff\omega \le
\sqrt{2\pi}\eta_B\,
(2\pi)^{-1/2}\sigma_B^{-1}
(2\zeta\sigma_B^{-1})^{r_1+r_2}
\le
(2\zeta\beta)^{r_1+r_2}.
\end{align*}
Hence, with \(2\zeta\beta\le 1/2\), we have the entry-wise bound
\[
4\delta\sum_{B\in\mathcal B_i}
\sum_{r_1+r_2\ge 1}
(2\zeta\beta)^{r_1+r_2}
E_{\ball(\{i\},\max\{r_1,r_2\})}
\le_v
16\abs{\calB_i}\delta
\sum_{r\ge 1}(2\zeta\beta)^rE_{\ball(\{i\},r)} .
\]
To see the last inequality, set \(a:=2\zeta\beta\le 1/2\) and group ordered pairs by
\(s:=\max\{r_1,r_2\}\ge 1\). For fixed \(s\),
\[
\sum_{\max\{r_1,r_2\}=s}a^{r_1+r_2}
\le
2a^s\sum_{t=0}^s a^t
\le
\frac{2a^s}{1-a}
\le
4a^s. 
\]
Combining all the cases and using \(\abs{\calB_i}=4^m\) proves the claimed update matrix.
\end{proof}

\subsection{Onsite contraction lemma}
\begin{lemma}[Onsite contraction]\label{lem:fermionic-onsite-contraction}
Fix a site $i\in\Lambda$ and use the field-resonant construction from
\Cref{def:field-resonant-lindblad}. For $\delta\ge0$, define the onsite
dissipative step
\[
\Phi_{i,0}^{(\delta)}(X):=
X+\delta\sum_{B\in\mathcal B_i}\int_{-\infty}^{\infty}
\gamma_B(\omega)\,D_{G_{i,0}^{B,\omega}}(X)\,d\omega,
\]
where $D_A(X):=AXA^\dagger-\frac12\{A^\dagger A,X\}$ and
$G_{i,0}^{B,\omega}$ is defined in \cref{eq:fermionic-onsite-jump}.
Let $z_i:=\sum_{\mu\in\mathcal F_i}e^{-\beta\varepsilon_\mu}$.
Then for any even operator $X$ satisfying $\operatorname{tr}_i(X)=0$, one has the exact identity
\[
\Phi_{i,0}^{(\delta)}(X)=\bigl(1-\kappa_\star z_i\delta\bigr)X.
\]
In particular, since $1\le z_i\le 2^m$, for every $0\le \delta\le 1/(\kappa_\star 2^m)$, one has
\[
\|\Phi_{i,0}^{(\delta)}(X)\|_1
=
\bigl(1-\kappa_\star z_i\delta\bigr)\|X\|_1
\le
\bigl(1-\kappa_\star\delta\bigr)\|X\|_1.
\]
\end{lemma}

\begin{proof}
Abbreviate $P_{\mu\nu}:=P^i_{\mu\nu}$. Fix $\mu,\nu\in\mathcal F_i$,
set $B:=P_{\mu\nu}$, and write
$\Omega_{\mu\nu}:=\varepsilon_\mu-\varepsilon_\nu$.
Since $v_i$ is diagonal in the occupation basis,
\[
e^{itv_i}P_{\mu\nu}e^{-itv_i}
=
e^{it(\varepsilon_\mu-\varepsilon_\nu)}P_{\mu\nu}
=
e^{it\Omega_{\mu\nu}}P_{\mu\nu} \,.
\]
Since the Fourier transform of $f_B$ is given by
\[
\widehat f_B(\xi):=
\frac{1}{\sqrt{2\pi}}
\int_{-\infty}^{\infty}f_B(t)e^{-i\xi t}\,dt =
(2\pi)^{-1/4}\sigma_{\mu\nu}^{-1/2}
\exp\!\left(-\frac{\xi^2}{4\sigma_{\mu\nu}^2}\right) \,,
\] 
then
\[
G_{i,0}^{B,\omega}
=
\widehat f_B(\omega-\Omega_{\mu\nu})\,P_{\mu\nu} \,.
\]
Therefore
\begin{equation*}
    D_{G_{i,0}^{B,\omega}}(X) = \bigl|\widehat f_B(\omega-\Omega_{\mu\nu})\bigr|^2 \cdot D_{P_{\mu\nu}}(X) \,,
\end{equation*}
and thus
\[
\int_{-\infty}^{\infty}
\gamma_B(\omega)\,
D_{G_{i,0}^{B,\omega}}(X)\,d\omega
=
\kappa_{\mu\nu}\,D_{P_{\mu\nu}}(X),
\]
where
\begin{align*}
    \kappa_{\mu\nu} &\coloneqq \int_{-\infty}^{\infty} \gamma_B(\omega)\, \bigl|\widehat f_B(\omega-\Omega_{\mu\nu})\bigr|^2\,d\omega \\
    &= (2\pi)^{-1/2}\sigma_{\mu\nu}^{-1} \int_{-\infty}^{\infty} \exp\!\left(-\frac{(\omega - \Omega_{\mu\nu})^2 + (\omega + \Delta_{\mu\nu})^2}{2\sigma_{\mu\nu}^2} \right) \, d\omega \tag{$\eta_{\mu\nu} = \sigma_{\mu\nu}$} \\
    &= (2\pi)^{-1/2}\sigma_{\mu\nu}^{-1} \exp \Paren{-\frac{(\Omega_{\mu\nu} + \Delta_{\mu\nu})^2}{4\sigma_{\mu\nu}^2}} \int_{-\infty}^{\infty}  \exp\!\left(-\frac{\Paren{\omega + \frac{\Delta_{\mu\nu}-\Omega_{\mu\nu}}{2}}^2}{\sigma_{\mu\nu}^2} \right) \, d\omega \tag{$(\omega - \Omega)^2 + (\omega + \Delta)^2 = 2\Paren{\omega + \frac{\Delta_{\mu\nu}-\Omega_{\mu\nu}}{2}}^2 + \frac{(\Delta_{\mu\nu}+\Omega_{\mu\nu})^2}{2}$} \\
    &= \frac{1}{\sqrt2}\exp\!\left(-\frac{\beta(\Delta_{\mu\nu}+\Omega_{\mu\nu})^2}{4\Delta_{\mu\nu}}\right) \,. \tag{$\sigma_{\mu\nu} = \sqrt{\Delta_{\mu\nu}/\beta}$}
\end{align*}
Using \Cref{def:field-resonant-lindblad} that $\Delta_{\mu\nu} = (a_\mu+a_\nu)^2$ where $a_\mu = \sqrt{\varepsilon_\mu + 1/(4\beta)}$ and $a_\nu = \sqrt{\varepsilon_\nu + 1/(4\beta)}$,
\begin{equation*}
    \Delta_{\mu\nu}+\Omega_{\mu\nu}
=
(a_\mu+a_\nu)^2+(a_\mu^2-a_\nu^2)
=
2a_\mu(a_\mu+a_\nu).
\end{equation*}
So
\begin{equation*}
    \frac{\beta(\Delta_{\mu\nu}+\Omega_{\mu\nu})^2}{4\Delta_{\mu\nu}}
=
\beta a_\mu^2
=
\beta\varepsilon_\mu+\frac14.
\end{equation*}
Then 
\begin{equation*}
    \kappa_{\mu\nu}
=
\frac1{\sqrt2}
\exp\left(
-\beta\varepsilon_\mu-\frac14
\right) 
= \kappa_\star e^{-\beta\varepsilon_\mu}.
\end{equation*}
Thus
\begin{equation}\label{eq:fermionic-onsite-step}
\Phi_{i,0}^{(\delta)}(X)
=
X+\kappa_\star\delta
\sum_{\mu,\nu\in\mathcal F_i}
e^{-\beta\varepsilon_\mu}D_{P_{\mu\nu}}(X) \,.
\end{equation}
It remains to sum $D_{P_{\mu\nu}}(X)$. Write
\[
X=\sum_{\mu,\lambda\in\mathcal F_i}P_{\mu\lambda}X_{\mu\lambda},
\qquad
X_{\mu\lambda}\in\mathfrak A_{\Lambda\setminus\{i\}}.
\]
Since $P_{\nu\nu} \in \mathfrak{A}_{\{i\}}$ is even,
it commutes with each $X_{\mu\lambda}\in\mathfrak A_{\Lambda\setminus\{i\}}$ \Cref{prop:supercommute}. Thus
\begin{equation}\label{eq:fermionic-diagonal-block}
P_{\nu\nu}XP_{\nu\nu}
=\sum_{\mu,\lambda\in\mathcal F_i}
P_{\nu\nu}P_{\mu\lambda}P_{\nu\nu}X_{\mu\lambda}=P_{\nu\nu}X_{\nu\nu}.
\end{equation}
Since $X$ is even, uniqueness of the expansion implies that each
$X_{\nu\nu}$ is even. Then
\begin{align}
P_{\mu\nu}XP_{\nu\mu}
&=P_{\mu\nu}(P_{\nu\nu}XP_{\nu\nu})P_{\nu\mu} \nonumber\\
&=P_{\mu\nu}P_{\nu\nu}X_{\nu\nu}P_{\nu\mu} \tag{\cref{eq:fermionic-diagonal-block}}\\
&=P_{\mu\nu}P_{\nu\nu}P_{\nu\mu}X_{\nu\nu} \tag{$X_{vv}$ commutes with $P_{\mu\nu}$ by \Cref{prop:supercommute}}\\
&=P_{\mu\mu}X_{\nu\nu}. \nonumber
\end{align}
Consequently,
\[
\sum_{\mu,\nu\in\mathcal F_i}
e^{-\beta\varepsilon_\mu}P_{\mu\nu}XP_{\nu\mu}
=e^{-\beta v_i}\sum_{\nu\in\mathcal F_i}X_{\nu\nu}
=e^{-\beta v_i}\tr_i(X).
\]
Also,
\[
\sum_{\mu,\nu\in\mathcal F_i}
e^{-\beta\varepsilon_\mu}P_{\mu\nu}^\dagger P_{\mu\nu}
=
\sum_{\mu,\nu\in\mathcal F_i}
e^{-\beta\varepsilon_\mu}P_{\nu\nu}
=
\Paren{\sum_{\mu\in\mathcal F_i}e^{-\beta\varepsilon_\mu}} \id = z_i\id\,.
\]
Therefore, using $\tr_i(X)=0$,
\[
\sum_{\mu,\nu\in\mathcal F_i}
e^{-\beta\varepsilon_\mu}D_{P_{\mu\nu}}(X)
=e^{-\beta v_i}\tr_i(X)-z_iX
=-z_iX.
\]
Substituting back into the formula for $\Phi_{i,0}^{(\delta)}(X)$ in
\cref{eq:fermionic-onsite-step} gives
\[
\Phi_{i,0}^{(\delta)}(X)
=
X-\kappa_\star z_i\delta\,X
=
\bigl(1-\kappa_\star z_i\delta\bigr)X.
\]
Because $\min_\mu\varepsilon_\mu=0$, we have
\[
1\le z_i=\sum_{\mu\in\mathcal F_i}e^{-\beta\varepsilon_\mu}\le |\mathcal F_i|=2^m.
\]
Hence, if $0\le \delta\le (\kappa_\star 2^m)^{-1}$, then $1-\kappa_\star z_i\delta\ge 0$, so
\[
\|\Phi_{i,0}^{(\delta)}(X)\|_1
=
\bigl(1-\kappa_\star z_i\delta\bigr)\|X\|_1
\le
\bigl(1-\kappa_\star\delta\bigr)\|X\|_1.
\]
This proves the lemma.
\end{proof}

\section{Rapid mixing from the fermionic Dobrushin condition}
\label{sec:fermionic-dobrushin}

We now have all the ingredients for proving a rapid mixing theorem for the field-resonant Lindbladian. We follow the quantum Dobrushin-condition framework of~\cite{blmt2025dobrushin}, in the form used in~\cite{BT26}, and replace the qubit Wasserstein machinery with its fermionic counterpart.

Throughout this section, $H = V + W$ is a local fermionic Hamiltonian with onsite terms $V$ (\Cref{def:local-fermionic-ham}) on $n$ sites in $\Lambda$, where each site carries $m$ complex fermionic modes. We assume that $W$ has bounded $L$-site degree-$D$ couplings as in \Cref{def:interaction-parameters}.

Let $\calL_{\mathsf{FR}} = \sum_{i\in \Lambda} \calL_{\mathsf{FR},i}$ be the associated field-resonant Lindbladian from \Cref{def:field-resonant-lindblad}.
We recall from \Cref{sec:evolution} the constant
$\kappa_\star = e^{-1/4}/\sqrt{2}$.

\subsection{Main theorem}

\begin{theorem}[Rapid mixing for fermionic Gibbs states with arbitrary onsite terms]
\label{thm:fermionic-rapid-mixing}
Let $H = V + W$ be a local fermionic Hamiltonian on $n$ sites, each carrying $m$ complex fermionic modes, where $W$ has bounded $L$-site degree-$D$ couplings and the onsite part $V = \sum_{i\in \Lambda} v_i$ is arbitrary.
For any inverse temperature
\[
0 < \beta \le 4^{-m}(DL)^{-3}/9600 \,,
\]
let $\sigma = e^{-\beta H}/\tr(e^{-\beta H})$ be the Gibbs state and let $\calL_{\mathsf{FR}}$ be the field-resonant Lindbladian of $H$ at inverse temperature $\beta$ as in \Cref{def:field-resonant-lindblad}. Then for every fermionic density operator $\rho$, every $\epsilon \in (0,1]$, and every time $t \ge 2\log\Paren{4n/\epsilon}$,
we have $\norm{e^{t\calL_{\mathsf{FR}}}(\rho) - \sigma}_1 \le \epsilon$.
\end{theorem}

Applying \Cref{thm:fermionic-rapid-mixing} to the Fermi--Hubbard model gives the following corollary.

\begin{corollary}[Rapid mixing for the Fermi--Hubbard model on a $d$-dimensional lattice]
\label{cor:fermionic-fermi-hubbard}
Let $\Lambda$ be $n$ sites on a $d$-dimensional lattice and let $E_{\Lambda}$ be the associated edge set.
Let $H_{\mathsf{FH}}$ be the spin-$\frac12$ Fermi--Hubbard Hamiltonian of \Cref{def:fermi-hubbard} on $(\Lambda,E_\Lambda)$, with hopping amplitudes $t_{ij}\in\R$ for each $(i,j)\in E_\Lambda$ and arbitrary onsite couplings $U_i,\mu_i\in\R$ for each $i\in \Lambda$.
Assume that
\[
t_{\max}\coloneqq\max_{\{i,j\}\in E_\Lambda}\abs{t_{ij}} >0.
\]
For any inverse temperature $\beta>0$ satisfying $\beta t_{\max} \leq d^{-3}/(2^{16}\cdot 300)$,
let $\sigma = e^{-\beta H_{\mathsf{FH}}}/\tr(e^{-\beta H_{\mathsf{FH}}})$ be the Gibbs state and let $\calL_{\mathsf{FR}}$ be the corresponding field-resonant Lindbladian. Then for every fermionic density operator $\rho$, every $\epsilon \in (0,1]$, and every time $t \ge 2\log(4n/\epsilon)$, we have $\norm{e^{t\calL_{\mathsf{FR}}}(\rho) - \sigma}_1 \le \epsilon$.
\end{corollary}
\begin{proof}
Each site carries $m=2$ modes, each hopping term has support size two and norm at most $2t_{\max}$, and each site has at most $2d$ neighbors. Rescale $\widetilde H=H_{\mathsf{FH}}/(2t_{\max})$ and
$\widetilde\beta=2t_{\max}\beta$, which leaves the Gibbs state
and the field-resonant Lindbladian unchanged.
Applying \Cref{thm:fermionic-rapid-mixing} with $L=2$ and $D=2d$
gives the stated temperature condition and mixing bound.
\end{proof}

The remainder of this section proves \Cref{thm:fermionic-rapid-mixing} by bounding the Dobrushin influence matrix of a one-step Euler map.

\subsection{Bounding the Dobrushin influence matrix}

\begin{lemma}[Bounding the fermionic Dobrushin influence matrix]
\label{lem:fermionic-dobrushin-influence-bound}
For any inverse temperature $0 < \beta \le 4^{-m}(DL)^{-3}/9600$ and every $\delta \in [0, 1/(\kappa_\star 2^{m+1})]$, define
\begin{equation*}
\Phi = \calI + \frac{\delta}{n}\calL_{\mathsf{FR}} = \frac1n \sum_{i\in \Lambda} (\underbrace{\calI + \delta \calL_{\mathsf{FR},i}}_{=\, \Phi_i^{(\delta)}}) ,
\end{equation*}
where $\calL_{\mathsf{FR}}$ is the associated field-resonant Lindbladian defined in \Cref{def:field-resonant-lindblad}, and let $D^{(\Phi)}$ be the associated Dobrushin influence matrix. Then
\begin{equation*}
\norm{D^{(\Phi)}}_{1\to 1} \le 1 - \frac{\delta}{2n} + 160\cdot16^m \cdot \delta^2 \,.
\end{equation*}
\end{lemma}
\begin{proof}
Set $c:=160\cdot16^m$.
We first verify the hypotheses of \Cref{thm:fermionic-wass-growth-per-site}. By \Cref{fact:interaction-growth}, $\zeta \le DL$, so
\[
\beta \le \frac{4^{-m}}{9600\,(DL)^3} \le \frac{1}{12 DL} \le \frac{1}{12\zeta} .
\]
Since each entry of $D^{(\Phi)}$ is nonnegative, $\norm{D^{(\Phi)}}_{1\to 1}$ equals its maximum column sum.

We first upper bound each entry of $D^{(\Phi)}$. Fix sites $i,j\in \Lambda$, and let $X\in \mathfrak A_\Lambda^{\mathrm{ev}}$ be Hermitian with $\operatorname{tr}_j(X) = 0$ and $\norm{X}_{W_1} = 1$.
By \Cref{prop:fermionic-trace-wasserstein-relation}, $\norm{X}_1=2$,
so there is a transport plan for $X$ with cost vector $e_j$.
Summing the entries of the cost-vector bound from
\Cref{thm:fermionic-wass-growth-per-site} and taking the maximum over $X$ gives
\[
n D^{(\Phi)}_{i,j}
\le 1 + c\delta^2 - \Iver{i=j}\cdot \kappa_\star \delta
+ 40\cdot 4^m \delta \sum_{r\ge 1}(6\beta\zeta)^r \cdot \Iver{\dist(i,j)\le r} \cdot \abs{\ball(\{i\},r)} \,,
\]
where we used $E_{\ball(\{i\},r)}\, e_j = \Iver{j \in \ball(\{i\},r)}\cdot e_{\ball(\{i\},r)}$ and $\norm{e_{\ball(\{i\},r)}}_1 = \abs{\ball(\{i\},r)}$.

To bound the column sums, \Cref{fact:interaction-growth} gives
\begin{equation*}
\sum_{i\in \Lambda:\ \dist(i,j)\le r} \abs{\ball(\{i\},r)}
\le \abs{\ball(\{j\},r)} \cdot \max_{i\in \Lambda} \abs{\ball(\{i\},r)}
\le (DL)^{2r} \,.
\end{equation*}
Together with $\zeta \le DL$ and $6\beta(DL)^3 \le 6/9600 \le 1/2$, we have
\begin{equation*}
\sum_{r\ge 1}(6\beta\zeta)^r (DL)^{2r}
\le \sum_{r\ge 1} \Paren{6\beta(DL)^3}^r
\le 12\beta(DL)^3 \,.
\end{equation*}
Summing the entrywise bound over $i$ therefore gives
\[
\sum_{i\in\Lambda} nD^{(\Phi)}_{i,j}
\le n(1+c\delta^2)-\kappa_\star\delta
+480\cdot4^m\delta\beta(DL)^3.
\]
Dividing by $n$ and taking the maximum over $j$, we obtain
\begin{equation*}
\norm{D^{(\Phi)}}_{1\to1}
\le 1 + c\delta^2 - \frac{\delta}{n}\Paren{\kappa_\star - 480\cdot4^m\beta(DL)^3} \le 1 - \frac{\delta}{2n} + c\delta^2 \,.
\end{equation*}
The last inequality uses $\kappa_\star>0.55$ and $480\cdot4^m\beta(DL)^3\le0.05$.
\end{proof}

We are now ready to prove our main theorem.

\begin{proof}[Proof of \Cref{thm:fermionic-rapid-mixing}]
By \Cref{prop:fermionic-detailed-balance}, $\calL_{\mathsf{FR}}(\sigma)=0$.
Since $H$ is even, $\sigma$ is a fermionic density operator, so
$X:=\rho-\sigma\in\mathfrak A_\Lambda^{\mathrm{ev}}$ is traceless and Hermitian.

Fix $t>0$. For every sufficiently large integer $k$, set
\[
\delta_k:=\frac{nt}{k}\le\frac{1}{\kappa_\star2^{m+1}},
\qquad
\Phi_k:=\calI+\frac{\delta_k}{n}\calL_{\mathsf{FR}}
=\calI+\frac{t}{k}\calL_{\mathsf{FR}}.
\]
Each $\Phi_k$ preserves the space of even traceless Hermitian operators.
Thus \Cref{lem:fermionic-contraction-dobrushin,lem:fermionic-dobrushin-influence-bound}
give, by iteration,
\[
\norm{\Phi_k^k(X)}_{W_1}
\le\norm{D^{(\Phi_k)}}_{1\to1}^k\norm{X}_{W_1}
\le\Paren{1-\frac{t}{2k}+\frac{cn^2t^2}{k^2}}^k\norm{X}_{W_1},
\]
where $c=160\cdot16^m$.
In finite dimension, the Euler-limit formula gives $\Phi_k^k\to e^{t\calL_{\mathsf{FR}}}$ as $k\to\infty$.
Together with the scalar limit
\[
\lim_{k\to\infty}\Paren{1-\frac{t}{2k}+\frac{cn^2t^2}{k^2}}^k=e^{-t/2},
\]
this implies
\[
\norm{e^{t\calL_{\mathsf{FR}}}(X)}_{W_1}\le e^{-t/2}\norm{X}_{W_1}.
\]
Finally, stationarity of $\sigma$ and
\Cref{prop:fermionic-trace-wasserstein-relation} give
\begin{align*}
\norm{e^{t\calL_{\mathsf{FR}}}(\rho)-\sigma}_1
&\le2\norm{e^{t\calL_{\mathsf{FR}}}(X)}_{W_1}\le2e^{-t/2}\norm{X}_{W_1}
\le2ne^{-t/2}\norm{X}_1
\le4ne^{-t/2}.
\end{align*}
Thus \(t\ge2\log(4n/\epsilon)\) gives the claim. In particular, this
implies that $\sigma$ is the unique stationary state of the semigroup
$(e^{t\calL_{\mathsf{FR}}})_{t\ge 0}$ among fermionic density operators.
\end{proof}

\section{Fermi--Hubbard Gibbs states outside Gaussian mixtures}

\begin{example}[A two-site Fermi--Hubbard example]
\label{ex:non-gaussian-gibbs}
For any $\beta>0$, there exists a Fermi--Hubbard Hamiltonian on $2$ sites with hopping amplitude $t_{1,2}=1$ such that its Gibbs state at inverse temperature $\beta$ is not a convex mixture of fermionic Gaussian states.
\end{example}

\begin{proof}
Fix $\beta>0$, and consider the Hamiltonian $H_L$ with
$U_1=U_2=\mu_1=L$, $\mu_2=-2L$, and $t_{1,2}=1$.
We show that sufficiently large $L$ works.
Order the modes as $(1,\uparrow),(1,\downarrow),(2,\uparrow),(2,\downarrow)$ and set
\[
\ket{D_1}=\ket{1100},\qquad
\ket{S}=\frac{\ket{1001}-\ket{0110}}{\sqrt2},\qquad
\ket{D_2}=\ket{0011}.
\]
Their span is invariant under $H_L$, whose matrix in this basis is
\[
\begin{pmatrix}
3L&-\sqrt2&0\\
-\sqrt2&-L&-\sqrt2\\
0&-\sqrt2&-3L
\end{pmatrix}.
\]
The onsite Hamiltonian has unique ground state $\ket{D_2}$ and gap $L$, while the hopping term has norm at most $2$.
Consequently, for large $L$, the full Hamiltonian has a unique ground state
$\ket{\psi_L}=a_L\ket{D_1}+b_L\ket{S}+c_L\ket{D_2}$ in this subspace, with positive coefficients, ground energy $E_L=-3L+O(1)$, and spectral gap $L+O(1)$.
Moreover, $c_L\to1$, so the first two eigenvector equations give
\[
L b_L\longrightarrow\frac1{\sqrt2},
\qquad
L^2a_L\longrightarrow\frac16.
\]

The four-mode Gaussian-overlap formula~\cite{OGK14} implies that every convex mixture $\omega$ of fermionic Gaussian states satisfies
\[
\bra{\psi_L}\omega\ket{\psi_L}
\leq\frac{1+\sqrt{1-C_L^2}}2,
\qquad
C_L:=|b_L^2-2a_Lc_L|
=\frac{1+o(1)}{6L^2}.
\]
This bound also allows odd-parity Gaussian components, which are orthogonal to $\ket{\psi_L}$, and mixed Gaussian components, which decompose into pure Gaussian states.
In particular, the largest possible overlap of a Gaussian mixture with $\ket{\psi_L}$ is at most
\[
1-\frac{1+o(1)}{144L^4}.
\]
On the other hand, the spectral gap and the fixed dimension of the Fock space imply that the Gibbs state $\rho_{\beta,L}=e^{-\beta H_L}/\tr(e^{-\beta H_L})$ satisfies
\[
1-\bra{\psi_L}\rho_{\beta,L}\ket{\psi_L}
=O(e^{-\beta L})=o(L^{-4}).
\]
For sufficiently large $L$, its overlap therefore exceeds the Gaussian-mixture bound, proving the claim.
\end{proof}

%% file: appendix_lieb_robinson.tex
\section{Lieb--Robinson bounds with arbitrary onsite terms}
\label{sec:LR_time_dependent_proof}
\label{sec:technical-proofs-fermionic-lr}

Although the rapid-mixing proofs only require the shell decompositions in \Cref{lem:LR_radius_r,lem:fermionic-shell}, we nevertheless include a proof for the general bounds in \Cref{thm:lr-potential,thm:field-independent-fermionic}.
The argument applies to both spin and fermionic systems; in the fermionic case, every Hamiltonian term is assumed to be even.

\subsection{Time-dependent setup and differential inequality}

For a time-dependent Hamiltonian $K(t)$, we write $U_K(t,s)$ for the unitary evolution from time $s$ to time $t$, determined by
\[
\ii\,\partial_t U_K(t,s) = K(t) U_K(t,s),
\qquad
U_K(s,s)=\id,
\]
and we write the corresponding Heisenberg evolution as
\[
\tau_{t,s}^K(A) := U_K(t,s)^\dagger A\, U_K(t,s).
\]
When $K$ is time independent, we abbreviate $\tau_t^K := \tau_{t,0}^K$.
We assume that all time-dependent Hamiltonian terms are piecewise continuous.
Throughout this appendix, $\mathcal A_X$ denotes the algebra of operators supported on $X$: the usual local operator algebra for spins, and $\mathfrak A_X$ for fermions.
The fermionic operators to which the bounds apply need not be even.

For $H(t) = \sum_{Z\subseteq\Lambda} h_Z(t)$ on a finite set of sites $\Lambda$, the distance $\dist$ refers to the hypergraph induced by
\[
\mathcal E_H := \{Z\subseteq \Lambda : h_Z(t)\not\equiv 0,\ \abs{Z}\ge 2\}.
\]
A key ingredient of the proof is to generalize the Lieb--Robinson bound by Haah, Hastings, Kothari, and Low~\cite[Lemma 5]{hhkl21} to time-dependent local Hamiltonians.
We start by proving a time-dependent version of their Eq.~(42).

\begin{lemma}[Generalized Lieb--Robinson differential inequality]\label{lem:generalize_eq_42}\label{lem:fermionic-differential}
    Let $H(t) = \sum_{Z\subseteq \Lambda} h_Z(t)$ be a time-dependent Hamiltonian with Hermitian terms supported on $Z$, assumed even in the fermionic case. Let $U(t, s)$ denote the unitary evolution from time $s$ to time $t$.
    For any region $X\subseteq \Lambda$, operator $B$ (even in the fermionic case), and times $t \geq s$, define
    \begin{equation*}
        C_B(X, t, s) \coloneqq \sup_{A\in \mathcal{A}_X, \ \norm{A}\leq 1} \norm{[A(t), B]},
    \end{equation*}
    where $\mathcal{A}_X$ is the algebra of all operators supported on $X$ and $A(t) = U(t, s)^\dagger A U(t, s)$ is the operator $A$ evolved from $s$ to $t$. Then,
    \begin{equation*}
        C_B(X, t, s) \leq C_B(X, s, s) + 2\sum_{Z: Z\sim X}\int_{s}^{t} \norm{h_Z(\tau)}\cdot C_B(Z, \tau, s)\diff \tau,
    \end{equation*}
    where $Z\sim X$ denotes $Z\cap X\neq \varnothing$. Note that $C_B(X, s, s) = 0$ if $\supp(B) \cap X = \varnothing$.
\end{lemma}
\begin{proof}
    Fix the initial time $s$ and $A\in\mathcal A_X$ with $\norm{A}\leq1$. We aim to bound the growth of $f(t) = [A(t), B]$.
    The time derivative of the Heisenberg operator is
    \begin{align*}
        \partial_t A(t) &= (\partial_t U(t,s))^\dagger A U(t, s) +  U(t,s)^\dagger A \partial_t U(t, s) \\
        &= (-\ii H(t) U(t, s))^\dagger A U(t, s) +  U(t,s)^\dagger A (-\ii H(t) U(t, s)) \\
        &= \ii U(t,s)^\dagger [H(t), A] U(t,s).
    \end{align*}
    We decompose the time-dependent Hamiltonian into terms overlapping with $X$ and those that do not:
    \begin{equation*}
        H(t) = I_X(t) + J_X(t), \quad \text{where}\quad  I_X(t) = \sum_{Z: Z\sim X} h_Z(t).
    \end{equation*}
    Since $A$ is supported on $X$, we have that $[J_X(t), A] = 0$ for all $t$.
    In the fermionic case, every term in $J_X(t)$ is even and supported away from $X$.
    By \Cref{prop:supercommute}, each such term commutes with both parity components of $A$, and therefore with $A$ itself.
    The same reasoning gives $C_B(X,s,s)=0$ when $B$ is supported away from $X$.
    Write $K(t) = U(t,s)^\dagger I_X(t) U(t,s)$.
    Then
    \begin{align*}
        \partial_t f(t) = [\partial_t A(t), B] &= \ii [U(t,s)^\dagger [I_X(t) , A]U(t,s), B] = \ii [[K(t), A(t)], B] .
    \end{align*}
    Using the Jacobi identity $[[X, Y], Z] = [X, [Y, Z]] - [Y, [X, Z]]$ with $X=K(t)$, $Y=A(t)$, and $Z=B$
    yields
    \begin{equation*}
        \partial_t f(t) = \ii [K(t), f(t)] - \ii [A(t), [K(t), B]].
    \end{equation*}
    We now work to remove the term $\ii [K(t), f(t)]$.
    Let $V(t)$ be the unitary solution of $\partial_t V(t) = -\ii V(t) K(t)$ and $V(s) = I$.
    Define $\widetilde{f}(t) = V(t)f(t)V(t)^\dagger$. Then
    \begin{align*}
        \partial_t \widetilde{f}(t) &= (\partial_t V(t))\cdot f(t)\cdot V(t)^\dagger + V(t) \cdot f(t)\cdot (\partial_t V(t)^\dagger) + V(t) \cdot (\partial_t f(t))\cdot V(t)^\dagger \\
        &= -\ii V(t) K(t) f(t) V(t)^\dagger + \ii V(t) f(t) K(t) V(t)^\dagger + V(t) \cdot (\partial_t f(t))\cdot V(t)^\dagger \\
        &= V(t) (-\ii [K(t), f(t)] ) V(t)^\dagger + V(t) (\ii [K(t), f(t)] - \ii [A(t), [K(t), B]]) V(t)^\dagger\\
        &= -\ii V(t) \cdot [A(t), [K(t), B]]\cdot V(t)^\dagger.
    \end{align*}
    Hence,
    \begin{align*}
        \norm{f(t)} = \norm{\widetilde{f}(t)} &\leq \norm{\widetilde{f}(s)} + \int_{s}^t \norm{V(\tau) \cdot [A(\tau), [K(\tau), B]]\cdot V(\tau)^\dagger} \diff \tau \\
        &\leq \norm{\widetilde{f}(s)} + \int_{s}^t \norm{ [A(\tau), [K(\tau), B]]} \diff \tau.
    \end{align*}
    Using the inequality $\norm{[P, Q]} \leq 2\norm{P} \cdot \norm{Q}$ and $\norm{A(\tau)} \leq 1$:
    \begin{align*}
       \norm{ [A(\tau), [K(\tau), B]]} &\leq 2 \norm{ [K(\tau), B] } \\
        &\leq 2 \sum_{Z: Z\sim X} \norm{ [U(\tau, s)^\dagger h_Z(\tau) U(\tau, s), B] } \\
        &\leq 2 \sum_{Z: Z\sim X} \norm{h_Z(\tau)} \cdot C_B(Z, \tau, s),
    \end{align*}
    where the last inequality follows from the definition of $C_B(Z, \tau, s)$. Substituting this back and then taking the supremum over $A$ completes the proof.
\end{proof}

\subsection{The Lieb--Robinson bound for time-dependent local Hamiltonians}

\begin{restatable}[Lieb--Robinson bound for time-dependent local Hamiltonians]{lemma}{lrhhkltimedependent}
\label{lem:LR_HHKL_time_dep}
\label{thm:time-dependent-fermionic-lr}
    Given a finite set of sites $\Lambda$,
    let $H(t) = \sum_{Z\subseteq \Lambda} h_Z(t)$ be a time-dependent local Hamiltonian where $h_Z(t)$ is Hermitian and supported on $Z$ for all $t$, and is even in the fermionic case.
    Suppose that the local interaction strength
    \begin{equation*}
        \zeta \coloneqq \sup_t \max_{i\in \Lambda} \sum_{Z\ni i} \abs{Z} \cdot \norm{h_Z(t)}
    \end{equation*}
    is finite.
    Let $X\subseteq\Omega\subseteq\Lambda$, let $O_X\in\mathcal A_X$, and let $\ell$ be an integer with $2\leq\ell\leq\dist(X,\Lambda\setminus\Omega)$.
    Restricting $H(t)$ to $\Omega$ gives $H_{\Omega}(t) = \sum_{Z \subseteq \Omega} h_Z(t)$.
    Let $U(t)$ and $U_{\Omega}(t)$ be the unitary evolutions of $H(t)$ and $H_{\Omega}(t)$ from time $0$ to time $t$.
    Then for any $T\in \R$,
    \begin{equation}\label{eq:time-dependent-truncation}
        \norm{U_{\Omega}(T)^\dagger O_X U_{\Omega}(T) - U(T)^\dagger O_X U(T)}\leq \norm{O_X}\cdot\abs{X}\cdot\frac{(2\zeta\abs{T})^\ell}{\ell!}.
    \end{equation}
    If $\dist(X,\Lambda\setminus\Omega)=\infty$, the two evolutions agree on $O_X$.
\end{restatable}

\begin{proof}
Without loss of generality, we assume $T\geq0$.
The case $T<0$ follows by time reversal, replacing $H(t)$ with $-H(-t)$.
Scalar terms do not affect Heisenberg evolution, so we omit them.

We first remove the onsite terms. Write
\[
V(t):=\sum_{i\in\Lambda}h_{\{i\}}(t),\qquad
Q(t):=U_V(t,0),\qquad
w_Z(t):=Q(t)^\dagger h_Z(t)Q(t)\quad (\abs{Z}\geq2).
\]
Since onsite terms at different sites commute, $Q(t)$ is a product of onsite evolutions.
In the fermionic case, this follows from their evenness and \Cref{prop:supercommute}.
Thus each $w_Z(t)$ is supported on $Z$, has norm $\norm{h_Z(t)}$, and remains even in the fermionic case.
Adding the onsite terms outside $\Omega$ to $H_\Omega(t)$ does not change its Heisenberg evolution on $O_X$.
Passing to the common rotating frame therefore reduces the desired bound to
\[
\norm{\tau_{T,0}^{\widetilde W}(O_T)-\tau_{T,0}^{\widetilde W_\Omega}(O_T)},
\qquad O_T:=Q(T)^\dagger O_X Q(T),
\]
where
\[
\widetilde W(t):=\sum_{\abs{Z}\geq2}w_Z(t),
\qquad
\widetilde W_\Omega(t):=\sum_{\substack{Z\subseteq\Omega\\\abs{Z}\geq2}}w_Z(t).
\]
Here $O_T$ is supported on $X$ and $\norm{O_T}=\norm{O_X}$.
All sums over $Z$ below are over sets with $\abs{Z}\geq2$ and $h_Z(t)\not\equiv0$.

To compare the two evolutions, introduce an auxiliary qubit $q$ and the Hamiltonian
\[
K(t):=\sum_{Z\subseteq\Omega}w_Z(t)\otimes I_q
+\sum_{Z\not\subseteq\Omega}w_Z(t)\otimes P_0,
\qquad P_0:=\ketbra{0}{0}_q.
\]
The two diagonal blocks of $K(t)$ are $\widetilde W(t)$ and $\widetilde W_\Omega(t)$.
With $B:=I\otimes\sigma_X^{(q)}$, this gives
\begin{equation}\label{eq:hhkl_target}
\norm{[\tau_{T,0}^K(O_T\otimes I_q),B]}
=\norm{\tau_{T,0}^{\widetilde W}(O_T)-\tau_{T,0}^{\widetilde W_\Omega}(O_T)}.
\end{equation}
The auxiliary qubit commutes with all operators of the original system.
In the fermionic case, parity acts trivially on this qubit, so every term of $K(t)$ and $B$ is even.
Consequently, the proof of \Cref{lem:generalize_eq_42} applies without change to this enlarged system.

Let $\widehat Z:=Z$ when $Z\subseteq\Omega$, and $\widehat Z:=Z\cup\{q\}$ otherwise.
A chain of overlapping supports of nonzero Hamiltonian terms starting at $X$ cannot reach $q$ in fewer than $\ell$ steps: the first support containing $q$ comes from a term whose physical support meets $\Lambda\setminus\Omega$.
Write $C_B^K$ for the quantity in \Cref{lem:generalize_eq_42} with Hamiltonian $K$.
By definition,
\[
\norm{[\tau_{T,0}^K(O_T\otimes I_q),B]}
\leq\norm{O_X}\,C_B^K(X,T,0).
\]
Iterating \Cref{lem:generalize_eq_42} $\ell$ times, all initial terms vanish before the last step, and hence
\begin{align*}
C_B^K(X,T,0)
&\leq 2^\ell\int_{\Delta_{0,T}^\ell}
\sum_{\substack{Z_1,\ldots,Z_\ell:\\X\sim\widehat Z_1\sim\cdots\sim\widehat Z_\ell}}
\prod_{j=1}^\ell\norm{w_{Z_j}(u_j)}\,
C_B^K(\widehat Z_\ell,u_\ell,0)\,\diff\bu\\
&\leq 2^{\ell+1}\int_{\Delta_{0,T}^\ell}
\sum_{\substack{Z_1,\ldots,Z_\ell:\\X\sim Z_1\sim\cdots\sim Z_\ell}}
\prod_{j=1}^\ell\norm{w_{Z_j}(u_j)}\,\diff\bu.
\end{align*}
Here $S\sim Z$ means $S\cap Z\neq\varnothing$, and
$\Delta_{0,T}^\ell:=\{(u_1,\ldots,u_\ell):0\leq u_\ell\leq\cdots\leq u_1\leq T\}$.
For the second inequality, we used $C_B^K\leq2\norm{B}=2$ and the fact that the first $\ell-1$ supports do not contain $q$, so all overlaps in these chains occur at physical sites.

It remains to bound the chain sum. Since every $Z$ has size at least two,
\[
\sum_{Z\ni i}\norm{w_Z(u)}\leq\frac\zeta2,
\qquad
\sum_{Z\ni i}\abs{Z}\cdot\norm{w_Z(u)}\leq\zeta.
\]
Summing over $Z_\ell$ first and then successively over $Z_{\ell-1},\ldots,Z_1$ gives
\[
\sum_{\substack{Z_1,\ldots,Z_\ell:\\X\sim Z_1\sim\cdots\sim Z_\ell}}
\prod_{j=1}^\ell\norm{w_{Z_j}(u_j)}
\leq\frac{\abs{X}\zeta^\ell}{2}.
\]
Using $\int_{\Delta_{0,T}^\ell}\diff\bu=T^\ell/\ell!$ and \eqref{eq:hhkl_target}, we obtain
\[
\norm{\tau_{T,0}^{\widetilde W}(O_T)-\tau_{T,0}^{\widetilde W_\Omega}(O_T)}
\leq\norm{O_X}\cdot\abs{X}\cdot\frac{(2\zeta T)^\ell}{\ell!}.
\]
If the distance is infinite, the bound holds for every integer $\ell\geq2$; letting $\ell\to\infty$ shows that the difference vanishes.
\end{proof}

\subsection{Field-independent bounds}

We are now ready to complete the proofs of \Cref{thm:lr-potential,thm:field-independent-fermionic}.
We use the notation $W=\sum_Zw_Z$ for both systems; in the spin case, $w_{\supp(H_a)}=H_a$.
The following argument compares any two nested interaction truncations, so taking the outer region to be the full set of sites gives the fermionic statement.
\begin{proof}[Proof of \Cref{thm:lr-potential,thm:field-independent-fermionic}]
    We first remove the effect of the unbounded onsite potentials $V$ by considering the unitary evolution of a time-dependent Hamiltonian instead.
    Notice that $U(t) \coloneqq e^{\ii t V} e^{-\ii t H}$ is the unitary time evolution of the time-dependent local Hamiltonian
    \begin{equation*}
        W(t) \coloneqq e^{\ii t V}\cdot W\cdot e^{-\ii t V} = \sum_Z \underbrace{e^{\ii t V} \cdot w_Z\cdot e^{-\ii t V}}_{\coloneqq w_Z(t)}.
    \end{equation*}
    This is because $U(0) = I$ and $\ii \partial_t U(t) = W(t) U(t)$, which can be verified by direct computation:
    \begin{align*}
        \partial_t U(t) &= \ii V e^{\ii t V} e^{-\ii t H} + e^{\ii t V} (-\ii H) e^{-\ii t H} \\*
        &= -\ii e^{\ii t V} W e^{-\ii t H} \\*
        &= -\ii W(t) \cdot U(t).
    \end{align*}
    Since $e^{\ii tV}$ is unitary, we have $\norm{w_Z(t)} = \norm{w_Z}$ for all $t$. Since $V$ is a sum of onsite terms, $w_Z(t)$ remains supported on $Z$.
    In the fermionic case, each $v_i$ is even and supported on $\{i\}$, so \Cref{prop:supercommute} gives $[v_i,v_j]=0$ for $i\neq j$.
    Thus $e^{\ii tV}=\prod_i e^{\ii t v_i}$.
    For any operator supported on $S$, the factors with $i\notin S$ commute with it by \Cref{prop:supercommute}, after decomposing the operator into its parity components if necessary; hence conjugation by $e^{\ii tV}$ preserves its support and norm, as well as its parity when it is parity homogeneous.
    In the fermionic case, every $w_Z(t)$ remains even; in both cases, the coupling strength remains $\zeta$.

    Similarly, let us define $U_{\Omega}(t) \coloneqq e^{\ii t V} e^{-\ii t H_{\Omega}^{\mathrm{int}}}$ and $U_{\Lambda}(t) \coloneqq e^{\ii t V} e^{-\ii t H_{\Lambda}^{\mathrm{int}}}$, which are, respectively, the unitary evolutions of the time-dependent local Hamiltonians
    \begin{equation*}
        W_\Omega(t) \coloneqq \sum_{Z\subseteq\Omega}w_Z(t), \quad \text{and} \quad W_\Lambda(t) \coloneqq \sum_{Z\subseteq\Lambda}w_Z(t).
    \end{equation*}
    Let $O_X(t) \coloneqq e^{\ii t V}\cdot O_X \cdot e^{-\ii t V}$. Then $\norm{O_X(t)} = \norm{O_X}$, $O_X(t)$ is supported on $X$, and
\begin{align*}
    &e^{\ii H_\Lambda^{\mathrm{int}} T}\cdot O_X \cdot e^{-\ii H_\Lambda^{\mathrm{int}} T} - e^{\ii H_\Omega^{\mathrm{int}} T} \cdot O_X \cdot e^{-\ii H_\Omega^{\mathrm{int}} T} \\
    &\qquad = U_{\Lambda}(T)^\dagger \cdot O_X(T) \cdot U_{\Lambda}(T) - U_{\Omega}(T)^\dagger \cdot O_X(T) \cdot U_{\Omega}(T),
\end{align*}
which is supported on $\Lambda$.
Removing terms from $W(t)$ can only increase hypergraph distances, and the local interaction strength of $W_\Lambda(t)$ is at most $\zeta$.
Thus \Cref{lem:LR_HHKL_time_dep} bounds the operator norm of this difference by $\norm{O_X}\cdot\abs{X}\cdot(2\zeta\abs{T})^\ell/\ell!$.
The same argument gives zero difference when $\ell=\infty$.
\end{proof}

\subsection{Shell decomposition for Heisenberg evolution}

\begin{lemma}[Shell decomposition for Heisenberg evolution]
\label{lem:fermionic-shell}
Let $H=V+W$ be a local spin or fermionic Hamiltonian with onsite terms and $\zeta$ be the associated local interaction strength as defined in \Cref{def:interaction-parameters}.
For any $i\in\Lambda$ and any integer $r\geq0$, use the interaction-truncated Hamiltonian $H_{i,r}$ from \Cref{def:interaction-truncated-Ham}.
For any $A\in \mathcal A_{\{i\}}$ and $t\in\R$, define
\[
F_{i,0}^{A,t}:=\tau_t^{H_{i,0}}(A),
\qquad
F_{i,r}^{A,t}:=\tau_t^{H_{i,r}}(A)-\tau_t^{H_{i,r-1}}(A)
\quad \text{for } r\ge 1.
\]
Then
\[
\tau_t^H(A)=\sum_{r\ge0}F_{i,r}^{A,t},
\]
with only finitely many nonzero terms. In the fermionic case, if $A$ is homogeneous, each nonzero shell has the same parity as $A$.
Moreover, $F_{i,0}^{A,t}=e^{\ii t v_i}A e^{-\ii t v_i}$ is supported on $\{i\}$ and $\norm{F_{i,0}^{A,t}}=\norm{A}$.
For every integer $r\ge 1$, $F_{i,r}^{A,t}$ is supported on $\ball(\{i\},r)$ and
\[
\norm{F_{i,r}^{A,t}}
\le
\norm{A}\,\frac{(2\zeta \abs{t})^r}{r!}.
\]
\end{lemma}

\begin{proof}
This also proves \Cref{lem:LR_radius_r} for spin systems, with $H_r=H_{i,r}$.
As in the rotating-frame proof above, the Heisenberg evolution under $H_{i,r}$ preserves support inside $\ball(\{i\},r)$, and the evolution under $H_{i,r-1}$ preserves support inside $\ball(\{i\},r-1)$.
Hence $F_{i,r}^{A,t}$ is supported on $\ball(\{i\},r)$.

The decomposition follows by telescoping: once the balls reach the connected component of $i$, the truncated and full evolutions agree on $A$, and all further shells vanish. In the fermionic case, each $H_{i,r}$ is even, so its evolution preserves parity.
For $r=0$, the claimed expression and norm follow from $H_{i,0}=V$.

For $r=1$ and $t\ge0$, Duhamel's formula and $H_{i,0}=V$ give
\[
F_{i,1}^{A,t}
=\ii\int_0^t\tau_{t-s}^{H_{i,1}}
\bigl([H_{i,1}-V,\tau_s^V(A)]\bigr)\,\diff s.
\]
Since $\tau_s^V(A)$ is supported on $\{i\}$, only terms with $Z\ni i$ contribute to the commutator; in the fermionic case this follows from the evenness of $w_Z$, \Cref{prop:supercommute}, and linearity. Thus
\[
\norm{F_{i,1}^{A,t}}
\le
\int_0^t\norm{[H_{i,1}-V,\tau_s^V(A)]}\,\diff s \le2t\norm{A}\sum_{Z\ni i}\norm{w_Z}
\le
2\zeta t\norm{A}.
\]
The case \(t<0\) follows in the same way with \(\abs{t}\) in place of \(t\).

If $r\ge 2$, apply the field-independent bound for nested regions proved above with
\[
X=\{i\},
\qquad
\Omega=\ball(\{i\},r-1), \qquad \Lambda = \ball(\{i\},r).
\]
If $\Omega = \Lambda$, then $F_{i,r}^{A,t} = 0$; otherwise, we have $\dist\Paren{X,\Lambda\setminus \Omega}=r$ and then
\begin{align*}
    \norm{F_{i,r}^{A,t}}
=
\norm{\tau_t^{H_{i,r}}(A)-\tau_t^{H_{i,r-1}}(A)}
\le
\norm{A}\,\frac{(2\zeta \abs{t})^r}{r!}. & \qedhere
\end{align*}
\end{proof}

%% file: appendix_qubit.tex
\section{Moments of the kernel functions}\label{subsec:kernels_moments}
In this appendix, we bound the moments of the kernel functions $b_{1,j}(t)$ and $b_{2,j}(t)$ from \Cref{def:general_lindblad}.
For simplicity, we will drop the site index $j$ and write them as
\begin{align*}
    b_1(t) &= \frac{\sigma}{\pi\beta} e^{\beta^2\sigma^2 / 8} \Paren{\sech\left(\frac{2\pi t}{\beta}\right) *_t \sin(-\beta \sigma^2 t) e^{-2\sigma^2 t^2}} \,, \\
    b_{2}(t) &= 2\eta \cdot  \exp\Paren{-4 \Delta t^2 /\beta - 2\ii \Delta t} \,.
\end{align*}
These estimates are used to prove quasi-locality of the coherent terms and to control the kernel normalization factors in the simulation analysis.
We first bound the moments of the second kernel $b_2(t)$, which simply follow from the Gaussian moments.
\begin{lemma}[Moments of the second kernel]\label{lem:b_2_moments}
For any integer $r\geq0$,
\begin{equation*}
    \int_{-\infty}^{\infty} \abs{b_2(t)} \cdot \abs{t}^r \; \diff t = 2^{-r}\cdot \eta \cdot \Paren{\frac{\beta}{\Delta}}^{\frac{r+1}{2}} \cdot \Gamma\Paren{\frac{r+1}{2}}  \,.
\end{equation*}
\end{lemma}
\begin{proof}
Observe $\abs{e^{x+\ii y}} = e^x$, and thus the expression reduces to evaluating a Gaussian moment:
\begin{equation*}
\begin{split}
    \int_{-\infty}^{\infty} \abs{b_2(t)} \cdot \abs{t}^{r} \; \diff t & = 2\eta \int_{-\infty}^{\infty} \abs{t}^{r} \cdot \exp(-4 \Delta t^2/\beta)  \diff t \\
    & = 4\eta \int_{0}^{\infty} t^{r} \cdot \exp(-4 \Delta t^2/\beta)  \diff t \\
    & = 4\eta \left(\frac12 \Paren{\frac{4\Delta}{\beta}}^{-\frac{r+1}{2}}\cdot \Gamma\Paren{\frac{r+1}{2}}\right)\\
    &=  2^{-r}\cdot \eta \cdot \Paren{\frac{\beta}{\Delta}}^{\frac{r+1}{2}} \cdot \Gamma\Paren{\frac{r+1}{2}} \,. \qedhere
\end{split}
\end{equation*}
\end{proof}

The moments for the first kernel $b_1(t)$ is more complicated. We set up the following lemmas to make the bound easier.
\begin{lemma}[Fourier transform of the first kernel]\label{lem:b_1_fourier}
Let $\widehat{b}_1(\omega) = \int_{-\infty}^{\infty} b_1(t) e^{-\ii \omega t}\diff t$ denote the Fourier transform of $b_1(t)$.
Then
\begin{equation*}
    \widehat{b}_1(\omega) = \frac{\ii}{2\sqrt{2\pi}} e^{-\omega^2/ (8\sigma^2)} \tanh\Paren{\frac{\beta \omega}{4}} \,.
\end{equation*}
\end{lemma}
\begin{proof}
    Recall, the Fourier transform of a function $f$ is defined as follows:
\begin{equation*}
    f(t)= \frac{1}{2\pi} \int_{-\infty}^{\infty} \widehat{f}(\omega) e^{\ii \omega t} \diff \omega \, \quad \textrm{where } \; \; \widehat{f}(\omega) = \int_{-\infty}^{\infty} f(t) e^{-\ii \omega t} \diff t\, .
\end{equation*}
For the convolution of two function $g_1(t), g_2(t)$, we have that $\widehat{g_1 * g_2} = \widehat{g}_1 \cdot \widehat{g}_2$. Setting $g_1(t) = \text{sech}(2\pi t/\beta)$ and $g_2(t) = - \sin(\beta\sigma^2 t) \; e^{-2\sigma^2 t^2}$, it suffices to compute the individual Fourier coefficients of these two functions. Using the substitution $u= 2\pi t/ \beta$,
\begin{equation*}
\begin{split}
    \widehat{g}_1(\omega) & = \int_{-\infty}^{\infty} \textrm{sech}\Paren{ \frac{2\pi t}{ \beta}} \;  e^{-\ii \omega t} \diff t
     = \Paren{\frac{\beta}{2\pi}} \int_{-\infty}^{\infty} \textrm{sech}(u) \; e^{-\ii \omega \beta u / (2\pi)} \; \diff u \\
     & = \Paren{\frac{\beta}{2\pi}} \cdot \pi \cdot \textrm{sech}\Paren{ \frac{\pi}{2} \cdot \frac{\beta \omega}{2\pi} } = \Paren{\frac{\beta}{2}} \cdot \textrm{sech}\Paren{\frac{\beta \omega}{4} } \,,
\end{split}
\end{equation*}
where the third equality uses that $\int_{-\infty}^{\infty} \textrm{sech}(u) e^{-\ii \zeta u} \diff u = \pi \textrm{sech}\Paren{ \frac{\pi \zeta}{2} }$. Next, observe
\begin{equation*}
    g_2(t) = - \Paren{ \frac{e^{\ii \beta \sigma^2 t } - e^{-\ii \beta \sigma^2 t}}{2\ii} } \Paren{ e^{-2\sigma^2 t^2} } = \Paren{ \frac{1}{2\ii} } \Paren{ e^{-\ii \beta \sigma^2 t} \cdot e^{-2\sigma^2 t^2}  -  e^{\ii \beta \sigma^2 t} \cdot e^{-2\sigma^2 t^2} }\,.
\end{equation*}
Therefore,
\begin{equation}
\label{eqn:fourier-coeff-g2}
    \begin{split}
        \widehat{g}_2(\omega) & = \Paren{\frac{1}{2\ii}}\Bigg( \underbrace{ \int_{-\infty}^{\infty } e^{-\ii(\omega + \beta \sigma^2) t} \cdot e^{-2\sigma^2 t^2} \diff t }_{\calI(\omega + \beta \sigma^2)}   - \underbrace{ \int_{-\infty}^{\infty }  e^{\ii (\omega-\beta \sigma^2) t} \cdot e^{-2\sigma^2 t^2} \diff t}_{\calI(\omega - \beta\sigma^2)} \Bigg)
    \end{split}
\end{equation}
We can evaluate $\calI(\omega + \beta\sigma^2)$ as follows:
\begin{equation*}
\begin{split}
     \calI(\omega + \beta\sigma^2) & = \int_{-\infty}^{\infty }  \exp\Paren{ -2\sigma^2 t^2 - \ii (\omega+\beta\sigma^2) t} \diff t \\
     & = \int_{-\infty}^{\infty } \exp\Paren{ -2\sigma^2 \Paren{ t + \frac{\ii (\omega + \beta \sigma^2)}{4\sigma^2}}^2 - \frac{(\omega + \beta \sigma^2)^2}{8\sigma^2} } \\
     & = \exp\Paren{- \frac{(\omega + \beta \sigma^2)^2}{8\sigma^2}} \int_{-\infty}^{\infty} \exp\Paren{ -2\sigma^2 \Paren{ t + \frac{\ii (\omega + \beta \sigma^2)}{4\sigma^2}}^2} \\
     & =  \exp\Paren{- \frac{(\omega + \beta \sigma^2)^2}{8\sigma^2}} \cdot \sqrt{\frac{\pi}{2\sigma^2}} \,.
\end{split}
\end{equation*}
Performing a similar caluclation and substitution back into \Cref{eqn:fourier-coeff-g2}, we have
\begin{equation*}
    \widehat{g}_2(\omega) = \Paren{\frac{1}{2\ii}} \Paren{\sqrt{\frac{\pi}{2\sigma^2}} } \Paren{ \exp\Paren{- \frac{(\omega + \beta \sigma^2)^2}{8\sigma^2}} - \exp\Paren{- \frac{(\omega - \beta \sigma^2)^2}{8\sigma^2}} }
\end{equation*}
Next observe $\frac{\beta^2 \sigma^2 }{8} - \frac{(\omega + \beta \sigma^2)^2}{8\sigma^2} = -\frac{\omega^2}{8} -\frac{\beta \omega}{4}$ and $\frac{\beta^2 \sigma^2 }{8} - \frac{(\omega - \beta \sigma^2)^2}{8\sigma^2} = -\frac{\omega^2}{8} +\frac{\beta \omega}{4}$, and thus
\begin{equation*}
\begin{split}
    \exp\Paren{\frac{\beta^2 \sigma^2}{8}} \widehat{g}_2(\omega) & =  \Paren{\frac{1}{2\ii}} \Paren{\sqrt{\frac{\pi}{2\sigma^2}} }  \Paren{ \exp\Paren{-\frac{\omega^2}{8\sigma^2}} } \Paren{ \exp\Paren{-\frac{\beta\omega}{4} } -  \exp\Paren{\frac{\beta\omega}{4} } } \\
    & = \ii \cdot \Paren{\sqrt{\frac{\pi}{2\sigma^2}} } \Paren{ \exp\Paren{-\frac{\omega^2}{8\sigma^2}} } \sinh\Paren{ \frac{\beta \omega}{4}} \, .
\end{split}
\end{equation*}
Therefore, we have
\begin{equation*}
\begin{split}
    \widehat{b}_1(\omega) & =  \Paren{\frac{\sigma \; \ii }{2\pi \beta}} \Paren{\frac{\beta}{2}} \cdot \textrm{sech}\Paren{\frac{\beta \omega}{4} }   \Paren{\sqrt{\frac{\pi}{2\sigma^2}} }   \Paren{ \exp\Paren{-\frac{\omega^2}{8\sigma^2}} } \sinh\Paren{ \frac{\beta \omega}{4}}   \\
    & = \Paren{ \frac{\ii}{2\sqrt{2\pi}} } \Paren{ \exp\Paren{-\frac{\omega^2}{8\sigma^2}} } \; \tanh\Paren{ \frac{\beta \omega }{4} } \,.
\end{split}
\end{equation*}
This completes the proof.
\end{proof}

Next, we provide an equivalent formulation of the first kernel that is more amenable to analytic moment bounds.

\begin{lemma}[Equivalent expression of the first kernel]
    \begin{equation*}
        b_1(t) = \frac{\sigma}{\pi \beta }\int_{0}^{\infty} \frac{e^{-2\sigma^2(t+u)^2} - e^{-2\sigma^2(t-u)^2}}{\sinh(2\pi u/\beta)} \diff u.
    \end{equation*}
\end{lemma}
\begin{proof}
    Let us use the following integral equality for any $a>0$ and $\omega \in \R$, which can be derived by countour integration:
    \begin{equation}\label{eq:tanh_integral_equality}
        \int_{0}^{\infty} \frac{\sin(\omega u)}{\sinh(a u)} \diff u = \frac{\pi}{2a} \tanh\Paren{\frac{\pi \omega}{2a}}.
    \end{equation}
    Then
    \begin{align*}
        b_1(t) &= \frac{1}{2\pi} \int_{-\infty}^{\infty} \widehat{b}_1(\omega) e^{\ii \omega t} \diff \omega \\
        &= \frac{1}{2\pi} \int_{-\infty}^{\infty} \frac{\ii}{2\sqrt{2\pi}} e^{-\omega^2/ (8\sigma^2)} \tanh\Paren{\frac{\beta \omega}{4}}\cdot e^{\ii \omega t} \diff \omega \tag{\Cref{lem:b_1_fourier}}\\
        &= \frac{1}{2\pi} \int_{-\infty}^{\infty} \frac{\ii}{2\sqrt{2\pi}} e^{-\omega^2/ (8\sigma^2)} \Paren{\frac{4}{\beta} \int_{0}^{\infty} \frac{\sin(\omega u)}{\sinh(2\pi u/\beta)} \diff u}\cdot e^{\ii \omega t} \diff \omega \tag{\Cref{eq:tanh_integral_equality} with $a=2\pi / \beta$}\\
        &= \frac{2\ii}{\beta (2\pi)^{3/2}} \int_{0}^{\infty} \frac{1}{\sinh(2\pi u/\beta)} \Paren{\int_{-\infty}^{\infty} e^{-\omega^2/ (8\sigma^2)} \sin(\omega u) e^{\ii \omega t} \diff \omega}\cdot  \diff u \,,
    \end{align*}
    where in the last equality, we switch the order of the two integrals over $\diff u$ and  $\diff \omega$.
    This requires Fubini's theorem, which follows from the absolute integrability of \begin{equation*}
        g(u, \omega) \coloneqq e^{-\omega^2/(8\sigma^2)} \frac{\sin(\omega u)}{\sinh(2\pi u/\beta)} e^{\ii \omega t} \,,
    \end{equation*}
    i.e.
    \begin{align*}
        &\int_{-\infty}^{\infty} \int_{0}^{\infty} \abs*{g(u, \omega)} \diff u \diff \omega\\
        &= \int_{0}^{\infty}\int_{-\infty}^{\infty}  \abs*{g(u, \omega)}  \diff \omega \diff u \tag{Tonelli's theorem}\\
        &= \int_{0}^{\infty} \frac{1}{\sinh(2\pi u/\beta)}\int_{-\infty}^{\infty}  e^{-\omega^2/(8\sigma^2)} \abs*{\sin(\omega u) }  \diff \omega \diff u \\
        &\leq \int_{0}^{1} \frac{1}{\sinh(2\pi u/\beta)} \int_{-\infty}^{\infty}  e^{-\omega^2/(8\sigma^2)} \cdot \abs{\omega} u  \diff \omega \diff u  + \int_{1}^{\infty} \frac{1}{\sinh(2\pi u/\beta)} \int_{-\infty}^{\infty}  e^{-\omega^2/(8\sigma^2)}  \diff \omega \diff u \tag{Split the integral at $u=1$ and use $\abs{\sin(\omega u)} \leq \min \{\abs{\omega} u, 1\}$}\\
        &\leq 8\sigma^2 \int_{0}^{1} \frac{u}{\sinh(2\pi u/\beta)} \diff u + \sqrt{8\pi} \sigma\int_{1}^{\infty} \frac{1}{\sinh(2\pi u/\beta)} \diff u \tag{Gaussian moments}\\
        &\leq 8\sigma^2 \int_{0}^{1} \frac{\beta}{2\pi } \diff u + \sqrt{8\pi} \sigma\int_{1}^{\infty} 2e^{-2\pi u / \beta} \diff u \tag{$1/\sinh(x)\leq 1/x$ and $1/\sinh(x)\leq 2e^{-x}$ for $x>0$} \\
        &= 8\sigma^2 \frac{\beta}{2\pi } + \sqrt{8\pi} \sigma \frac{\beta}{\pi } e^{-2\pi/\beta} < \infty.
    \end{align*}
    Now for the inner integral,
    \begin{align*}
        \int_{-\infty}^{\infty} e^{-\omega^2/ (8\sigma^2)} \sin(\omega u) e^{\ii \omega t} \diff \omega
        &= \frac{1}{2\ii}\int_{-\infty}^{\infty} e^{-\omega^2/ (8\sigma^2)} (e^{\ii\omega u} - e^{-\ii \omega u}) e^{\ii \omega t} \diff \omega \tag{$\sin(\omega u) = \frac{1}{2\ii} (e^{\ii\omega u} - e^{-\ii \omega u})$} \\
        &= \frac{1}{2\ii}\sqrt{2\pi}\cdot 2\sigma \Paren{e^{-2\sigma^2 (t+u)^2} - e^{-2\sigma^2 (t-u)^2}} \,,
    \end{align*}
    where the last equality is because for any $k>0$, the Fourier transform of $e^{-k\omega^2}$ is given by $\int_{-\infty}^{\infty} e^{-k\omega^2} e^{-\ii\omega t} \diff \omega = \sqrt{\pi / k} \cdot e^{-t^2/(4k)}$ and we set $k = 1/(8\sigma^2)$.
    Plugging this back to $b_1(t)$ completes the proof.
\end{proof}

We are now ready to bound the moments of $b_1(t)$, starting with the $0$-th moment:
\begin{lemma}[$L^1$-norm of the first kernel]\label{lem:b_1_zero_moment}
\begin{equation*}
    \int_{-\infty}^{\infty} \abs{b_1(t)} \diff t \leq \frac{1}{\sqrt{2} \pi^{3/2}} \Paren{\frac{2}{\sqrt{\pi}} + \log \Paren{1+\frac{\sqrt{2}}{\pi}\cdot \sigma \beta}}\, .
\end{equation*}
\end{lemma}
\begin{proof}
    Note that $b_1(t) = -b_1(-t)$.
    Since $u\geq 0$, when $t>0$, we always have $(t+u)^2 \geq (t-u)^2$ and thus $e^{-2\sigma^2(t+u)^2} - e^{-2\sigma^2(t-u)^2} \leq 0$. So
    \begin{align*}
        \int_{-\infty}^{\infty} \abs{b_1(t)} \diff t
        &= 2 \int_{0}^{\infty} \abs{b_1(t)} \diff t \\
        &= \frac{2\sigma}{\pi \beta }\iint_{0}^{\infty} \frac{ e^{-2\sigma^2(t-u)^2}- e^{-2\sigma^2(t+u)^2}}{\sinh(2\pi u/\beta)}  \diff u \diff t \\
        &= \frac{2\sigma}{\pi \beta }\int_{0}^{\infty} \frac{ \int_{0}^{\infty}  \Paren{e^{-2\sigma^2(t-u)^2}- e^{-2\sigma^2(t+u)^2}} \diff t}{\sinh(2\pi u/\beta)} \diff u \tag{Tonelli's theorem}.
    \end{align*}
    By the definition of the Gauss error function $\mathrm{erf}(z) = \frac{2}{\sqrt{\pi}} \int_{0}^z e^{-t^2} \diff t$, we have
    \begin{align*}
        \int_{0}^{\infty} e^{-2\sigma^2(t-u)^2}- e^{-2\sigma^2(t+u)^2} \diff t &= 2\int_{0}^{u} e^{-2\sigma^2 t^2} \diff t \\
        &= \frac{\sqrt{2}}{ \sigma} \int_{0}^{\sqrt{2}\sigma u} e^{-s^2} \diff s \tag{Set $s = \sqrt{2}\sigma t$}\\
        &= \frac{\sqrt{\pi}}{\sqrt{2} \cdot \sigma}\mathrm{erf}(\sqrt{2}\sigma u).
    \end{align*}
    So
    \begin{align*}
        \int_{-\infty}^{\infty} \abs{b_1(t)} \diff t
        &= \frac{\sqrt{2}}{\sqrt{\pi}\cdot \beta }\int_{0}^{\infty} \frac{ \mathrm{erf}(\sqrt{2}\sigma u)}{\sinh(2\pi u/\beta)} \diff u \\
        &= \frac{1}{\sqrt{2\pi}\cdot \pi }\int_{0}^{\infty} \frac{ \mathrm{erf}(ax)}{\sinh(x)} \diff x \tag{Set $x = 2\pi u /\beta$ and $a=\frac{\sigma \beta}{\sqrt{2} \cdot \pi}$}.
    \end{align*}
    Since for any $z\geq 0$, we have
    \begin{align*}
        \mathrm{erf}(z)= \frac{2}{\sqrt{\pi}} \int_{0}^z e^{-t^2} \diff t \leq \frac{2}{\sqrt{\pi}} \int_{0}^z 1 \diff t \leq \frac{2z}{\sqrt{\pi}}
    \end{align*}
    and $\mathrm{erf}(z)\leq \frac{2}{\sqrt{\pi}} \int_{0}^{\infty} e^{-t^2} \diff t = 1$. We can split the integral at $x=1/a$:
    \begin{align*}
        \int_{0}^{\infty} \frac{ \mathrm{erf}(ax)}{\sinh(x)} \diff x &= \int_{0}^{1/a} \frac{ \mathrm{erf}(ax)}{\sinh(x)} \diff x + \int_{1/a}^{\infty} \frac{ \mathrm{erf}(ax)}{\sinh(x)} \diff x \\
        &\leq \frac{2a}{\sqrt{\pi}}\int_{0}^{1/a} \frac{ x}{\sinh(x)} \diff x + \int_{1/a}^{\infty} \frac{ 1}{\sinh(x)} \diff x \\
        &\leq \frac{2a}{\sqrt{\pi}}\int_{0}^{1/a} 1 \diff x \tag{$\frac{x}{\sinh(x)}\leq 1$ for $x>0$} + \int_{1/a}^{\infty} \frac{ 1}{\sinh(x)} \diff x \\
        &= \frac{2}{\sqrt{\pi}} + \log \coth\Paren{\frac{1}{2a}}.
    \end{align*}
    Since $\coth(y) = \frac{\cosh(y)}{\sinh(y)} = 1 + \frac{2}{e^{2y} - 1}$ and $e^{2y} - 1 \geq 2y$ for any $y>0$, we have
    \begin{align*}
        \log \coth\Paren{\frac{1}{2a}}  \leq \log (1+2a)      .
    \end{align*}
    So overall,
    \begin{align*}
        \int_{-\infty}^{\infty} \abs{b_1(t)} \diff t \leq \frac{1}{\sqrt{2} \pi^{3/2}} \Paren{\frac{2}{\sqrt{\pi}} + \log \Paren{1+ \frac{\sqrt{2}}{\pi}\cdot \sigma \beta}}.
    \end{align*}
    This completes the proof.
\end{proof}

A similar argument enables us to bound higher moments of this kernel function.

\begin{lemma}[Higher moments of the first kernel]\label{lem:b_1_higher_moments}
For any integer $r\geq 1$,
\begin{equation*}
    \int_{-\infty}^{\infty} \abs{b_1(t)} \cdot \abs{t}^r\; \diff t \leq \frac{2^{(r+5)/2} \cdot r!}{ \pi^{3/2}} \Paren{ \frac{1}{\sigma}
        +  \frac{\beta}{\pi} }^r\, .
\end{equation*}
\end{lemma}
\begin{proof}
    Let us write $g(t) = \sigma \sqrt{\frac{2}{\pi}}\cdot e^{-2\sigma^2 t^2}$. Then
    \begin{align*}
        \int_{-\infty}^{\infty} \abs{b_1(t)} \cdot \abs{t}^r \diff t
        &= \frac{\sigma}{\pi \beta }\int_{0}^{\infty} \frac{ \int_{-\infty}^{\infty} \abs{t}^r\cdot \abs*{e^{-2\sigma^2(t-u)^2}- e^{-2\sigma^2(t+u)^2}} \diff t}{\sinh(2\pi u/\beta)} \diff u \tag{Tonelli's theorem} \\
        &= \frac{1}{\sqrt{2\pi}} \frac{1}{\beta} \int_{0}^{\infty} \frac{ \int_{-\infty}^{\infty} \abs{t}^r\cdot \abs*{g(t-u)- g(t+u)} \diff t}{\sinh(2\pi u/\beta)} \diff u \\
        &= \frac{1}{\sqrt{2\pi}} \frac{1}{\beta} \int_{0}^{\infty} \frac{ D_r(u)}{\sinh(2\pi u/\beta)} \diff u \,,
    \end{align*}
    where $D_r(u) \coloneqq \int_{-\infty}^{\infty} \abs{t}^r \cdot \abs*{g(t-u)- g(t+u)} \diff t$.
    We split the integral over $u$ at $u=1/\sigma$:
    \begin{align}\label{eq:b_1_moments_split}
        \int_{-\infty}^{\infty} \abs{b_1(t)} \cdot \abs{t}^r \diff t
        &= \frac{1}{\sqrt{2\pi}} \frac{1}{\beta} \Bigg(\underbrace{\int_{0}^{1/\sigma} \frac{ D_r(u)}{\sinh(2\pi u/\beta)} \diff u }_{  \eqref{eq:b_1_moments_split}.(1)} + \underbrace{\int_{1/\sigma}^{\infty} \frac{ D_r(u)}{\sinh(2\pi u/\beta)} \diff u}_{  \eqref{eq:b_1_moments_split}.(2)}\Bigg).
    \end{align}
    Let us first bound \eqref{eq:b_1_moments_split}.(2).
    Since $g(t)>0$ for all $t\in \R$, we have
    \begin{align*}
        D_r(u) &\leq \int_{-\infty}^{\infty} \abs{t}^r \cdot g(t-u)\diff t + \int_{-\infty}^{\infty} \abs{t}^r \cdot g(t+u) \diff t \tag{triangle inequality} \\
        &= \int_{-\infty}^{\infty} \abs{t + u}^r \cdot g(t) \diff t + \int_{-\infty}^{\infty} \abs{t-u}^r \cdot g(t) \diff t \\
        &\leq 2^{r} \int_{-\infty}^{\infty} \Paren{\abs{t}^r + u^r} \cdot g(t)\diff t \tag{$\abs{t+u}^r\leq 2^{r-1} (\abs{t}^r + \abs{u}^r)$ and $u\geq 0$} \\
        &= 2^{r} (m_r + m_0 u^r) \tag{Set $m_k \coloneqq \int_{-\infty}^{\infty} \abs{t}^k \cdot g(t) \diff t$}.
    \end{align*}
    Therefore,
    \begin{align*}
        \eqref{eq:b_1_moments_split}.(2) &\leq \int_{1/\sigma}^{\infty} \frac{ 2^{r} (m_r + m_0 u^r)}{\sinh(2\pi u/\beta)} \diff u \\
        &\leq 2^{r} \Paren{ \frac{m_r \beta}{2\pi}\int_{\frac{2\pi}{\sigma \beta}}^{\infty}\frac{ 1}{\sinh(x)} \diff x + m_0\Paren{\frac{\beta}{2\pi}}^{r+1}\int_{\frac{2\pi}{\sigma \beta}}^{\infty}\frac{ x^r}{\sinh(x)} \diff x } \tag{Set $x = 2\pi u/ \beta$}
    \end{align*}
    Standard Gaussian moment calculations show that $m_0 =1$ and $m_r = \pi^{-1/2}2^{-r/2} \sigma^{-r}\cdot \Gamma\Paren{\frac{r+1}{2}}$.
    Since for any $a>0$,
    \begin{align*}
        \int_{a}^{\infty} \frac{1}{\sinh(x)} \diff x = \log \coth(a/2)\leq \log(1+2/a) \,,
    \end{align*}
    and for any integer $r\geq 1$,
    \begin{align*}
        \int_{0}^{\infty} \frac{x^r}{\sinh(x)} \diff x = 2r! \sum_{k=0}^{\infty} \frac{1}{(2k+1)^{r+1}} < 4 r! \,,
    \end{align*}
    then we have
    \begin{align*}
        \eqref{eq:b_1_moments_split}.(2) &\leq 2^{r} \Paren{ \frac{2^{-r/2} \sigma^{-r}\cdot r! \cdot \beta}{2\pi} \log\Paren{1+\frac{\sigma\beta}{\pi}}  + \Paren{\frac{\beta}{2\pi}}^{r+1} 4r!} \tag{$\Gamma(\frac{r+1}{2})\leq \sqrt{\pi} \cdot r!$} \\
        &= \frac{\beta}{\pi} r! \Paren{ 2^{r/2-1} \sigma^{-r} \log\Paren{1+\frac{\sigma\beta}{\pi}}  + 2\Paren{\frac{\beta}{\pi}}^{r} } \,.
    \end{align*}

    We now bound \eqref{eq:b_1_moments_split}.(1).
    Since $g(t+u)- g(t-u) = \int_{-u}^{u} g'(t+s) \diff s$,
    then
    \begin{align*}
        D_r(u) &\leq \int_{-\infty}^{\infty} \abs{t}^r \Paren{\int_{-u}^{u} \abs*{g'(t+s)} \diff s} \diff t \\
        &= \int_{-u}^{u} \int_{-\infty}^{\infty} \abs{t}^r \cdot \abs*{g'(t+s)}  \diff t \diff s \\
        &= \int_{-u}^{u} \int_{-\infty}^{\infty} \abs{w-s}^r \cdot \abs*{g'(w)}  \diff w \diff s \tag{Set $w=t+s$} \\
        &\leq \int_{-u}^{u} \int_{-\infty}^{\infty} 2^{r-1} (\abs{w}^r + \abs{s}^r) \cdot \abs*{g'(w)}  \diff w \diff s \\
        &= 2^{r-1} \Paren{2u \int_{-\infty}^{\infty} \abs{w}^r  \cdot \abs*{g'(w)}  \diff w  + \int_{-u}^{u} \abs{s}^r \cdot \int_{-\infty}^{\infty} \abs*{g'(w)}  \diff w \diff s} \\
        &= 2^{r-1} \Paren{2u \mu_r  + 2\mu_0 \int_{0}^{u} s^r \diff s } \tag{Set $\mu_k = \int_{-\infty}^{\infty} \abs{w}^k \cdot \abs*{g'(w)} \diff w$} \\
        &= u 2^r \left(\mu_r + \mu_0 \frac{u^r}{r+1}\right) \\
        &= u 2^{r+2} \sigma^2 \left(m_{r+1} + m_1 \frac{u^r}{r+1}\right),
    \end{align*}
    where the last equality is because $g'(t) = (-4\sigma^2 t)g(t)$ and thus
    $\mu_k = 4\sigma^2 m_{k+1}$ for any integer $k\geq 0$.
    Then similarly, we can bound
    \begin{align*}
        \eqref{eq:b_1_moments_split}.(1) &\leq  2^{r+2} \sigma^2 \left(m_{r+1} + m_1 \frac{\sigma^{-r}}{r+1}\right) \int_{0}^{1/\sigma} \frac{u}{\sinh(2\pi u / \beta)}\diff u \tag{since $u\leq 1/\sigma$}\\
        &\leq 2^{r+2} \sigma^2 \left(m_{r+1} + m_1 \frac{\sigma^{-r}}{r+1}\right) \Paren{\frac{\beta}{2\pi}}^2 \int_{0}^{\frac{2\pi}{\sigma \beta}} \frac{x}{\sinh(x)}\diff x \tag{Set $x =2\pi u/\beta$}\\
        &\leq 2^{r+2} \sigma^2 \left(m_{r+1} + m_1 \frac{\sigma^{-r}}{r+1}\right) \Paren{\frac{\beta}{2\pi}}^2 \frac{2\pi}{\sigma \beta} \tag{Since $\frac{x}{\sinh(x)} \leq 1$ for $x>0$}\\
        &\leq 2^{r+1} \frac{\sigma \beta}{\pi} \left(\pi^{-1/2}2^{-(r+1)/2} \sigma^{-(r+1)}\cdot \Gamma\Paren{\frac{r+2}{2}}  + 2^{-1/2} \sigma^{-1}  \frac{\sigma^{-r}}{r+1}\right) \\
        &\leq 2^{r+1/2} \frac{\sigma^{-r} \beta}{\pi} \left(\pi^{-1/2}2^{-r/2}\cdot r!  +   \frac{1}{r+1}\right) \tag{$\Gamma(\frac{r+2}{2})\leq r!$ for any integer $r\geq 1$}\\
        &\leq 2^{r+1/2} \frac{\sigma^{-r} \beta}{\pi} \left(4\pi^{-1/2}2^{-r/2}\cdot r! \right)
        \tag{$3\pi^{-1/2}2^{-r/2}\cdot r!  \geq   \frac{1}{r+1}$ for $r\geq 1$}\\
        &= 2^{(r+5)/2} \frac{\sigma^{-r} \beta}{\pi^{3/2}}\cdot r!
    \end{align*}
    Overall,
    \begin{align*}
        \int_{-\infty}^{\infty} \abs{b_1(t)} \cdot \abs{t}^r \diff t
        &= \frac{1}{\sqrt{2\pi}} \frac{1}{\beta} \Paren{ \eqref{eq:b_1_moments_split}.(1) +  \eqref{eq:b_1_moments_split}.(2)} \\
        &= \frac{1}{\sqrt{2\pi}} \frac{1}{\beta} \Paren{ 2^{(r+5)/2} \frac{\sigma^{-r} \beta}{\pi^{3/2}}\cdot r!
        + \frac{\beta}{\pi} r! \Paren{ 2^{r/2-1} \sigma^{-r} \log\Paren{1+\frac{\sigma\beta}{\pi}}  + 2\Paren{\frac{\beta}{\pi}}^{r} }} \\
        &= \frac{r!}{\sqrt{2\pi} \pi} \Paren{  \frac{ 2^{(r+5)/2}}{\pi^{1/2}} \sigma^{-r}
        +  2^{r/2-1} \sigma^{-r} \log\Paren{1+\frac{\sigma\beta}{\pi}}  + 2\Paren{\frac{\beta}{\pi}}^{r} } \\
        &\leq \frac{2^{(r+3)/2} \cdot r!}{ \pi^{3/2} \cdot \sigma^r} \Paren{ 1
        +  \frac{\sigma\beta}{\pi}  + \Paren{\frac{\sigma\beta}{\pi}}^{r} } \,,
    \end{align*}
    where the last inequality uses $\log(1+x)\leq x$ for $x\geq 0$, $\frac{2^{1/2}}{\pi^{1/2}} \leq 1$, and $2^{(r+4)/2} \geq 2^{r/2-1}$ and $2^{(r+4)/2} \geq 2$ for $r\geq 1$.

    Let $x = \frac{\sigma \beta}{\pi} > 0$. Since for any $r\geq 1$, we have $1+x\leq (1+x)^r$ and $x^r \leq (1+x)^r$. Then $1+x+x^r \leq 2(1+x)^r$. Hence
    \begin{align*}
        \int_{-\infty}^{\infty} \abs{b_1(t)} \cdot \abs{t}^r \diff t \leq \frac{2^{(r+5)/2} \cdot r!}{ \pi^{3/2} \cdot \sigma^r} \Paren{ 1
        +  \frac{\sigma\beta}{\pi} }^r  \leq \frac{2^{(r+5)/2} \cdot r!}{ \pi^{3/2}} \Paren{ \frac{1}{\sigma}
        +  \frac{\beta}{\pi} }^r\,.
    \end{align*}
    This completes the proof.
\end{proof}

\section{KMS Detailed Balance}
\label{sec:kms-detailed-balance}

Recall the Kubo-Martin-Schwiger (KMS) detailed-balance condition.
\begin{definition}[KMS detailed balance]
\label{def:kms-db}
Let $\rho$ be a density matrix, and let $\mathcal{L}$ be a linear super-operator.
Let $\mathcal{L}^\dagger$ denote the adjoint of $\mathcal{L}$ with respect to the
Hilbert-Schmidt inner product,
i.e.\
$\tr(X^\dagger \mathcal{L}[Y])=\tr((\mathcal{L}^\dagger[X])^\dagger Y)$ for all $X,Y$.
We say that $\mathcal{L}$ satisfies \emph{KMS detailed balance} if
\begin{equation*}
\mathcal{L}^\dagger[\cdot]
\;=\;
\rho^{-1/2}\,\mathcal{L}\!\big[\rho^{1/2}(\cdot)\rho^{1/2}\big]\,\rho^{-1/2}.
\end{equation*}
\end{definition}

In this section, we will prove \Cref{prop:kms} that the general Lindbladian defined in \Cref{def:general_lindblad} satisfies KMS detailed balance.

\begin{definition}[Bohr frequencies]
    Given a Hamiltonian $H$, write $\spec(H)$ as the set of eigenvalues of $H$.
    We call the set of energy changes $B(H)\coloneqq \{E_i-E_j: E_i, E_j\in \spec(H)\}$ the Bohr frequencies of $H$.
    For any operator $A$ and $\nu \in B(H)$, define the $\nu$-frequency component in $A$ as
    \begin{equation*}
        A_{\nu} \coloneqq \sum_{\substack{E_1, E_2\in \spec(H):\\ E_2-E_1 = \nu}} \Pi_{E_2} A \Pi_{E_1} \,,
    \end{equation*}
    where $\Pi_E$ is the orthogonal projector onto the eigensubspace of $H$ with exact energy $E$. Hence we can rewrite $A$ in the energy basis and regroup in terms of the Bohr frequencies, i.e.
    \begin{equation*}
        A = \sum_{E_1, E_2\in \spec(H)} \Pi_{E_2} A \Pi_{E_1} = \sum_{\nu\in B(H)} A_{\nu} \,.
    \end{equation*}
\end{definition}

\begin{theorem}[A detailed-balanced Lindbladian in the Bohr frequency domain {\cite[Corollary II.2]{ckg23}}]\label{thm:ckg_detailed_balance}
    Let $\calL$ be a Lindbladian given by
    \begin{equation*}
        \calL[\cdot] \coloneqq -\ii \sum_{\nu \in B(H)} [C_{\nu}, \cdot] + \sum_{\nu_1, \nu_2\in B(H)} \alpha_{\nu_1, \nu_2} \Paren{A_{\nu_1} (\cdot) A_{\nu_2}^\dagger - \frac12 \{A_{\nu_2}^\dagger A_{\nu_1}, \cdot \}} \,,
    \end{equation*}
    where each $\alpha_{\nu_1, \nu_2}\in \C$ satisfies $\alpha_{\nu_1, \nu_2}^* = \alpha_{\nu_2, \nu_1}$ and $A_{\nu}$ are the associated Bohr frequency components of an operator $A$.
    Suppose that
    \begin{equation*}
        \alpha_{\nu_1, \nu_2} e^{\frac{\beta}{2}(\nu_1 + \nu_2)} = \alpha_{-\nu_2, -\nu_1} \,,
    \end{equation*}
    and further
    \begin{equation*}
        C^P_\nu = \sum_{\nu_1, \nu_2\in B(H): \nu_1-\nu_2 = \nu} \underbrace{\frac{\ii}{2} \tanh\Paren{\frac{\beta}{4}(\nu_1 - \nu_2)} \alpha_{\nu_1, \nu_2}}_{\coloneqq \widehat{g}(\nu_1, \nu_2)} \cdot (A_{\nu_2}^P)^\dagger A_{\nu_1}^P \,.
    \end{equation*}
    Then $\calL$ satisfies KMS detailed balance and the Gibbs state $e^{-\beta H}/ \tr(e^{-\beta H})$ is a stationary state.
\end{theorem}

\subsection{Proof of \texorpdfstring{\Cref{prop:kms}}{Proposition~\ref{prop:kms}}}
\kms*
\begin{proof}
It suffices to show that for each Pauli operator $P$ at any site $j$, the associated Lindbladian $\calL^P = -\ii [C^P, \cdot] + \calL_{\mathrm{diss}}^P(\cdot)$ as defined in \Cref{def:general_lindblad} satisfies detailed balance.
Thereore, for simplicity, we will drop the site index $j$ and write, for example, $f_j(t)$ as $f(t)$ and $\sigma_j$ as $\sigma$.

Since $e^{\ii Ht} = \sum_{E\in \spec(H)} e^{\ii E t}P_E$, by the definition of $A_\nu$, we have
\begin{equation*}
    e^{\ii Ht} {A} e^{-\ii H t} = \sum_{\nu \in B(H)} e^{\ii Ht} A_{\nu} e^{-\ii H t} = \sum_{\nu \in B(H)} A_{\nu} e^{\ii \nu t} \,.
\end{equation*}
Therefore, the jump operator can be written as
\begin{align*}
    A^P(\omega) &= \frac{1}{\sqrt{2\pi}} \int_{-\infty}^{\infty} e^{\ii Ht} {P} e^{-\ii H t} e^{-\ii \omega t}  f(t) \diff t \\
    &= \frac{1}{\sqrt{2\pi}}
    \sum_{\nu \in B(H)} P_{\nu} \int_{-\infty}^{\infty}  e^{-\ii (\omega - \nu) t}  f(t) \diff t \\
    &= \frac{1}{\sqrt{2\pi}} \sum_{\nu \in B(H)} P_{\nu} \widehat{f}(\omega - \nu)  \,,
\end{align*}
where $\widehat{f}(\omega) = \int_{-\infty}^{\infty} f(t) e^{-\ii \omega t} \diff t$.
We can rewrite the associated dissipative evolution $\calL_{\mathrm{diss}}^P(\rho)$ in the energy basis as
\begin{align*}
    \calL_{\mathrm{diss}}^P(\rho) &= \int_{-\infty}^{\infty} \gamma(\omega) \cdot \Paren{  A^{P}(\omega) \rho A^{P}(\omega)^\dagger - \frac12\braces{A^{P}(\omega)^\dagger A^{P}(\omega), \rho} } \diff \omega \\
    &= \sum_{\nu_1, \nu_2\in B(H)} \underbrace{\int_{-\infty}^{\infty} \gamma(\omega) \cdot\frac{\widehat{f}(\omega - \nu_1) \widehat{f}(\omega - \nu_2)}{2\pi}\diff \omega}_{\coloneqq \alpha_{\nu_1, \nu_2}} \Paren{  P_{\nu_1} \rho P_{\nu_2}^\dagger - \frac12\braces{ P_{\nu_2}^\dagger P_{\nu_1}, \rho} }  \,.
\end{align*}
By the symmetry in $\nu_1, \nu_2$ and the fact that $\alpha_{\nu_1, \nu_2}\in \R$, we have $\alpha_{\nu_1, \nu_2}^* = \alpha_{\nu_2, \nu_1}$.
We now show that $\alpha_{\nu_1, \nu_2} e^{\frac{\beta}{2}(\nu_1 + \nu_2)} = \alpha_{-\nu_2, -\nu_1}$.
Since
\begin{equation*}
    \gamma(\omega) \coloneqq \exp\left( - \frac{(\omega + \Delta)^2}{ 2\eta^2} \right) \,, \quad \text{and} \quad \widehat{f}(\omega) = \frac{(2\pi)^{1/4}}{\sqrt{\sigma}} \exp\Paren{-\frac{\omega^2}{4\sigma^2}} \,,
\end{equation*}
we have
\begin{align*}
    \alpha_{\nu_1, \nu_2} &= \frac{1}{\sqrt{2\pi} \sigma}  \int_{-\infty}^{\infty} \exp\Paren{- \frac{(\omega + \Delta)^2}{ 2\eta^2} - \frac{(\omega - \nu_1)^2 + (\omega - \nu_2)^2}{4\sigma^2}}\diff \omega \,.
\end{align*}
Set $s_+ \coloneqq \nu_1 + \nu_2$ and $s_- \coloneqq \nu_1 - \nu_2$. Observe that
\begin{equation*}
    (\omega - \nu_1)^2 + (\omega - \nu_2)^2 = 2\Paren{\omega - \frac{s_+}{2}}^2 + \frac{s_-^2}{2} \,.
\end{equation*}
Hence
\begin{align*}
    \alpha_{\nu_1, \nu_2} &= \frac{\exp(-\frac{s_-^2}{8\sigma^2})}{\sqrt{2\pi} \sigma}  \int_{-\infty}^{\infty} \exp\Paren{- \frac{(\omega + \Delta)^2}{ 2\eta^2} - \frac{(\omega - \frac{s_+}{2})^2}{2\sigma^2}}\diff \omega \\
    &= \frac{\exp(-\frac{s_-^2}{8\sigma^2})}{\sqrt{2\pi} \sigma} \sqrt{2\pi} \frac{\eta \sigma}{\sqrt{\eta^2 + \sigma^2}} \exp\Paren{ - \frac{(\Delta + \frac{s_+}{2})^2}{2(\eta^2 + \sigma^2)}} \\
    &= \exp\Paren{-\frac{s_-^2}{8\sigma^2}} \frac{\eta}{\sqrt{\eta^2 + \sigma^2}} \exp\Paren{ - \frac{(\Delta + \frac{s_+}{2})^2}{2(\eta^2 + \sigma^2)}} \,,
\end{align*}
where the second equality is due to the formula for the integral of the product of two Gaussian functions in \Cref{eq:int_prod_gaussians}.
Since $\beta(\eta^2 + \sigma^2) = 2\Delta$, we have
\begin{align*}
    \alpha_{\nu_1, \nu_2} e^{\frac{\beta}{2}(\nu_1 + \nu_2)} &= \exp\Paren{-\frac{s_-^2}{8\sigma^2}} \frac{\eta}{\sqrt{\eta^2 + \sigma^2}} \exp\Paren{ - \frac{(\Delta + \frac{s_+}{2})^2}{2(\eta^2 + \sigma^2)}}\exp\Paren{\frac{\Delta s_+}{\eta^2 + \sigma^2}} \\
    &= \exp\Paren{-\frac{s_-^2}{8\sigma^2}} \frac{\eta}{\sqrt{\eta^2 + \sigma^2}} \exp\Paren{ - \frac{(\Delta - \frac{s_+}{2})^2}{2(\eta^2 + \sigma^2)}} = \alpha_{-\nu_1, -\nu_2} \,.
\end{align*}
It remains to check the coherent term.
Let us first derive the Fourier transform of $b_2(t)$:
\begin{align*}
    \widehat{b}_2(\omega) = \int_{-\infty}^{\infty} b_2(t) e^{-\ii \omega t} \diff t &= 2\eta \int_{-\infty}^{\infty} \exp \Paren{-\frac{4\Delta}{\beta}t^2 - \ii (\omega + 2\Delta) t} \diff t \\
    &= \eta \sqrt{\frac{\pi \beta}{\Delta}} \exp \Paren{-\frac{\beta(\omega + 2\Delta)^2}{16\Delta}} \,.
\end{align*}
Recall the Fourier transform of $b_1(t)$ in \Cref{lem:b_1_fourier}. Then
\begin{align*}
    \widehat{g}(\nu_1, \nu_2) &\coloneqq \frac{\ii}{2} \tanh\Paren{\frac{\beta}{4}(\nu_1 - \nu_2)} \alpha_{\nu_1, \nu_2} \\
    &= \underbrace{\frac{\ii}{2} \tanh\Paren{\frac{\beta s_-}{4}} \exp\Paren{-\frac{s_-^2}{8\sigma^2}}}_{= (2\pi)^{1/2}\cdot \widehat{b}_1(s_-)} \cdot \underbrace{\eta
    \sqrt{\frac{\beta}{2\Delta}} \exp\Paren{ - \frac{(\Delta + \frac{s_+}{2})^2}{2(\eta^2 + \sigma^2)}}}_{= (2\pi)^{-1/2} \cdot \widehat{b}_2(s_+)} \\
    &= \widehat{b}_1(\nu_1 - \nu_2) \cdot \widehat{b}_2(\nu_1 + \nu_2) \,.
\end{align*}
Therefore,
\begin{align*}
    \sum_{\nu_1, \nu_2\in B(H)} \widehat{g}(\nu_1, \nu_2) P_{\nu_2}^\dagger P_{\nu_1} &= \sum_{\nu_1, \nu_2\in B(H)} \widehat{b}_1(\nu_1 - \nu_2) \cdot \widehat{b}_2(\nu_1 + \nu_2)\cdot P_{\nu_2}^\dagger P_{\nu_1} \\
    &= \sum_{\nu_1, \nu_2\in B(H)} \int_{-\infty}^{\infty} b_1(t) e^{-\ii (\nu_1 - \nu_2) t} \diff t \cdot \int_{-\infty}^{\infty} b_2(t') e^{-\ii (\nu_1 + \nu_2) t'} \diff t' \cdot P_{\nu_2}^\dagger P_{\nu_1} \\
    &= \iint_{-\infty}^{\infty} b_1(t) b_2(t') \sum_{\nu_1, \nu_2\in B(H)}   e^{-\ii \nu_2  (t'-t)} P_{\nu_2}^\dagger \cdot e^{-\ii \nu_1(t+t')} P_{\nu_1} \diff t \diff t'  \\
    &= \iint_{-\infty}^{\infty} b_1(t) b_2(t') \cdot e^{-\ii H(t-t')} P^\dagger e^{\ii H(t-t')} \cdot e^{-\ii H(t+t')} P e^{\ii H(t+t')}\diff t \diff t'  \\
    &=\int_{-\infty}^\infty b_1(t) e^{-\ii H t} \Paren{\int_{-\infty}^\infty b_2(t') e^{\ii Ht'} P^\dagger e^{-\ii Ht'} e^{-\ii H t'} P e^{\ii Ht'} \diff t' } e^{\ii H t} \diff t \\
    &= C^P \,,
\end{align*}
which means that the condition for the coherent term in \Cref{thm:ckg_detailed_balance} is also satisfied. We then conclude that each $\calL^P$ satisfies KMS detailed balance and the Gibbs state is a stationary state.
\end{proof}

%% file: refs.bib
@misc{BT26,
      title={Rapid mixing for high-temperature Gibbs states with arbitrary external fields}, 
      author={Ainesh Bakshi and Xinyu Tan},
      year={2026},
      eprint={2604.08408},
      archivePrefix={arXiv},
      primaryClass={quant-ph},
      url={https://arxiv.org/abs/2604.08408}, 
}

@article{pilatowsky2026one,
  title={One-dimensional quantum Gibbs states in constant circuit depth},
  author={Pilatowsky-Cameo, Sa{\'u}l and Styliaris, Georgios and Bakshi, Ainesh and Malz, Daniel},
  journal={arXiv preprint arXiv:2609.35973},
  year={2026}
}

@Article{dmtl21,
  author       = {De Palma, Giacomo and Marvian, Milad and Trevisan, Dario and Lloyd, Seth},
  date         = {2021-10},
  journaltitle = {IEEE Transactions on Information Theory},
  title        = {The quantum {W}asserstein distance of order 1},
  doi          = {10.1109/tit.2021.3076442},
  eprint       = {2009.04469},
  eprintclass  = {quant-ph},
  eprinttype   = {arXiv},
  issn         = {1557-9654},
  number       = {10},
  pages        = {6627--6643},
  volume       = {67},
}

@InProceedings{blmt2025dobrushin,
  author      = {Bakshi, Ainesh and Liu, Allen and Moitra, Ankur and Tang, Ewin},
  booktitle   = {Proceedings of the 58\textsuperscript{th} Annual ACM Symposium on Theory of Computing},
  date        = {2026-06},
  title       = {A Dobrushin condition for quantum Markov chains: Rapid mixing and conditional mutual information at high temperature},
  eprint      = {2510.08542},
  eprintclass = {quant-ph},
  eprinttype  = {arXiv},
  publisher   = {ACM},
  series      = {STOC ’26},
}

@Article{NRSS09,
    author={Nachtergaele, Bruno
    and Raz, Hillel
    and Schlein, Benjamin
    and Sims, Robert},
    title={Lieb-Robinson Bounds for Harmonic and Anharmonic Lattice Systems},
    journal={Communications in Mathematical Physics},
    year={2009},
    month={3},
    day={01},
    volume={286},
    number={3},
    pages={1073-1098},
    issn={1432-0916},
    doi={10.1007/s00220-008-0630-2},
    url={https://doi.org/10.1007/s00220-008-0630-2}
}

@Article{hhkl21,
  author       = {Haah, Jeongwan and Hastings, Matthew B. and Kothari, Robin and Low, Guang Hao},
  date         = {2021-01},
  journaltitle = {SIAM Journal on Computing},
  title        = {Quantum algorithm for simulating real time evolution of lattice {H}amiltonians},
  doi          = {10.1137/18m1231511},
  eprint       = {1801.03922},
  eprintclass  = {quant-ph},
  eprinttype   = {arXiv},
  issn         = {1095-7111},
  number       = {6},
  volume       = {52},
}

@Unpublished{ckg23,
  author      = {Chen, Chi-Fang and Kastoryano, Michael J. and Gilyén, András},
  date        = {2023-11-15},
  title       = {An efficient and exact noncommutative quantum {G}ibbs sampler},
  eprint      = {2311.09207},
  eprintclass = {quant-ph},
  eprinttype  = {arXiv},
  year        = {2023},
}

@Article{rfa24a,
  author       = {Rouzé, Cambyse and França, Daniel Stilck and Alhambra, Álvaro M.},
  date         = {2026-02},
  title        = {Optimal quantum algorithm for {G}ibbs state preparation},
  journaltitle = {Physical Review Letters},
  volume       = {136},
  eid          = {060601},
  doi          = {10.1103/lhht-svmn},
  eprint       = {2411.04885},
  eprintclass  = {quant-ph},
  eprinttype   = {arXiv},
}

@InProceedings{rfa24,
  author      = {Rouzé, Cambyse and França, Daniel Stilck and Alhambra, Álvaro M.},
  booktitle   = {Proceedings of the 57\textsuperscript{th} Annual ACM Symposium on Theory of Computing},
  date        = {2025-06},
  title       = {Efficient thermalization and universal quantum computing with quantum {G}ibbs samplers},
  doi         = {10.1145/3717823.3718268},
  eprint      = {2403.12691},
  eprintclass = {quant-ph},
  eprinttype  = {arXiv},
  pages       = {1488--1495},
  publisher   = {ACM},
  series      = {STOC ’25},
}

@Article{Has19,
  author       = {Hastings, Matthew B.},
  title        = {The stability of free {Fermi} {Hamiltonians}},
  journaltitle = {Journal of Mathematical Physics},
  year         = {2019},
  volume       = {60},
  number       = {4},
  eid          = {042201},
  doi          = {10.1063/1.5053869},
  eprint       = {1706.02270},
  eprintclass  = {quant-ph},
  eprinttype   = {arXiv},
}

@article{TZ25,
  title = {Fast Mixing of Weakly Interacting Fermionic Systems at Any Temperature},
  author = {Tong, Yu and Zhan, Yongtao},
  journal = {PRX Quantum},
  volume = {6},
  number = {3},
  pages = {030301},
  year = {2025},
  month = {7},
  doi = {10.1103/h1dx-ps5p},
  url = {https://link.aps.org/doi/10.1103/h1dx-ps5p}
}

@misc{LW18,
      title={Hamiltonian Simulation in the Interaction Picture},
      author={Guang Hao Low and Nathan Wiebe},
      year={2019},
      eprint={1805.00675},
      archivePrefix={arXiv},
      primaryClass={quant-ph},
      url={https://arxiv.org/abs/1805.00675},
}

@Book{watrous18,
  author    = {Watrous, John},
  date      = {2018-04},
  title     = {The theory of quantum information},
  doi       = {10.1017/9781316848142},
  isbn      = {9781107180567},
  publisher = {Cambridge University Press},
}

@article{vE24,
   title={Almost All Quantum Channels Are Diagonalizable},
   volume={31},
   ISSN={1793-7191},
   url={http://dx.doi.org/10.1142/S1230161224500124},
   DOI={10.1142/s1230161224500124},
   number={03},
   journal={Open Systems \& Information Dynamics},
   author={vom Ende, Frederik},
   year={2024},
   month={8}
}

@misc{BCP26,
      title={Entanglement in quantum spin chains is strictly finite at any temperature},
      author={Ainesh Bakshi and Soonwon Choi and Saúl Pilatowsky-Cameo},
      year={2026},
      eprint={2602.13386},
      archivePrefix={arXiv},
      primaryClass={quant-ph},
      url={https://arxiv.org/abs/2602.13386},
}

@article{rajakumar2026gibbs,
  title={Gibbs sampling gives quantum advantage at constant temperatures with O (1)-local Hamiltonians},
  author={Rajakumar, Joel and Watson, James D},
  journal={Quantum},
  volume={10},
  pages={1981},
  year={2026},
}

@phdthesis{zhao2024learning,   title={Learning, Optimizing, and Simulating Fermions with Quantum Computers},   author={Zhao, Andrew},   year={2024},   school={The University of New Mexico} }

@Article{AM03,
  author       = {Araki, Huzihiro and Moriya, Hajime},
  title        = {Equilibrium statistical mechanics of fermion lattice systems},
  journaltitle = {Reviews in Mathematical Physics},
  year         = {2003},
  volume       = {15},
  number       = {2},
  pages        = {93--198},
  doi          = {10.1142/S0129055X03001606},
  eprint       = {math-ph/0211016},
  eprinttype   = {arXiv},
}

@Article{Szalay21,
  author       = {Szalay, Szilárd and Zimborás, Zoltán and Máté, Mihály and Barcza, Gergely and Schilling, Christian and Legeza, Örs},
  title        = {Fermionic systems for quantum information people},
  journaltitle = {Journal of Physics A: Mathematical and Theoretical},
  year         = {2021},
  volume       = {54},
  eid          = {393001},
  doi          = {10.1088/1751-8121/ac0646},
  eprint       = {2006.03087},
  eprinttype   = {arXiv},
  eprintclass  = {quant-ph},
}

@Article{CM14,
  author       = {Carlen, Eric A. and Maas, Jan},
  title        = {An Analog of the 2-{Wasserstein} Metric in Non-commutative Probability under which the Fermionic {Fokker--Planck} Equation is Gradient Flow for the Entropy},
  journaltitle = {Communications in Mathematical Physics},
  year         = {2014},
  volume       = {331},
  pages        = {887--926},
  eprint       = {1203.5377},
  eprinttype   = {arXiv},
  eprintclass  = {math.FA},
}

@Unpublished{DvZZ25,
  author       = {Duvenhage, Rocco and van Zyl, Dylan and Zurlo, Paola},
  title        = {Fermionic optimal transport},
  year         = {2025},
  eprint       = {2510.02888},
  eprinttype   = {arXiv},
  eprintclass  = {math-ph},
}

@Article{BDSP25,
  author       = {Bertini, Lorenzo and De Sole, Alberto and Posta, Gustavo and Presilla, Carlo},
  title        = {Perturbative criteria for the ergodicity of interacting dissipative quantum lattice systems},
  journaltitle = {Communications in Mathematical Physics},
  year         = {2025},
  volume       = {406},
  eid          = {140},
  doi          = {10.1007/s00220-025-05305-6},
  eprint       = {2402.15488},
  eprinttype   = {arXiv},
  eprintclass  = {math-ph},
}

@Misc{Lin24,
  author  = {Lin, Lin},
  title   = {Notes on {L}indblad equations and {D}avies generator},
  date    = {2024-01-26},
  year    = {2024},
  url     = {https://github.com/lin-lin/Quantum290/blob/master/asset/Lindblad_Notes.pdf},
  urldate = {2026-10-04},
}

@Unpublished{lidar19,
  author      = {Lidar, Daniel A.},
  title       = {Lecture notes on the theory of open quantum systems},
  date        = {2019-02-03},
  year        = {2019},
  eprint      = {1902.00967},
  eprintclass = {quant-ph},
  eprinttype  = {arXiv},
}

@Unpublished{BC25,
  author      = {Bergamaschi, Thiago and Chen, Chi-Fang},
  title       = {Fast mixing of quantum spin chains at all temperatures},
  date        = {2025-10-09},
  eprint      = {2510.08533},
  eprintclass = {quant-ph},
  eprinttype  = {arXiv},
  note        = {Version 2, revised January 2026},
}

@Unpublished{Kim26,
  author      = {Kim, Leeseok},
  title       = {Rapid mixing of quantum spin chains at any finite temperature},
  date        = {2026-10-01},
  eprint      = {2610.01190},
  eprintclass = {quant-ph},
  eprinttype  = {arXiv},
}

@Unpublished{RF26,
  author      = {Rouzé, Cambyse and França, Daniel Stilck},
  title       = {Rapid mixing of {Gibbs} samplers via quantum {Dobrushin--Shlosman} conditions},
  date        = {2026-10-02},
  eprint      = {2610.03578},
  eprintclass = {quant-ph},
  eprinttype  = {arXiv},
}

@Unpublished{Ber26,
  author      = {Bergamaschi, Thiago},
  title       = {Fast mixing of all-to-all quantum systems at high temperatures},
  date        = {2026-06-24},
  eprint      = {2606.26090},
  eprintclass = {quant-ph},
  eprinttype  = {arXiv},
}

@inproceedings{blmt24b,
  title={High-temperature Gibbs states are unentangled and efficiently preparable},
  author={Bakshi, Ainesh and Liu, Allen and Moitra, Ankur and Tang, Ewin},
  booktitle={2024 IEEE 65th Annual Symposium on Foundations of Computer Science (FOCS)},
  pages={1027--1036},
  year={2024},
  organization={IEEE}
}

@Article{kh11,
  author       = {Kuwahara, Tomotaka and Hatano, Naomichi},
  date         = {2011-06},
  journaltitle = {Physical Review A},
  title        = {Maximization of thermal entanglement of arbitrarily interacting two qubits},
  doi          = {10.1103/physreva.83.062311},
  eprint       = {1104.0485},
  eprintclass  = {quant-ph},
  eprinttype   = {arXiv},
  issn         = {1094-1622},
  number       = {6},
  pages        = {062311},
  volume       = {83},
}

@inproceedings{bergamaschi2024quantum,
  title={Quantum computational advantage with constant-temperature Gibbs sampling},
  author={Bergamaschi, Thiago and Chen, Chi-Fang and Liu, Yunchao},
  booktitle={2024 IEEE 65th Annual Symposium on Foundations of Computer Science (FOCS)},
  pages={1063--1085},
  year={2024},
  organization={IEEE}
}

@Article{CL93,
  author       = {Carlen, Eric A. and Lieb, Elliott H.},
  title        = {Optimal hypercontractivity for {Fermi} fields and related non-commutative integration inequalities},
  journaltitle = {Communications in Mathematical Physics},
  year         = {1993},
  volume       = {155},
  number       = {1},
  pages        = {27--46},
  doi          = {10.1007/BF02100048},
  eprint       = {hep-th/9209053},
  eprinttype   = {arXiv},
}

@Article{TPK14,
  author       = {Temme, Kristan and Pastawski, Fernando and Kastoryano, Michael J.},
  title        = {Hypercontractivity of quasi-free quantum semigroups},
  journaltitle = {Journal of Physics A: Mathematical and Theoretical},
  year         = {2014},
  volume       = {47},
  number       = {40},
  eid          = {405303},
  doi          = {10.1088/1751-8113/47/40/405303},
  eprint       = {1403.5224},
  eprinttype   = {arXiv},
  eprintclass  = {quant-ph},
}

@Article{CM17,
  author       = {Carlen, Eric A. and Maas, Jan},
  title        = {Gradient flow and entropy inequalities for quantum {Markov} semigroups with detailed balance},
  journaltitle = {Journal of Functional Analysis},
  year         = {2017},
  volume       = {273},
  number       = {5},
  pages        = {1810--1869},
  doi          = {10.1016/j.jfa.2017.05.003},
  eprint       = {1609.01254},
  eprinttype   = {arXiv},
  eprintclass  = {math.OA},
}

@Article{SMBB25,
  author       = {Šmíd, Štěpán and Meister, Richard and Berta, Mario and Bondesan, Roberto},
  title        = {Polynomial-time quantum {Gibbs} sampling for the weak and strong coupling regime of the {Fermi--Hubbard} model at any temperature},
  journaltitle = {Nature Communications},
  date         = {2025-11-28},
  volume       = {16},
  eid          = {10736},
  doi          = {10.1038/s41467-025-65765-1},
  eprint       = {2501.01412},
  eprinttype   = {arXiv},
  eprintclass  = {quant-ph},
}

@Unpublished{SMBB26rapid,
  author      = {Šmíd, Štěpán and Meister, Richard and Berta, Mario and Bondesan, Roberto},
  title       = {Rapid mixing of quantum {Gibbs} samplers for weakly-interacting quantum systems},
  date        = {2025-10-06},
  eprint      = {2510.04954},
  eprinttype  = {arXiv},
  eprintclass = {quant-ph},
  note        = {Version 2, revised April 2026},
}

@Unpublished{DZPL25,
  author      = {Ding, Zhiyan and Zhan, Yongtao and Preskill, John and Lin, Lin},
  title       = {Simple and efficient end-to-end quantum thermal and ground state preparation},
  year        = {2025},
  eprint      = {2508.05703},
  eprinttype  = {arXiv},
  eprintclass = {quant-ph},
  note        = {Version 3, revised August 2026},
}

@Unpublished{Slezak26,
  author      = {Slezak, Samuel and Scandi, Matteo and Alhambra, Álvaro M. and França, Daniel Stilck and Rouzé, Cambyse},
  title       = {Polynomial-time thermalization and {Gibbs} sampling from system-bath couplings},
  year        = {2026},
  eprint      = {2601.16154},
  eprinttype  = {arXiv},
  eprintclass = {quant-ph},
}

@Unpublished{WD25,
  author      = {Wang, Ke and Ding, Zhiyan},
  title       = {Beyond {Lindblad} approximation: Rigorous guarantees for thermal and ground state preservation under system-bath interactions},
  year        = {2025},
  eprint      = {2512.03457},
  eprinttype  = {arXiv},
  eprintclass = {quant-ph},
  note        = {Version 4, revised September 2026},
}

@Article{OGK14,
  author       = {Oszmaniec, Michał and Gutt, Jan and Kuś, Marek},
  title        = {Classical simulation of fermionic linear optics augmented with noisy ancillas},
  journaltitle = {Physical Review A},
  year         = {2014},
  volume       = {90},
  number       = {2},
  eid          = {020302},
  doi          = {10.1103/PhysRevA.90.020302},
  eprint       = {1406.1577},
  eprinttype   = {arXiv},
  eprintclass  = {quant-ph},
}

@Unpublished{HZL26,
  author      = {Han, Zhengyi and Zhao, Yuanchen and Liu, Zi-Wen},
  title       = {Weakly interacting fermionic {Gibbs} states are {Gaussian} mixtures and classically simulable},
  year        = {2026},
  eprint      = {2609.35159},
  eprinttype  = {arXiv},
  eprintclass = {quant-ph},
}

@Article{RCTJ26,
  author       = {Ramkumar, Akshar and Cai, Yiyi and Tong, Yu and Jiang, Jiaqing},
  title        = {High-Temperature Fermionic {G}ibbs States Are Mixtures of {G}aussian States},
  journaltitle = {Physical Review Letters},
  date         = {2026-07-07},
  volume       = {137},
  eid          = {020601},
  doi          = {10.1103/qprk-k7wn},
  eprint       = {2505.09730},
  eprintclass  = {quant-ph},
  eprinttype   = {arXiv},
}

@Article{jordens2008mott,
  author       = {Jördens, Robert and Strohmaier, Niels and Günter, Kenneth and Moritz, Henning and Esslinger, Tilman},
  title        = {A {M}ott insulator of fermionic atoms in an optical lattice},
  journaltitle = {Nature},
  year         = {2008},
  volume       = {455},
  pages        = {204--207},
  doi          = {10.1038/nature07244},
  eprint       = {0804.4009},
  eprintclass  = {cond-mat.other},
  eprinttype   = {arXiv},
}

@Article{mazurenko2017antiferromagnet,
  author       = {Mazurenko, Anton and Chiu, Christie S. and Ji, Geoffrey and Parsons, Maxwell F. and Kanász-Nagy, Márton and Schmidt, Richard and Grusdt, Fabian and Demler, Eugene and Greif, Daniel and Greiner, Markus},
  title        = {A cold-atom {Fermi--Hubbard} antiferromagnet},
  journaltitle = {Nature},
  year         = {2017},
  volume       = {545},
  pages        = {462--466},
  doi          = {10.1038/nature22362},
  eprint       = {1612.08436},
  eprintclass  = {cond-mat.quant-gas},
  eprinttype   = {arXiv},
}

@Article{zheng2017stripe,
  author       = {Zheng, Bo-Xiao and Chung, Chia-Min and Corboz, Philippe and Ehlers, Georg and Qin, Ming-Pu and Noack, Reinhard M. and Shi, Hao and White, Steven R. and Zhang, Shiwei and Chan, Garnet Kin-Lic},
  title        = {Stripe order in the underdoped region of the two-dimensional {Hubbard} model},
  journaltitle = {Science},
  year         = {2017},
  volume       = {358},
  number       = {6367},
  pages        = {1155--1160},
  doi          = {10.1126/science.aam7127},
  eprint       = {1701.00054},
  eprintclass  = {cond-mat.str-el},
  eprinttype   = {arXiv},
}

@Article{maier2006pairing,
  author       = {Maier, Thomas A. and Jarrell, M. and Scalapino, D. J.},
  title        = {Pairing interaction in the two-dimensional {Hubbard} model studied with a dynamic cluster quantum {Monte Carlo} approximation},
  journaltitle = {Physical Review B},
  year         = {2006},
  volume       = {74},
  eid          = {094513},
  doi          = {10.1103/PhysRevB.74.094513},
  eprint       = {cond-mat/0606003},
  eprintclass  = {cond-mat.supr-con},
  eprinttype   = {arXiv},
}

@Article{dong2022superconductivity,
  author       = {Dong, Xinyang and Del Re, Lorenzo and Toschi, Alessandro and Gull, Emanuel},
  title        = {Mechanism of superconductivity in the {Hubbard} model at intermediate interaction strength},
  journaltitle = {Proceedings of the National Academy of Sciences},
  year         = {2022},
  volume       = {119},
  number       = {33},
  eid          = {e2205048119},
  doi          = {10.1073/pnas.2205048119},
  eprint       = {2205.06286},
  eprintclass  = {cond-mat.supr-con},
  eprinttype   = {arXiv},
}

@InCollection{lieb1995hubbard,
  author    = {Lieb, Elliott H.},
  title     = {The {Hubbard} Model: Some Rigorous Results and Open Problems},
  booktitle = {The {Hubbard} Model: Its Physics and Mathematical Physics},
  series    = {NATO ASI Series},
  volume    = {343},
  publisher = {Springer},
  year      = {1995},
  pages     = {1--19},
  doi       = {10.1007/978-1-4899-1042-4_1},
  eprint    = {cond-mat/9311033},
  eprinttype = {arXiv},
}

@Article{qin2022hubbard,
  author       = {Qin, Mingpu and Sch\"afer, Thomas and Andergassen, Sabine and Corboz, Philippe and Gull, Emanuel},
  title        = {The {Hubbard} Model: A Computational Perspective},
  journaltitle = {Annual Review of Condensed Matter Physics},
  year         = {2022},
  volume       = {13},
  pages        = {275--302},
  doi          = {10.1146/annurev-conmatphys-090921-033948},
  eprint       = {2104.00064},
  eprintclass  = {cond-mat.str-el},
  eprinttype   = {arXiv},
}

@Article{Uel99,
  author       = {Ueltschi, Daniel},
  title        = {Analyticity in {H}ubbard models},
  journaltitle = {Journal of Statistical Physics},
  date         = {1999-05},
  volume       = {95},
  pages        = {693--717},
  doi          = {10.1023/A:1004599410952},
  eprint       = {cond-mat/9810320},
  eprintclass  = {cond-mat.stat-mech},
  eprinttype   = {arXiv},
}

@Article{ding2024kms,
  author       = {Ding, Zhiyan and Li, Bowen and Lin, Lin},
  title        = {Efficient quantum {Gibbs} samplers with {Kubo--Martin--Schwinger} detailed balance condition},
  journaltitle = {Communications in Mathematical Physics},
  date         = {2025-02-27},
  volume       = {406},
  eid          = {67},
  doi          = {10.1007/s00220-025-05235-3},
  eprint       = {2404.05998},
  eprintclass  = {quant-ph},
  eprinttype   = {arXiv},
}

@Article{temme2013lower,
  author       = {Temme, Kristan},
  title        = {Lower bounds to the spectral gap of {Davies} generators},
  journaltitle = {Journal of Mathematical Physics},
  year         = {2013},
  volume       = {54},
  number       = {12},
  eid          = {122110},
  doi          = {10.1063/1.4850896},
  eprint       = {1305.5591},
  eprintclass  = {quant-ph},
  eprinttype   = {arXiv},
}

@Article{bardet2024entropy,
  author       = {Bardet, Ivan and Capel, {\'A}ngela and Gao, Li and Lucia, Angelo and P{\'e}rez-Garc{\'i}a, David and Rouz{\'e}, Cambyse},
  title        = {Entropy decay for {Davies} semigroups of a one dimensional quantum lattice},
  journaltitle = {Communications in Mathematical Physics},
  date         = {2024-02-10},
  volume       = {405},
  eid          = {42},
  doi          = {10.1007/s00220-023-04869-5},
  eprint       = {2112.00601},
  eprintclass  = {quant-ph},
  eprinttype   = {arXiv},
}

@Unpublished{BGSS+25,
  author       = {Basso, Joao and Ganguly, Shirshendu and Sinclair, Alistair and Srivastava, Nikhil and Stier, Zachary and Vuong, Thuy-Duong},
  title        = {On quantum to classical comparison for {Davies} generators},
  date         = {2025-10-08},
  eprint       = {2510.07267},
  eprintclass  = {quant-ph},
  eprinttype   = {arXiv},
}

@Article{crosson2025classical,
  author       = {Crosson, Elizabeth and Slezak, Samuel},
  title        = {Classical simulation of high temperature quantum {Ising} models},
  journaltitle = {Quantum},
  date         = {2025-07-09},
  volume       = {9},
  eid          = {1788},
  doi          = {10.22331/q-2025-07-09-1788},
  eprint       = {2002.02232},
  eprintclass  = {quant-ph},
  eprinttype   = {arXiv},
}

@InProceedings{HMS20,
  author       = {Harrow, Aram W. and Mehraban, Saeed and Soleimanifar, Mehdi},
  title        = {Classical algorithms, correlation decay, and complex zeros of partition functions of quantum many-body systems},
  booktitle    = {Proceedings of the 52nd Annual {ACM SIGACT} Symposium on Theory of Computing},
  year         = {2020},
  pages        = {378--386},
  publisher    = {ACM},
  doi          = {10.1145/3357713.3384322},
  eprint       = {1910.09071},
  eprintclass  = {quant-ph},
  eprinttype   = {arXiv},
}

@InProceedings{YYZ22,
  author       = {Yao, Penghui and Yin, Yitong and Zhang, Xinyuan},
  title        = {Polynomial-time approximation of zero-free partition functions},
  booktitle    = {49th International Colloquium on Automata, Languages, and Programming ({ICALP} 2022)},
  series       = {Leibniz International Proceedings in Informatics (LIPIcs)},
  volume       = {229},
  year         = {2022},
  pages        = {108:1--108:20},
  publisher    = {Schloss Dagstuhl -- Leibniz-Zentrum f{\"u}r Informatik},
  doi          = {10.4230/LIPIcs.ICALP.2022.108},
  eprint       = {2201.12772},
  eprintclass  = {cs.DS},
  eprinttype   = {arXiv},
}

@Unpublished{ZK26,
  author       = {Zlokapa, Alexander and Kiani, Bobak T.},
  title        = {{SYK} thermal expectations are classically easy at any temperature},
  year         = {2026},
  eprint       = {2602.22619},
  eprintclass  = {quant-ph},
  eprinttype   = {arXiv},
}
